\documentclass[twoside,11pt]{article}

\usepackage{blindtext}

\usepackage{mystyle}
\usepackage{microtype}
\usepackage{graphicx}
\usepackage{subfigure}
\usepackage{booktabs} 
\usepackage{fullpage}
\usepackage{natbib}
\usepackage{hyperref}
\usepackage{tikz}
\usepackage{acronym}
\usetikzlibrary{positioning, shapes.geometric}

\acrodef{marl}[MARL]{Multi-Agent Reinforcement Learning}
\acrodef{mfg}[MFG]{Mean-Field Game}
\acrodef{mfc}[MFC]{Mean-Field Control}
\acrodef{gmfg}[GMFG]{Graphon Mean-Field Game}
\acrodef{mdp}[MDP]{Markov Decision Process}
\acrodef{sbm}[SBM]{Stochastic Block Model}
\acrodef{ne}[NE]{Nash Equilibrium}
\acrodef{rkhs}[RKHS]{Reproducing Kernel Hilbert Space}
\acrodef{gmfgppo}[GMFG-PPO]{Proximal Policy Optimization for GMFG}
\acrodef{ppo}[PPO]{Proximal Policy Optimization}
\acrodef{sbm}[SBM]{Stochastic Block Model}
\acrodef{sis}[SIS]{Susceptible-Infectious-Susceptible}

\usetikzlibrary{decorations.pathreplacing}

\title{\Huge Learning Regularized Graphon Mean-Field Games with Unknown Graphons}
\author{Fengzhuo Zhang$^{1}$ \quad Vincent Y.~F.~Tan$^1$ \quad Zhaoran Wang$^2$\quad Zhuoran Yang$^3$\\
$^1$National University of Singapore \quad $^2$ Northwestern University   \quad $^3$Yale University\\
\texttt{fzzhang@u.nus.edu}, \texttt{vtan@nus.edu.sg}, \\
\texttt{zhaoranwang@gmail.com}, \texttt{zhuoranyang.work@gmail.com}
}
\date{}
\begin{document}

\maketitle
\begin{abstract}
We design and analyze reinforcement learning algorithms for Graphon Mean-Field Games (GMFGs). In contrast to previous works that require the precise values of the graphons, we aim to learn the Nash Equilibrium (NE) of the regularized GMFGs when the graphons are unknown. Our contributions are threefold. First, we propose the Proximal Policy Optimization for GMFG (GMFG-PPO) algorithm and show that it converges at a rate of $\tilO(T^{-1/3})$ after $T$ iterations with an estimation oracle, improving on a previous work by Xie~{\em et al.}\ (ICML, 2021). Second, using kernel embedding of distributions, we design efficient algorithms to estimate the transition kernels, reward functions, and graphons from sampled agents. Convergence rates are then derived when the positions of the agents are either known or unknown. Results for the combination of the optimization algorithm GMFG-PPO and the estimation algorithm are then provided. These algorithms are the first specifically designed for learning graphons from sampled agents. Finally, the efficacy of the proposed algorithms are corroborated through simulations. These simulations demonstrate that learning the unknown graphons  reduces the {\em exploitability} effectively.
\end{abstract}
\section{Introduction}
\ac{marl} aims to solve sequential decision-making problems in  multi-agent systems~\citep{zhang2021multi,gronauer2022multi,oroojlooy2022review}. Although \ac{marl} has enjoyed tremendous successes across a wide range of real-world applications~\citep{tang2021sensory,wang2022mathrm,wang2022equivariant,xu2021relation}, it suffers from the ``curse of many agents'' where the sizes of the state  and action spaces increase exponentially with the number of  agents~\citep{menda2018deep,wang2020breaking}. A potential remedy is to use the {\em mean-field approximation}~\citep{yang2018mean,carmona2019model}. It assumes that the agents are \emph{homogeneous}, and each agent is influenced only by the {\em common} state distribution of agents. This assumption mitigates the exponential growth of the state and action spaces~\citep{wang2020breaking,guo2022mf}. However, the homogeneity assumption heavily restricts the applicability of the \ac{mfg}. As a result, the \ac{gmfg} is proposed as a means to relax the homogeneity assumption. It captures the heterogeneity of agents through  {\em graphons} and allows the number of agents to be potentially uncountably infinite~\citep{parise2019graphon,carmona2022stochastic}. \ac{gmfg}s have achieved great successes in a  wide range of applications~\citep{gao2019graphon,aurell2022finite}.

However, learning algorithms for \ac{gmfg} require significantly more efforts to design and analyze. \cite{cui2021learning} proposed to learn the \ac{ne} of \ac{gmfg}s by modifying existing \ac{mfg} learning algorithms. However, these model-free algorithms suffer from the fact that  the distribution flow estimation in \ac{gmfg} requires a  large number of samples due to the   heterogeneity of the agents. In addition, these algorithms potentially necessitate the use of  a very large class of  value functions. In particular, this function class   should include the nominal value function in \ac{gmfg} with any graphons to satisfy the realizability assumption~\citep{jin2021bellman,zhan2022offline}. Moreover, existing works only prove the consistency of learning algorithms with rather stringent assumptions~\citep{cui2021learning,fabian2022learning}. These assumptions include the contractivity of the estimated operators and the access to the nominal value functions. The {\em convergence rates} of algorithms in \ac{gmfg}s with milder assumptions are currently lacking in the literature.


In this paper, we focus on learning the \ac{ne} from the collected data of sampled agents. Concretely, we have access to a simulator of the \ac{gmfg} which generates the states and rewards of agents with the agent policies as its inputs. However, only the states and rewards of only a {\em finite} set of agents are revealed to the learner. Compared with the settings in \cite{cui2021learning} and \cite{fabian2022learning}, our setting is more relevant in real-world applications where the number of agents is  always finite. We aim to learn the \ac{ne} of the \ac{gmfg} from the states and rewards  of these sampled agents.  

Learning the \ac{ne}s in our problem  involves overcoming difficulties from   the {\em statistical} and {\em optimization} perspectives. From the statistical side, we suffer from the lack of information about the inputs of the functions to estimate. The transition kernels and the reward functions of each agent take as inputs the collective behavior of all the other agents and the graphon. In contrast, we do not know the graphons and only have information provided by  a {\em  finite} subset of agents. From the optimization perspective, each agent is faced with a non-stationary environment formed by other agents. Thus, we should design policy optimization procedures that ensure that  the policy of each agent converges to the optimal one in a time-varying environment, while also ensuring that the non-stationary environment converges to a  \ac{ne}.

\paragraph{Main Contributions} 

Addressing these difficulties, we summarize our main contributions and results in Table~\ref{table:sum} and in more details as follows:
\begin{itemize}
    \item We propose and analyze the \ac{gmfgppo} algorithm to learn the \ac{ne}. Given an estimate oracle, our algorithm implements a \ac{ppo}-like algorithm to update the agents' policies~\citep{schulman2017proximal}. The environment is simultaneously updated with a carefully designed learning rate. These strategies overcome the {\em optimization-related} hurdles. \ac{gmfgppo} achieves a convergence rate $\tilO(T^{-1/3})$, where $T$ is the number of iterations. This convergence rate is faster than that of the algorithm in \citet{xie2021learning} and is proved under fewer assumptions. This improvement is attributed to our carefully designed policy and environment update rates. In addition, the analysis of our optimization leads to a faster convergence of the mirror descent algorithm on a fixed MDP. As a byproduct, we generalize the result in \citet{lan2022policy} to inhomogeneous MDPs with a finite horizon.
    \item We design and analyze the model learning algorithm of \ac{gmfg} under three different agent sampling schemes, as shown in Table~\ref{table:sum}. The algorithm first incorporates the graphon with the empirical measure to estimate the mean-embedding of each agent's influence. Then we take this estimate as the input and then perform a  regression task; this resolves the {\em statistical} difficulties mentioned above. In the case where sampled agents have known and fixed positions, Theorem~\ref{thm:fixest} shows that the convergence rate for the model estimate is $O((NL)^{-1}+N^{-1/2})$, where $N$ is the number of sampled agents, and $L$ is the number of sampled from each agent.   We also consider two additional scenarios---the case in which the agents are {\em  randomly sampled} from the unit interval but their positions are known, and learning from sample agents with {\em unknown} grid positions. Pertaining to the final scenario, Theorem~\ref{thm:ufixest} indicates that the lack of  information of the position of the agents results in the sample complexity being degraded by an additional factor of $O(N\log N)$.
    \item Our model estimation learning algorithm is the first one proposed for   \ac{gmfg}s. It recovers the underlying graphons from the states sampled from a  finite number of  agents. This model-learning problem is a considerable generalization of the distribution regression problem~\citep{szabo2016learning}. Detailed discussions are provided in Section~\ref{sec:knownest}.  Also, our graphon learning setting can be regarded as a novel addition to the existing graphon estimation literature, as discussed in Section~\ref{sec:relatedwork}. 
\end{itemize}

\begin{table}[t]
	\centering
	\caption{Summary of the theoretical results}
	\setlength\tabcolsep{2pt}
	\begin{tabular}{|c|c|}
		\hline
	    Results & Description  \\ 
		\hline
		Theorem~\ref{thm:contractopt}& \thead[c]{Convergence rate of  GMFG-PPO, \\ when  an estimation oracle is assumed.}\\
        \hline
		Theorem~\ref{thm:fixest}& \thead[c]{Convergence rate of the model estimation procedure,\\ when the  agents have {\em known fixed} positions.}\\
        \hline
		Theorem~\ref{thm:randest}& \thead[c]{Convergence rate of the model estimation procedure, \\ when the agents have \emph{known random} positions.}\\
        \hline
		Theorem~\ref{thm:ufixest}& \thead[c]{Convergence rate of the model estimation procedure, \\ when the  agents have \emph{unknown  fixed} positions.}\\
        \hline
		Corollary~\ref{coro:optknownfix}& \thead[c]{Convergence rate of the \ac{ne} learning algorithm that implements GMFG-PPO\\ and collects data from agents with \emph{known fixed} positions.}\\
        \hline
		Corollary~\ref{coro:optknownrand}& \thead[c]{Convergence rate of the \ac{ne} learning algorithm that implements GMFG-PPO\\ and collects data from agents with \emph{known random} positions.}\\
        \hline
		Corollary~\ref{coro:optunknownfix}& \thead[c]{Convergence rate of the \ac{ne} learning algorithm that implements GMFG-PPO\\ and collects data from agents with \emph{unknown    fixed} positions.}\\
		\hline 
	\end{tabular}
	\label{table:sum}
	\vspace{-6mm}
\end{table} 

\vspace{-4mm}
\paragraph{Paper Outline} 
The rest of the paper is organized as follows. We discuss  related works in Section~\ref{sec:relatedwork}. In Section~\ref{sec:prelim}, we introduce the \ac{gmfg}s and a key property that they possess, namely {\em equivariance}. Our three sampling schemes are also introduced. In Section~\ref{sec:opt}, we propose \ac{gmfgppo} and analyze its convergence rate  assuming an estimation oracle. In Section~\ref{sec:modellearn}, we first introduce our mean-embedding procedure. Then we propose and analyze the model-learning algorithms for three sampling schemes. In Section~\ref{sec:combine}, we combine the results from Sections~\ref{sec:opt} and \ref{sec:modellearn}. In Section~\ref{sec:expts}, we provide  the numerical simulation results to corroborate our theoretical findings. In Section~\ref{sec:conclusion}, we conclude our paper.

\section{Related Work}\label{sec:relatedwork}
The \ac{gmfg} has been proposed to study the games played between a large number of heterogenous agents for several year. \citet{parise2019graphon} first formulated the static \ac{gmfg} and proved that it is the limit of finite-agent games with graph structure. \citet{carmona2022stochastic} then generalized these results to the Bayesian setting. \citet{caines2019graphon} analyzed the existence and uniqueness of \ac{ne} of the continuous-time \ac{gmfg}. As a special case, the continuous-time linear-quadratic \ac{gmfg} was studied by \citet{aurell2022stochastic,tchuendom2020master,gao2020linear,gao2021lqg}, where the existence and uniqueness of \ac{ne} were established, and the convergence of finite-agent games to \ac{gmfg} was analyzed. Learning of the \ac{ne} on the discrete-time \ac{gmfg} was first considered in \citet{vasal2020master} via the master equation. After that, \citet{cui2021learning} and \citet{fabian2022learning} proposed algorithms to learn the \ac{ne} of discrete-time \ac{gmfg} with dense and sparse graphons, respectively.

As a special case, \ac{mfg} studies the game between a large number of homogeneous agents. \ac{ne} learning algorithms for the continuous-time \ac{mfg} have been designed via fictitious play~\citep{cardaliaguet2017learning}, mirror descent~\citep{hadikhanloo2017learning}, generalized conditional gradient ~\citep{lavigne2022generalized}, and policy gradient~\citep{guo2022entropy}. For discrete-time \ac{mfg}, efficient algorithms have been proposed based on the notion of contraction~\citep{guo2019learning,xie2021learning,anahtarci2022q,yardim2022policy,guo2023general}. With the monotonicity condition, \citet{perrin2020fictitious} and \citet{perolat2021scaling} propose fictitious play and mirror descent algorithms for learning the \ac{ne}, respectively. Readers are encouraged to refer to \citet{lauriere2022learning} for a comprehensive survey of \ac{mfg}s.

The graphon estimation problem has been studied for a decade under different classes of graphons and different performance metrics. Existing works mainly focus on the estimation of graphons from the random graphs generated from it. \citet{gao2015rate} first proposed a rate-optimal algorithm to estimate the graphon at sampled points. The graphon estimation is then studied under $L_{2}$ norm~\citep{klopp2017oracle,wolfe2013nonparametric}, and cut distance~\citep{klopp2019optimal}. The spectral method for graphon estimation was also studied in \citet{xu2018rates}. For a comprehensice survey of graphon estimation, readers are encouraged to refer to \citet{gao2021minimax}. Different from these works, we aim to estimate the graphons without the graphs generated from them. Instead, we only have access to the state and action samples of agents, who interact with each other according to an unknown graphon structure.

\paragraph{Notations}
We denote $\{1,\cdots,N\}$ as $[N]$. For a set $\calS$, we denote the collection of all the measures and the probability measures on $\calS$ as $\calM(\calS)$ and $\Delta(\calS)$, respectively. For a measurable space $(\calX,\calF)$ and two distributions $P,Q\in\Delta(\calX)$ supported on $\calX$, the total variation distance between them is defined as $\tv(P,Q)=\sup_{A\in\calF}|P(A)-Q(A)|$. For to random variables $X,Y$ supported on $(\calX,\calF)$, we write $\tv(X,Y)$ to denote the total variation between their distributions. For a graphon $W$, we define its infinity norm  as $\|W\|_{\infty}=\sup_{x,y\in[0,1]}|W(x,y)|$.

\section{Preliminaries}\label{sec:prelim}
{\em Graphons} are {\em symmetric} functions that map $[0,1]^2$ to $[0,1]$. By symmetry, we mean that  $W(\alpha,\beta)=W(\beta,\alpha)$ for any $\alpha,\beta\in[0,1]$. The set of all graphons is denoted as $\calW=\{W:[0,1]^{2}\rightarrow[0,1]\,|\, W$ is symmetric$\}$. In the following, graphons are used to represent interactions between agents. We consider a finite horizon \ac{gmfg}  $(\calI,\calS,\calA,\mu_{1},H,P^{*},r^{*},W^{*})$. In this game, each agent is indexed by $\alpha\in\calI=[0,1]$. The state space and the action space of each agent are respectively denoted as $\calS\subseteq \bbR^{d_{\rms}}$ and $\calA\subseteq \bbR^{d_{\rma}}$. We assume that 
$\calS$ is a compact subset of $\bbR^{d_{\rms}}$ and $\calA$ is a finite subset of $\bbR^{d_{\rma}}$. The horizon of the game is denoted as $H\in\bbN$. The initial state distribution of each  agent is $\mu_{1}\in\Delta(\calS)$, where $\Delta(\calS)$ is the set of probability measures on $\calS$. 
The state transition kernels $P^{*}=\{P_{h}^{*}\}_{h=1}^{H}$ are functions $P_{h}^{*}:\calS\times\calA\times\calM(\calS) \rightarrow\Delta(\calS)$ for all $h\in[H]$, where $\calM(\calS)$ is the  set of measures on $\calS$. In contrast to  the single-agent \ac{mdp}, the state dynamics of each agent in a \ac{gmfg} depends on an \emph{aggregate} $z\in\calM(\calS)$, which reflects the influence of other agents on it.   For time $h$ and   agent $\alpha\in\calI$, given a graphon $W_{h}^{*}\in W^{*}=\{W_{h}^{*}\}_{h=1}^{H}$, the aggregate $z_{h}^{\alpha}$ for agent $\alpha$ is defined as
\vspace{-2mm}
\begin{align}
    z^{\alpha}_{h}=\int_{0}^{1} W_{h}^{*}(\alpha,\beta)\calL(s^{\beta}_{h}) \, \rmd \beta, \label{eq:1}
\end{align}
where $\calL(s)\in\Delta(\calS)$ denotes the law of the random variable $s$. We note that the agents in this game are \emph{heterogeneous}. This means that each agent is affected differently by other agents or, in other words, the aggregates $z_{h}^{\alpha}$ for different $\alpha\in\calI$ are, in general, different. Given the state $s^{\alpha}_{h}\in\calS$ and the action $a^{\alpha}_{h}\in\calA$ of agent $\alpha$, the agent transitions to a new state
$s^{\alpha}_{h+1}\sim P_{h}^{*}(\cdot\, |\, s^{\alpha}_{h},a^{\alpha}_{h},z^{\alpha}_{h})$. The reward functions $r^{*}=\{r_{h}^{*}\}_{h=1}^{H}$ are deterministic functions $r_{h}^{*}:\calS\times\calA\times \calM(\calS) \rightarrow\bbR$ for all $h\in[H]$. For agent $\alpha\in\calI$ at time $h$, taking the action $a_{h}^{\alpha}$ under the state $s_{h}^{\alpha}$ and the aggregate $z_{h}^{\alpha}$ earns the agent a reward of $r_{h}^{*}(s_{h}^{\alpha}, a_{h}^{\alpha},z_{h}^{\alpha})$.

    We remark that the above \ac{gmfg} subsumes the \ac{mfg}~\citep{xie2021learning,anahtarci2022q} as a special case. To see this, let $W_{h}(\alpha,\beta)=1$ for all $\alpha,\beta\in\calI$ and $h\in[H]$, then the agents are \emph{homogeneous}. The aggregate $z_{h}^{\alpha}$ in Eqn.~\eqref{eq:1} is simply the state distributions of these homogeneous agents.

A \emph{Markov policy} for the agent $\alpha\in\calI$ is characterized by $\pi^{\alpha}=\{\pi_{h}^{\alpha}\}_{h=1}^{H}\in\Pi^{H}$, where $\pi_{h}^{\alpha}:\calS\rightarrow \Delta(\calA)$ lies in the class $\Pi=\{\pi:\calS\rightarrow\Delta(\calA)\}$. The collection of policies of all agents is denoted as $\pi^{\calI}=(\pi^{\alpha})_{\alpha\in\calI}\in\Pi^{\calI\times H}=\tilde{\Pi}$. We let $\mu^{\alpha}_{h}=\calL(s^{\alpha}_{h})\in \Delta(\calS)$ be the state distribution of the agent $\alpha$ at time $h$. Then $\mu^{\calI}_{h}=(\mu^{\alpha}_{h})_{\alpha\in\calI}\in\Delta(\calS)^{\calI}$ is the set of state distributions of all agents at time $h$. Note that the aggregate $z^{\alpha}_{h}$ is a function of the distributions $\mu^{\calI}_{h}$ and the graphon $W_{h}^{*}$, so we may write it more explicitly as $z^{\alpha}_{h}(\mu^{\calI}_{h},W_{h})$. The distribution flow $\mu^{\calI}=(\mu^{\calI}_{h})_{h=1}^{H}\in\Delta(\calS)^{\calI\times H}=\tilde{\Delta}$ consists of the state distributions of all agents at any given time.

In this work, we focus on the regularized problem~\citep{nachum2017bridging,cui2021approximately}. This setting augments standard reward functions with the entropy of the implemented policy. Some recent works have shown that entropy regularization can accelerate the convergence of the policy gradient methods~\citep{shani2020adaptive,cen2022fast}. In a $\lambda$-regularized \ac{gmfg}, when agent $\alpha$ implements policy $\pi_{h}^{\alpha}$ at time $h$, she will receive a reward $r_{h}^{*}(s^{\alpha}_{h},a^{\alpha}_{h},z^{\alpha}_{h})-\lambda \log\pi_{h}^{\alpha}(a_{h}^{\alpha}\,|\,s_{h}^{\alpha})$ by taking action $a_{h}^{\alpha}$ at state $s_{h}^{\alpha}$. Given the underlying distribution flow $\mu^{\calI}$ and the policy $\pi^{\calI}$, the \emph{value function} and the \emph{action-value function} for agent $\alpha\in\calI$ in the $\lambda$-regularized game with $\lambda>0$ are respectively defined as
\begin{align*}
    V_{h}^{\lambda,\alpha}(s,\pi^{\alpha},\mu^{\calI},W^{*})&=\bbE^{\pi^{\alpha}}\bigg[\sum_{t=h}^{H}r_{t}^{*}\big(s^{\alpha}_{t},a^{\alpha}_{t},z^{\alpha}_{t}(\mu_{t}^{\calI},W_{t}^{*})\big)-\lambda\log\pi_{t}^{\alpha}(a_{t}^{\alpha}\,|\,s_{t}^{\alpha})\,\bigg|\, s^{\alpha}_{h}=s\bigg],\\
    Q_{h}^{\lambda,\alpha}(s,a,\pi^{\alpha},\mu^{\calI},W^{*})&=r_{h}^{*}\big(s,a,z^{\alpha}_{h}(\mu_{h}^{\calI},W_{h}^{*})\big)+\bbE\big[V_{h+1}^{\lambda,\alpha}(s_{h+1}^{\alpha},\pi^{\alpha},\mu^{\calI},W^{*})\,|\, s_{h}^{\alpha}=s,a_{h}^{\alpha}=a\big],
\end{align*}
where the expectation $\bbE^{\pi^{\alpha}}[\cdot]$ is taken with respect to $a_{t}^{\alpha}\sim\pi_{t}^{\alpha}(\cdot\,|\, s_{t}^{\alpha})$ and $s^{\alpha}_{t+1}\sim P_{t}(\cdot\,|\,s_{t}^{\alpha},a_{t}^{\alpha},z_{t}^{\alpha})$ for all $t\in[H]$. The cumulative reward of agent $\alpha \in \mathcal{I}$ under policy $\pi^{\calI}$ is defined as $J^{\lambda,\alpha}(\pi^{\alpha},\mu^{\calI},W^{*})=\bbE_{\mu_{1}^{\alpha}}[V_{1}^{\lambda,\alpha}(s,\pi^{\alpha},\mu^{\calI},W^{*})]$, where the expectation is taken 
with respect to $s\sim \mu_{1}^{\alpha}$.

\begin{definition}\label{def:regne}
    A \ac{ne} of the $\lambda$-regularized \ac{gmfg} is a pair $(\pi^{*,\calI},\mu^{*,\calI})\in \tilde{\Pi}\times\tilde{\Delta}$ that satisfies the following two conditions:
    \begin{itemize}
        \item (Agent rationality) $J^{\lambda,\alpha}(\pi^{*,\alpha},\mu^{*,\calI},W^{*})\geq J^{\lambda,\alpha}(\tilde{\pi}^{\alpha},\mu^{*,\calI},W^{*})$ for all $\alpha\in\calI$ and $\tilde{\pi}^{\alpha}\in\Pi^{ H}$.
        \item (Distribution consistency) The distribution flow $\mu^{*,\calI}$ is equal to the distribution flow $\mu^{\pi^{*,\calI},\calI}$ induced by the policy $\pi^{*,\calI}$.
    \end{itemize}
\end{definition}

We define the operator that returns the optimal policy when the underlying distribution flow is $\mu^{\calI}$ and the graphon is $W$ as $\Gamma_{1}^{\lambda}(\mu^{\calI},W)\in\tilde{\Pi}$, i.e., $\pi^{\calI}=\Gamma_{1}^{\lambda}(\mu^{\calI},W)$ if $J^{\lambda,\alpha}(\pi^{\alpha},\mu^{\calI},W)=\sup_{\tilde{\pi}^{\alpha}\in\tilde{\Pi}} J^{\lambda,\alpha}(\tilde{\pi}^{\calI},\mu^{\calI},W)$ for all $\alpha\in\calI$. We also define the operator that returns the distribution flow induced by the policy $\pi^{\calI}$ as $\Gamma_{2}(\pi^{\calI},W^{*})\in\tilde{\Delta}$, i.e., $\tilde{\mu}^{\calI}=\Gamma_{2}(\pi^{\calI},W^{*})$ if
\begin{align*}
     \tilde{\mu}_{h+1}^{\alpha}(s^{\prime})&=\int_{\calS}\sum_{a\in\calA}\tilde{\mu}_{h}^{\alpha}(s)\pi_{h}^{\alpha}(a\,|\, s)P_{h}\big(s^{\prime}\, |\, s,a,z_{h}^{\alpha}(\tilde{\mu}_{h}^{\calI},W_{h}^{*})\big) \rmd s  \\ 
     &\hspace{2in}\text{ for all }s^{\prime}\in\calS,h\in[H-1]\text{ and }\alpha\in\calI,
\end{align*}
and $\tilde{\mu}_{1}^{\calI}=\mu_{1}^{\calI}$.  Our goal in this paper is to learn the \ac{ne} of the $\lambda$-regularized \ac{gmfg} from the data collected of the sampled agents. Before giving an overview of our agent sampling schemes, we first introduce the equivariance property of \ac{gmfg}.

\subsection{Equivariance Property of \ac{gmfg}s}
We now argue that \ac{gmfg} is \emph{equivariant} to the measure-preserving bijection imposed on agents. In the \ac{gmfg}, all the interactions among agents are captured by the underlying graphons. For agents $\alpha,\beta\in\calI$, the value $W(\alpha,\beta)$ represents the strength of interactions between $\alpha$ and $\beta$. 
Intuitively, if we ``permute'' the positions of agents in the graphon (i.e., we ``permute'' the values of $\alpha\in\calI$) and transform the graphons accordingly, the resultant game  remains the same up to this permutation. However, given an uncountable number of agents in $[0,1]$, the concept of ``permutation" of finite objects should be more precisely stated. This is formalized  by the notion of {\em measure-preserving bijections} from $[0,1]$ to $[0,1]$. Given a measure-preserving bijection $\phi:[0,1]\rightarrow[0,1]$, the transformation of a graphon $W^{\phi}$ is defined as
\begin{align*}
    W^{\phi}(x,y)=W\big(\phi(x),\phi(y)\big).
\end{align*}
We denote the set of all the measure-preserving bijections as $\calB_{[0,1]}$. Then the equivariance property of the \ac{gmfg} can be stated as follows. 
\begin{proposition}\label{prop:equivari}
    For any policy $\pi^{\calI}\in\tilde{\Pi}$, let its distribution flow on $(\calS,\calA,\mu_{1},H,P^{*},r^{*},W^{*})$ be $\mu^{\calI}\in\tilde{\Delta}$. In other words, $\mu^{\calI}=\Gamma_{2}(\pi^{\calI},W^{*})$. For any $\phi\in\calB_{[0,1]}$, define the $\phi$-transformed policy $\pi^{\phi,\calI}$ as $\pi^{\phi,\alpha}=\pi^{\phi(\alpha)}$ for all $\alpha\in\calI$. Then we denote its distribution flow on $(\calS,\calA,\mu_{1},H,P^{*},r^{*},W^{\phi,*})$ as $\mu^{\phi,\calI}\in\tilde{\Delta}$, i.e., $\mu^{\phi,\calI}=\Gamma_{2}(\pi^{\phi,\calI},W^{\phi,*})$. We have
    \begin{align*}
        \mu^{\phi,\alpha}=\mu^{\phi(\alpha)} \text{ for all }\alpha\in\calI.
    \end{align*}
\end{proposition}
Proposition~\ref{prop:equivari} shows that the graphons transformed by a measure-preserving bijections defines the same game as the original graphons up to the bijection. 

\subsection{Overview of Sampling Schemes}\label{sec:sampling}
Our goal is to learn the \ac{ne} from the data collected from a subset of sampled agents. We sample $N$ agents from $[0,1]$ and collect their states, actions, and rewards when implementing behavior policies. In this work, we consider three types of agent sampling procedures
\begin{itemize}
    \item Agents are sampled from \emph{known grid} positions. In particular, we sample the agents at grid positions $\{i/N\}_{i=1}^{N}\subset [0,1]$, and we know the position of each agent; 
    \item Agents are sampled from \emph{known random} positions. In particular, we sample the agents from $N$ i.i.d.\ samples of $\unif([0,1])$, and the positions of agents are also known;
    \item  Agents are sampled from \emph{grid} positions, but the positions of the sampled agents are \emph{unknown}. For example, we know the positions of the sampled agents belong to the set $\{i/N\}_{i=1}^{N}$. However, the position of each agent within the set $\{i/N\}_{i=1}^{N}$ is unknown. 
\end{itemize}

In Section~\ref{sec:modellearn}, we design and analyze a model learning algorithm that estimates the transition kernel $P^{*}$, the reward function $r^{*}$, and the underlying graphons $W^{*}$ for each of these three sampling schemes. In the first two cases, we design a model learning algorithm that estimates the transition kernel $P^{*}$, the reward function $r^{*}$, and the underlying graphons $W^{*}$. However, in the third case, we cannot estimate the original graphons, since the positions of the agents are unknown. Instead, we can only estimate the original graphons up to a measure-preserving bijection. In this case, we need to recover the ``relative positions'' of sampled agents to select the graphons from set $\tilde{\calW}$. For $N$ agents, there are $N!$ potential cases for their relative positions. The super-exponential size of the search space makes the problem statistically challenging. To complete the story, there is a sampling scheme where the positions of agents are unknown and random. However, the analysis of algorithms in this case is difficult due to the need to carefully analyze the order statistics which is rather different from the abovementioned three cases. We leave this case for future work.

\section{Learning Algorithm for \ac{gmfg}}\label{sec:opt}

\subsection{Design of the \ac{gmfgppo} Algorithm}
In this section, we design an algorithm called \ac{gmfgppo} (Algorithm~\ref{algo:MDcontract}) to learn an \ac{ne} of the $\lambda$-regularized \ac{gmfg} with $\lambda>0$.  \ac{gmfgppo}, which is an iterative algorithm,   involves three main steps in each iteration. First, it evaluates the distribution flow and the action-value function (Line 4), assuming the access to a sub-module for computing these. In Section~\ref{sec:modellearn}, we design this sub-module as a model-based learning algorithm. Second, it updates the distribution flow as a mixture of the distribution flow of the current policy $\pi_{t}^{\calI}$ and the current distribution flow $\hbmu_{t}^{\calI}$ (Line 5). This procedure is called \emph{fictitious play} in \citet{xie2021learning} and \citet{perrin2020fictitious}. It slows down the update of the distribution flow. In our analysis, this deceleration is shown to be important for learning the optimal policy with respect to the current distribution flow. Finally, we improve the policy with one-step mirror descent (Line~6). We note that Line 6 is in fact the closed-form solution to the optimization
\begin{align*}
    \hpi_{t+1,h}^{\alpha}(\cdot\,|\,s)=\argmax_{p\in\Delta(\calA)}\eta_{t+1}\Big[\big\langle \hatQ_{h}^{\lambda,\alpha}(s,\cdot,\pi_{t}^{\alpha},\hbmu_{t}^{\calI},\hatW),p \big\rangle-\lambda \barh(p)\Big]-\kl\big(p\|\pi_{t,h}^{\alpha}(\cdot\,|\,s)\big)\quad \forall\, s\in\calS,
\end{align*}
where $\barh(p) = \langle p , \log p\rangle$ is the negative entropy function. This procedure is one-step policy mirror descent in \cite{lan2022policy}, and it also corresponds to the \ac{ppo} algorithm in \cite{schulman2017proximal}. This policy improvement procedure aims to optimize the policy in the \ac{mdp} induced by $\hbmu_{t}^{\calI}$. With the convergence of $\hbmu_{t}^{\calI}$ to $\mu^{*,\calI}$, this procedure can  learn the optimal policy on $\mu^{*,\calI}$, i.e., the policy $\pi^{*,\calI}$ in the \ac{ne}.
\begin{algorithm}[t]
	\caption{\ac{gmfgppo}}
	\textbf{Procedure:}
	\begin{algorithmic}[1]\label{algo:MDcontract}
        \STATE Initialize $\pi_{1,h}^{\alpha}(\cdot\,|\,s)=\unif(\calA)$ for all $s\in\calS$, $h\in[H]$ and $\alpha\in\calI$.
        \STATE Initialize $\hbmu_{1}^{\calI}=\hat{\Gamma}_{2}(\pi_{1}^{\calI},\hatW)$.
        \FOR{$t=1,2,\cdots,T$}
            \STATE Compute the distribution flow $\hat{\mu}_{t}^{\calI}=\hat{\Gamma}_{2}(\pi_{t}^{\calI},\hatW)$ induced by policy $\pi_{t}^{\calI}$ and corresponding action-value function $\hatQ_{h}^{\lambda,\alpha}(s,a,\pi_{t}^{\alpha},\hbmu_{t}^{\calI},\hatW)$ for all $\alpha\in\calI$ and $h\in[H]$.
            \STATE $\hbmu_{t+1}^{\calI}=(1-\alpha_{t})\hbmu_{t}^{\calI}+\alpha_{t}\hmu_{t}^{\calI}$.
            \STATE $\hpi_{t+1,h}^{\alpha}(\cdot\,|\,s)\propto \big(\pi_{t,h}^{\alpha}(\cdot\,|\,s)\big)^{1-\frac{\lambda\eta_{t+1}}{1+\lambda\eta_{t+1}}}\exp\Big(\frac{\eta_{t+1}}{1+\lambda\eta_{t+1}} \hatQ_{h}^{\lambda,\alpha}(s,\cdot,\pi_{t}^{\alpha},\hbmu_{t}^{\calI},\hatW)\Big)$ for all $\alpha\in\calI$ and  $h\in[H]$
            \STATE $\pi_{t+1,h}(\cdot\,|\,s)=(1-\beta_{t+1})\hpi_{t+1,h}^{\alpha}(\cdot\,|\,s)+\beta_{t+1}\unif(\calA)$
        \ENDFOR
        \STATE Output $\bar{\pi}^{\calI}=\unif\big(\pi_{[1:T]}^{\calI}\big)$ and $\barmu^{\calI}=\unif\big(\hbmu_{[1:T]}^{\calI}\big)$
	\end{algorithmic}
\end{algorithm}

\ac{gmfgppo} differs from the \ac{ne} learning algorithm of regularized \ac{mfg} in \cite{xie2021learning} in three aspects. First, \ac{gmfgppo} is designed to learn the \ac{ne} of the regularized \ac{gmfg}. It involves  graphon learning and requires the policy and action-value function updates for all the agents. In contrast, the algorithm in \cite{xie2021learning} can only learn the \ac{ne} of the regularized \ac{mfg}, which is a special case of \ac{gmfg} with {\em constant} graphons. It only keeps track of the policy and action-value function of a representative agent. Second, \ac{gmfgppo} learns a non-stationary \ac{ne}, whereas the algorithm in \cite{xie2021learning} learns a  {\em stationary} \ac{ne}.  Finally, the stepsize $\eta_{t}$ used in the policy improvement (Line 6) will be set to be a (non-vanishing) constant in Section~\ref{sec:optanaly}. In contrast, the algorithm in \cite{xie2021learning} sets $\eta_{t}=o(1)$. Our 
choice of $\eta_{t}$ is the chief reason for the  improved convergence rate. 

\subsection{Convergence Analysis of \ac{gmfgppo}}\label{sec:optanaly}
Assuming that an \ac{ne} exists~\citep{cui2021learning, fabian2022learning}, we now present convergence results for learning it. We denote an \ac{ne} of the $\lambda$-regularized \ac{gmfg} as $(\pi^{*,\calI},\mu^{*,\calI})$. We measure the distances between policies and distribution flows with
\begin{align*}D(\pi^{\calI},\tilde{\pi}^{\calI}) &=\int_{0}^{1}\sum_{h=1}^{H}\bbE_{\mu_{h}^{*,\alpha}}\Big[\big\|\pi_{h}^{\alpha}(\cdot\,|\,s)-\tilde{\pi}_{h}^{\alpha}(\cdot\,|\,s)\big\|_{1}\Big]\, \rmd\alpha,\qquad\text{and} \\ d(\mu^{\calI},\tilde{\mu}^{\calI})&=\int_{0}^{1}\sum_{h=1}^{H}\|\mu_{h}^{\alpha}-\tilde{\mu}_{h}^{\alpha}\|_{1}\, \rmd\alpha.
\end{align*}
For the purpose of our convergence results, we make a few assumptions about the $\lambda$-regularized \ac{gmfg}. We first assume the Lipschitz continuity of transition kernels and reward functions.
\begin{assumption}\label{assump:lipconti}
    The reward function $r_{h}(s,a,z)$ is Lipschitz continuous in $z$ for all $h\in[H]$, that is $|r_{h}(s,a,z)-r_{h}(s,a,z^{\prime})|\leq L_{r} \|z-z^{\prime}\|_{1}$ for all $h\in[H]$, $s\in\calS$ and $a\in\calA$. The transition kernel $P_{h}(\cdot\,|\, s,a,z)$ is Lipschitz continuous in $z$ with respect to the total variation, that is $\tv(P_{h}(\cdot\,|\, s,a,z),P_{h}(\cdot\,|\, s,a,z^{\prime}))\leq L_{P}\|z-z^{\prime}\|_{1}$ for all $h\in[H]$, $s\in\calS$ and $a\in\calA$.
\end{assumption}
This assumption is common in the \ac{mfg} and \ac{gmfg} literature~\citep{cui2021learning,anahtarci2022q}. We then assume that the composition of the operators $\Gamma_{1}^{\lambda}$ and $\Gamma_{2}$ is contractive in the following sense.
\begin{assumption}\label{assump:contract}
    There exist constants $d_{1},d_{2}>0$ and $d_{1}d_{2}<1$ such that for any policies $\pi^{\calI},\tilde{\pi}^{\calI}$ and distribution flows $\mu^{\calI},\tilde{\mu}^{\calI}$, it holds that
    \begin{align*}
        D\big(\Gamma_{1}^{\lambda}(\mu^{\calI},W^{*}),\Gamma_{1}^{\lambda}(\tilde{\mu}^{\calI},W^{*})\big)&\leq d_{1}\, d(\mu^{\calI},\tilde{\mu}^{\calI}),\quad \mbox{and}\\
        d\big(\Gamma_{2}(\pi^{\calI},W^{*}),\Gamma_{2}(\tilde{\pi}^{\calI},W^{*})\big)&\leq d_{2}\, D(\pi^{\calI},\tilde{\pi}^{\calI}).
    \end{align*}
\end{assumption}
This ``contractive'' assumption plays an important role in the design of efficient algorithms, since it guarantees the convergence of both $\pi^\calI$ and $\mu^{\calI}$ using simple fixed point iterations. This assumption is widely adopted in the \ac{mfg} literature~\citep{xie2021learning,guo2019learning}, and it holds if the regularization $\lambda$ is higher enough than $L_{r}$ and $L_{P}$~\citep{anahtarci2022q,cui2021approximately}. The uniqueness of the \ac{ne} implied by Assumption~\ref{assump:contract} is proved in Appendix~\ref{app:uniquene}.

For a policy $\pi^{\calI}$ and any distribution flow $\mu^{\calI}$, we define the operator  $\Gamma_3$ that satisfies $\mu^{+,\calI}=\Gamma_{3}(\pi^{\calI},\mu^{\calI},W)$ as
\begin{align*}
    \mu_{1}^{+,\calI}=\mu_{1}^{\calI}, \qquad\mu_{h+1}^{+,\alpha}(s^{\prime})=\sum_{a\in\calA}\int_{\calS}\mu_{h}^{+,\alpha}(s)\pi_{h}^{\alpha}(a\,|\, s)P_{h}\big(s^{\prime}\, |\, s,a,z_{h}^{\alpha}(\mu_{h}^{\calI},W_{h})\big) \, \rmd s ,
\end{align*}
  for all $s^{\prime}\in\calS,\alpha\in\calI$,   and $h\geq 1$.
The operator $\Gamma_{3}$ outputs the distribution flow $\mu^{+,\calI}$ for implementing the policy $\pi^{\calI}$ on the \ac{mdp} induced by $\mu^{\calI}$. We now make an assumption about certain concentrability coefficients.
\begin{assumption}\label{assump:optconcen}
    For any distribution flow $\mu^{\calI}$, we define its induced optimal policy on the \ac{mdp} induced by it as $\pi^{*,\calI}=\Gamma_{1}^{\lambda}(\mu^{\calI},W^{*})$ and the induced distribution flow as $\tilde{\mu}^{*,\calI}=\Gamma_{3}(\pi^{*,\calI},\mu^{\calI},W^{*})$. Then there exists a constant $C_{\mu}>0$ such that for any distribution flow $\mu^{\calI}$, it hold that
    \begin{align*}\sup_{\alpha\in\calI,h\in[H]}\bbE_{s\sim\tilde{\mu}_{h}^{*,\alpha}}\bigg[\bigg|\frac{\mu_{h}^{*,\alpha}(s)}{\tilde{\mu}_{h}^{*,\alpha}(s)}\bigg|^{2}\bigg] \leq C_{\mu}^2.
    \end{align*}
\end{assumption}
This assumption concerns the boundedness of concentrability coefficients. This type of assumption are standard in the policy optimization literatures~\citep{shani2020adaptive,bhandari2019global,agarwal2020optimality}. We then make an assumption about the accuracy about our distribution flow and action-value function estimates in Line 4 of Algorithm~\ref{algo:MDcontract}.
\begin{assumption}\label{assump:est}
    We have access to the estimator $\hatP=\{\hatP_{h}\}_{h=1}^{H}$, $\hatr=\{\hatr_{h}\}_{h=1}^{H}$, and $\hatW=\{W_{h}\}_{h=1}^{H}$ and corresponding operator estimate $\hat{\Gamma}_{2}(\cdot,\hatW)$ and action-value function estimator $\hatQ_{h}^{\lambda,\alpha}(\cdot,\hatW)$. These estimates satisfy that
     for any policy $\pi^{\calI}$, we have that
    \begin{align*}
        d\big(\hat{\Gamma}_{2}(\pi^{\calI},\hatW),\Gamma_{2}(\pi^{\calI},W^{*})\big)\leq \varepsilon_{\mu},
    \end{align*}
    and that for any policy $\pi^{\calI}$ and distribution flow $\mu^{\calI}$
    \begin{align*}
        \sup_{\tilde{\pi}^{\calI},\alpha}\bbE_{\tilde{\pi}^{\alpha},\mu^{\calI}}\big\|\hatQ_{h}^{\lambda,\alpha}(s,\cdot,\pi^{\alpha},\mu^{\calI},\hatW)-Q_{h}^{\lambda,\alpha}(s,\cdot,\pi^{\alpha},\mu^{\calI},W^{*})\big\|_{\infty}\leq \varepsilon_{Q}.
    \end{align*}
    for some constants $\varepsilon_{\mu}$ and $\varepsilon_{Q}$.
\end{assumption}
We make this assumption only for ease of the presentation of the analysis of our algorithm. In Section~\ref{sec:combine}, we will replace this assumption with the actual performance guarantee of our model learning algorithms, in which $\varepsilon_{\mu}$ and $\varepsilon_{Q}$ will be quantified.

\begin{theorem}\label{thm:contractopt}
    We set $\alpha_{t}=O(T^{-2/3})$, $\beta_{t}=O(T^{-1})$, and $\eta_{t}$ to a constant that only depends on $\lambda$, $H$ and $|\calA|$. Under Assumptions~\ref{assump:lipconti}, \ref{assump:contract},  \ref{assump:optconcen},  and \ref{assump:est}, Algorithm~\ref{algo:MDcontract} returns the policy $\bar{\pi}^{\calI}$ and the distribution flow $\bar{\mu}^{\calI}$ that satisfies
    \begin{align*}
        &D\bigg(\frac{1}{T}\sum_{t=1}^{T}\pi_{t}^{\calI},\pi^{*,\calI}\bigg)+d\bigg(\frac{1}{T}\sum_{t=1}^{T}\hbmu_{t}^{\calI},\mu^{*,\calI}\bigg)=O\bigg(\frac{\sqrt{\log T}}{T^{1/3}}\bigg)+O(\varepsilon_{\mu}+\sqrt{\varepsilon_{Q}+\varepsilon_{\mu}}).
    \end{align*}
\end{theorem}
There are two main differences in Theorem~\ref{thm:contractopt} and~\citet[Theorem 1]{xie2021learning}. First, we achieve a faster rate $\tilO(T^{-1/3})$ than the rate $\tilO(T^{-1/5})$ in \cite{xie2021learning}. This improvement is attributed to the newly designed stepsize $\eta_{t}$, which is a {\em constant}, but the algorithm in \cite{xie2021learning} sets $\eta_{t}$ to be $O(T^{-2/5})$. Intuitively, a stepsize $\eta_{t}$ that is independent of $T$ will result in faster convergence of an algorithm compared to one that decays as $T$ grows. However, the proof involves a novel optimization error recursion analysis for this new stepsize. This novel optimization error recursion analysis also generalizes  \citet[Theorem~1]{lan2022policy} to the time-inhomogeneous \ac{mdp} with a finite horizon. See Appendix~\ref{app:singlepmd} for the statement. Second, Theorem~\ref{thm:contractopt} does not require the first condition in   Assumptions~4 and~5 in \cite{xie2021learning}. Instead, we adopt the more realistic Assumption~\ref{assump:lipconti} concerning the Lipschitzness of transition kernels and reward functions to control the difference between the \ac{mdp} induced by difference distribution flows.

\section{Model Estimation From Datasets}\label{sec:modellearn}
We assume that the state space $\calS\subseteq\bbR^{d_{\rms}}$ is a subset of $\bbR$, i.e., $d_{\rms}=1$. Our results can be  extended to the case $d_{\rms}>1$ by using kernels of functions with multiple outputs. Since $\calS$ is compact, there exists a constant $B_{S}>0$ such that $|s|\leq B_{S}$ for all $s\in\calS$.

\subsection{Dataset Collection}\label{sec:datacollect}
Since the \ac{gmfg} involves uncountably infinite agents, it is impossible to collect the trajectories of all the agents. Thus, we sample $N$ agents $\{\xi_{i}\}_{i=1}^{N}$ in $[0,1]$ to collect their states, actions, and rewards in each episode. We consider three sampling methods: (i) agents' positions $\{\xi_{i}\}_{i=1}^{N}$ are known grids, namely, $\xi_{i}=i/N$ for all $i\in[N]$. Furthermore, the map between the identity of each agent to the grid $\{i/N\}_{i=1}^{N}$ is known. (ii) $\{\xi_{i}\}_{i=1}^{N}$ are known i.i.d.\ samples  of the uniform distribution $\unif([0,1])$. (iii) agents' positions $\{\xi_{i}\}_{i=1}^{N}$ are grid points, and these positions are \emph{unknown}. Then we acquire the states and actions of these sampled agents. For notational simplicity, we denote the state $s_{h}^{\xi_{i}}$ and action $a_{h}^{\xi_{i}}$ of the agent $\xi_{i}$ as $s_{h}^{i}=s_{h}^{\xi_{i}}$ and $a_{h}^{i}=a_{h}^{\xi_{i}}$, respectively. To collect these data, we implement $L$ behavior policies $\pi_{\tau}^{\calI}$ for all $\tau\in[L]$. In the $\tau^{\text{th}}$ episode, a trajectory of these agents is $\calD_{\tau}=\{(s_{\tau,h}^{[N]},a_{\tau,h}^{[N]},r_{\tau,h}^{[N]},s_{\tau,h+1}^{[N]})\}_{h=1}^{H}$. The dataset consists of $L$ trajectories, i.e., $\calD=\{\calD_{\tau}\}_{\tau=1}^{L}$.

We note that once the behavior policy $\pi_{\tau}^{\calI}$ is determined, the distribution flow $\mu_{\tau}^{\calI}$ is fixed. Then the influence aggregate on the $i^{th}$ agent $z_{\tau,h}^{i}(W_{h}^{*})$ is a function only of $\xi_{i}$, which is independent of the states of other agents. Thus, the distribution of $s_{\tau,h}^{[N]}$ is $\prod_{i=1}^{N}\mu_{\tau,h}^{\xi_{i}}=\prod_{i=1}^{N}\mu_{\tau,h}^{i}$.

\subsection{Mean-Embedding of Distribution Flows}
The transition kernels and the reward functions both take $(s,a,z)$ as their inputs. However, the aggregate $z\in\calM(\calS)$ for an agent which is defined in Eqn.~\eqref{eq:1} is not available to us, since it requires the unknown values of graphons $W^{*}$ and the distribution flow $\mu^{\calI}$. From the collected data, we only have the states $\{s_{\tau,h}^{i}\}_{i=1}^{N}$ sampled from distributions $\{\mu_{\tau,h}^{i}\}_{i=1}^{N}$. Thus, we first need to estimate the distribution flow $\mu^{\calI}$ from these sample. We handle this by using a mean-embedding, which is a widely adopted method in distribution regression~\citep{szabo2016learning,szabo2015two}. Define $\Xi=\calS\times\calA\times\calS$, then $d_{s,a,z}=\delta_{s}\times\delta_{a}\times z$ is measure on $\Xi$. Given a positive definite kernel $k:\Xi\times\Xi\rightarrow\bbR$, we denote the \ac{rkhs} spanned by kernel $k$ as $\calH$. Then we embed the measure $d_{s,a,z}$ with the kernel $k$ as
\begin{align*}
    \omega_{d_{s,a,z}}=\int_{\calS}k\big(\cdot,(s,a,s^{\prime})\big)z(\rmd s^{\prime}).
\end{align*}
We have   $\omega_{d_{s,a,z}}\in\calH$. We note that such mean-embedding procedure will not cause the problem to be degenerate, since the embedding with the identity kernel degenerates to $d_{s,a,z}$. For our regression setting, we will embed the measure $\delta_{s_{h}^{\alpha}}\times\delta_{a_{h}^{\alpha}}\times z_{h}^{\alpha}(W_{h}^{*})$ for all $\alpha\in\calI$, and $h\in[H]$. Here the aggregate $z_{h}^{\alpha}$ is the influence aggregate for agent $\alpha$ at time $h$ defined in Eqn.~\eqref{eq:1}. Then the mean-embedding of the measure $\delta_{s_{h}^{\alpha}}\times\delta_{a_{h}^{\alpha}}\times z_{h}^{\alpha}(W_{h}^{*})$ is
\begin{align*}
    \omega_{h}^{\alpha}(W_{h}^{*})=\int_{0}^{1}\int_{\calS}W_{h}^{*}(\alpha,\beta)k\big(\cdot,(s_{h}^{\alpha},a_{h}^{\alpha},s)\big)\mu_{h}^{\beta}(s)\, \rmd s\, \rmd\beta.
\end{align*}
Given such embedding representation, we reformulate the transition kernels and the reward functions as functions $f_{h}^{*},g_{h}^{*}:\calH\rightarrow\bbR$ that is defined as
\begin{align}
    s_{h+1}^{\alpha}=f_{h}^{*}\big(\omega_{h}^{\alpha}(W_{h}^{*})\big)+\varepsilon_{h},\quad r_{h}^{\alpha}=g_{h}^{*}\big(\omega_{h}^{\alpha}(W_{h}^{*})\big) \text{ for all }h\in[H], \alpha\in\calI,\label{eq:sd}
\end{align}
where $\{\varepsilon_{h}^{\alpha}\}_{\alpha\in\calI}$ are independent zero-mean noises. Since $|s|\leq B_{S}$, we have  $|\varepsilon_{h}^{\alpha}|\leq 2B_{S}$.

\subsection{Assumptions for Model Learning}
In the following, we will estimate the transition kernels $\{f_{h}^{*}\}_{h=1}^{H}$, the reward functions $\{g_{h}^{*}\}_{h=1}^{H}$ and the graphons $\{W_{h}^{*}\}_{h=1}^{H}$ from the collected data. With nonparametric regression methods, we adopt a general graphon class $\tilde{\calW}$ to estimate the underlying graphons and adopt the kernels $K:\calH\times\calH\rightarrow\bbR$ and $\tilK:\calH\times\calH\rightarrow\bbR$ to estimate the transition kernel and reward functions, respectively. The space spanned by the kernels $K$ and $\tilK$ are respectively denoted as $\bar{\calH}$ and $\tilde{\calH}$. We postpone the details of the estimation algorithms for three sampling schemes to the following sections, and we first state the assumptions needed for the convergence of all these estimation algorithms.

First, we assume the Lipschitz continuity of the graphon class $\tilde{\calW}$ and the nominal graphons $W^{*}=\{W_{h}^{*}\}_{h=1}^{H}$. This assumption will help us to generalize the estimate from the sampled agents to the unobserved agents.
\begin{assumption}[Lipschitzness of Graphons]\label{assump:graphon}
    For any $W\in\tilde{\calW}$ (resp. $\{W_{h}^{*}\}_{h=1}^{H}$), we have that $|W(\alpha,\beta)-W(\alpha^{\prime},\beta^{\prime})|\leq L_{\tilde{\calW}}(|\alpha-\alpha^{\prime}|+|\beta-\beta^{\prime}|)$ (resp. $|W(\alpha,\beta)-W(\alpha^{\prime},\beta^{\prime})|\leq L_{W^{*}}(|\alpha-\alpha^{\prime}|+|\beta-\beta^{\prime}|)$) for all $\alpha,\alpha^{\prime},\beta,\beta^{\prime}\in[0,1]$, where $L_{\tilde{\calW}}>0$ (resp. $L_{W^{*}}>0$) is a constant.
\end{assumption}
For ease of notation, we define $L_{\bar{\calW}}=\max\{L_{\tilde{\calW}},L_{W^{*}}\}$. Second, we assume the boundedness and the Lipschitz continuity of the kernels. Similar as Assumption~\ref{assump:graphon}, this assumption is helpful to guarantee the boundedness of estimates and generalize the estimates from the sampled agents to the unobserved agents.
\begin{assumption}[Boundedness and Lipschitzness of Kernels]\label{assump:kernel}
    The reproducing kernels $k$, $K$ and $\tilK$ satisfy
    \begin{itemize}
        \item The kernel $k$ is bounded, i.e., there exists $B_{k}>0$ such that $k(x,x)\leq B_{k}^{2}$ for all $x\in\Xi$.
        \item The kernel $K$ (resp. $\tilK$) is bounded, i.e., there exists $B_{K}>0$ (resp. $B_{\tilK}>0$) such that $K(\omega,\omega)\leq B_{K}^{2}$ (resp. $\tilK(\omega,\omega)\leq B_{\tilK}^{2}$) for all $\omega\in\calH$.
        \item The kernel $K$ (resp. $\tilK$) is $L_{K}$-Lipschitz (resp. $L_{\tilK}$) continuous, i.e., $\|K(\cdot,\omega)-K(\cdot,\omega^{\prime})\|_{\bar{\calH}}\leq L_{K}\|\omega-\omega^{\prime}\|_{\calH}$ (resp. $\|\tilK(\cdot,\omega)-\tilK(\cdot,\omega^{\prime})\|_{\bar{\calH}}\leq L_{\tilK}\|\omega-\omega^{\prime}\|_{\calH}$) for all $\omega,\omega^{\prime}\in\calH$.
    \end{itemize}
\end{assumption}
For ease of notation, we define the maximal boundedness parameter $\barB_{K}=\max\{B_{K},B_{\tilK}\}$ and the maximal Lipschitz constant $\barL_{K}=\max\{L_{K},L_{\tilK}\}$. Finally, we state the realizability assumption. It guarantees that we choose the proper function class for our regression task. We define the $r$-ball in a \ac{rkhs} $\bar{\calH}$ as $\bbB(r,\bar{\calH})=\{f\in\bar{\calH}\,|\, \|f\|_{\bar{\calH}}\leq r\}$.
\begin{assumption}[Realizability]\label{assump:real}
    The nominal transition functions $f_{h}^{*}$, reward functions $g_{h}^{*}$ and graphons $W_{h}^{*}$ satisfy that $f_{h}^{*}\in \bbB(r,\bar{\calH})$, $g_{h}^{*}\in \bbB(\tilr,\tilde{\calH})$ and $W_{h}^{*}\in\tilde{\calW}$ for all $h\in[H]$, where $r,\tilr>0$ are some constants.
\end{assumption}
For ease of notation, we define the maximal radius as $\barr=\max\{r,\tilr\}$. We note that our algorithms and analysis are also applicable to the general function class $\calF$ and $\tilde{\calF}$, replacing $\calH$ and $\tilde{\calH}$. Here we adopt the \ac{rkhs} for the ease of representation.

\subsection{Learning from Sampled Agents with Known Positions}\label{sec:knownest}
In this section, we design regression algorithms when the positions of sampled agents are known. From the data collection procedure in Section~\ref{sec:datacollect}, the values of the distribution flows $\mu_{\tau}^{\calI}$ for $\tau\in[L]$ are not directly accessible. For the $i^{\rm{th}}$ agent, the mean-embedding of her state, action and the aggregate at time $h$ in the $\tau^{\rm{th}}$ episode is
\begin{align}
    \omega_{\tau,h}^{i}(W_{h}^{*})=\int_{\calS}k\big(\cdot,(s_{\tau,h}^{i},a_{\tau,h}^{i},s)\big)z_{\tau,h}^{i}(\rmd s)=\int_{0}^{1}\int_{\calS}W_{h}^{*}(\xi_{i},\beta)k\big(\cdot,(s_{\tau,h}^{i},a_{\tau,h}^{i},s)\big)\mu_{\tau,h}^{\beta}(s)\, \rmd s\, \rmd\beta.\label{eq:meanembed}
\end{align}
Thus, the input of $f_{h}^{*}$ and $g_{h}^{*}$, i.e., $\omega_{\tau,h}^{i}(W_{h}^{*})$, needs to be estimated. Given any graphon $W\in\tilde{\calW}$, we derive the empirical estimate of the aggregate of the $i^{\rm{th}}$ agent at time $h$ as
\begin{align*}
    \hat{z}_{\tau,h}^{i}(W)=\frac{1}{N-1}\sum_{j\neq i}W(\xi_{i},\xi_{j})\delta_{s_{\tau,h}^{j}}.
\end{align*}
This estimate involves three kinds of error sources. The first is the {\em graphon estimation error}, which originates from the difference between $W$ and $W_{h}^{*}$. The second is the {\em agent sampling error} which originates from the approximation of uncountably many agents in $[0,1]$ with $N-1$ of them, i.e., an integral over $[0,1]$ is replaced by a sum over $N-1$ terms. The last is the {\em state sampling error} in which   we replace the integral of $\mu_{\tau,h}^{\xi_{j}}$ over state space $\calS$ with the singleton $\delta_{s_{\tau,h}^{j}}$. In the analysis, we handle these three errors separately. Given the aggregate estimate $\hat{z}_{\tau,h}^{i}(W)$, the corresponding mean-embedding of the state, action, and the aggregate for the $i^{\rm{th}}$ agent is
\begin{align}
    \homega_{\tau,h}^{i}(W)=\frac{1}{N-1}\sum_{j\neq i}W(\xi_{i},\xi_{j})k\big(\cdot,(s_{\tau,h}^{i},a_{\tau,h}^{i},s_{\tau,h}^{j})\big).\label{eq:homega}
\end{align}
Taking this estimate as the input of $f_{h}^{*}$ and $g_{h}^{*}$, we evaluate the square error of the prediction and derive the estimates by minimizing the error. Thus, the estimation procedure for learning the system dynamics, the reward functions, and the graphons can be expressed as
\begin{align}
    (\hatf_{h},\hatg_{h},\hatW_{h})=\!\!\!\!\argmin_{\substack{f\in\bbB(r,\bar{\calH}),g\in\bbB(\tilr,\tilde{\calH}),\\ W\in\tilde{\calW}}} \!\!\frac{1}{NL}\sum_{\tau=1}^{L}\sum_{i=1}^{N}\Big(\!s_{\tau,h+1}^{i}\!-\!f\big(\homega_{\tau,h}^{i}(W)\big)\!\Big)^{2}\!\!\!+\!\Big(\!r_{\tau,h}^{i}\!-\!g\big(\homega_{\tau,h}^{i}(W)\big)\!\Big)^{2}.\label{algo:estalgo}
\end{align}
We note that the above optimization problem is, in general, non-convex. However, we focus on the statistical property of it in this work, and the practical implementation can be done with the help of non-convex optimization algorithms. In this estimation procedure, we form our predictions of states/rewards via the composition of two procedures, i.e., 
\begin{align}\label{eq:twostage}
    \{(s_{\tau,h}^{i},a_{\tau,h}^{i})\}_{i=1}^{N}\overset{k,W}{\longrightarrow}\homega_{\tau,h}^{i}(W)\overset{f \text{ or } g}{\longrightarrow}s_{\tau,h+1}^{i}/r_{\tau,h}^{i}
\end{align}
In the first stage, the states and actions are embedded with the kernel $k$ and a selected graphon $W$. In the second stage, the mean-embedding $\homega_{\tau,h}^{i}(W)$ is forwarded by the functions in $\calH$ or $\tilde{\calH}$. 

This two-stage prediction distinguishes our estimation procedure from the algorithms designed for the distribution regression problem~\citep{szabo2016learning,fang2020optimal,meunier2022distribution}. In the distribution regression problem, the covariate, i.e., the input of $f$ or $g$ in Eqn.~\eqref{eq:twostage}, is an unknown distribution. In this problem, we are tasked with performing a regression from the data of the response variable and the i.i.d.\ samples of the unknown distribution. Although the distribution regression problem also requires a two-stage prediction similarly as Eqn.~\eqref{eq:twostage}, i.e., the covariate should be first estimated from i.i.d.\ samples drawn from itself, our problem setting involving graphons is a strict generalization of distribution regression. First, the input of $f$ or $g$ in our problem is a function of a set of distributions $\{\mu_{\tau,h}^{\alpha}\}_{\alpha\in\calI}$. In contrast, the covariate of the distribution regression problem is a single distribution. Second, in addition to the recovery of $\mu_{\tau,h}^{\calI}$ from its samples, our problem requires the estimation of the graphon $W$ to form $\homega_{\tau,h}^{i}(W)$. However, the distribution regression problem only requires the recovery of a distribution from its i.i.d.\ samples, which corresponds to the case that $W$ is a constant function.

\subsubsection{Agents with Known Grid Positions}\label{sec:knownfixest}
In this section, we provide the convergence result of the estimation procedure in Eqn.~\eqref{algo:estalgo} in the setting where the agents' positions $\{\xi_{i}\}_{i=1}^{N}$ form a known grid on $[0,1]$. Without loss of generality, we assume that $\xi_{i}\leq\xi_{j}$ for any $i\leq j$ in $[N]$, and denote the set of positions as $\bar{\xi}=\{\xi_{i}\}_{i=1}^{N}$. In this section, our behavior policies $\pi_{\tau}^{\calI}$ for $\tau\in[L]$ are set as $L_{\pi}$-Lipschitz policies. It means that $\|\pi_{h}^{\alpha}(\cdot\,|\,s)-\pi_{h}^{\beta}(\cdot\,|\,s)\|_{1}\leq L_{\pi}|\alpha-\beta|$ for all $h\in[H]$ and $\alpha,\beta\in\calI$. We note that setting the behavior policies as Lipschitz policies will not restrict the applicability of our estimation procedure, since the \ac{ne} is shown to be Lipschitz under Assumptions~\ref{assump:lipconti} and \ref{assump:graphon} in Appendix~\ref{app:nelip}.

Then we introduce the performance metric for our estimates. Given $\bar{\xi}$, the joint distribution of $(s_{\tau,h}^{i},a_{\tau,h}^{i},\mu_{\tau,h}^{\calI},r_{\tau,h}^{i},s_{\tau,h+1}^{i})_{i=1}^{N}$ is $\prod_{i=1}^{N}\rho_{\tau,h}^{i}$, where $\rho_{\tau,h}^{i}=\mu_{\tau,h}^{i}\times\pi_{\tau,h}^{i}\times\delta_{\mu_{\tau,h}^{\calI}}\times \delta_{r_{h}^{*}}\times P_{h}^{*}$. Here $\delta_{\mu_{\tau,h}^{\calI}}$ is the delta distribution induced by the deterministic function $r_{h}^{*}$. We define the risk of $(f,g,W)$ given $\bar{\xi}$ as
\begin{align}
    \calR_{\bar{\xi}}(f,g,W)=\frac{1}{NL}\sum_{\tau=1}^{L}\sum_{i=1}^{N}\bbE_{\rho_{\tau,h}^{i}}\bigg[\Big(s_{\tau,h+1}^{i}-f\big(\omega_{\tau,h}^{i}(W)\big)\Big)^{2}+\Big(r_{\tau,h}^{i}-g\big(\omega_{\tau,h}^{i}(W)\big)\Big)^{2}\bigg].\label{eq:condirisk}
\end{align}
The risk $\calR_{\bar{\xi}}(f,g,W)$ measures the mean square error of the estimates $f,g,W$ with respect to the distributions of states, actions and distribution flow on the sampled agents. The convergence rate of our estimates $(\hatf_{h},\hatg_{h},\hatW_{h})$ is stated as follows.
\begin{theorem}\label{thm:fixest}
    Under Assumptions~\ref{assump:lipconti}, \ref{assump:graphon}, \ref{assump:kernel}, and \ref{assump:real}, if $\{\xi_{i}\}_{i=1}^{N}$ are known grid positions such that $\xi_{i}=i/N$ for $i\in[N]$, then with probability at least $1-\delta$, the risk of the estimates in Eqn.~\eqref{algo:estalgo} can be bounded as
    \begin{align*}
        &\calR_{\bar{\xi}}(\hatf_{h},\hatg_{h},\hatW_{h})-\calR_{\bar{\xi}}(f_{h}^{*},g_{h}^{*},W_{h}^{*})\\
        &\quad= O\bigg(\underbrace{\frac{(B_{S}+\barr\barB_{K})^{4}}{NL}\log\frac{N_{\bbB_{r}}N_{\tilde{\bbB}_{\tilr}}N_{\tilde{\calW}}}{\delta}}_{\text{generalization error}}+\underbrace{\frac{(B_{S}+\barr\barB_{K})\barr\barL_{K}B_{k}}{\sqrt{N}}\log\frac{NL\calN_{\infty}(1/\sqrt{N},\tilde{\calW})}{\delta}}_{\text{mean-embedding estimation error}}\bigg),
    \end{align*}
    where
    \begin{align*}
        N_{\bbB_{r}}=\calN_{\bar{\calH}}\bigg(\frac{3}{NL},\bbB(r,\bar{\calH})\bigg),\quad N_{\tilde{\bbB}_{\tilr}}=\calN_{\bar{\calH}}\bigg(\frac{3}{NL},\bbB(\tilr,\tilde{\calH})\bigg),\quad N_{\tilde{\calW}}=\calN_{\infty}\bigg(\frac{3}{L_{K}NL},\tilde{\calW}\bigg).
    \end{align*}
\end{theorem}
The estimation error in Theorem~\ref{thm:fixest} consists of two terms: the first term corresponds to the generalization error, and the second term corresponds to the mean-embedding estimation error. The generalization error involves the error from optimizing over the empirical mean of the risk in Eqn.~\eqref{algo:estalgo} instead of the population risk in Eqn.~\eqref{eq:condirisk}. The mean-embedding estimation error comes from the fact that we cannot directly observe the distribution flow $\mu_{\tau}^{\calI}$, but we need to estimate it from the states of sampled agents. As discussed in Section~\ref{sec:knownest}, the mean-embedding estimation error consists of the agent sampling error and the state sampling error. If we use finite general function classes, then the covering number in the bound will be replaced by the cardinalities of these function classes. The resultant convergence rate would thus be $O(1/\sqrt{N})$.

The model learning algorithm in \citet{pasztor2021efficient} for the \ac{mfg} assumes access to the nominal value of the distribution flow. Such an assumption can be achieved in \ac{mfg}s by sampling a large number of agents at each time, since all the agents are homogeneous and have the same state distribution flow. This estimation procedure will however, come at a cost of $O(1/\sqrt{N})$, which is not reflected in their results. What's more, such an assumption is no longer realistic in the \ac{gmfg}, since the agents in \ac{gmfg} are heterogeneous, and the state distributions of agents are different. Our estimation procedure in~\eqref{algo:estalgo} does not require the access to the nominal value of the distribution flow $\mu_{h}^{\calI}$. Instead, we estimate this quantity from states of sampled agents and prove that such an estimate works for the heterogeneous agents.

Next, we derive a corollary for the setting where we implement a single behavior policy for $L$ independent times to collect the data, i.e., $\pi_{\tau}^{\calI}=\pi^{\calI}$ for all $\tau\in[L]$. As such, instead of Eqn.~\eqref{eq:homega}, we estimate the mean-embedding via
\begin{align}
    \htomega_{\tau,h}^{i}(W)=\frac{1}{(N-1)L}\sum_{j\neq i}\sum_{\tau^{\prime}=1}^{L}W(\xi_{i},\xi_{j})k\big(\cdot,(s_{\tau,h}^{i},a_{\tau,h}^{i},s_{\tau^{\prime},h}^{j})\big).\label{eq:homega1}
\end{align}
We note that Eqn.~\eqref{eq:homega1} averages the states over $L$ episodes, since the distribution flows of these $L$ episodes are same. Correspondingly, the estimation procedure in Eqn.~\eqref{algo:estalgo} is modified to be
\begin{align}(\hatf_{h},\hatg_{h},\hatW_{h})=\!\!\!\!\argmin_{\substack{f\in\bbB(r,\bar{\calH}),g\in\bbB(\tilr,\tilde{\calH}),\\W\in\tilde{\calW}}} \!\!\frac{1}{NL}\!\sum_{\tau=1}^{L}\sum_{i=1}^{N}\!\Big(\!s_{\tau,h+1}^{i}-f\big(\htomega_{\tau,h}^{i}(W)\big)\!\Big)^{2}\!\!\!+\!\Big(\!r_{\tau,h}^{i}-g\big(\htomega_{\tau,h}^{i}(W)\big)\!\Big)^{2}.\label{algo:estalgo1}
\end{align}
The convergence rate of the corresponding estimates $(\hatf_{h},\hatg_{h},\hatW_{h})$ can be derived as follows.
\begin{corollary}\label{coro:fixest}
    Under Assumptions~\ref{assump:lipconti}, \ref{assump:graphon}, \ref{assump:kernel}, and \ref{assump:real}, if we implement a policy $L$ independent times to collect the data, and $\xi_{i}=i/N$ for $i\in[N]$, then with probability at least $1-\delta$, the risk of the estimates in Eqn.~\eqref{algo:estalgo1} can be bounded as
    \begin{align*}
        &\calR_{\bar{\xi}}(\hatf_{h},\hatg_{h},\hatW_{h})-\calR_{\bar{\xi}}(f_{h}^{*},g_{h}^{*},W_{h}^{*})\\*
        &\quad= O\bigg(\frac{(B_{S}+\barr\barB_{K})^{4}}{NL}\log\frac{N_{\bbB_{r}}N_{\tilde{\bbB}_{\tilr}}N_{\tilde{\calW}}}{\delta}+\frac{(B_{S}+\barr\barB_{K})\barr\barL_{K}B_{k}}{\sqrt{NL}}\log\frac{NL\calN_{\infty}(1/\sqrt{NL},\tilde{\calW})}{\delta}\bigg),
    \end{align*}
    where
    \begin{align*}
        N_{\bbB_{r}}=\calN_{\bar{\calH}}\bigg(\frac{3}{NL},\bbB(r,\bar{\calH})\bigg),\quad N_{\tilde{\bbB}_{\tilr}}=\calN_{\bar{\calH}}\bigg(\frac{3}{NL},\bbB(\tilr,\tilde{\calH})\bigg),\quad N_{\tilde{\calW}}=\calN_{\infty}\bigg(\frac{3}{L_{K}NL},\tilde{\calW}\bigg).
    \end{align*}
\end{corollary}
Compared to the result in Theorem~\ref{thm:fixest}, the mean-embedding estimation error, i.e., the second term, is improved from $O(1/\sqrt{N})$ to $O(1/\sqrt{NL})$. Such an improvement is intuitive, since we now utilize the data from $L$ episodes to estimate the distribution flow, but the estimation procedure in Theorem~\ref{thm:fixest} only uses the data from a \emph{single} episode for the same purpose.

\subsubsection{Agents with Known Random Positions}\label{sec:knownrandest}
In this section, we provide the convergence result of estimation procedure in Eqn.~\eqref{algo:estalgo} in the setting where the agent positions $\{\xi_{i}\}_{i=1}^{N}$ are known realizations of i.i.d.\ samples drawn from $\unif([0,1])$. The set of positions is denoted as $\bar{\xi}=\{\xi_{i}\}_{i=1}^{N}$. We first specify the performance metric in this section. For an agent $\alpha\in\calI$, we denote the joint distribution of $(s_{\tau,h}^{\alpha},a_{\tau,h}^{\alpha},\mu_{\tau,h}^{\calI},r_{\tau,h}^{\alpha},s_{\tau,h+1}^{\alpha})$ as $\rho_{\tau,h}^{\alpha}$, where $\rho_{\tau,h}^{\alpha}=\mu_{\tau,h}^{\alpha}\times\pi_{\tau,h}^{\alpha}\times\delta_{\mu_{\tau,h}^{\calI}}\times \delta_{r_{h}^{*}}\times P_{h}^{*}$. Then the risk of   $f\in\calH$, $g\in\tilde{\calH}$, and $W\in\tilde{\calW}$ is defined as
\begin{align}
    \calR(f,g,W)=\frac{1}{L}\sum_{\tau=1}^{L}\int_{0}^{1}\bbE_{\rho_{\tau,h}^{\alpha}}\bigg[\Big(s_{\tau,h+1}^{\alpha}-f\big(\omega_{\tau,h}^{\alpha}(W)\big)\Big)^{2}+\Big(r_{\tau,h}^{\alpha}-g\big(\omega_{\tau,h}^{\alpha}(W)\big)\Big)^{2}\bigg]\rmd \alpha. \label{eq:popurisk}
\end{align}
Compared to the risk with grid positions defined in Eqn.~\eqref{eq:condirisk}, the risk defined in Eqn.~\eqref{eq:popurisk} can be derived by taking expectation with respect the distribution of the positions, i.e., $\calR(f,g,W)=\bbE_{\bar{\xi}}\big[\calR_{\bar{\xi}}(f,g,W)\big]$. The convergence rate of our estimates can be stated as follows. 
\begin{theorem}\label{thm:randest}
     Under Assumptions~\ref{assump:graphon}, \ref{assump:kernel}, \ref{assump:real}, and \ref{assump:lipconti}, if $\{\xi_{i}\}_{i=1}^{N}$ are known i.i.d.\ samples of $\unif([0,1])$, then with probability at least $1-\delta$, the risk of the estimates in Eqn.~\eqref{algo:estalgo} can be bounded as
    \begin{align*}
        &\calR(\hatf_{h},\hatg_{h},\hatW_{h})-\calR(f_{h}^{*},g_{h}^{*},W_{h}^{*})\\
        &\quad= O\bigg(\frac{(B_{S}+\barr\barB_{K})^{2}}{\sqrt{N}}\log\frac{\tilN_{\bbB_{r}}\tilN_{\tilde{\bbB}_{\tilr}}\tilN_{\tilde{\calW}}}{\delta}\\
        &\quad\qquad+\frac{(B_{S}+\barr\barB_{K})\barr\barL_{K}B_{k}}{\sqrt{N}}\log\frac{NL\calN_{\infty}\big(1/\sqrt{N},\tilde{\calW}\big)}{\delta}+\frac{(B_{S}+\barr\barB_{K})^{4}}{NL}\log\frac{N_{\bbB_{r}}N_{\tilde{\bbB}_{\tilr}}N_{\tilde{\calW}}}{\delta}\bigg),
    \end{align*}
    where
    \begin{align*}
        \tilN_{\bbB_{r}}=\calN_{\bar{\calH}}\bigg(\!\frac{1}{16\sqrt{N}},\bbB(r,\bar{\calH})\!\bigg), \tilN_{\tilde{\bbB}_{\tilr}}=\calN_{\tilde{\calH}}\bigg(\!\frac{1}{16\sqrt{N}},\bbB(\tilr,\tilde{\calH})\!\bigg), \tilN_{\tilde{\calW}}=\calN_{\infty}\bigg(\!\frac{1}{16rL_{K}B_{k}\sqrt{N}},\tilde{\calW}\!\bigg).
    \end{align*}
\end{theorem}

The estimation error in Theorem~\ref{thm:randest} consists of three terms: the first term corresponds to the approximation error, the second term corresponds to the generalization error, and the third term corresponds to the mean-embedding estimation error. The first term comes from the fact that we can only approximate the risk $\calR(f,g,W)$ by $\calR_{\bar{\xi}}(f,g,W)$ in the estimation procedure specified in Eqn.~\eqref{algo:estalgo}. The second and the third terms can be explained in the same way as for the terms in Theorem~\ref{thm:fixest}.

\subsection{Learning from Sampled Agents with Unknown Positions}\label{sec:unknownfixest}
We now consider the setting where the positions of the sampled agents $\{\xi_{i}\}_{i=1}^{N}$ are on the grid in $[0,1]$, but are \emph{unknown}. This means that the set of sampled positions $\{\xi_{i}\}_{i=1}^{N}$ is equal to  $\{i/N\}_{i=1}^{N}$, but we do not know which $i/N$ each $\xi_{i}$ corresponds to. In addition to the data collection procedures in Section~\ref{sec:datacollect}, we assume that we implement the \emph{same} policy over $L$ independent rounds. This sampling method implies that the distribution defined in Section~\ref{sec:unknownfixest} satisfies  $\rho_{\tau,h}^{\alpha}=\rho_{\tau^{\prime},h}^{i}$ for all $\tau,\tau^{\prime}\in[L]$, $\alpha\in\calI$ and $h\in[H]$.

Intuitively, since the position information is missing from our observations, we cannot estimate the precise values of graphons. For example, the collected data from the agents in Figure~\ref{fig:relapos1} is same as that in Figure~\ref{fig:relapos2}, so we cannot distinguish between these two different graphons. However, we can see that these two graphons are the same up to a measure-preserving bijection. Proposition~\ref{prop:equivari} shows that the model with transformed graphons is same as the original model up to a measure-preserving bijection. Thus, in this section, our goal is to estimate the model of \ac{gmfg} up to a measure-preserving bijection.

\begin{figure}[t]
	\centering
	\subfigure[The \ac{sbm} graphon and four sampled agents.]{
	\begin{minipage}[t]{0.44\linewidth}
	\centering
	\includegraphics[width=1.5in]{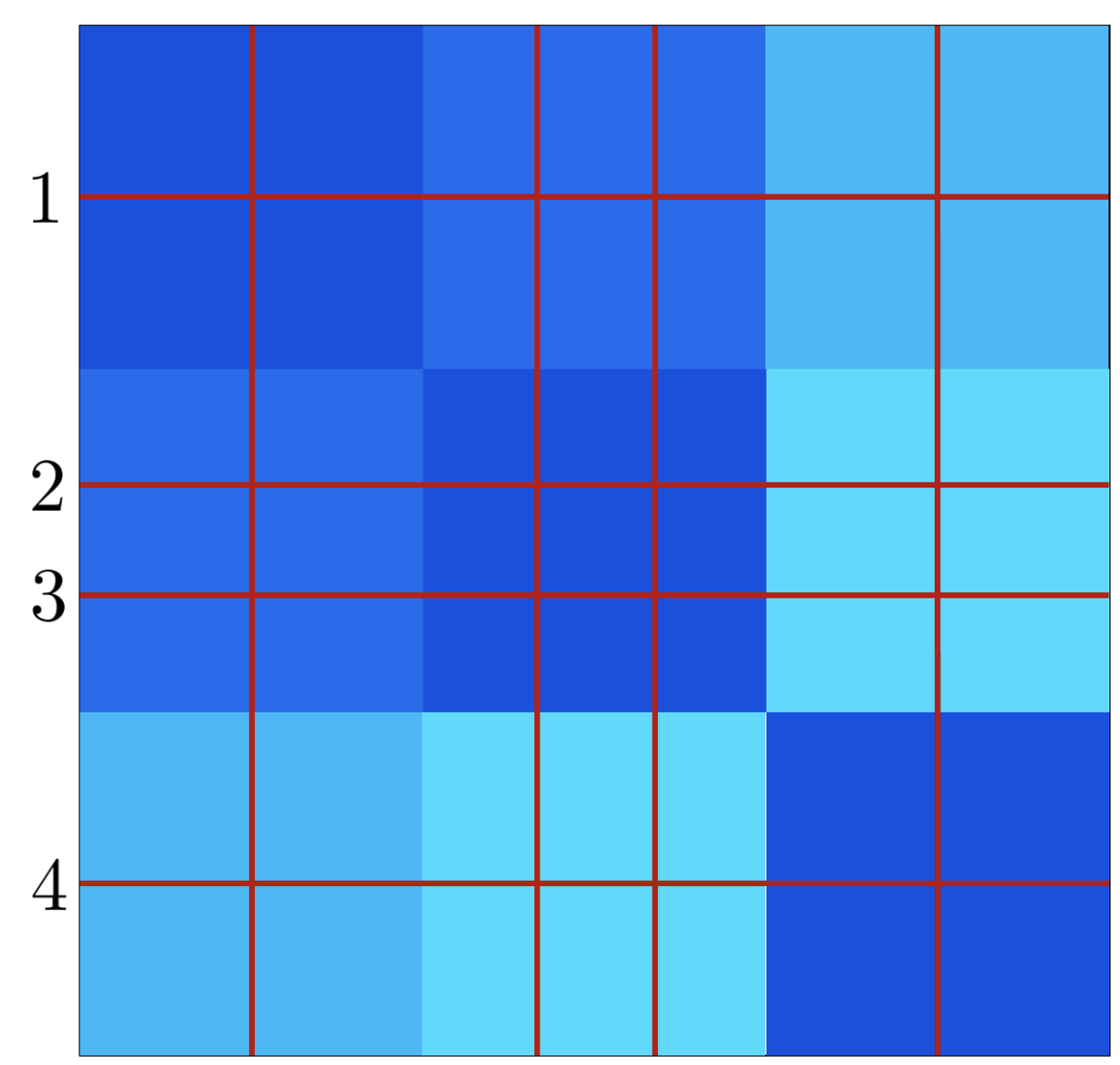}
	\end{minipage}%
	\label{fig:relapos1}
	}%
	\hspace{1cm}
	\centering
	\subfigure[Transformed \ac{sbm} graphon and the correspondingly transformed agents.]{
	\begin{minipage}[t]{0.44\linewidth}
	\centering
	\includegraphics[width=1.5in]{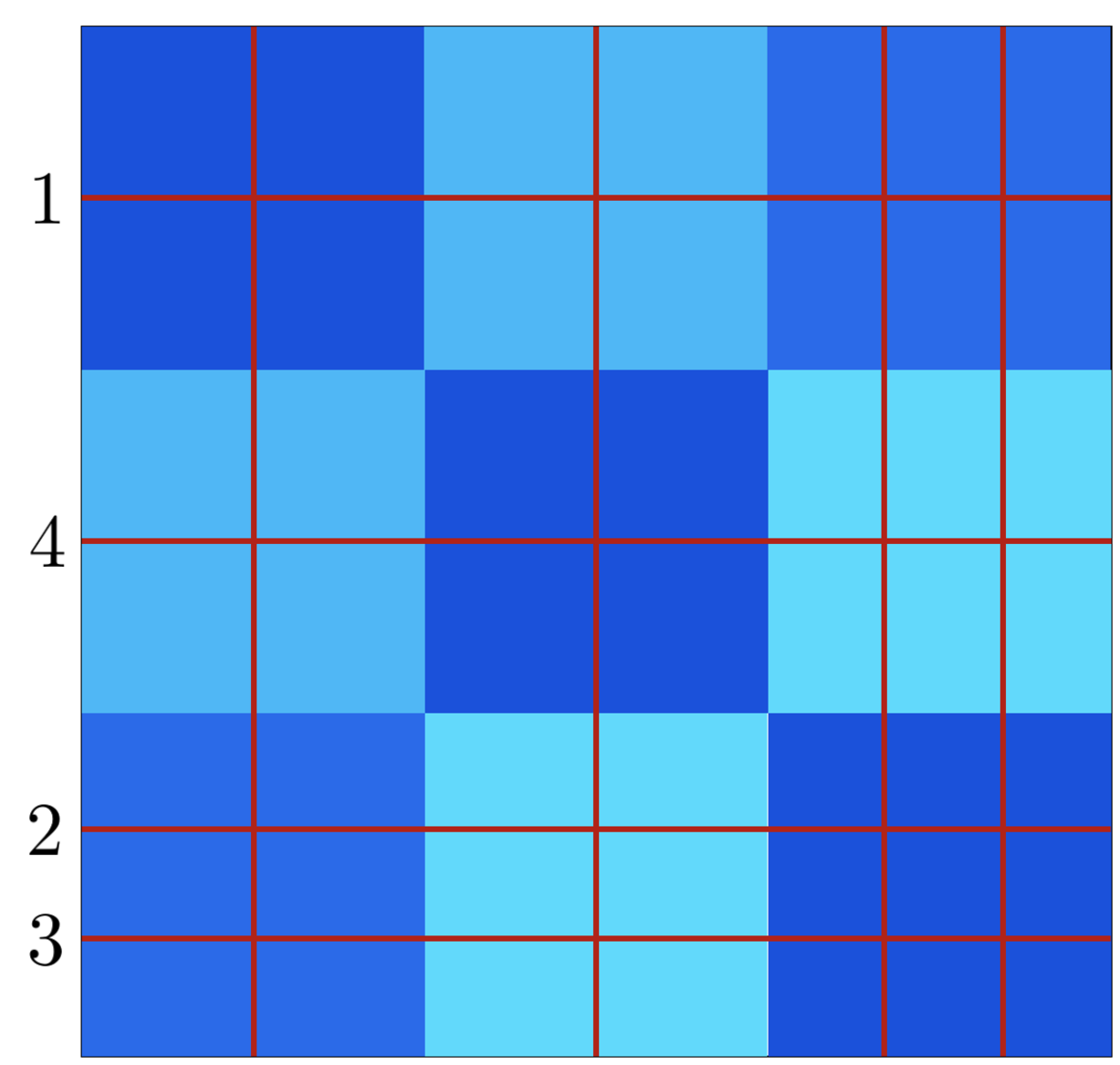}
	\end{minipage}%
	\label{fig:relapos2}
	}%
	\vspace{-0.3cm}
	\caption{The left figure shows the \ac{sbm} graphon and four sampled agents. Swapping the second and the third communities, we obtain the graphon on the right. The sampled agents are correspondingly swapped. Although the graphons and agent positions in the left and the right figures are not the same, the agents in both figures retain the same ``relative positions'' with the underlying graphons.}
	\label{fig:relapos}
\end{figure}

In this setting, we cannot estimate the mean-embedding $\omega_{\tau,h}^{i}(W_{h}^{*})$ as Eqn.~\eqref{eq:homega}, since we do not know the agents' positions $\{\xi_{i}\}_{i=1}^{N}$. Instead, we need to estimate the ``relative positions'' of these agents. Here the relative positions refer to the relationship between the agents' positions and the underlying graphon. For example, in Figure~\ref{fig:relapos}, the agents retain the same relative positions in different graphons. With $N$ sampled agents, the relative positions can be represented by the permutation of these agents. We denote the set of all the permutations of $N$ objects as $\calC^{N}$, where $|\calC^{N}|=N!$. For a permutation $\sigma\in\calC^{N}$ and a graphon $W$, we estimate the relative position of $i^{\rm{th}}$ agent as $\sigma(i)/N$ for all $i\in[N]$. Then mean-field embedding estimate can be derived as
\begin{align}
    \hbomega_{\tau,h}^{i,\sigma}(W)=\frac{1}{(N-1)L}\sum_{j\neq i}\sum_{\tau^{\prime}=1}^{L}W\bigg(\frac{\sigma(i)}{N},\frac{\sigma(j)}{N}\bigg)k\big(\cdot,(s_{\tau,h}^{i},a_{\tau,h}^{i},s_{\tau^{\prime},h}^{j})\big).\label{eq:hbomega}
\end{align}
Similar as Eqn.~\eqref{eq:homega1}, Eqn.~\eqref{eq:hbomega} is also an  average  over $L$ episodes, since we implement the same policy for $L$  independent times. In this estimate, only the relative positions between agents and the underlying graphon matters, so we can equivalently express such estimate with a transformed graphon. We define $\hbomega_{\tau,h}^{i}(W)$ as $\hbomega_{\tau,h}^{i,\sigma}(W)$ with the identity map $\sigma$. The set of measure-preserving bijections that are permutations of the intervals $[(i-1)/N,i/N]$ for $i\in[N]$ is denoted as $\calC_{[0,1]}^{N}$. Then for some $\phi\in\calC_{[0,1]}^{N}$, the estimate in Eqn.~\eqref{eq:hbomega} can be reformulated as
\begin{align*}
    \hbomega_{\tau,h}^{i}(W^{\phi})=\frac{1}{(N-1)L}\sum_{j\neq i}\sum_{\tau^{\prime}=1}^{L}W\Big(\phi\big(i/N\big),\phi\big(j/N\big)\Big)k\big(\cdot,(s_{\tau,h}^{i},a_{\tau,h}^{i},s_{\tau^{\prime},h}^{j})\big).
\end{align*}
Given this mean-embedding estimate, our model estimation estimation procedure can be stated as
\begin{align}
    (\hatf_{h},\hatg_{h},\hatW_{h},\hphi_{h})\!=\!\!\argmin_{\substack{f\in\bbB(r,\bar{\calH}), \\ g\in\bbB(\tilr,\tilde{\calH}),\\W\in\tilde{\calW},\phi\in\calC_{[0,1]}^{N}}}\! \frac{1}{NL}\sum_{\tau=1}^{L}\sum_{i=1}^{N} \Big(s_{\tau,h+1}^{i}-f\big(\hbomega_{\tau,h}^{i}(W^{\phi})\big) \Big)^{2}+\Big(\!r_{\tau,h}^{i} - g\big( \hbomega_{\tau,h}^{i}(W^{\phi})\big)\Big)^{2}\!\!\!.\label{algo:inestalgo2}    
\end{align} 
We then specify the performance metric under this setting. As mentioned earlier, we cannot estimate the precise values of graphons. Thus, we measure the accuracy of our estimates by transforming the graphon estimate with the optimal measure-preserving bijections. Such a risk is known as the \emph{permutation-invariant risk}, which is defined as
\begin{align}
    \bar{\calR}_{\bar{\xi}}(f,g,W)&=\inf_{\phi\in\calB_{[0,1]}}\frac{1}{NL}\sum_{\tau=1}^{L}\sum_{i=1}^{N}\bbE_{\rho_{\tau,h}^{i}}\bigg[\Big(s_{\tau,h+1}^{i}-f\big(\omega_{\tau,h}^{i}(W^{\phi})\big)\Big)^{2}+\Big(r_{\tau,h}^{i}-g\big(\omega_{\tau,h}^{i}(W^{\phi})\big)\Big)^{2}\bigg]\nonumber\\
    &=\inf_{\phi\in\calB_{[0,1]}}\frac{1}{N}\sum_{i=1}^{N}\bbE_{\rho_{h}^{i}}\bigg[\Big(s_{h+1}^{i}-f\big(\omega_{h}^{i}(W^{\phi})\big)\Big)^{2}+\Big(r_{\tau,h}^{i}-g\big(\omega_{\tau,h}^{i}(W^{\phi})\big)\Big)^{2}\bigg],\label{eq:pincondirisk}
\end{align}
where $\rho_{h}^{i}=\rho_{\tau,h}^{i}$ for all $\tau\in[L]$. The term ``permutation-invariant'' comes from the the analogy between the permutations and measure-preserving bijections and the fact that $\bar{\calR}_{\bar{\xi}}(f,g,W)=\bar{\calR}_{\bar{\xi}}(f,g,W^{\phi})$ for any $\phi\in\calB_{[0,1]}$. Our convergence guarantee of the estimation procedure can be stated as follows.
\begin{theorem}\label{thm:ufixest}
    Under Assumptions~\ref{assump:graphon}, \ref{assump:kernel}, \ref{assump:real}, and \ref{assump:lipconti}, if $\{\xi_{i}\}_{i=1}^{N}=\{i/N\}_{i=1}^{N}$, then with probability at least $1-\delta$, the risk of the estimates in Eqn.~\eqref{algo:inestalgo2} can be bounded as
    \begin{align*}
        &\bar{\calR}_{\bar{\xi}}(\hatf_{h},\hatg_{h},\hatW_{h})-\bar{\calR}_{\bar{\xi}}(f_{h}^{*},g_{h}^{*},W_{h}^{*})\nonumber\\
        &\quad=O\bigg(\underbrace{\frac{B_{k}\barr\barL_{K}(B_{S}+\barr\barB_{K})}{N}}_{\text{agent sampling error}}+\underbrace{(B_{S}+\barr\barB_{K})\barr\barL_{K}\barB_{k}\sqrt{\frac{N}{L}}\log\frac{NL \calN_{\infty}(\sqrt{N/L},\tilde{\calW})}{\delta}}_{\text{state sampling error}}\\
        &\quad\qquad+\underbrace{\frac{(B_{S}+\barr\barB_{K})^{4}}{L}\log\frac{N\tilN_{\bbB_{r}}\tilN_{\tilde{\bbB}_{\tilr}}\tilN_{\infty}}{\delta}}_{\text{generalization error}}\bigg),
    \end{align*}
    where
    \begin{align*}
        \tilN_{\bbB_{r}}=\calN_{\bar{\calH}}\bigg(\frac{3}{L},\bbB(r,\bar{\calH})\bigg),\quad\tilN_{\tilde{\bbB}_{\tilr}}=\calN_{\tilde{\calH}}\bigg(\frac{3}{L},\bbB(\tilr,\tilde{\calH})\bigg),\quad \tilN_{\tilde{\calW}}=\calN_{\infty}\bigg(\frac{3}{L_{K}L},\tilde{\calW}\bigg).
    \end{align*}
\end{theorem}
The estimation error in Theorem~\ref{thm:ufixest} consists of three terms: the first two terms correspond to the mean-embedding estimation error, and the last term corresponds to the generalization error. As mentioned in Section~\ref{sec:knownest}, the mean-embedding estimation error consists of agent sampling error and the state sampling error. The first term in the bound represents the agent sampling error. Since the distance between adjacent agents is $1/N$, this approximation error is of order $O(1/N)$. The second term represents the state sampling error. The term $\sqrt{N}$ in the numerator comes from the estimation of relative positions from $\calC^{N}$, whose size is $N!$, and the union bound among this set. The third term, which is the generalization error, also suffers from the union bound of $N!$ relative positions. Compared with Corollary~\ref{coro:fixest} in Section~\ref{sec:knownest}, the result in Theorem~\ref{thm:ufixest} suffers from a multiplicative factor $\log N!$. When the function classes are finite and $L=\Theta(N^{\beta})$ with $\beta>1$, the convergence rate in Theorem~\ref{thm:ufixest} is $O(\max\{N^{-(\beta-1)/2},N^{-1}\})$. In contrast, the convergence rate Corollary~\ref{coro:fixest} is $O(N^{-(\beta+1)/2})$.

Theorem~\ref{thm:ufixest} states the estimate error in the permutation-invariant risk. In fact, we can also derive the convergence rate of our estimation of relative positions $\hphi_{h}$. This means that for some unknown correction $\psi^{*}\in \calC_{[0,1]}^{N}$, the risk defined in Eqn.~\eqref{eq:condirisk} of our estimate $(\hatf_{h},\hatg_{h},\hatW_{h}^{\hphi_{h}\circ\psi^{*}})$ vanishes.
\begin{corollary}\label{coro:learnequi}
    Given $\{\xi_{i}\}_{i=1}^{N}=\{i/N\}_{i=1}^{N}$, we adopt $\psi^{*}\in \calC_{[0,1]}^{N}$ to denote the mapping that $\psi^{*}(\xi_{i})=i/N$ for all $i\in[N]$. Under Assumptions~\ref{assump:graphon}, \ref{assump:kernel}, \ref{assump:real}, and \ref{assump:lipconti},  the risk of estimate can be bounded as
    \begin{align*}
        &\calR_{\bar{\xi}}(\hatf_{h},\hatg_{h},\hatW_{h}^{\hphi_{h}\circ\psi^{*}})-\calR_{\bar{\xi}}(f_{h}^{*},g_{h}^{*},W_{h}^{*})\nonumber\\
        &\quad=O\bigg(\frac{B_{k}\barr\barL_{K}(B_{S}+\barr\barB_{K})}{N}+(B_{S}+\barr\barB_{K})\barr\barL_{K}\barB_{k}\sqrt{\frac{N}{L}}\log\frac{NL \calN_{\infty}(\sqrt{N/L},\tilde{\calW})}{\delta}\\
        &\quad\qquad+\frac{(B_{S}+\barr\barB_{K})^{4}}{L}\log\frac{N\tilN_{\bbB_{r}}\tilN_{\tilde{\bbB}_{\tilr}}\tilN_{\infty}}{\delta}\bigg)
    \end{align*}
    with probability at least $1-\delta$.
\end{corollary}
Combined with Proposition~\ref{prop:equivari}, Corollary~\ref{coro:learnequi} shows that the model estimate $(\hatf_{h},\hatg_{h},\hatW_{h}^{\hphi_{h}})$ converges to the nominal model in the sense that they are shown to be equivalent up to an unknown  measure-preserving bijection $\psi^{*}\in \calC_{[0,1]}^{N}$.

\section{Combination of Optimization and Estimation Results}\label{sec:combine}
In this section, we make use of the estimator we constructed and analyzed in Section~\ref{sec:modellearn} to derive estimates in Step 4 in Algorithm~\ref{algo:MDcontract}. We assume that one has access to a population simulator; this assumption is commonly made in the \ac{mfg} literature~ \citep{guo2019learning,anahtarci2019fitted,anahtarci2022q}. This simulator is able to generate data according to two types of requests: (i) implement policy $\pi^{\calI}$ on the \ac{mdp} induced by a pre-specified distribution flow $\mu^{\calI}$, (ii) implement policy $\pi^{\calI}$ directly. In the latter case, the \ac{mdp} is induced by the distribution flow of $\pi^{\calI}$ itself. 

\begin{algorithm}[t]
	\caption{Estimation of $\hmu_{t}^{\calI}$, $\hbmu_{t+1}^{\calI}$, and $\hatQ_{h}^{\lambda,\alpha}(s,a,\pi_{t}^{\alpha},\hbmu_{t}^{\calI},\hatW)$}
	\textbf{Inputs:} the current policy $\pi_{t}^{\calI}$ and the past distribution flow estimate $\hbmu_{t}^{\calI}$\\
	\textbf{Outputs:} $\hmu_{t}^{\calI}$, $\hbmu_{t+1}^{\calI}$, and $\hatQ_{h}^{\lambda,\alpha}(\cdot,\cdot,\pi_{t}^{\alpha},\hbmu_{t}^{\calI},\hatW)$ for all $h\in[H]$, $\alpha\in\calI$\\
	\textbf{Procedure:}
	\begin{algorithmic}[1]\label{algo:gammaqest}
        \STATE Implement policy $\pi_{t}^{\calI}$ for $L$ times and collect the data $\{\calD_{\tau}\}_{\tau=1}^{L}$ (with any kind of sampled agents in Section~\ref{sec:modellearn})
        \STATE Derive the \ac{mdp} estimate $(\hatP,\hatr,\hatW)$ with the estimation procedures in Section~\ref{sec:modellearn}, where $\hatP=\{\hatP_{h}\}_{h=1}^{H}$, $\hatP=\{\hatP_{h}\}_{h=1}^{H}$, $\hatr=\{\hatr_{h}\}_{h=1}^{H}$, and $\hatW=\{\hatW_{h}\}_{h=1}^{H}$
        \STATE Derive $\hmu_{t}^{\calI}$ as the distribution flow of implementing $\pi_{t}^{\calI}$ on the \ac{mdp} estimate.
        \STATE Derive $\hbmu_{t+1}^{\calI}$ as $\hbmu_{t+1}^{\calI}=(1-\alpha_{t})\hbmu_{t}^{\calI}+\alpha_{t}\hmu_{t}^{\calI}$.
        \STATE Implement a behavior policy $\pi_{t}^{\rmb,\calI}$ on the \ac{mdp} induced by $\hbmu_{t}^{\calI}$ for $L$ times and collect the data $\{\calD_{\tau}^{\prime}\}_{\tau=1}^{L}$ (with any kind of sampled agents in Section~\ref{sec:modellearn})
        \STATE Derive the \ac{mdp} estimate $(\hatP^{\prime},\hatr^{\prime},\hatW^{\prime})$ with the estimation procedures in Section~\ref{sec:modellearn}
        \STATE Derive $\hatQ_{h}^{\lambda,\alpha}(\cdot,\cdot,\pi_{t}^{\alpha},\hbmu_{t}^{\calI},\hatW)$ as action-value functions of $\pi_{t}^{\calI}$ on the \ac{mdp} estimate $(\hatP^{\prime},\hatr^{\prime},\hatW^{\prime})$.
	\end{algorithmic}
\end{algorithm}
We use Algorithm~\ref{algo:gammaqest} to derive the distribution flow and action-value function estimate in Line 4 of Algorithm~\ref{algo:MDcontract}. In this algorithm, we call the simulator twice. First, we directly implement the policy $\pi_{t}^{\calI}$ for $L$ times independently. With the collected data, we can estimate the distribution flow $\mu_{t}^{\calI}$. Second, we implement a behavior policy $\pi_{t}^{\rmb,\calI}$ on the \ac{mdp} induced by $\hbmu_{t}^{\calI}$ for $L$ times. Then estimate the action-value functions with the collected data.

One natural question is that why we need to estimate the transition kernels and underlying graphons to estimate $\mu_{t}^{\calI}$. An alternative is to implement $\pi_{t}^{\calI}$ for $L$ times and estimate the distribution flow of the sampled agents as their empirical distribution. In fact, the convergence rate of the alternative will be $O(1/\sqrt{L})$ from central limit theorem. However, our estimate will shown to have risk bounded by $O(1/\sqrt{NL})$. This improvement is because our algorithm makes use of the information of all the agents, but the alternative only uses the information of single agent for the estimation.

To derive the theoretical guarantees on the accuracy of the distribution flow and action-value function estimates, we make the following assumptions.
\begin{assumption}\label{assump:noise}
    There exist $L_{\varepsilon}>0$ such that the noises $\varepsilon_{h}$ for $h\in[H]$ satisfy that for any $a\in\bbR$, $\tv(\varepsilon_{h}+a,\varepsilon_{h})\leq L_{\varepsilon}a$ for all $h\in[H]$.
\end{assumption}
This assumption enables us to control the total variation error of our transition kernels $P_{h}^{*}$ by the estimation error of $f_{h}^{*}$. We note that Assumption~\ref{assump:noise} is satisfied for a wide range of distributions, including the uniform distribution, the centralized Beta distributions for $\alpha>1,\beta>1$, and the truncated Gaussian distribution. We then assume that the behavior policy $\pi_{t}^{\rmb,\calI}$ satisfies the following assumptions.
\begin{assumption}\label{assump:policycover1}
    There exist two constants $C_{\pi},C_{\pi}^{\prime}>0$ such that for all $t\in[T]$
    \begin{align*}
        \sup_{s\in\calS,a\in\calA,\alpha\in\calI,h\in[H]}\frac{\bar{\pi}_{t,h}^{*,\alpha}(a\,|\,s)}{\pi_{t,h}^{\rmb,\alpha}(a\,|\,s)}\leq C_{\pi} \quad \mbox{and}\quad\sup_{s\in\calS,a\in\calA,\alpha\in\calI,h\in[H]}\frac{\pi_{t+1,h}^{\alpha}(a\,|\,s)}{\pi_{t,h}^{\rmb,\alpha}(a\,|\,s)}\leq C_{\pi}^{\prime}.
    \end{align*}
\end{assumption}
This assumption states that the behavior policy should explore the actions of the \ac{ne} and the policy $\pi_{t+1}^{\calI}$. It is quite natural since we want to estimate the action-value function of $\pi_{t+1}^{\calI}$ from the data collected by $\pi_{t}^{\rmb,\calI}$. Similar assumptions have been commonly made in the off-policy evaluation literature~\citep{kallus2021optimal,uehara2020minimax}.
\begin{assumption}\label{assump:policycover2}
    For any policy $\pi^{\calI}\in\tilde{\Pi}$, we define $\mu^{+,\calI}=\Gamma_{3}(\pi^{\calI},\barmu_{t}^{\calI},W^{*})$. We also define $\barmu_{t}^{\rmb,\calI}=\Gamma_{3}(\pi_{t}^{\rmb,\calI},\barmu_{t}^{\calI},W^{*})$. There exists a constant $C_{\pi}^{\prime\prime}>0$ such that for any $t\in[T]$ and any policy $\pi^{\calI}$ specified above, we have
    \begin{align*}
        \sup_{s\in\calS,h\in[H],\alpha\in\calI}\frac{\mu_{h}^{+,\alpha}(s)}{\barmu_{t,h}^{\rmb,\alpha}(s)}\leq C_{\pi}^{\prime\prime}.
    \end{align*}
\end{assumption}
This assumption states that the behavior policy should be sufficiently exploratory such that the induced distribution of other policies can be covered by that of the behavior policy. Similar assumptions haven been made in the policy optimization literatures~\citep{shani2020adaptive,agarwal2020optimality}. We note that if we take the behavior policy $\pi_{t}^{\rmb,\calI}=\unif(\calA)^{\calI\times H}$ to be the uniform distribution on the action space, then the constants in Assumptions~\ref{assump:policycover1} and \ref{assump:policycover2} can be set as $C_{\pi}=C_{\pi}^{\prime}=|\calA|$ and $C_{\pi}^{\prime\prime}=|\calA|^{H}$.

\subsection{Known-position Case}\label{sec:knownopt}
In this section, we analyze Algorithm~\ref{algo:MDcontract} and Algorithm~\ref{algo:gammaqest} when we know the positions (grid or random) of the sampled agents. In Algorithm~\ref{algo:gammaqest}, we know the distribution flow $\hbmu_{t}^{\calI}$ during our second call of the simulator. Thus, in Line 6 of Algorithm~\ref{algo:gammaqest}, we estimate the model from the collected and the precise value of the distribution flows. This estimation procedure can be acquired by simplifying the estimation procedure in Section~\ref{sec:knownfixest} as
\begin{align}
    (\hatf_{h},\hatg_{h},\hatW_{h})=\argmin_{\substack{f\in\bbB(r,\bar{\calH}), \\ g\in\bbB(\tilr,\tilde{\calH}), \\ W\in\tilde{\calW}}} \frac{1}{NL}\sum_{\tau=1}^{L}\sum_{i=1}^{N}\Big(s_{\tau,h+1}^{i}-f\big(\omega_{\tau,h}^{i}(W)\big)\Big)^{2}+\Big(r_{\tau,h}^{i}-g\big(\omega_{\tau,h}^{i}(W)\big)\Big)^{2},\label{eq:fixdistest}
\end{align}
where $\omega_{\tau,h}^{i}(W)$ is the mean-embedding calculated by Eqn.~\eqref{eq:meanembed} and the known distribution flow. Then the result for the agents with known grid positions is stated as
\begin{corollary}\label{coro:optknownfix}
    If we sample agents with known grid positions and adopt Algorithms~\eqref{algo:estalgo} and \eqref{eq:fixdistest} to estimate the \ac{mdp}, then under Assumptions~\ref{assump:noise}, \ref{assump:policycover1}, \ref{assump:policycover2} and assumptions in Theorems~\ref{thm:contractopt} and \ref{thm:fixest}, we have that \ac{gmfgppo} return the following estimates with probability at least $1-\delta$
    \begin{align*}
        &D\bigg(\frac{1}{T}\sum_{t=1}^{T}\pi_{t}^{\calI},\pi^{*,\calI}\bigg)+d\bigg(\frac{1}{T}\sum_{t=1}^{T}\hbmu_{t}^{\calI},\mu^{*,\calI}\bigg)\\
        &\quad=O\bigg(\frac{B_{S}+\barr\barB_{K}}{(NL)^{1/4}}\log^{1/4}\frac{TN_{\bbB_{r}}N_{\tilde{\bbB}_{\tilr}}N_{\tilde{\calW}}}{\delta}+\frac{(B_{S}+\barr\barB_{K})^{1/4}(\barr\barL_{K}B_{k})^{1/4}}{(NL)^{1/8}}\log^{1/4}\frac{TNL\calN_{\infty}(1/\sqrt{N},\tilde{\calW})}{\delta}\\
        &\quad\qquad+\frac{(B_{S}+\barr\barB_{K})^{1/4}(B_{S}+\barr\barB_{K}+\barr\barL_{K}B_{k})^{1/4}}{N^{1/4}}\bigg)+O\bigg(\frac{\sqrt{\log T}}{T^{1/3}}\bigg).
    \end{align*}
\end{corollary}
The error of learning \ac{ne} consists of two types of terms. The first originates from the estimation error of the distribution flow and the action-value function. It involves the number of sampled agents $N$ and the number of episodes $L$. The second represents the optimization error and  involves the number of iterations $T$. Consider the case where the function classes are finite. To learn a \ac{ne} with error $\varepsilon$ measured according to $D(\cdot,\cdot)$ and $d(\cdot,\cdot)$, we can run Algorithms~\ref{algo:MDcontract} and \ref{algo:gammaqest} with $T=\tilO(\varepsilon^{-3})$ iterations and $O((NL)^{-1/8}+N^{-1/4})=\varepsilon$. The second condition can be achieved by several parameter settings, e.g., $L=1$, $N=O(\varepsilon^{-8})$ and $L=O(\varepsilon^{-4})$, $N=O(\varepsilon^{-4})$.

The result for the agents with known random positions is stated as follows.
\begin{corollary}\label{coro:optknownrand}
    If we sample agents with known random positions and adopt Algorithms~\eqref{algo:estalgo} and \eqref{eq:fixdistest} to estimate the \ac{mdp}, then under Assumptions~\ref{assump:noise}, \ref{assump:policycover1}, \ref{assump:policycover2} and assumptions in Theorems~\ref{thm:contractopt} and \ref{thm:randest}, we have that \ac{gmfgppo} return the following estimates with probability at least $1-\delta$
    \begin{align*}
        &D\bigg(\frac{1}{T}\sum_{t=1}^{T}\pi_{t}^{\calI},\pi^{*,\calI}\bigg)+d\bigg(\frac{1}{T}\sum_{t=1}^{T}\hbmu_{t}^{\calI},\mu^{*,\calI}\bigg)\\
        &\quad=O\bigg(\frac{B_{S}+\barr\barB_{K}}{(NL)^{1/4}}\log^{1/4}\frac{TN_{\bbB_{r}}N_{\tilde{\bbB}_{\tilr}}N_{\tilde{\calW}}}{\delta}+\frac{(B_{S}+\barr\barB_{K})^{1/4}(\barr\barL_{K}B_{k})^{1/4}}{(NL)^{1/8}}\log^{1/4}\frac{TNL\calN_{\infty}(1/\sqrt{NL},\tilde{\calW})}{\delta}\\
        &\quad\qquad+\frac{(B_{S}+\barr\barB_{K})^{1/2}}{N^{1/8}}\log^{1/4}\frac{\tilN_{\bbB_{r}}\tilN_{\tilde{\bbB}_{\tilr}}\tilN_{\tilde{\calW}}}{\delta}\bigg)+O\bigg(\frac{\sqrt{\log T}}{T^{1/3}}\bigg).
    \end{align*}
\end{corollary}
Similar as Corollary~\ref{coro:optknownfix}, error of learning \ac{ne} consists of the estimation error and the optimization error. To learn a \ac{ne} with error $\varepsilon$ measured according to $D(\cdot,\cdot)$ and $d(\cdot,\cdot)$, we can run Algorithms~\ref{algo:MDcontract} and \ref{algo:gammaqest} with $T=\tilO(\varepsilon^{-3})$ iterations and $N=O(\varepsilon^{-8})$ sampled agents.

\subsection{Unknown-position Case}\label{sec:unknownopt}
In this section, we analyze Algorithms~\ref{algo:MDcontract} and~\ref{algo:gammaqest} when we do not know the grid positions of the sampled agents. In Algorithm~\ref{algo:gammaqest}, we need to specify the policy $\pi_{t}^{\calI}$ and distribution flow $\hbmu_{t}^{\calI}$, which requires the  information of agents' positions. Thus, we additionally assume that for a specific agent $\alpha\in\calI$, we know which sampled agent is closest to $\alpha$ and the relative position to the closest sampled agent. This assumption holds in many realistic scenarios. For example, we may wish to find the \ac{ne} of the U.S.\ social welfare problem, which can be formulated as a \ac{gmfg}. The simulator can be a computer program that can simulate the influence of the policies $\pi^{\alpha}$ and the distributions of state $\mu^{\alpha}$ of the people in the U.S. In this case, there is one sampled person from each state, and we assume that each person knows which state she belongs to, i.e., which sampled person is the closest person to her. 

\begin{corollary}\label{coro:optunknownfix}
    If we sample agents with known grid positions and adopt Algorithms~\eqref{algo:estalgo} and \eqref{eq:fixdistest} to estimate the \ac{mdp}, then under Assumptions~\ref{assump:noise}, \ref{assump:policycover1}, \ref{assump:policycover2} and assumptions in Theorems~\ref{thm:contractopt} and \ref{thm:ufixest}, we have that \ac{gmfgppo} return the following estimates with probability at least $1-\delta$
    \begin{align*}
        &D\bigg(\frac{1}{T}\sum_{t=1}^{T}\pi_{t}^{\calI},\pi^{*,\calI}\bigg)+d\bigg(\frac{1}{T}\sum_{t=1}^{T}\hbmu_{t}^{\calI},\mu^{*,\calI}\bigg)\\
        &\quad=O\bigg(\frac{B_{k}\barr\barL_{K}(B_{S}+\barr\barB_{K})}{N^{1/4}}+\frac{(B_{S}+\barr\barB_{K})^{1/4}(\barr\barL_{K}\barB_{k}N)^{1/4}N^{1/8}}{L^{1/8}}\log^{1/4}\frac{NL \calN_{\infty}(\sqrt{N/L},\tilde{\calW})}{\delta}\\
        &\quad\qquad+\frac{B_{S}+\barr\barB_{K}}{L^{1/4}}\log^{1/4}\frac{N\tilN_{\bbB_{r}}\tilN_{\tilde{\bbB}_{\tilr}}\tilN_{\infty}}{\delta}\bigg)+O\bigg(\frac{\sqrt{\log T}}{T^{1/3}}\bigg).
    \end{align*}
\end{corollary}
Similar to Corollaries~\ref{coro:optknownfix} and \ref{coro:optknownrand}, the learning error in Corollary~\ref{coro:optunknownfix} consists of the estimation error and the optimization error. To learn a \ac{ne} with error $\varepsilon$ measured according to $D(\cdot,\cdot)$ and $d(\cdot,\cdot)$, we can run Algorithms~\ref{algo:MDcontract} and \ref{algo:gammaqest} with $T=\tilO(\varepsilon^{-3})$ iterations, $N=O(\varepsilon^{-4})$ sampled agents, and $L=O(\varepsilon^{-12})$ episodes.

\section{Experiments}\label{sec:expts}
In this section, we utilize simulations to demonstrate the importance of learning the underlying graphons, thus corroborating our theoretical results. We simulate our algorithm on the \ac{sis} problem, which is widely adopted in  previous \ac{gmfg} and \ac{mfg} works~\citep{cui2021learning,cui2021approximately}. The underlying graphons take the form of exp-graphons and \ac{sbm} graphons. The exp-graphons are defined as
\begin{align*}
    W_{\theta}^{\rm exp}(\alpha,\beta)=\frac{2\exp(\theta\cdot \alpha\beta)}{1+\exp(\theta\cdot \alpha\beta)}-1,
\end{align*}
which is parameterized by $\theta>0$. The details of the experiments are provided in Appendix~\ref{app:experiment}. Since we do not know the nominal value of the \ac{ne}, we adopt the notion of \emph{exploitability} to measure the closeness between a policy and the \ac{ne}. For a policy $\pi^{\calI}$ and its induced distribution flow $\mu^{\calI}$, the exploitability is defined as~\citep{fabian2022learning}
\begin{align*}
    \Delta(\pi^{\calI})=\int_{0}^{1}\sup_{\tilde{\pi}^{\alpha}\in\Pi^{H}}J^{\lambda,\alpha}(\tilde{\pi}^{\alpha},\mu^{\calI},W^{*})-J^{\lambda,\alpha}(\pi^{\alpha},\mu^{\calI},W^{*})\, \rmd\alpha.
\end{align*}

If we do not learn the underlying graphons, reasonable guesses for them would be constant graphons $W(\alpha,\beta)=p$ for all $\alpha,\beta\in\calI$, corresponding to the \ac{mfg}. In the simulations, we choose the constant $p$ to be  $0,0.5$ and $1$. These values model the cases from the independent agents to the most intensely interacting agents. To learn the system model, we sample $N=7$ agents with known positions. The number of episodes for data collection $L$ is set to $125$ and $500$. 
\begin{figure}[t]
	\centering
	\subfigure[\ac{sis} problem with exp-graphons]{
	\begin{minipage}[t]{0.4\linewidth}
	\centering
	\includegraphics[width=2.9in]{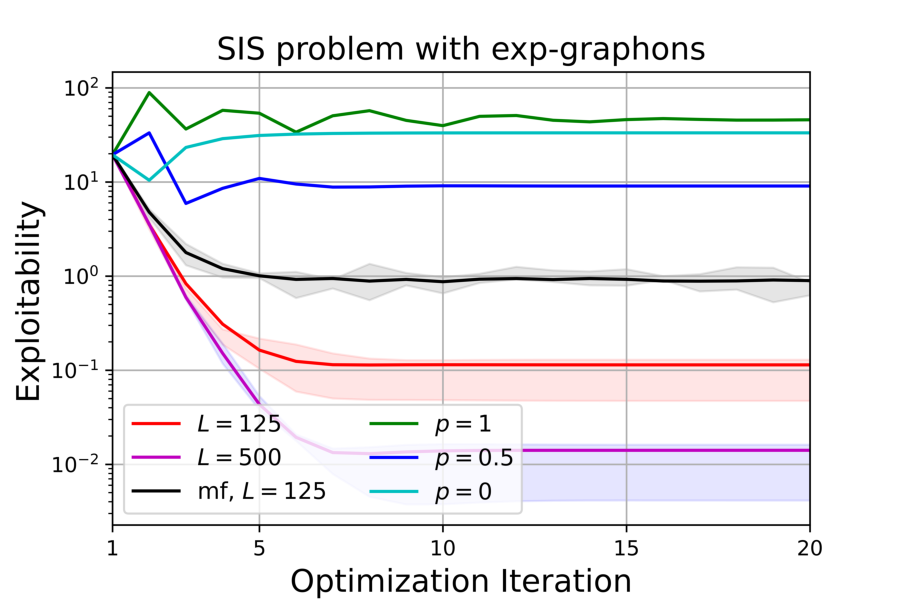}
	\end{minipage}%
	\label{fig:simu_exp}
	}%
	\hspace{1cm}
	\centering
	\subfigure[\ac{sis} problem with \ac{sbm} graphons.]{
	\begin{minipage}[t]{0.4\linewidth}
	\centering
	\includegraphics[width=2.9in]{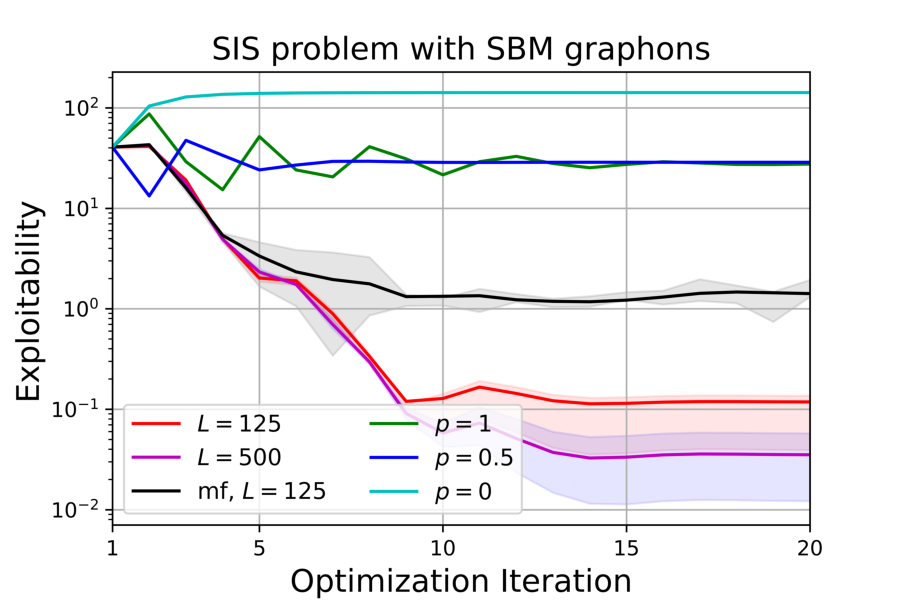}
	\end{minipage}%
	\label{fig:simu_sbm}
	}%
	\vspace{-0.3cm}
	\centering
	\caption{Simulation results for \ac{sis} problem with \ac{sbm} and exp-graphons.}
	\label{fig:simu}
\end{figure}

Figure~\ref{fig:simu} displays the exploitability for the algorithms in the \ac{sis} problem with different graphons. The line ``mf, $L=125$'' refers to the model-free algorithm in \cite{cui2021learning} that uses 125 trajectories for distribution flow and value function estimation in each round. The lines ``$L=125$'' and ``$L=500$'' refer to our algorithms that estimate with 125 and 500 samples in each round. Figure~\ref{fig:simu} demonstrates that our model-based algorithm achieves lower exploitability than the model-free algorithm. The reason is that the estimation error of model-based algorithm is smaller, as mentioned in Section~\ref{sec:combine}. Figure~\ref{fig:simu} also shows that when we assume that the heterogeneous agents are homogeneous, the learning algorithm for \ac{ne} will suffer from a large error (large exploitability). In contrast, learning the graphons will enable us to learn the \ac{ne} more accurately. These results demonstrate the necessity of our model learning algorithm in Algorithm~\ref{algo:gammaqest}. We can also observe that the learning error for $L=500$ is less than that for $L=125$, which corroborates Corollary~\ref{coro:optknownfix}.

\section{Conclusion}\label{sec:conclusion}
In this paper, we investigated learning the \ac{ne} of \ac{gmfg} in the graphons incognizant case. Provably efficient optimization algorithms were designed and analyzed with an estimation oracle, which improved on the previous works in convergence rate. In addition, adopting the mean-embedding ideas, we designed and analyzed the model-based estimation algorithms with sampled agents. Here, the sampled agents have known or unknown positions. These estimation algorithms feature as the first model-based algorithms in \ac{gmfg} without the distribution flow information. We leave the analysis of more complex agent sampling schemes for future works.

\bibliographystyle{ims}
\vskip 0.2in
\bibliography{ref}

\clearpage
\appendix

\begin{center}
{\Large {\bf Appendix for\\ ``Learning Graphon Mean-Field Games with Unknown Graphons''}}
\end{center}

\section{Experiment Details}\label{app:experiment}
In this section, we provide the details of our experiments shown in Section~\ref{sec:expts} of the main paper. We first define the susceptible–infected–susceptible (\ac{sis}) problem. The state space of this problem $\calS=\{S,I\}$ consists of the states $S$ (susceptible) and $I$ (infected). The action space $\calA=\{U,D\}$ consists of the actions  $U$ (going out) and $D$ (keeping distance). The horizon is $H=50$. The reward functions are defined as $r_{h}^{*}(s,a,z)=-10\cdot\mathbb{I}_{s=I}-2.5\cdot\mathbb{I}_{s=D}$ for all $h\in[H]$. The transition kernels are defined as
\begin{align*}
    P_{h}^{*}(S\,|\,I,\cdot,\cdot)=0.2,\quad P_{h}^{*}(I\,|\,S,U,z)=0.8\cdot z(I),\quad P_{h}^{*}(I\,|\,S,D,\cdot)=0 \text{ for all }h\in[H].
\end{align*}
We set the regularization parameter as $\lambda=1$. We set $\theta=3$ for exp-graphons in our simulation. For the \ac{sbm} graphon, we set the community number as $2$ and set the intra-community and inter-community interaction strengths as $0.9$ and $0.3$, respectively. We adopt exploitability as the performance metric. For a policy $\pi^{\calI}$ and its induced distribution flow $\mu^{\calI}$, the exploitability is defined as~\citep{fabian2022learning}
\begin{align}
\Delta(\pi^{\calI})=\int_{0}^{1}\sup_{\tilde{\pi}^{\alpha}\in\Pi^{H}}J^{\lambda,\alpha}(\tilde{\pi}^{\alpha},\mu^{\calI},W^{*})-J^{\lambda,\alpha}(\pi^{\alpha},\mu^{\calI},W^{*})\, \rmd\alpha. \label{eqn:exploit}
\end{align}

To shorten the simulation time and convey the main massage, we only estimate the model in the beginning of the first iteration round and reuse this estimate in the following iterations to generate action-value function estimates. Figure~\ref{fig:simu} is derived from five Monte-Carlo implementations of the algorithms. The error bar indicates the 25\% and the 75\% quantile of the errors. When simulating the cases with constant graphons, we implement the fixed point iteration to find the \ac{ne}, and the calculations of the optimal policy and the induced distribution flow are implemented via the dynamical programming and direct calculation with the nominal transition kernels and reward functions. Thus, there is no error bar for these cases.

The code used in our simulations uses the code in \cite{fabian2022learning} and \cite{cui2021learning} for building the simulation environment. We run our simulations on Intel(R) Core(TM) i5-8257U CPU @ 1.40GHz, and each Monte-Carlo experiment takes about ninety minutes.

\section{Proof of Proposition~\ref{prop:equivari}}
\begin{proof}[Proof of Proposition~\ref{prop:equivari}]
    We prove the desired results by induction on $h\in[H]$. When $h=1$, $\mu_{1}^{\phi(\alpha)}=\mu_{1}^{\phi,\alpha}$ holds trivially for all $\alpha\in\calI$. Assume that $\mu_{h}^{\phi(\alpha)}=\mu_{h}^{\phi,\alpha}$ holds for all $\alpha\in\calI$, then for $h+1$ and any $\alpha\in\calI$ we have that
    \begin{align*}
        \mu_{h+1}^{\phi(\alpha)}(s^{\prime})&=\sum_{a\in\calA}\int_{\calS}\mu_{h}^{\phi(\alpha)}(s)\pi_{h}^{\phi(\alpha)}(a\,|\,s)P_{h}^{*}(s^{\prime}\,|\,s,a,z_{h}^{\phi(\alpha)}(\mu_{h}^{\calI},W_{h}^{*}))\, \rmd s\qquad\mbox{and}\\
        \mu_{h+1}^{\phi,\alpha}(s^{\prime})&=\sum_{a\in\calA}\int_{\calS}\mu_{h}^{\phi,\alpha}(s)\pi_{h}^{\phi,\alpha}(a\,|\,s)P_{h}^{*}(s^{\prime}\,|\,s,a,z_{h}^{\alpha}(\mu_{h}^{\phi,\calI},W_{h}^{\phi,*}))\, \rmd s\\
        &=\sum_{a\in\calA}\int_{\calS}\mu_{h}^{\phi(\alpha)}(s)\pi_{h}^{\phi(\alpha)}(a\,|\,s)P_{h}^{*}(s^{\prime}\,|\,s,a,z_{h}^{\alpha}(\mu_{h}^{\phi,\calI},W_{h}^{\phi,*}))\, \rmd s,
    \end{align*}
    where the last equality results from the definition of $\pi^{\phi,\calI}$ and the hypothesis. To show that $\mu_{h+1}^{\phi(\alpha)}=\mu_{h+1}^{\phi,\alpha}$, it remains to show $z_{h}^{\phi(\alpha)}(\mu_{h}^{\calI},W_{h}^{*})=z_{h}^{\alpha}(\mu_{h}^{\phi,\calI},W_{h}^{\phi,*})$. In fact, we have that
    \begin{align*}
        z_{h}^{\alpha}(\mu_{h}^{\phi,\calI},W_{h}^{\phi,*})&=\int_{0}^{1}W_{h}^{*}(\phi(\alpha),\phi(\beta))\mu_{h}^{\phi(\beta)}\, \rmd\beta
        =\int_{0}^{1}W_{h}^{*}(\phi(\alpha),\gamma)\mu_{h}^{\gamma}\, \rmd\gamma,
    \end{align*}
    where the last equality results from setting $\gamma=\phi(\beta)$. Thus, we conclude the proof of Proposition~\ref{prop:equivari}.
\end{proof}

\section{Proof of Theorem~\ref{thm:contractopt}}
\begin{proof}[Proof of Theorem~\ref{thm:contractopt}]
    For the analysis of the Algorithm~\ref{algo:MDcontract}, we define the nominal distribution flows as
    \begin{align*}
        \mu_{t}^{\calI}=\Gamma_{2}(\pi_{t}^{\calI},W^{*}),\quad        \barmu_{t+1}^{\calI}=(1-\alpha_{t})\barmu_{t}^{\calI}+\alpha_{t}\mu_{t}^{\calI} \text{ for all }t\in[T].
    \end{align*}
    
    Our proof of Theorem~\ref{thm:contractopt} involves four distinct steps:
    \begin{itemize}
        \item First, we derive the first-order optimality condition of the policy $\hpi_{t+1}^{\calI}$ derived in Line~6 of Algorithm~\ref{algo:MDcontract}.
        \item Second, we derive the recurrence relationship of the policy learning error from the relationship the second step.
        \item Third, we derive the convergence rate of the learned mean-field.
        \item Finally, we obtain the desired result by combining step 2 and step 3.
    \end{itemize}
    
    \textbf{Step 1: Analyze the property of the policy $\hpi_{t+1}^{\calI}$ derived in Line 6 of Algorithm~\ref{algo:MDcontract} }
    
    We first note that the update of $\hpi_{t+1,h}^{\alpha}(\cdot\,|\,s)$ in Line 6 of Algorithm~\ref{algo:MDcontract} can be equivalently defined as
    \begin{align}
        \hpi_{t+1,h}^{\alpha}(\cdot\,|\,s)=\argmax_{p\in\Delta(\calA)}\eta_{t+1}\Big[\big\langle \hatQ_{h}^{\lambda,\alpha}(s,\cdot,\pi_{t}^{\alpha},\hbmu_{t}^{\calI},\hatW),p \big\rangle-\lambda \barh(p)\Big]-\kl\big(p\|\pi_{t,h}^{\alpha}(\cdot\,|\,s)\big),\label{eq:16}
    \end{align}
    which can be proved using Lagrangian multipliers.
    \begin{proposition}\label{prop:firstordopt}
        For the policy $\hpi_{t+1,h}^{\alpha}(\cdot\,|\,s)$, which is defined in Eqn.~\eqref{eq:16}, we have that for all $s\in\calS$, $p\in\Delta(\calA)$, and $h\in[H]$
        \begin{align*}
            &\eta_{t+1}\big\langle \hatQ_{h}^{\lambda,\alpha}(s,\cdot,\pi_{t}^{\alpha},\hbmu_{t}^{\calI},\hatW),p-\hpi_{t+1,h}^{\alpha}(\cdot\,|\,s) \big\rangle+\lambda \eta_{t+1}\Big[ R\big(\hpi_{t+1,h}^{\alpha}(\cdot\,|\,s)\big)-\barh(p)\Big]\\
            &\quad\leq \kl\big(p\|\pi_{t,h}^{\alpha}(\cdot\,|\,s)\big)-(1+\lambda\eta_{t+1})\kl\big(p\|\hpi_{t+1,h}^{\alpha}(\cdot\,|\,s)\big)-\kl\big(\hpi_{t+1,h}^{\alpha}(\cdot\,|\,s)\|\pi_{t,h}^{\alpha}(\cdot\,|\,s)\big).
        \end{align*}
    \end{proposition}
    \begin{proof}[Proof of Proposition~\ref{prop:firstordopt}]
        See Appendix~\ref{app:firstordopt}.
    \end{proof}
    Proposition~\ref{prop:firstordopt} shows that
    \begin{align}
        &\eta_{t+1}\big\langle Q_{h}^{\lambda,\alpha}(s_{h},\cdot,\pi_{t}^{\alpha},\barmu_{t}^{\calI},W^{*}),p-\pi_{t+1,h}^{\alpha}(\cdot\,|\,s_{h}) \big\rangle+\lambda \eta_{t+1}\Big[ R\big(\pi_{t+1,h}^{\alpha}(\cdot\,|\,s_{h})\big)-\barh(p)\Big]\nonumber\\
        &\quad\leq \kl\big(p\|\pi_{t,h}^{\alpha}(\cdot\,|\,s_{h})\big)-(1+\lambda\eta_{t+1})\kl\big(p\|\pi_{t+1,h}^{\alpha}(\cdot\,|\,s_{h})\big)-\kl\big(\pi_{t+1,h}^{\alpha}(\cdot\,|\,s_{h})\|\pi_{t,h}^{\alpha}(\cdot\,|\,s_{h})\big)\nonumber\\
        &\quad\qquad+\text{(I)}+\text{(II)}+\text{(III)}+\text{(IV)},\label{ieq:55}
    \end{align}
    where the term (I) is the combination of the action-value function estimation error and the difference between $\hpi_{t+1}^{\calI}$ and $\pi_{t+1}^{\calI}$ that is defined as
    \begin{align*}
        \text{(I)}=\eta_{t+1}\Big[\big\langle Q_{h}^{\lambda,\alpha}(s_{h},\cdot,\pi_{t}^{\alpha},\barmu_{t}^{\calI},W^{*}),p-\pi_{t+1,h}^{\alpha}(\cdot\,|\,s_{h}) \big\rangle-\big\langle \hatQ_{h}^{\lambda,\alpha}(s_{h},\cdot,\pi_{t}^{\alpha},\hbmu_{t}^{\calI},\hatW),p-\hpi_{t+1,h}^{\alpha}(\cdot\,|\,s_{h}) \big\rangle\Big].
    \end{align*}
    The term (II) is the entropy difference between $\hpi_{t+1}^{\calI}$ and $\pi_{t+1}^{\calI}$ that is defined as
    \begin{align*}
        \text{(II)}=\lambda\eta_{t+1}\Big(R\big(\pi_{t+1,h}^{\alpha}(\cdot\,|\,s_{h})\big)-R\big(\hpi_{t+1,h}^{\alpha}(\cdot\,|\,s_{h})\big)\Big).
    \end{align*}
    The term (III) is the KL divergence difference between $\hpi_{t+1}^{\calI}$ and $\pi_{t+1}^{\calI}$ that is defined as
    \begin{align*}
        \text{(III)}=\kl\big(\pi_{t+1,h}^{\alpha}(\cdot\,|\,s_{h})\|\pi_{t,h}^{\alpha}(\cdot\,|\,s_{h})\big)-\kl\big(\hpi_{t+1,h}^{\alpha}(\cdot\,|\,s_{h})\|\pi_{t,h}^{\alpha}(\cdot\,|\,s_{h})\big).
    \end{align*}
    The term (IV) is also the KL divergence difference between $\hpi_{t+1}^{\calI}$ and $\pi_{t+1}^{\calI}$ that is defined as
    \begin{align*}
        \text{(IV)}=(1+\lambda\eta_{t+1})\Big[\kl\big(p\|\pi_{t+1,h}^{\alpha}(\cdot\,|\,s_{h})\big)-\kl\big(p\|\hpi_{t+1,h}^{\alpha}(\cdot\,|\,s_{h})\big)\Big].
    \end{align*}
    We define
    \begin{align}
        \Lambda_{t+1,h}^{\alpha}&=2\eta_{t+1}\big\|Q_{h}^{\lambda,\alpha}(s_{h},\cdot,\pi_{t}^{\alpha},\hbmu_{t}^{\calI},W^{*})-\hatQ_{h}^{\lambda,\alpha}(s_{h},\cdot,\pi_{t}^{\alpha},\hbmu_{t}^{\calI},\hatW)\big\|_{\infty}+2\eta_{t+1}H(1+\lambda\log|\calA|)\beta_{t+1}\nonumber\\
        &\qquad+2\eta_{t+1}\big[L_{r}+H(1+\lambda\log|\calA|)L_{P}\big]\varepsilon_{\mu}+2\beta_{t+1}\log \frac{|\calA|}{\beta_{t}}+2(1+\lambda\eta_{t+1})\beta_{t+1},\label{eq:17}
    \end{align}
    Then we can show the following bound.
    \begin{proposition}\label{prop:esterrsum}
            Under assumptions in Theorem~\ref{thm:contractopt}, $\text{(I)}+\text{(II)}+\text{(III)}+\text{(IV)}\leq\Lambda_{t+1,h}^{\alpha}$.
    \end{proposition}
    \begin{proof}
        See Appendix~\ref{app:esterrsum}.
    \end{proof}
    Then inequality~\eqref{ieq:55} shows that
    \begin{align}
        &\eta_{t+1}\big\langle Q_{h}^{\lambda,\alpha}(s_{h},\cdot,\pi_{t}^{\alpha},\barmu_{t}^{\calI},W^{*}),p-\pi_{t+1,h}^{\alpha}(\cdot\,|\,s_{h}) \big\rangle+\lambda \eta_{t+1}\Big[ R\big(\pi_{t+1,h}^{\alpha}(\cdot\,|\,s_{h})\big)-\barh(p)\Big]\nonumber\\
        &\quad\qquad+\kl\big(\pi_{t+1,h}^{\alpha}(\cdot\,|\,s_{h})\|\pi_{t,h}^{\alpha}(\cdot\,|\,s_{h})\big)\nonumber\\
        &\quad\leq \kl\big(p\|\pi_{t,h}^{\alpha}(\cdot\,|\,s_{h})\big)-(1+\lambda\eta_{t+1})\kl\big(p\|\pi_{t+1,h}^{\alpha}(\cdot\,|\,s_{h})\big)+\Lambda_{t+1,h}^{\alpha}.\label{ieq:56}
    \end{align}
    
    \textbf{Step 2: Derive the recurrence relationship of the policy learning error from the relationship the second step, and bound the dynamical error in such recurrence relationship.}
    
    Inequality~\eqref{ieq:56} implies that the improvement of $\pi_{t+1}^{\calI}$ of the \ac{mdp} induced by $\barmu_{t}^{\calI}$ over $\pi_{t}^{\calI}$ can be lower bounded as
    \begin{align}
        &V_{m}^{\lambda,\alpha}(s,\pi_{t+1}^{\alpha},\barmu_{t}^{\calI},W^{*})-V_{m}^{\lambda,\alpha}(s,\pi_{t}^{\alpha},\barmu_{t}^{\calI},W^{*})\nonumber\\
        &\quad= \bbE_{\pi_{t+1}^{\alpha},\barmu_{t}^{\calI}}\bigg[\sum_{h=m}^{H}\big\langle Q_{h}^{\lambda,\alpha}(s_{h},\cdot,\pi_{t}^{\alpha},\barmu_{t}^{\calI},W^{*}),\pi_{t+1,h}^{\alpha}(\cdot\,|\,s_{h})-\pi_{t,h}^{\alpha}(\cdot\,|\,s_{h})\big\rangle\nonumber\\
        &\quad\qquad+\lambda\Big[R\big(\pi_{t,h}^{\alpha}(\cdot\,|\,s_{h})\big)- R\big(\pi_{t+1,h}^{\alpha}(\cdot\,|\,s_{h})\big)\Big] \,\bigg|\,s_{m}=s\bigg]\nonumber\\
        &\quad\geq \big\langle Q_{m}^{\lambda,\alpha}(s,\cdot,\pi_{t}^{\alpha},\barmu_{t}^{\calI},W^{*}),\pi_{t+1,m}^{\alpha}(\cdot\,|\,s)-\pi_{t,m}^{\alpha}(\cdot\,|\,s)\big\rangle+\lambda\Big[R\big(\pi_{t,h}^{\alpha}(\cdot\,|\,s)\big)- R\big(\pi_{t+1,h}^{\alpha}(\cdot\,|\,s)\big)\Big]\nonumber\\
        &\quad\qquad-\frac{1}{\eta_{t+1}}\bbE_{\pi_{t+1}^{\alpha},\barmu_{t}^{\calI}}\bigg[\sum_{h=m}^{H}\Lambda_{t+1,h}^{\alpha}\,\bigg|\,s_{m}=s\bigg],\label{ieq:57}
    \end{align}
    where the equality results from Lemma~\ref{lem:pdl}, and the inequality results from inequality~\eqref{ieq:56} and that KL divergence is non-negative.
    
    We denote the optimal policy on the \ac{mdp} induced by $\barmu_{t}^{\calI}$ as $\barpi_{t}^{*,\calI}=\Gamma_{1}^{\lambda}(\barmu_{t}^{\calI},W^{*})$. Then Lemma~\ref{lem:pdl} and inequality~\eqref{ieq:57} implies that
    \begin{align}
        &\eta_{t+1}\bbE_{\barpi_{t}^{*,\alpha},\barmu_{t}^{\calI}}\bigg[\sum_{h=1}^{H}\big\langle Q_{h}^{\lambda,\alpha}(s_{h},\cdot,\pi_{t}^{\alpha},\barmu_{t}^{\calI},W^{*}),\barpi_{t,h}^{*,\alpha}(\cdot\,|\,s_{h})-\pi_{t+1,h}^{\alpha}(\cdot\,|\,s_{h})\big\rangle\nonumber\\
        &\qquad\qquad\qquad\qquad\qquad\qquad\qquad\qquad\qquad\qquad +\lambda\Big[R\big(\pi_{t+1,h}^{\alpha}(\cdot\,|\,s_{h})\big)-R\big(\barpi_{t,h}^{*,\alpha}(\cdot\,|\,s_{h})\big)\Big]\bigg]\nonumber\\
        &\quad\geq\eta_{t+1}\bbE_{\barpi_{t}^{*,\alpha},\barmu_{t}^{\calI}}\bigg[\sum_{h=1}^{H}V_{h}^{\lambda,\alpha}(s_{h},\pi_{t}^{\alpha},\barmu_{t}^{\calI},W^{*})-V_{h}^{\lambda,\alpha}(s_{h},\pi_{t+1}^{\alpha},\barmu_{t}^{\calI},W^{*})\bigg]\nonumber\\
        &\quad\qquad-\bbE_{\barpi_{t}^{*,\alpha},\barmu_{t}^{\calI}}\bigg[\sum_{h=1}^{H}\bbE_{\pi_{t+1}^{\alpha},\barmu_{t}^{\calI}}\bigg[\sum_{m=h}^{H}\Lambda_{t+1,m}^{\alpha}\,\bigg|\,s_{h}\bigg]\bigg]\nonumber\\
        &\quad\qquad+\eta_{t+1}\bbE_{\mu_{1}^{\alpha}}\big[V_{1}^{\lambda,\alpha}(s_{1},\barpi_{t}^{*,\alpha},\barmu_{t}^{\calI},W^{*})-V_{1}^{\lambda,\alpha}(s_{1},\pi_{t}^{\alpha},\barmu_{t}^{\calI},W^{*})\big].\label{ieq:58}
    \end{align}
    Applying inequality~\eqref{ieq:56} with $p=\barpi_{t,h}^{*,\alpha}(\cdot\,|\,s_{h})$ to the left-hand side of inequality~\eqref{ieq:58} and rearranging the terms, we have that
    \begin{align}
        &\eta_{t+1}\bbE_{\barpi_{t}^{*,\alpha},\barmu_{t}^{\calI}}\bigg[\sum_{h=1}^{H}V_{h}^{\lambda,\alpha}(s_{h},\barpi_{t}^{*,\alpha},\barmu_{t}^{\calI},W^{*})-V_{h}^{\lambda,\alpha}(s_{h},\pi_{t+1}^{\alpha},\barmu_{t}^{\calI},W^{*})\bigg]\nonumber\\
        &\quad\qquad+(1+\lambda\eta_{t+1})\bbE_{\barpi_{t}^{*,\alpha},\barmu_{t}^{\calI}}\bigg[\sum_{h=1}^{H}\kl\big(\barpi_{t,h}^{*,\alpha}(\cdot\,|\,s_{h})\|\pi_{t+1,h}^{\alpha}(\cdot\,|\,s_{h})\big)\bigg]\nonumber\\
        &\leq \eta_{t+1}\bbE_{\barpi_{t}^{*,\alpha},\barmu_{t}^{\calI}}\bigg[\sum_{h=1}^{H}V_{h}^{\lambda,\alpha}(s_{h},\barpi_{t}^{*,\alpha},\barmu_{t}^{\calI},W^{*})-V_{h}^{\lambda,\alpha}(s_{h},\pi_{t}^{\alpha},\barmu_{t}^{\calI},W^{*})\bigg]-\eta_{t+1}\bbE_{\mu_{1}^{\alpha}}\big[V_{1}^{\lambda,\alpha}(s_{1},\barpi_{t}^{*,\alpha},\barmu_{t}^{\calI},W^{*})\nonumber\\
        &\quad\qquad-V_{1}^{\lambda,\alpha}(s_{1},\pi_{t}^{\alpha},\barmu_{t}^{\calI},W^{*})\big]+\bbE_{\barpi_{t}^{*,\alpha},\barmu_{t}^{\calI}}\bigg[\sum_{h=1}^{H}\kl\big(\barpi_{t,h}^{*,\alpha}(\cdot\,|\,s_{h})\|\pi_{t,h}^{\alpha}(\cdot\,|\,s_{h})\big)\bigg]\nonumber\\
        &\quad\qquad+\bbE_{\barpi_{t}^{*,\alpha},\barmu_{t}^{\calI}}\bigg[\sum_{h=1}^{H}\bbE_{\pi_{t+1}^{\alpha},\barmu_{t}^{\calI}}\bigg[\sum_{m=h}^{H}\Lambda_{t+1,m}^{\alpha}\,\bigg|\,s_{h}\bigg]\bigg]+\bbE_{\barpi_{t}^{*,\alpha},\barmu_{t}^{\calI}}\bigg[\sum_{h=1}^{H}\Lambda_{t+1,h}^{\alpha}\bigg].\label{ieq:76}
    \end{align}
    To handle the right-hand side of this inequality, we utilize the following proposition.
    \begin{proposition}\label{prop:valuefraction}
        For a $\lambda$-regularized finite-horizon \ac{mdp} $(\calS,\calA,H,\{r_{h}\}_{h=1}^{H},\{P_{h}\}_{h=1}^{H})$ with $|r_{h}|\leq 1$ for all $h\in[H]$, we denote the optimal policy as $\pi^{*}=\{\pi^{*}_{h}\}_{h=1}^{H}$. Then for any policy $\pi$, we have that
        \begin{align*}
            \bbE_{\pi^{*}}\big[V_{1}^{\lambda}(s_{1},\pi^{*})-V_{1}^{\lambda}(s_{1},\pi)\big]\geq \beta^{*}\bbE_{\pi^{*}}\bigg[\sum_{h=2}^{H}V_{h}^{\lambda}(s_{h},\pi^{*})-V_{h}^{\lambda}(s_{h},\pi)\bigg],
        \end{align*}
        where the expectation is taken with respect to the state distribution induced by $\pi^{*}$, and  $\beta^{*}>0$ is a constant that only depends on $\lambda,H$ and $|\calA|$.
    \end{proposition}
    \begin{proof}[Proof of Proposition~\ref{prop:valuefraction}]
        See Appendix~\ref{app:valuefraction}.
    \end{proof}
    Define $\theta^{*}=1/(1+\beta^{*})<1$ and let $\eta_{t}=\eta$, where $1+\lambda\eta=1/\theta^{*}$. Proposition~\ref{prop:valuefraction} shows that
    \begin{align}
        &\bbE_{\barpi_{t}^{*,\alpha},\barmu_{t}^{\calI}}\bigg[\sum_{h=1}^{H}V_{h}^{\lambda,\alpha}(s_{h},\barpi_{t}^{*,\alpha},\barmu_{t}^{\calI},W^{*})-V_{h}^{\lambda,\alpha}(s_{h},\pi_{t+1}^{\alpha},\barmu_{t}^{\calI},W^{*})\bigg]\nonumber\\
        &\quad\qquad+\frac{1}{\eta\theta^{*}}\bbE_{\barpi_{t}^{*,\alpha},\barmu_{t}^{\calI}}\bigg[\sum_{h=1}^{H}\kl\big(\barpi_{t,h}^{*,\alpha}(\cdot\,|\,s_{h})\|\pi_{t+1,h}^{\alpha}(\cdot\,|\,s_{h})\big)\bigg]\nonumber\\
        &\quad\leq \theta^{*}\bigg\{\bbE_{\barpi_{t}^{*,\alpha},\barmu_{t}^{\calI}}\bigg[\sum_{h=1}^{H}V_{h}^{\lambda,\alpha}(s_{h},\barpi_{t}^{*,\alpha},\barmu_{t}^{\calI},W^{*})-V_{h}^{\lambda,\alpha}(s_{h},\pi_{t}^{\alpha},\barmu_{t}^{\calI},W^{*})\bigg]\nonumber\\
        &\quad\qquad+\frac{1}{\eta\theta^{*}}\bbE_{\barpi_{t}^{*,\alpha},\barmu_{t}^{\calI}}\bigg[\sum_{h=1}^{H}\kl\big(\barpi_{t,h}^{*,\alpha}(\cdot\,|\,s_{h})\|\pi_{t,h}^{\alpha}(\cdot\,|\,s_{h})\big)\bigg]\bigg\}\nonumber\\
        &\quad\qquad+\frac{1}{\eta}\bbE_{\barpi_{t}^{*,\alpha},\barmu_{t}^{\calI}}\bigg[\sum_{h=1}^{H}\bbE_{\pi_{t+1}^{\alpha},\barmu_{t}^{\calI}}\bigg[\sum_{m=h}^{H}\Lambda_{t+1,m}^{\alpha}\,\bigg|\,s_{h}\bigg]\bigg]\frac{1}{\eta}\bbE_{\barpi_{t}^{*,\alpha},\barmu_{t}^{\calI}}\bigg[\sum_{h=1}^{H}\Lambda_{t+1,h}^{\alpha}\bigg].\label{ieq:59}
    \end{align}
    In the following, we will derive the rate of convergence of the following term
    \begin{align}
        X_{t}^{\alpha}&=\bbE_{\barpi_{t}^{*,\alpha},\barmu_{t}^{\calI}}\bigg[\sum_{h=1}^{H}V_{h}^{\lambda,\alpha}(s_{h},\barpi_{t}^{*,\alpha},\barmu_{t}^{\calI},W^{*})-V_{h}^{\lambda,\alpha}(s_{h},\pi_{t}^{\alpha},\barmu_{t}^{\calI},W^{*})\bigg]\nonumber\\
        &\qquad+\frac{1}{\eta\theta^{*}}\bbE_{\barpi_{t}^{*,\alpha},\barmu_{t}^{\calI}}\bigg[\sum_{h=1}^{H}\kl\big(\barpi_{t,h}^{*,\alpha}(\cdot\,|\,s_{h})\|\pi_{t,h}^{\alpha}(\cdot\,|\,s_{h})\big)\bigg].\label{eq:defx}
    \end{align}
    We note that $X_{t}^{\calI}$ is a good quantity to measure the ``distance'' between $\pi_{t}^{\calI}$ and \ac{ne}. For \ac{ne}, $\pi^{*,\calI}$ is the optimal policy on the \ac{mdp} induced by the distribution flow $\mu^{*,\calI}$ of itself. Since $\barmu_{t}^{\calI}$ is close to $\mu_{t}^{\calI}$, we expect that $\pi_{t}^{\calI}$ achieves high rewards on the \ac{mdp} induced by $\barmu_{t}^{\calI}$ if it is close to the \ac{ne}. Inequality~\eqref{ieq:59} shows that the recurrence relationship of $X_{t}^{\alpha}$ is
    \begin{align}
        X_{t+1}^{\alpha}\leq \theta^{*}X_{t}^{\alpha}+\frac{1}{\eta}\bbE_{\barpi_{t}^{*,\alpha},\barmu_{t}^{\calI}}\bigg[\sum_{h=1}^{H}\bbE_{\pi_{t+1}^{\alpha},\barmu_{t}^{\calI}}\bigg[\sum_{m=h}^{H}\Lambda_{t+1,m}^{\alpha}\,\bigg|\,s_{h}\bigg]\bigg]+\frac{1}{\eta}\bbE_{\barpi_{t}^{*,\alpha},\barmu_{t}^{\calI}}\bigg[\sum_{h=1}^{H}\Lambda_{t+1,h}^{\alpha}\bigg]+\Delta_{t+1}^{\alpha},\label{ieq:69}
    \end{align}
    where $\Delta_{t+1}^{\alpha}$ is the error introduced by the change of the environment, which is also called the dynamical error, and it is defined as
    \begin{align}
        \Delta_{t+1}^{\alpha}&=X_{t+1}^{\alpha}-\bbE_{\barpi_{t}^{*,\alpha},\barmu_{t}^{\calI}}\bigg[\sum_{h=1}^{H}V_{h}^{\lambda,\alpha}(s_{h},\barpi_{t}^{*,\alpha},\barmu_{t}^{\calI},W^{*})-V_{h}^{\lambda,\alpha}(s_{h},\pi_{t+1}^{\alpha},\barmu_{t}^{\calI},W^{*})\bigg]\nonumber\\
        &\qquad-\frac{1}{\eta\theta^{*}}\bbE_{\barpi_{t}^{*,\alpha},\barmu_{t}^{\calI}}\bigg[\sum_{h=1}^{H}\kl\big(\barpi_{t,h}^{*,\alpha}(\cdot\,|\,s_{h})\|\pi_{t+1,h}^{\alpha}(\cdot\,|\,s_{h})\big)\bigg]\nonumber.
    \end{align}
    \begin{proposition}\label{prop:dynamerr}
        Under assumptions in Theorem~\ref{thm:contractopt}, we have
        \begin{align*}
         \Delta_{t+1}^{\alpha}
         &\leq\bigg[H\bigg(2H(1+\lambda\log|\calA|)+\lambda L_{R}+\frac{1}{\eta\theta^{*}}\log\frac{|\calA|^{2}}{\beta_{t+1}}\bigg)+\frac{2}{\eta\theta^{*}}\max\bigg\{\log\frac{|\calA|}{\beta_{t+1}},L_{R}\bigg\}\bigg]\\
         &\quad\quad\qquad\cdot\bbE_{\barpi_{t+1}^{*,\alpha},\barmu_{t+1}^{\calI}}\bigg[\sum_{m=1}^{H}\big\|\barpi_{t+1,m}^{*,\alpha}(\cdot\,|\,s_{m})-\barpi_{t,m}^{*,\alpha}(\cdot\,|\,s_{m})\big\|_{1}\bigg]\\
         &\quad\qquad+\bigg[H\bigg(H\big(1+\lambda\log|\calA|\big)+\frac{1}{\eta\theta^{*}}\log\frac{|\calA|^{2}}{\beta_{t+1}}\bigg)L_{P}+2H\big[L_{r}+H(1+\lambda\log|\calA|)L_{P}\big]\bigg]\\
         &\quad\quad\qquad\cdot\sum_{m=1}^{H}\int_{0}^{1}\|\barmu_{t+1,m}^{\beta}-\barmu_{t,m}^{\beta}\|_{1}\rmd \beta\\
         &= C_{1}(\eta,\beta_{t+1})\bbE_{\barpi_{t+1}^{*,\alpha},\barmu_{t+1}^{\calI}}\bigg[\sum_{m=1}^{H}\big\|\barpi_{t+1,m}^{*,\alpha}(\cdot\,|\,s_{m})-\barpi_{t,m}^{*,\alpha}(\cdot\,|\,s_{m})\big\|_{1}\bigg]\nonumber\\
         &\quad\qquad+C_{2}(\eta,\beta_{t+1})\sum_{m=1}^{H}\int_{0}^{1}\|\barmu_{t+1,m}^{\beta}-\barmu_{t,m}^{\beta}\|_{1}\rmd \beta.
    \end{align*}
    In the above, we defined $C_{1}(\eta,\beta_{t+1})$ and $C_{2}(\eta,\beta_{t+1})$ for  ease of notation subsequently.
    \end{proposition}
    \begin{proof}
        See Appendix~\ref{app:dynamerr}.
    \end{proof}

    We need the following proposition to relate the difference between the optimal policies $\barpi_{t+1,m}^{*,\alpha}(\cdot\,|\,s_{m})$ and $\barpi_{t,m}^{*,\alpha}(\cdot\,|\,s_{m})$ in Proposition~\ref{prop:dynamerr} to the distribution flows $\barmu_{t+1}^{\calI}$ and $\barmu_{t}^{\calI}$.
    \begin{proposition}\label{prop:optpolicylip}
        For any two distribution flows $\mu^{\calI}$ and $\tilde{\mu}^{\calI}$, we define the optimal policies $\pi^{*,\calI}=\Gamma_{1}^{\lambda}(\mu^{\calI},W^{*})$ and $\tilde{\pi}^{*,\calI}=\Gamma_{1}^{\lambda}(\tilde{\mu}^{\calI},W^{*})$. Under Assumption~\ref{assump:lipconti}, we have that for any $h\in[H]$ and $\alpha\in[0,1]$
        \begin{align*}
            &\max_{s\in\calS}\big|V_{h}^{\lambda,\alpha}(s,\pi^{*,\calI},\mu^{\calI},W^{*})-V_{h}^{\lambda,\alpha}(s,\tilde{\pi}^{*,\calI},\tilde{\mu}^{\calI},W^{*})\big|\\
            &\quad\leq \big(H(1+\lambda\log|\calA|)L_{P}+L_{r}\big)\sum_{m=h}^{H}\int_{0}^{1}\|\mu_{m}^{\beta}-\tilde{\mu}_{m}^{\beta}\|_{1}\rmd\beta,\\
            &\max_{s\in\calS}\big\|\pi_{h}^{*,\alpha}(\cdot\,|\,s)-\tilde{\pi}_{h}^{*,\alpha}(\cdot\,|\,s)\big\|_{1}\\
            &\quad\leq 2\big(H(1+\lambda\log|\calA|)L_{P}+L_{r}\big)\sum_{m=h}^{H}\int_{0}^{1}\|\mu_{m}^{\beta}-\tilde{\mu}_{m}^{\beta}\|_{1}\rmd\beta.
        \end{align*}
    \end{proposition}
    \begin{proof}[Proof of Proposition~\ref{prop:optpolicylip}]
        See Appendix~\ref{app:optpolicylip}.
    \end{proof}
    Propositions~\ref{prop:dynamerr} and \ref{prop:optpolicylip} shows that
    \begin{align}
        \Delta_{t+1}^{\alpha}&\leq \Big(2H\big(H(1+\lambda\log|\calA|)L_{P}+L_{r}\big)C_{1}(\eta,\beta_{t+1})+C_{2}(\eta,\beta_{t+1})\Big)\sum_{m=1}^{H}\int_{0}^{1}\|\barmu_{t+1,m}^{\beta}-\barmu_{t,m}^{\beta}\|_{1}\rmd \beta\nonumber\\
        &\leq 2H\Big(2H\big(H(1+\lambda\log|\calA|)L_{P}+L_{r}\big)C_{1}(\eta,\beta_{t+1})+C_{2}(\eta,\beta_{t+1})\Big)\alpha_{t},\label{ieq:70}
    \end{align}
    where the inequality results from the definition of $\barmu_{t+1}^{\calI}$. 
    
    Next, we will combine Eqn.~\eqref{eq:17} and inequalities~\eqref{ieq:69} and \eqref{ieq:70} to derive a relationship between $X_{t+1}^{\alpha}$ and $X_{t}^{\alpha}$. Adopting Assumption~\ref{assump:est} to control the estimation error of the action-value functions in inequality~\eqref{ieq:69}, we have that
    \begin{align}
        \int_{0}^{1}X_{t+1}^{\alpha}\rmd\alpha&\leq \theta^{*}\int_{0}^{1}X_{t}^{\alpha}\rmd\alpha+\frac{H(H+1)}{\eta}\bigg[2\eta\varepsilon_{Q}+2\eta H(1+\lambda\log|\calA|)\beta_{t+1}+2\beta_{t+1}\log \frac{|\calA|}{\beta_{t}}\nonumber\\
        &\qquad+2\eta\big[L_{r}+H(1+\lambda\log|\calA|)L_{P}\big]\varepsilon_{\mu}+2(1+\lambda\eta)\beta_{t+1}\bigg]\nonumber\\
        &\qquad+2H\Big(2H\big(H(1+\lambda\log|\calA|)L_{P}+L_{r}\big)C_{1}(\eta,\beta_{t+1})+C_{2}(\eta,\beta_{t+1})\Big)\alpha_{t},\label{ieq:78}
    \end{align}
    where the inequality results from Assumption~\ref{assump:est}.
    
    We set $\alpha_{t}=O(T^{-2/3})$ and $\beta_{t}=O(T^{-1})$ for all $t\in[T]$. Lemma~\ref{lem:seqconv} shows that
    \begin{align*}
        \int_{0}^{1}X_{t}^{\alpha}\rmd\alpha=O\bigg((\theta^{*})^{t}+(\varepsilon_{Q}+\varepsilon_{\mu})(\theta^{*})^{t/2}+\frac{\log T}{T^{2/3}}\bigg)+O(\varepsilon_{Q}+\varepsilon_{\mu}).
    \end{align*}
    Thus, we have
    \begin{align}
        \frac{1}{T}\sum_{t=1}^{T}\int_{0}^{1}X_{t}^{\alpha}\rmd\alpha=O\bigg(\frac{\log T}{T^{2/3}}\bigg)+O(\varepsilon_{Q}+\varepsilon_{\mu}).\label{eq:18}
    \end{align}

    \textbf{Step 3: Derive the convergence rate of the learned mean-field.}
    
    To derive the convergence behavior of $\hbmu_{t}^{\calI}$, we define the distribution flow induced by $\barpi_{t}^{*,\calI}$ as $\barmu_{t}^{*,\calI}=\Gamma_{2}(\barpi_{t}^{*,\calI},W^{*})$. Then we have that
    \begin{align}
        d(\hbmu_{t+1}^{\calI},\mu^{*,\calI}) &=d\big((1-\alpha_{t})\hbmu_{t}^{\calI}+\alpha_{t}\hat{\mu}_{t}^{\calI},\mu^{*,\calI}\big)\nonumber\\
        &\leq (1-\alpha_{t})d(\hbmu_{t}^{\calI},\mu^{*,\calI})+\alpha_{t}d(\hat{\mu}_{t}^{\calI},\mu_{t}^{\calI})+\alpha_{t}d(\mu_{t}^{\calI},\barmu_{t}^{*,\calI})+\alpha_{t}d(\barmu_{t}^{*,\calI},\mu^{*,\calI}),\label{ieq:66}
    \end{align}
    where the equality results from the definition of $\hbmu_{t+1}^{\calI}$, and the inequality results from the triangle inequality. For the fourth term in the right-hand side of inequality~\eqref{ieq:66}, we have that
    \begin{align}
        d(\barmu_{t}^{*,\calI},\mu^{*,\calI})&=d\Big(\Gamma_{2}\big(\Gamma_{1}^{\lambda}(\barmu_{t}^{\calI},W^{*}),W^{*}\big),\Gamma_{2}\big(\Gamma_{1}^{\lambda}(\mu^{*,\calI},W^{*}),W^{*}\big)\Big)\nonumber\\
        &\leq d_{1}d_{2}d(\barmu_{t}^{\calI},\mu^{*,\calI})\nonumber\\
        & \leq d_{1}d_{2}\big(d(\barmu_{t}^{\calI},\hbmu_{t}^{\calI})+d(\hbmu_{t}^{\calI},\mu^{*,\calI})\big),\label{ieq:67}
    \end{align}
    where the equality results from the definitions of $\barmu_{t}^{*,\calI}$ and $\mu^{*,\calI}$, the first inequality results from Assumption~\ref{assump:contract}, and the last inequality results from the triangle inequality. We then define $\tilde{\mu}_{t}^{*,\calI}=\Gamma_{3}(\barpi_{t}^{*,\calI},\barmu_{t}^{\calI},W^{*})$. For the third term in the right-hand side of inequality~\eqref{ieq:66}, we have that
    \begin{align}
        d(\mu_{t}^{\calI},\barmu_{t}^{*,\calI}) 
        &=d\big(\Gamma_{2}(\pi_{t}^{\calI},W^{*}),\Gamma_{2}(\barpi_{t}^{*,\calI},W^{*})\big)\nonumber\\
        &\leq d_{2}\int_{0}^{1}\sum_{h=1}^{H}\bbE_{\mu_{h}^{*,\alpha}}\Big[\big\|\pi_{t,h}^{\alpha}(\cdot\,|\,s)-\barpi_{t,h}^{*,\alpha}(\cdot\,|\,s)\big\|_{1}\Big]\rmd\alpha \nonumber\\
        &= d_{2}\int_{0}^{1}\sum_{h=1}^{H}\bbE_{\tilde{\mu}_{t,h}^{*,\alpha}}\bigg[\frac{\mu_{h}^{*,\alpha}(s)}{\tilde{\mu}_{t,h}^{*,\alpha}(s)}\big\|\pi_{t,h}^{\alpha}(\cdot\,|\,s)-\barpi_{t,h}^{*,\alpha}(\cdot\,|\,s)\big\|_{1}\bigg]\rmd\alpha\nonumber\\
        &\leq d_{2}C_{\mu}\sqrt{H}\sqrt{2\int_{0}^{1}\sum_{h=1}^{H}\bbE_{\tilde{\mu}_{t,h}^{*,\alpha}}\Big[\kl\big(\barpi_{t,h}^{*,\alpha}(\cdot\,|\,s)\|\pi_{t,h}^{\alpha}(\cdot\,|\,s)\big)\Big]\rmd\alpha},\label{ieq:68}
    \end{align}
    where the first inequality results from Assumption~\ref{assump:contract}, and the second inequality results from Assumption~\ref{assump:optconcen} and the Cauchy--Schwarz inequality. Define $Y_{t}=d(\hbmu_{t}^{\calI},\mu^{*,\calI})$. Combining inequalities~\eqref{ieq:66}, \eqref{ieq:67}, and \eqref{ieq:68}, we have that
    \begin{align*}
        Y_{t+1}\leq (1-\alpha_{t}\bard)Y_{t}+\alpha_{t}d(\hat{\mu}_{t}^{\calI},\mu_{t}^{\calI})+\alpha_{t}d_{1}d_{2}d(\barmu_{t}^{\calI},\hbmu_{t}^{\calI})+\alpha_{t}d_{2}C_{\mu}\sqrt{H}\sqrt{2\eta\theta^{*}\int_{0}^{1}X_{t}^{\alpha}\rmd\alpha},
    \end{align*}
    where $\bard=1-d_{1}d_{2}$. 
    
    Recall the expressions of $\barmu_{t}^{\calI}$ and $\hbmu_{t}^{\calI}$ in Eqn.~\eqref{eq:19}, we have that
    \begin{align*}
        d(\barmu_{t}^{\calI},\hbmu_{t}^{\calI})\leq \sum_{m=1}^{t-1}\alpha_{m,t-1}d(\hat{\mu}_{m}^{\calI},\mu_{m}^{\calI})\leq \varepsilon_{\mu},
    \end{align*}
    where the first inequality results from the triangle inequality, and the second inequality results from Assumption~\ref{assump:est}. Take $\alpha_{t}=\alpha$. we have that
    \begin{align*}
        \frac{1}{T}\sum_{t=1}^{T}Y_{t}&\leq \frac{1}{\bard \alpha T}Y_{1}+\frac{(1+d_{1}d_{2})}{\bard}\varepsilon_{\mu}+\frac{d_{2}C_{\mu}\sqrt{H}}{\bard}\cdot\frac{1}{T}\sum_{t=1}^{T}\sqrt{2\eta\theta^{*}\int_{0}^{1}X_{t}^{\alpha}\rmd\alpha}\\
        &\leq \frac{1}{\bard \alpha T}Y_{1}+\frac{(1+d_{1}d_{2})}{\bard}\varepsilon_{\mu}+\frac{d_{2}C_{\mu}\sqrt{H}}{\bard}\cdot\sqrt{2\eta\theta^{*}\frac{1}{T}\sum_{t=1}^{T}\int_{0}^{1}X_{t}^{\alpha}\rmd\alpha},
    \end{align*}
    where the last inequality results from Eqn.~\eqref{eq:18}. Thus, we have
    \begin{align*}
        \frac{1}{T}\sum_{t=1}^{T}Y_{t}=O\bigg(\frac{\sqrt{\log T}}{T^{1/3}}\bigg)+O(\varepsilon_{\mu}+\sqrt{\varepsilon_{Q}+\varepsilon_{\mu}}).
    \end{align*}
    
    \textbf{Step 4: Build the desired result from step 2 and step 3.}
    
    From the definition of $X_{t}$, i.e., Eqn.~\eqref{eq:defx}, and Eqn.~\eqref{eq:18}, we have that
    \begin{align*}
        \frac{1}{T}\sum_{t=1}^{T}\int_{0}^{1}\bbE_{\barpi_{t}^{*,\alpha},\barmu_{t}^{\calI}}\bigg[\sum_{h=1}^{H}\kl\big(\barpi_{t,h}^{*,\alpha}(\cdot\,|\,s_{h})\|\pi_{t,h}^{\alpha}(\cdot\,|\,s_{h})\big)\bigg]\rmd\alpha =O\bigg(\frac{\log T}{T^{2/3}}\bigg)+O(\varepsilon_{Q}+\varepsilon_{\mu}).
    \end{align*}
    Recall that we defined $\tilde{\mu}_{t}^{*,\calI}=\Gamma_{3}(\barpi_{t}^{*,\calI},\barmu_{t}^{\calI},W^{*})$. Then we  bound  $D(\cdot,\cdot)$ as follows
    \begin{align*}
        \frac{1}{T}\sum_{t=1}^{T}D(\pi_{t}^{\calI},\barpi_{t}^{*,\calI})
        &=\frac{1}{T}\sum_{t=1}^{T}\int_{0}^{1}\sum_{h=1}^{H}\bbE_{\tilde{\mu}_{t,h}^{*,\alpha}}\bigg[\frac{\mu_{h}^{*,\alpha}(s)}{\tilde{\mu}_{t,h}^{*,\alpha}(s)}\big\|\pi_{t,h}^{\alpha}(\cdot\,|\,s)-\barpi_{t,h}^{*,\alpha}(\cdot\,|\,s)\big\|_{1}\bigg]\rmd\alpha\\
        &\leq C_{\mu}\sqrt{H}\sqrt{\frac{2}{T}\sum_{t=1}^{T}\int_{0}^{1}\sum_{h=1}^{H}\bbE_{\tilde{\mu}_{t,h}^{*,\alpha}}\Big[\kl\big(\barpi_{t,h}^{*,\alpha}(\cdot\,|\,s)\|\pi_{t,h}^{\alpha}(\cdot\,|\,s)\big)\Big]\rmd\alpha}\\
        &\leq O\bigg(\frac{\sqrt{\log T}}{T^{1/3}}\bigg)+O(\sqrt{\varepsilon_{Q}+\varepsilon_{\mu}}),
    \end{align*}
    where the first inequality results from the same arguments in inequality~\eqref{ieq:68}. To bound the distance between $\pi_{t}^{\calI}$ and $\pi^{*,\calI}$, we adopt the triangle inequality as
    \begin{align*}
        D(\pi_{t}^{\calI},\pi^{*,\calI})\leq D(\pi_{t}^{\calI},\barpi_{t}^{*,\calI})+D(\barpi_{t}^{*,\calI},\pi^{*,\calI})\leq D(\pi_{t}^{\calI},\barpi_{t}^{*,\calI})+d_{1}d(\barmu_{t}^{\calI},\mu^{*,\calI}).
    \end{align*}
    Thus, we have that
    \begin{align*}
        &D\bigg(\frac{1}{T}\sum_{t=1}^{T}\pi_{t}^{\calI},\pi^{*,\calI}\bigg)+d\bigg(\frac{1}{T}\sum_{t=1}^{T}\hbmu_{t}^{\calI},\mu^{*,\calI}\bigg)\\
        &\quad\leq \frac{1}{T}\sum_{t=1}^{T}D(\pi_{t}^{\calI},\barpi_{t}^{*,\calI})+d_{1}d(\barmu_{t}^{\calI},\mu^{*,\calI})+d(\hbmu_{t}^{\calI},\mu^{*,\calI})\\
        &\quad\leq \frac{1}{T}\sum_{t=1}^{T}D(\pi_{t}^{\calI},\barpi_{t}^{*,\calI})+d_{1}d(\hbmu_{t}^{\calI},\mu^{*,\calI})+d(\hbmu_{t}^{\calI},\mu^{*,\calI})+d_{1}d(\barmu_{t}^{\calI},\hbmu_{t}^{\calI})\\
        &\quad=O\bigg(\frac{\sqrt{\log T}}{T^{1/3}}\bigg)+O(\varepsilon_{\mu}+\sqrt{\varepsilon_{Q}+\varepsilon_{\mu}}),
    \end{align*}
    where the first inequality results from Jensen's inequality, and the second inequality results from the triangle inequality. Thus, we conclude the proof of Theorem~\ref{thm:contractopt}.
\end{proof}
\section{Corollary for Single-Agent \ac{mdp}}\label{app:singlepmd}
In this section, we state and prove our corollary for the policy mirror descent algorithm on single-agent \ac{mdp}. A single-agent \ac{mdp} is defined through a tuple $(\calS,\calA,\mu_{1},P,r,H)$. The state space and the action space are denoted respectively as $\calS$ and $\calA$. The initial state distribution $\mu_{1}\in\Delta(\calS)$ is state distribution at time $h=1$. The transition kernels $P_{h}:\calS\times\calA\rightarrow\Delta(\calS))$ and reward functions $r_{h}:\calS\times\calA\rightarrow[0,1]$ for $h\in[H]$ describe the state transition behavior and the reward generation process. A policy $\pi=\{\pi_{h}\}_{h=1}^{H}$ is the set of mappings $\pi_{h}:\calS\rightarrow\calA$ for $h\in[H]$. Similar as the value function defined in Section~\ref{sec:prelim}, the value function and the action-value function of a policy $\pi$ on a $\lambda$-regularized \ac{mdp} are defined as
\begin{align*}
    V_{h}^{\lambda}(s,\pi)&=\bbE^{\pi}\bigg[\sum_{t=h}^{H}r_{t}(s_{t},a_{t})-\lambda\log\pi_{t}(a_{t}\,|\,s_{t})\,\bigg|\, s_{h}=s\bigg],\\
    Q_{h}^{\lambda}(s,a,\pi)&=r_{h}(s,a)+\bbE\big[V_{h+1}^{\lambda}(s_{h+1},\pi)\,|\, s_{h}=s,a_{h}=a\big].
\end{align*}
The cumulative reward function is $J^{\lambda}(\pi)=\bbE_{\mu_{1}}[V_{1}^{\lambda}(s,\pi)]$, where the expectation is taken with respect to $s\sim \mu_{1}$. We denote the optimal policy as $\pi^{*}=\argmax_{\pi\in\Pi^{H}}J^{\lambda}(\pi)$. The policy mirror descent algorithm is implementing
\begin{align*}
    \pi_{t+1,h}(\cdot\,|\,s)\propto \big(\pi_{t,h}(\cdot\,|\,s)\big)^{1-\frac{\lambda\eta_{t+1}}{1+\lambda\eta_{t+1}}}\exp\Big(\frac{\eta_{t+1}}{1+\lambda\eta_{t+1}} Q_{h}^{\lambda}(s,\cdot,\pi_{t})\Big) \text{ for all }h\in[H].
\end{align*}
for $t\in[T]$, and we set $\pi_{1,h}(\cdot\,|\,s)=\unif(\calA)$ for all $s\in\calS$. Then the convergence result of this algorithm is
\begin{corollary}\label{coro:singlepmd}
    Suppose that $\eta_{t}=\eta>0$ for all $t\in[T]$, and we set this as some function of $\lambda$, $H$ and $|\calA|$. Then we have
    \begin{align*}
        &\bbE_{\pi^{*}}\bigg[\sum_{h=1}^{H}V_{h}^{\lambda}(s_{h},\pi^{*})-V_{h}^{\lambda}(s_{h},\pi_{t+1})\bigg]+\frac{1}{\eta\theta^{*}}\bbE_{\pi^{*}}\bigg[\sum_{h=1}^{H}\kl\big(\pi_{h}^{*}(\cdot\,|\,s_{h})\|\pi_{t+1,h}(\cdot\,|\,s_{h})\big)\bigg]\nonumber\\
        &\quad\leq \theta^{*}\bigg\{\bbE_{\pi^{*}}\bigg[\sum_{h=1}^{H}V_{h}^{\lambda}(s_{h},\pi^{*})-V_{h}^{\lambda}(s_{h},\pi_{t})\bigg]+\frac{1}{\eta\theta^{*}}\bbE_{\pi_{t}^{*}}\bigg[\sum_{h=1}^{H}\kl\big(\pi_{t,h}^{*}(\cdot\,|\,s_{h})\|\pi_{t,h}(\cdot\,|\,s_{h})\big)\bigg]\bigg\},
    \end{align*}
    where $0<\theta^{*}<1$ is a function of $\lambda$, $H$ and $|\calA|$, and $\bbE_{\pi^{*}}$ refers to the expectation with respect to the state distribution induced by $\pi^{*}$.
\end{corollary}
\begin{proof}[Proof of Corollary~\ref{coro:singlepmd}]
    Similarly as Step 1 of the proof of Theorem~\ref{thm:contractopt}, we have
    \begin{align*}
        &\eta_{t+1}\big\langle Q_{h}^{\lambda}(s_{h},\cdot,\pi_{t}),p-\pi_{t+1,h}(\cdot\,|\,s_{h}) \big\rangle+\lambda \eta_{t+1}\Big[ R\big(\pi_{t+1,h}(\cdot\,|\,s_{h})\big)-R(p)\Big]+\kl\big(\pi_{t+1,h}(\cdot\,|\,s_{h})\|\pi_{t,h}(\cdot\,|\,s_{h})\big)\nonumber\\
        &\quad\leq \kl\big(p\,\|\,\pi_{t,h}(\cdot\,|\,s_{h})\big)-(1+\lambda\eta_{t+1})\kl\big(p\,\|\,\pi_{t+1,h}(\cdot\,|\,s_{h})\big)
    \end{align*}
    for any $p\in\Delta(\calA)$. Following the same pipeline to inequality~\eqref{ieq:59}, we have that
    \begin{align*}
        &\bbE_{\pi^{*}}\bigg[\sum_{h=1}^{H}V_{h}^{\lambda}(s_{h},\pi^{*})-V_{h}^{\lambda}(s_{h},\pi_{t+1})\bigg]+\frac{1}{\eta\theta^{*}}\bbE_{\pi^{*}}\bigg[\sum_{h=1}^{H}\kl\big(\pi_{h}^{*}(\cdot\,|\,s_{h})\|\pi_{t+1,h}(\cdot\,|\,s_{h})\big)\bigg]\nonumber\\
        &\quad\leq \theta^{*}\bigg\{\bbE_{\pi^{*}}\bigg[\sum_{h=1}^{H}V_{h}^{\lambda}(s_{h},\pi^{*})-V_{h}^{\lambda}(s_{h},\pi_{t})\bigg]+\frac{1}{\eta\theta^{*}}\bbE_{\pi_{t}^{*}}\bigg[\sum_{h=1}^{H}\kl\big(\pi_{t,h}^{*}(\cdot\,|\,s_{h})\|\pi_{t,h}(\cdot\,|\,s_{h})\big)\bigg]\bigg\},
    \end{align*}
    where $\beta^{*}$ is defined in Proposition~\ref{prop:valuefraction}, $\theta^{*}=1/(1+\beta^{*})<1$, $\eta_{t}=\eta$ is defined through $1+\lambda\eta=1/\theta^{*}$. We note that in this single-agent setting, we do not have the $\barmu_{t}^{\calI}$, which is adopted to calculate the influence from others. Thus, the optimal policy $\pi^{*}$ is not changed over iterations. At the same time, we do not include the estimation error in the above algorithm. We conclude the proof of Corollary~\ref{coro:singlepmd}.
\end{proof}

\section{Proof of Theorem~\ref{thm:fixest}}
\begin{proof}[Proof of Theorem~\ref{thm:fixest}]
We first decompose the difference between the risk as the sum of the generalization error of risk, the Estimation Error of Mean-embedding, and the empirical risk difference. Given the fact that the empirical risk difference is equal and less to zero, Our proof involves two steps:
\begin{itemize}
    \item Bound the Estimation Error of Mean-embedding.
    \item Bound the generalization error of risk.
\end{itemize}

\begin{align*}
    &\calR_{\bar{\xi}}(\hatf_{h},\hatg_{h},\hatW_{h})-\calR_{\bar{\xi}}(f_{h}^{*},g_{h}^{*},W_{h}^{*})\nonumber\\
    &\quad=\text{Generalization Error of Risk}+\text{Estimation Error of Mean-embedding}+\text{Empirical Risk Difference},
\end{align*}
where each term is defined as
\begin{align*}
    &\text{Generalization Error of Risk}\\
    &\quad=\frac{1}{NL}\sum_{\tau=1}^{L}\sum_{i=1}^{N}\bbE_{\rho_{\tau,h}^{i}}\bigg[\Big(s_{\tau,h+1}^{i}-\hatf_{h}\big(\omega_{\tau,h}^{i}(\hatW_{h})\big)\Big)^{2}-\Big(s_{\tau,h+1}^{i}-f_{h}^{*}\big(\omega_{\tau,h}^{i}(W_{h}^{*})\big)\Big)^{2}\bigg]\nonumber\\
    &\quad\qquad -2\frac{1}{NL}\sum_{\tau=1}^{L}\sum_{i=1}^{N}\Big(s_{\tau,h+1}^{i}-\hatf_{h}\big(\omega_{\tau,h}^{i}(\hatW_{h})\big)\Big)^{2}-\Big(s_{\tau,h+1}^{i}-f_{h}^{*}\big(\omega_{\tau,h}^{i}(W_{h}^{*})\big)\Big)^{2}\nonumber\\
    &\quad\qquad+\frac{1}{NL}\sum_{\tau=1}^{L}\sum_{i=1}^{N}\bbE_{\rho_{\tau,h}^{i}}\bigg[\Big(r_{\tau,h}^{i}-\hatg_{h}\big(\omega_{\tau,h}^{i}(\hatW_{h})\big)\Big)^{2}-\Big(r_{\tau,h}^{i}-g_{h}^{*}\big(\omega_{\tau,h}^{i}(W_{h}^{*})\big)\Big)^{2}\bigg]\nonumber\\
    &\quad\qquad -2\frac{1}{NL}\sum_{\tau=1}^{L}\sum_{i=1}^{N}\Big(r_{\tau,h}^{i}-\hatg_{h}\big(\omega_{\tau,h}^{i}(\hatW_{h})\big)\Big)^{2}-\Big(r_{\tau,h}^{i}-g_{h}^{*}\big(\omega_{\tau,h}^{i}(W_{h}^{*})\big)\Big)^{2}.
\end{align*}
This generalization error of risk represents the error due to the fact that we optimize over the empirical estimation of the risk not the population risk.
\begin{align*}
    &\text{Estimation Error of Mean-embedding}\\
    &\quad=2\frac{1}{NL}\sum_{\tau=1}^{L}\sum_{i=1}^{N}\Big(s_{\tau,h+1}^{i}-\hatf_{h}\big(\omega_{\tau,h}^{i}(\hatW_{h})\big)\Big)^{2}-\Big(s_{\tau,h+1}^{i}-\hatf_{h}\big(\homega_{\tau,h}^{i}(\hatW_{h})\big)\Big)^{2}\nonumber\\
    &\quad\qquad+2\frac{1}{NL}\sum_{\tau=1}^{L}\sum_{i=1}^{N}\Big(s_{\tau,h+1}^{i}-f_{h}^{*}\big(\homega_{\tau,h}^{i}(W_{h}^{*})\big)\Big)^{2}-\Big(s_{\tau,h+1}^{i}-f_{h}^{*}\big(\omega_{\tau,h}^{i}(W_{h}^{*})\big)\Big)^{2}\nonumber\\
    &\quad\qquad+2\frac{1}{NL}\sum_{\tau=1}^{L}\sum_{i=1}^{N}\Big(r_{\tau,h}^{i}-\hatg_{h}\big(\omega_{\tau,h}^{i}(\hatW_{h})\big)\Big)^{2}-\Big(r_{\tau,h}^{i}-\hatg_{h}\big(\homega_{\tau,h}^{i}(\hatW_{h})\big)\Big)^{2}\nonumber\\
    &\quad\qquad+2\frac{1}{NL}\sum_{\tau=1}^{L}\sum_{i=1}^{N}\Big(r_{\tau,h}^{i}-g_{h}^{*}\big(\homega_{\tau,h}^{i}(W_{h}^{*})\big)\Big)^{2}-\Big(r_{\tau,h}^{i}-g_{h}^{*}\big(\omega_{\tau,h}^{i}(W_{h}^{*})\big)\Big)^{2}.
\end{align*}
Estimation error of mean-embedding represents the error due to the fact that we cannot observe the value of $\homega_{\tau,h}^{i}(\hatW_{h})$. Instead, we can only estimate the value of it through the states of sampled agents.
\begin{align*}
    &\text{Empirical Risk Difference}\\
    &\quad=2\frac{1}{NL}\sum_{\tau=1}^{L}\sum_{i=1}^{N}\Big(s_{\tau,h+1}^{i}-\hatf_{h}\big(\homega_{\tau,h}^{i}(\hatW_{h})\big)\Big)^{2}-\Big(s_{\tau,h+1}^{i}-f_{h}^{*}\big(\homega_{\tau,h}^{i}(W_{h}^{*})\big)\Big)^{2}\\
    &\quad\qquad+2\frac{1}{NL}\sum_{\tau=1}^{L}\sum_{i=1}^{N}\Big(r_{\tau,h}^{i}-\hatg_{h}\big(\homega_{\tau,h}^{i}(\hatW_{h})\big)\Big)^{2}-\Big(r_{\tau,h}^{i}-g_{h}^{*}\big(\homega_{\tau,h}^{i}(W_{h}^{*})\big)\Big)^{2}.
\end{align*}
Empirical risk difference represents the error from that fact that we choose $(\hatf_{h},\hatg_{h},\hatW_{h})$ not $(f_{h}^{*},g_{h}^{*},W_{h}^{*})$ by minimizing the empirical risk. From the procedure of Algorithm~\eqref{algo:estalgo}, we have
\begin{align*}
    \text{Empirical Risk Difference}\leq 0.
\end{align*}
Thus, we have that
\begin{align}
    &\calR_{\bar{\xi}}(\hatf_{h},\hatg_{h},\hatW_{h})-\calR_{\bar{\xi}}(f_{h}^{*},g_{h}^{*},W_{h}^{*})\nonumber\\
    &\quad\leq\bigg\{\frac{1}{NL}\sum_{\tau=1}^{L}\sum_{i=1}^{N}\bbE_{\rho_{\tau,h}^{i}}\bigg[\Big(s_{\tau,h+1}^{i}-\hatf_{h}\big(\omega_{\tau,h}^{i}(\hatW_{h})\big)\Big)^{2}-\Big(s_{\tau,h+1}^{i}-f_{h}^{*}\big(\omega_{\tau,h}^{i}(W_{h}^{*})\big)\Big)^{2}\bigg]\nonumber\\
    &\quad\qquad -2\frac{1}{NL}\sum_{\tau=1}^{L}\sum_{i=1}^{N}\Big(s_{\tau,h+1}^{i}-\hatf_{h}\big(\omega_{\tau,h}^{i}(\hatW_{h})\big)\Big)^{2}-\Big(s_{\tau,h+1}^{i}-f_{h}^{*}\big(\omega_{\tau,h}^{i}(W_{h}^{*})\big)\Big)^{2}\nonumber\\
    &\quad\qquad+\frac{1}{NL}\sum_{\tau=1}^{L}\sum_{i=1}^{N}\bbE_{\rho_{\tau,h}^{i}}\bigg[\Big(r_{\tau,h}^{i}-\hatg_{h}\big(\omega_{\tau,h}^{i}(\hatW_{h})\big)\Big)^{2}-\Big(r_{\tau,h}^{i}-g_{h}^{*}\big(\omega_{\tau,h}^{i}(W_{h}^{*})\big)\Big)^{2}\bigg]\nonumber\\
    &\quad\qquad -2\frac{1}{NL}\sum_{\tau=1}^{L}\sum_{i=1}^{N}\Big(r_{\tau,h}^{i}-\hatg_{h}\big(\omega_{\tau,h}^{i}(\hatW_{h})\big)\Big)^{2}-\Big(r_{\tau,h}^{i}-g_{h}^{*}\big(\omega_{\tau,h}^{i}(W_{h}^{*})\big)\Big)^{2}\bigg\}\nonumber\\
    &\quad\qquad+2\sup_{f\in\bbB(r,\bar{\calH}),W\in\tilde{\calW}}\bigg|\frac{1}{NL}\sum_{\tau=1}^{L}\sum_{i=1}^{N}\Big(s_{\tau,h+1}^{i}-f\big(\homega_{\tau,h}^{i}(W)\big)\Big)^{2}-\Big(s_{\tau,h+1}^{i}-f\big(\omega_{\tau,h}^{i}(W)\big)\Big)^{2}\bigg|\nonumber\\
    &\quad\qquad+2\sup_{g\in\bbB(\tilr,\tilde{\calH}),W\in\tilde{\calW}}\bigg|\frac{1}{NL}\sum_{\tau=1}^{L}\sum_{i=1}^{N}\Big(r_{\tau,h}^{i}-g\big(\homega_{\tau,h}^{i}(W)\big)\Big)^{2}-\Big(r_{\tau,h}^{i}-g\big(\omega_{\tau,h}^{i}(W)\big)\Big)^{2}\bigg|\nonumber\\
    &\quad=\text{(I)}+\text{(II)},\label{ieq:8}
\end{align}
We note that the terms related to the transition kernels and reward functions are similar. In the following, we will only present the bounds for the terms related to the transition kernels, and the bounds for the reward functions can be similarly derived.

\textbf{Step 1: Bound the Estimation Error of Mean-embedding.}

Considering term (II), we have that 
\begin{align}
    &\sup_{f\in\bbB(r,\bar{\calH}),W\in\tilde{\calW}}\bigg|\frac{1}{NL}\sum_{\tau=1}^{L}\sum_{i=1}^{N}\Big(s_{\tau,h+1}^{i}-f\big(\homega_{\tau,h}^{i}(W)\big)\Big)^{2}-\Big(s_{\tau,h+1}^{i}-f\big(\omega_{\tau,h}^{i}(W)\big)\Big)^{2}\bigg|\nonumber\\
    &\quad\leq \sup_{f\in\bbB(r,\bar{\calH}),W\in\tilde{\calW}}\frac{1}{NL}\sum_{\tau=1}^{L}\sum_{i=1}^{N}\Big|f\big(\homega_{\tau,h}^{i}(W)\big)-f\big(\omega_{\tau,h}^{i}(W)\big)\Big|\cdot\Big|2s_{\tau,h+1}^{i}-f\big(\homega_{\tau,h}^{i}(W)-f\big(\omega_{\tau,h}^{i}(W)\big)\Big|\nonumber\\
    &\quad \leq 2(B_{S}+rB_{K})rL_{K}\sup_{W\in\tilde{\calW}}\frac{1}{NL}\sum_{\tau=1}^{L}\sum_{i=1}^{N}\big\|\homega_{\tau,h}^{i}(W)-\omega_{\tau,h}^{i}(W)\big\|_{\calH},\label{ieq:9}
\end{align}
where the first inequality results from the triangle inequality, and the second inequality results from Assumption~\ref{assump:kernel} and Lemma~\ref{lem:rkhs}. Recall the definitions of $\homega_{\tau,h}^{i}(W)$ and $\omega_{\tau,h}^{i}(W)$ are
\begin{align*}
    \omega_{\tau,h}^{i}(W)&=\int_{0}^{1}\int_{\calS}W(\xi_{i},\beta)k\big(\cdot,(s_{\tau,h}^{i},a_{\tau,h}^{i},s)\big)\mu_{\tau,h}^{\beta}(s)\, \rmd s\, \rmd\beta,\\
    \homega_{\tau,h}^{i}(W)&=\frac{1}{N-1}\sum_{j\neq i}W(\xi_{i},\xi_{j})k\big(\cdot,(s_{\tau,h}^{i},a_{\tau,h}^{i},s_{\tau,h}^{j})\big),
\end{align*}
respectively. We decompose the error between them as
\begin{align}
    \sup_{W\in\tilde{\calW}}\big\|\homega_{\tau,h}^{i}(W)-\omega_{\tau,h}^{i}(W)\big\|_{\calH}&\leq \sup_{W\in\tilde{\calW}}\big\|\bar{\omega}_{\tau,h}^{i}(W)-\omega_{\tau,h}^{i}(W)\big\|_{\calH}+\sup_{W\in\tilde{\calW}}\big\|\homega_{\tau,h}^{i}(W)-\bar{\omega}_{\tau,h}^{i}(W)\big\|_{\calH}\nonumber\\
    &=\text{(III)}+\text{(IV)},\label{ieq:16}
\end{align}
where
\begin{align*}
    \bar{\omega}_{\tau,h}^{i}(W)&=\frac{1}{N-1}\sum_{j\neq i}W(\xi_{i},\xi_{j})\int_{\calS}k\big(\cdot,(s_{\tau,h}^{i},a_{\tau,h}^{i},s)\big)\mu_{\tau,h}^{j}(s)\rmd s.
\end{align*}

For term (III)$\;=\sup_{W\in\tilde{\calW}}\big\|\omega_{\tau,h}^{i}(W)-\bar{\omega}_{\tau,h}^{i}(W)\big\|_{\calH}$, we have that
\begin{align}
    \big\|\omega_{\tau,h}^{i}(W)-\bar{\omega}_{\tau,h}^{i}(W)\big\|_{\calH}
    &\leq\bigg\|\int_{0}^{1}\int_{\calS}W(\xi_{i},\beta)k\big(\cdot,(s_{\tau,h}^{i},a_{\tau,h}^{i},s)\big)\mu_{\tau,h}^{\beta}(s)\, \rmd s\, \rmd\beta\nonumber\\
    &\quad\qquad-\frac{1}{N-1}\sum_{j=1}^{N-1}W\bigg(\xi_{i},\frac{j}{N-1}\bigg)\int_{\calS}k\big(\cdot,(s_{\tau,h}^{i},a_{\tau,h}^{i},s)\big)\mu_{\tau,h}^{\frac{j}{N-1}}(s)\, \rmd s\bigg\|_{\calH}\nonumber\\
    &\quad\qquad+\bigg\|\frac{1}{N-1}\sum_{j=1}^{N-1}W\bigg(\xi_{i},\frac{j}{N-1}\bigg)\int_{\calS}k\big(\cdot,(s_{\tau,h}^{i},a_{\tau,h}^{i},s)\big)\mu_{\tau,h}^{\frac{j}{N-1}}(s)\, \rmd s\nonumber\\
    &\quad\qquad-\frac{1}{N-1}\sum_{j\neq i}W(\xi_{i},\xi_{j})\int_{\calS}k\big(\cdot,(s_{\tau,h}^{i},a_{\tau,h}^{i},s)\big)\mu_{\tau,h}^{j}(s)\, \rmd s\bigg\|_{\calH}\nonumber\\
    & =\text{(V)}+\text{(VI)}\label{ieq:14}
\end{align}

For term (V), we have that
\begin{align}
    &\bigg\|\int_{0}^{1}\int_{\calS}W(\xi_{i},\beta)k\big(\cdot,(s_{\tau,h}^{i},a_{\tau,h}^{i},s)\big)\mu_{\tau,h}^{\beta}(s)\rmd s\rmd\beta\nonumber\\
    &\quad\qquad-\frac{1}{N-1}\sum_{j=1}^{N-1}W\bigg(\xi_{i},\frac{j}{N-1}\bigg)\int_{\calS}k\big(\cdot,(s_{\tau,h}^{i},a_{\tau,h}^{i},s)\big)\mu_{\tau,h}^{\frac{j}{N-1}}(s)\, \rmd s\bigg\|_{\calH}\nonumber\\
    &\quad\leq \sum_{j=1}^{N-1}\int_{\frac{j-1}{N-1}}^{\frac{j}{N-1}}\bigg\|\int_{\calS}W(\xi_{i},\beta)k\big(\cdot,(s_{\tau,h}^{i},a_{\tau,h}^{i},s)\big)\mu_{\tau,h}^{\beta}(s)\, \rmd s\nonumber\\
    &\quad\qquad-W\bigg(\xi_{i},\frac{j}{N-1}\bigg)\int_{\calS}k\big(\cdot,(s_{\tau,h}^{i},a_{\tau,h}^{i},s)\big)\mu_{\tau,h}^{\frac{j}{N-1}}(s)\rmd s\bigg\|_{\calH}\, \rmd\beta,\label{ieq:11}
\end{align}
where the inequality results from the triangle inequality. For each term in the sum, we have that
\begin{align}
    &\bigg\|\int_{\calS}W(\xi_{i},\beta)k\big(\cdot,(s_{\tau,h}^{i},a_{\tau,h}^{i},s)\big)\mu_{\tau,h}^{\beta}(s)\rmd s\nonumber\\
    &\quad\qquad-W\bigg(\xi_{i},\frac{j}{N-1}\bigg)\int_{\calS}k\big(\cdot,(s_{\tau,h}^{i},a_{\tau,h}^{i},s)\big)\mu_{\tau,h}^{\frac{j}{N-1}}(s)\rmd s\bigg\|_{\calH}\nonumber\\
    &\quad\leq \bigg\|\bigg(W(\xi_{i},\beta)-W\bigg(\xi_{i},\frac{j}{N-1}\bigg)\bigg)\int_{\calS}k\big(\cdot,(s_{\tau,h}^{i},a_{\tau,h}^{i},s)\big)\mu_{\tau,h}^{\beta}(s)\rmd s\bigg\|_{\calH}\nonumber\\
    &\quad\qquad+\bigg\|W\bigg(\xi_{i},\frac{j}{N-1}\bigg)\int_{\calS}k\big(\cdot,(s_{\tau,h}^{i},a_{\tau,h}^{i},s)\big)\big(\mu_{\tau,h}^{\beta}(s)-\mu_{\tau,h}^{\frac{j}{N-1}}(s)\big)\rmd s\bigg\|_{\calH}\nonumber\\
    &\quad \leq B_{k}L_{\bar{\calW}}\bigg|\beta-\frac{j}{N-1}\bigg|+B_{k}\Big\|\mu_{\tau,h}^{\beta}-\mu_{\tau,h}^{\frac{j}{N-1}}\Big\|_{1},\label{ieq:12}
\end{align}
where the first inequality results from the triangle inequality, and the second results from Assumptions~\ref{assump:graphon} and \ref{assump:kernel}.

\begin{proposition}\label{prop:agentlip}
    Under Assumptions~\ref{assump:graphon} and \ref{assump:lipconti}, we have that
    \begin{align*}
        \|\mu_{h}^{\alpha}-\mu_{h}^{\beta}\|_{1}\leq (h-1)L_{P}L_{\bar{\calW}}|\alpha-\beta|+\sum_{t=1}^{h-1}\sup_{s\in\calS}\big\|\pi_{t}^{\alpha}(\cdot\,|\,s)-\pi_{t}^{\beta}(\cdot\,|\,s)\big\|_{1} \text{ for all }h\in[H].
    \end{align*}
\end{proposition}
\begin{proof}[Proof of Proposition~\ref{prop:agentlip}]
    See Appendix~\ref{app:agentlip}.
\end{proof}
Thus, we bound the second term of inequality~\eqref{ieq:12} as
\begin{align}
    \Big\|\mu_{\tau,h}^{\beta}-\mu_{\tau,h}^{\frac{j}{N-1}}\Big\|_{1}
    &\leq H L_{P}L_{\bar{\calW}}\bigg|\beta-\frac{j}{N-1}\bigg|+\sum_{t=1}^{h-1}\sup_{s\in\calS}\big\|\pi_{t}^{\beta}(\cdot\,|\,s)-\pi_{t}^{\frac{j}{N-1}}(\cdot\,|\,s)\big\|_{1}\nonumber\\
    &  \leq  (H L_{P}L_{\bar{\calW}}+HL_{\pi})\bigg|\beta-\frac{j}{N-1}\bigg|,\label{ieq:13}
\end{align}
where the first inequality results from Proposition~\ref{prop:agentlip}, and the second inequality results from the Lipschitzness of behavior policies. Substituting inequalities~\eqref{ieq:12} and \eqref{ieq:13} into inequality~\eqref{ieq:11}, we have that
\begin{align*}
    \text{(V)}&=\bigg\|\int_{0}^{1}\int_{\calS}W(\xi_{i},\beta)k\big(\cdot,(s_{\tau,h}^{i},a_{\tau,h}^{i},s)\big)\mu_{\tau,h}^{\beta}(s)\rmd s\rmd\beta\nonumber\\
    &\quad\qquad-\frac{1}{N-1}\sum_{j=1}^{N-1}W\bigg(\xi_{i},\frac{j}{N-1}\bigg)\int_{\calS}k\big(\cdot,(s_{\tau,h}^{i},a_{\tau,h}^{i},s)\big)\mu_{\tau,h}^{\frac{j}{N-1}}(s)\rmd s\bigg\|_{\calH}\\
    &\leq\sum_{j=1}^{N-1}\int_{\frac{j-1}{N-1}}^{\frac{j}{N-1}}B_{k}(L_{\bar{\calW}}+H L_{P}L_{\bar{\calW}}+HL_{\pi})\bigg|\beta-\frac{j}{N-1}\bigg|\rmd\beta\\
    &=\frac{1}{2(N-1)}B_{k}(L_{\bar{\calW}}+H L_{P}L_{\bar{\calW}}+HL_{\pi}).
\end{align*}

For term (VI), we have that
\begin{align*}
    \text{(VI)}&=\bigg\|\frac{1}{N-1}\sum_{j=1}^{N-1}W\bigg(\xi_{i},\frac{j}{N-1}\bigg)\int_{\calS}k\big(\cdot,(s_{\tau,h}^{i},a_{\tau,h}^{i},s)\big)\mu_{\tau,h}^{\frac{j}{N-1}}(s)\rmd s\nonumber\\
    &\quad\qquad-\frac{1}{N-1}\sum_{j\neq i}W(\xi_{i},\xi_{j})\int_{\calS}k\big(\cdot,(s_{\tau,h}^{i},a_{\tau,h}^{i},s)\big)\mu_{\tau,h}^{j}(s)\rmd s\bigg\|_{\calH}\\
    &\leq \frac{1}{N-1}\sum_{j\neq 1}\bigg[B_{k}L_{\bar{\calW}}\bigg(\bigg|\xi_{j}-\frac{j}{N-1}\bigg|+\bigg|\xi_{j}-\frac{j-1}{N-1}\bigg|\bigg)+B_{k}\big\|\mu_{\tau,h}^{\frac{j}{N-1}}-\mu_{\tau,h}^{j}\big\|_{1}\bigg]\\
    &\leq \frac{1}{N-1}\bigg[3+2B_{k}(L_{\bar{\calW}}+HL_{\pi}+HL_{p}L_{\bar{\calW}})\sum_{i=1}^{N}\bigg|\xi_{i}-\frac{i}{N}\bigg|\bigg],
\end{align*}
where the first results from triangle inequality, and the second inequality results from Proposition~\ref{prop:agentlip}. Substituting the bounds for terms (V) and (VI) into inequality~\eqref{ieq:14}, we have that
\begin{align}
    \text{(III)}&=\sup_{W\in\tilde{\calW}}\big\|\omega_{\tau,h}^{i}(W)-\bar{\omega}_{\tau,h}^{i}(W)\big\|_{\calH}\nonumber\\
    &\leq \frac{1}{2(N-1)}B_{k}(L_{\bar{\calW}}+H L_{P}L_{\bar{\calW}}+HL_{\pi})+\frac{1}{N-1}\bigg[3+2B_{k}(L_{\bar{\calW}}+HL_{\pi}+HL_{p}L_{\bar{\calW}})\sum_{i=1}^{N}\bigg|\xi_{i}-\frac{i}{N}\bigg|\bigg]\nonumber\\
    &=\frac{1}{2(N-1)}B_{k}(L_{\bar{\calW}}+H L_{P}L_{\bar{\calW}}+HL_{\pi})+\frac{3}{N-1}.\label{ieq:17}
\end{align}

For term (IV), we have that
\begin{align*}
    &\big\|\homega_{\tau,h}^{i}(W)-\bar{\omega}_{\tau,h}^{i}(W)\big\|_{\calH}\\
    &\quad=\bigg\|\frac{1}{N-1}\sum_{j\neq i}W(\xi_{i},\xi_{j})k\big(\cdot,(s_{\tau,h}^{i},a_{\tau,h}^{i},s_{\tau,h}^{j})\big)-\frac{1}{N-1}\sum_{j\neq i}W(\xi_{i},\xi_{j})\int_{\calS}k\big(\cdot,(s_{\tau,h}^{i},a_{\tau,h}^{i},s)\big)\mu_{\tau,h}^{j}(s)\rmd s\bigg\|_{\calH}.
\end{align*}
To derive a concentration inequality for term (IV), we first construct the minimal $\varepsilon-$cover of $\tilde{\calW}$ with respect to $\|\cdot\|_{\infty}$. The covering number is denoted as $\calN_{\infty}(\varepsilon,\tilde{\calW})$. Then for any $W\in\tilde{\calW}$, there exists a graphon $W_{i}$ for $i\in \{1,\cdots,\calN_{\infty}(\varepsilon,\tilde{\calW})\}$ such that $\|W-W_{i}\|_{\infty}\leq\varepsilon$. Then we have that
\begin{align*}
    \big\|\homega_{\tau,h}^{i}(W)-\bar{\omega}_{\tau,h}^{i}(W)\big\|_{\calH}\leq \big\|\homega_{\tau,h}^{i}(W_{i})-\bar{\omega}_{\tau,h}^{i}(W_{i})\big\|_{\calH}+2\varepsilon B_{k},
\end{align*}
where the inequality results from the triangle inequality. In the following, we set $\varepsilon=t/(4B_{k})$. Then the concentration inequality for term (IV) can be derived as
\begin{align}
    &\bbP\Big(\exists W\in\tilde{\calW} ,i\in[N],\tau\in[L], \big\|\homega_{\tau,h}^{i}(W)-\bar{\omega}_{\tau,h}^{i}(W)\big\|_{\calH}\geq t)\nonumber\\
    &\quad \leq \bbP\Big(\exists j\in \big[\calN_{\infty}(\varepsilon,\tilde{\calW})\big],i\in[N],\tau\in[L], \big\|\homega_{\tau,h}^{i}(W_{j})-\bar{\omega}_{\tau,h}^{i}(W_{j})\big\|_{\calH}\geq t-2\varepsilon B_{k}\Big)\nonumber\\
    &\quad\leq NL\calN_{\infty}(t/(4B_{k}),\tilde{\calW})\max_{j\in[\calN_{\infty}],i\in[N],\tau\in[L]}\bbP\Big( \big\|\homega_{\tau,h}^{i}(W_{j})-\bar{\omega}_{\tau,h}^{i}(W_{j})\big\|_{\calH}\geq t/2\Big)\nonumber\\
    &\quad\leq 2NL\calN_{\infty}(t/(4B_{k}),\tilde{\calW})\exp\bigg(-\frac{(N-1)t^{2}}{32B_{k}^{2}}\bigg),\nonumber
\end{align}
where the first inequality results from the construction of the cover, the second inequality results from the union bound, and the last inequality results from Lemma~\ref{lem:hilconcen} and that $\|W(\xi_{i},\xi_{j})k\big(\cdot,(s_{\tau,h}^{i},a_{\tau,h}^{i},s_{\tau,h}^{j})\|\leq B_{k}$ for any $W\in\tilde{\calW}$. For $t\geq 4B_{k}/\sqrt{N}$, we have that
\begin{align*}
    \bbP\Big(\exists W\in\tilde{\calW} ,i\in[N],\tau\in[L], \big\|\homega_{\tau,h}^{i}(W)-\bar{\omega}_{\tau,h}^{i}(W)\big\|_{\calH}\geq t)\leq 2NL\calN_{\infty}(1/\sqrt{N},\tilde{\calW})\exp\bigg(-\frac{(N-1)t^{2}}{32B_{k}^{2}}\bigg).
\end{align*}
Thus, term (IV) can be bounded as
\begin{align}
    \text{(IV)}=\sup_{W\in\tilde{\calW}}\big\|\homega_{\tau,h}^{i}(W)-\bar{\omega}_{\tau,h}^{i}(W)\big\|_{\calH}\leq \frac{4\sqrt{2}B_{k}}{\sqrt{N-1}}\log\frac{2NL\calN_{\infty}(1/\sqrt{N},\tilde{\calW})}{\delta},\label{ieq:15}
\end{align}
with probability at least $1-\delta$. Substituting inequalities~\eqref{ieq:15} and \eqref{ieq:17} into inequalities~\eqref{ieq:9} and \eqref{ieq:16}, we have that
\begin{align}
    &\sup_{f\in\bbB(r,\bar{\calH}),W\in\tilde{\calW}}\bigg|\frac{1}{NL}\sum_{\tau=1}^{L}\sum_{i=1}^{N}\Big(s_{\tau,h+1}^{i}-f\big(\homega_{\tau,h}^{i}(W)\big)\Big)^{2}-\Big(s_{\tau,h+1}^{i}-f\big(\omega_{\tau,h}^{i}(W)\big)\Big)^{2}\bigg|\nonumber\\
    &\quad\leq O\bigg(\frac{(B_{S}+rB_{K})rL_{K}B_{k}}{\sqrt{N}}\log\frac{NL\calN_{\infty}(1/\sqrt{N},\tilde{\calW})}{\delta}\bigg),\label{ieq:33}
\end{align}
with probability at least $1-\delta$.

\textbf{Step 2: Bound the generalization error of risk.}

Considering term (I), for ease of notation, we denote the quadruple $(s_{\tau,h}^{i},a_{\tau,h}^{i},\mu_{\tau,h}^{\calI},s_{\tau,h+1}^{i})$ as $e_{\tau,h}^{i}$. We define the function $f_{W}$ as 
\begin{align*}
    f_{W}(e_{\tau,h}^{i})=\Big(s_{\tau,h+1}^{i}-f\big(\omega_{\tau,h}^{i}(W)\big)\Big)^{2}-\Big(s_{\tau,h+1}^{i}-f_{h}^{*}\big(\omega_{\tau,h}^{i}(W_{h}^{*})\big)\Big)^{2}.
\end{align*}
The correspond function class is defined as $\calF_{\tilde{\calW}}=\{f_{W}\,|\, f\in\bbB(r,\bar{\calH}),W\in\tilde{\calW}\}$. Then we have that
\begin{align*}
    \text{(I)}=\frac{1}{NL}\sum_{\tau=1}^{L}\sum_{i=1}^{N}\bbE_{\rho_{\tau,h}^{i}}\big[\hatf_{\hatW}(e_{\tau,h}^{i})\big]-2\frac{1}{NL}\sum_{\tau=1}^{L}\sum_{i=1}^{N}\hatf_{\hatW}(e_{\tau,h}^{i}).
\end{align*}

\begin{proposition}\label{prop:frconcen}
    With Assumption~\ref{assump:kernel}, we have that
    \begin{align*}
        &\bbP\bigg(\exists f_{W}\in\calF_{\tilde{\calW}}, \frac{1}{NL}\sum_{\tau=1}^{L}\sum_{i=1}^{N}\bbE_{\rho_{\tau,h}^{i}}\big[f_{W}(e_{\tau,h}^{i})\big]-\frac{1}{NL}\sum_{\tau=1}^{L}\sum_{i=1}^{N}f_{W}(e_{\tau,h}^{i})\\
        &\quad\qquad\geq \varepsilon\bigg(\alpha+\beta+\frac{1}{NL}\sum_{\tau=1}^{L}\sum_{i=1}^{N}\bbE_{\rho_{\tau,h}^{i}}\big[f_{W}(e_{\tau,h}^{i})\big]\bigg)\bigg)\\
        &\quad\leq 14\calN_{\bar{\calH}}\bigg(\frac{\varepsilon\beta}{40(B_{S}+rB_{K})^{3}B_{K}},\bbB(r,\bar{\calH})\bigg)\cdot \calN_{\infty}\bigg(\frac{\varepsilon\beta}{40(B_{S}+rB_{K})^{3}rL_{K}B_{k}},\tilde{\calW}\bigg)\\
        &\quad\qquad\cdot\exp\bigg(-\frac{\varepsilon^{2}(1-\varepsilon)\alpha NL}{20(B_{S}+rB_{K})^{4}(1+\varepsilon)}\bigg),
    \end{align*}
    where $\alpha,\beta>0$ and $0<\varepsilon\leq 1/2$.
\end{proposition}
\begin{proof}[Proof of Proposition~\ref{prop:frconcen}]
    See Appendix~\ref{app:frconcen}.
\end{proof}
Now consider,
\begin{align*}
    &\bbP\bigg(\frac{1}{NL}\sum_{\tau=1}^{L}\sum_{i=1}^{N}\bbE_{\rho_{\tau,h}^{i}}\big[\hatf_{\hatW}(e_{\tau,h}^{i})\big]-2\frac{1}{NL}\sum_{\tau=1}^{L}\sum_{i=1}^{N}\hatf_{\hatW}(e_{\tau,h}^{i})\geq t\bigg)\\
    &\quad\leq \bbP\bigg(\exists f_{W}\in\calF_{\tilde{\calW}},\frac{1}{NL}\sum_{\tau=1}^{L}\sum_{i=1}^{N}\bbE_{\rho_{\tau,h}^{i}}\big[f_{W}(e_{\tau,h}^{i})\big]-2\frac{1}{NL}\sum_{\tau=1}^{L}\sum_{i=1}^{N}f_{W}(e_{\tau,h}^{i})\geq t\bigg)\\
    &\quad=\bbP\bigg(\exists f_{W}\in\calF_{\tilde{\calW}}, \frac{1}{NL}\sum_{\tau=1}^{L}\sum_{i=1}^{N}\bbE_{\rho_{\tau,h}^{i}}\big[f_{W}(e_{\tau,h}^{i})\big]-\frac{1}{NL}\sum_{\tau=1}^{L}\sum_{i=1}^{N}f_{W}(e_{\tau,h}^{i})\\
    &\quad\qquad\geq \frac{1}{2}\bigg(t+\frac{1}{NL}\sum_{\tau=1}^{L}\sum_{i=1}^{N}\bbE_{\rho_{\tau,h}^{i}}\big[f_{W}(e_{\tau,h}^{i})\big]\bigg)\bigg)\\
    &\quad\leq 14\calN_{\bar{\calH}}\bigg(\frac{t}{160(B_{S}+rB_{K})^{3}B_{K}},\bbB(r,\bar{\calH})\bigg)\cdot \calN_{\infty}\bigg(\frac{t}{160(B_{S}+rB_{K})^{3}rL_{K}B_{k}},\tilde{\calW}\bigg)\\
    &\quad\qquad\cdot\exp\bigg(-\frac{tNL}{480(B_{S}+rB_{K})^{4}}\bigg),
\end{align*}
where the last inequality results from Proposition~\ref{prop:frconcen}. We define that
\begin{align*}
    N_{\bbB_{r}}=\calN_{\bar{\calH}}\bigg(\frac{3}{NL},\bbB(r,\bar{\calH})\bigg), N_{\tilde{\calW}}=\calN_{\infty}\bigg(\frac{3}{L_{K}NL},\tilde{\calW}\bigg).
\end{align*}
For $\delta>0$, we set
\begin{align*}
    t=\frac{480(B_{S}+rB_{K})^{4}}{NL}\log\frac{14N_{\bbB_{r}}N_{\tilde{\calW}}}{\delta},
\end{align*}
then we have that
\begin{align}
    \bbP\bigg(\frac{1}{NL}\sum_{\tau=1}^{L}\sum_{i=1}^{N}\bbE_{\rho_{\tau,h}^{i}}\big[\hatf_{\hatW}(e_{\tau,h}^{i})\big]-2\frac{1}{NL}\sum_{\tau=1}^{L}\sum_{i=1}^{N}\hatf_{\hatW}(e_{\tau,h}^{i})\geq t\bigg)\leq \delta.\label{ieq:32}
\end{align}
Combining inequalities~\eqref{ieq:8}, \eqref{ieq:33}, and \eqref{ieq:32}, we have that the following holds with probability at least $1-\delta$
\begin{align*}
    &\calR_{\bar{\xi}}(\hatf_{h},\hatg_{h},\hatW_{h})-\calR_{\bar{\xi}}(f_{h}^{*},g_{h}^{*},W_{h}^{*})\\
    &\quad\leq O\bigg(\frac{(B_{S}+rB_{K})^{4}}{NL}\log\frac{N_{\bbB_{r}}N_{\tilde{\bbB}_{\tilr}}N_{\tilde{\calW}}}{\delta}+\frac{(B_{S}+rB_{K})rL_{K}B_{k}}{\sqrt{N}}\log\frac{NL\calN_{\infty}(1/\sqrt{N},\tilde{\calW})}{\delta}\bigg),
\end{align*}
where
\begin{align*}
    N_{\tilde{\bbB}_{\tilr}}=\calN_{\bar{\calH}}\bigg(\frac{3}{NL},\bbB(\tilr,\tilde{\calH})\bigg).
\end{align*}
Thus, we conclude the proof of Theorem~\ref{thm:fixest}.
\end{proof}

\section{Proof of Theorem~\ref{thm:randest}}
\begin{proof}[Proof of Theorem~\ref{thm:randest}]
We first decompose the difference between the risk as the sum of the generalization error of risk from position sampling and the difference between the risk given the positions. Our proof involves two steps:
\begin{itemize}
    \item Bound the generalization error of risk from position sampling.
    \item Bound the difference between the risk given positions.
\end{itemize}
\begin{align}
    &\calR(\hatf_{h},\hatg_{h},\hatW_{h})-\calR(f_{h}^{*},g_{h}^{*},W_{h}^{*})\nonumber\\
    &\quad=\calR(\hatf_{h},\hatg_{h},\hatW_{h})-\calR_{\bar{\xi}}(\hatf_{h},\hatg_{h},\hatW_{h})-\big(\calR(f_{h}^{*},g_{h}^{*},W_{h}^{*})-\calR_{\bar{\xi}}(f_{h}^{*},g_{h}^{*},W_{h}^{*})\big)\nonumber\\
    &\quad\qquad+\calR_{\bar{\xi}}(\hatf_{h},\hatg_{h},\hatW_{h})-\calR_{\bar{\xi}}(f_{h}^{*},g_{h}^{*},W_{h}^{*})\nonumber\\
    &\quad\leq 2\sup_{f\in\bbB(r,\bar{\calH}),g\in\bbB(\tilr,\tilde{\calH}),W\in\tilde{\calW}}\big|\calR(f,g,W)-\calR_{\bar{\xi}}(f,g,W)\big|+\calR_{\bar{\xi}}(\hatf_{h},\hatg_{h},\hatW_{h})-\calR_{\bar{\xi}}(f_{h}^{*},g_{h}^{*},W_{h}^{*})\nonumber\\
    &\quad=\text{(IX)}+\text{(X)},\label{ieq:34}
\end{align}
where (IX) is the generalization error of risk from position sampling, and (X) is the difference between the risk given positions. Similar as the proof of Theorem~\ref{thm:fixest}, the terms related to the transition kernels and reward functions in inequality~\eqref{ieq:34} are analogous. In the following, we will only present the proof for the terms related to the transition kernel, and the results for the terms related to the reward functions can be similarly derived.

\textbf{Step 1: Bound the generalization error of risk from position sampling.}

We first define that
\begin{align*}
    g_{f,W}(\alpha)=\frac{1}{L}\sum_{\tau=1}^{L}\bbE_{\rho_{\tau,h}^{\alpha}}\bigg[\Big(s_{\tau,h+1}^{\alpha}-f\big(\omega_{\tau,h}^{\alpha}(W)\big)\Big)^{2}\bigg].
\end{align*}
The correspond function class for $g_{f,W}$ is $\calG_{\calF,\tilde{\calW}}=\{g_{f,W}\,|\,f\in\bbB(r,\bar{\calH}),W\in\tilde{\calW}\}$. Then term in (IX) that is related to the transition kernels can be expressed as
\begin{align*}
    2\sup_{g_{f,W}\in\calG_{\calF,\tilde{\calW}}}\bigg|\int_{0}^{1}g_{f,W}(\alpha)\rmd\alpha-\frac{1}{N}\sum_{i=1}^{N} g_{f,W}(\xi_{i})\bigg|.
\end{align*}
Let $\delta>0$, $\calG_{\delta}$ be a minimal $L_{\infty}$ $\delta-$cover of $\calG_{\calF,\tilde{\calW}}$. Then for any $g_{f,W}\in\calG_{\calF,\tilde{\calW}}$, there exists $\barg_{f,W}\in\calG_{\delta}$ such that $|g_{f,W}(\alpha)-\barg_{f,W}(\alpha)|\leq\delta$ for all $\alpha\in\calI$. For any $t>0$, we set $\delta=t/4$. Then we have that
\begin{align}
    &\bbP\bigg(\sup_{g_{f,W}\in\calG_{\calF,\tilde{\calW}}}\bigg|\int_{0}^{1}g_{f,W}(\alpha)\rmd\alpha-\frac{1}{N}\sum_{i=1}^{N} g_{f,W}(\xi_{i})\bigg|\geq t\bigg)\nonumber\\
    &\quad\leq \calN_{\infty}\bigg(\frac{t}{4},\calG_{\calF,\tilde{\calW}}\bigg)\max_{g_{f,W}\in\calG_{\frac{t}{4}}}\bbP\bigg(\bigg|\int_{0}^{1}g_{f,W}(\alpha)\rmd\alpha-\frac{1}{N}\sum_{i=1}^{N} g_{f,W}(\xi_{i})\bigg|\geq \frac{t}{2}\bigg)\nonumber\\
    &\quad\leq 2\calN_{\infty}\bigg(\frac{t}{4},\calG_{\calF,\tilde{\calW}}\bigg)\exp\bigg(-\frac{Nt^{2}}{2(B_{S}+rB_{K})^{4}}\bigg),\label{ieq:35}
\end{align}
where the first inequality results from the union bound, and the second inequality results from that $0\leq g_{f,W}(\alpha)\leq (B_{S}+rB_{K})^{2}$ and Hoeffding inequality. To upper bound the covering number in the tail probability, we note that
\begin{align*}
    \big|g_{f,W}(\alpha)-\barg_{f,W}(\alpha)\big|
    & \leq 2(B_{S}+rB_{K})\frac{1}{L}\sum_{\tau=1}^{L}\bbE_{\rho_{\tau,h}^{\alpha}}\bigg[\Big|f\big(\omega_{\tau,h}^{\alpha}(W)\big)-\barf\big(\omega_{\tau,h}^{\alpha}(\barW)\big)\Big|\bigg]\\
    & \leq 2(B_{S}+rB_{K})\big(B_{K}\|f-\barf\|_{\bar{\calH}}+rL_{K}B_{k}\|W-\barW\|_{\infty}\big),
\end{align*}
where the first inequality results from the definition of $g_{f,W}$, and the second inequality results from Lemma~\ref{lem:rkhs} and the triangle inequality. This inequality implies that
\begin{align*}
    \calN_{\infty}\bigg(\frac{t}{4},\calG_{\calF,\tilde{\calW}}\bigg)\leq \calN_{\bar{\calH}}\bigg(\frac{t}{16(B_{S}+rB_{K})B_{K}},\bbB(r,\bar{\calH})\bigg)\cdot \calN_{\infty}\bigg(\frac{t}{16(B_{S}+rB_{K})rL_{K}B_{k}},\tilde{\calW}\bigg).
\end{align*}
For $1>\delta>0$, we take
\begin{align*}
    t=\frac{\sqrt{2}(B_{S}+rB_{K})^{2}}{\sqrt{N}}\log\frac{2\calN_{\bar{\calH}}\bigg(\frac{1}{16\sqrt{N}},\bbB(r,\bar{\calH})\bigg)\cdot \calN_{\infty}\bigg(\frac{1}{16rL_{K}B_{k}\sqrt{N}},\tilde{\calW}\bigg)}{\delta}.
\end{align*}
Then inequality~\eqref{ieq:35} shows that
\begin{align}
    &\sup_{g_{f,W}\in\calG_{\calF,\tilde{\calW}}}\bigg|\int_{0}^{1}g_{f,W}(\alpha)\rmd\alpha-\frac{1}{N}\sum_{i=1}^{N} g_{f,W}(\xi_{i})\bigg|\nonumber\\
    &\quad= O\bigg(\frac{(B_{S}+rB_{K})^{2}}{\sqrt{N}}\log\frac{\calN_{\bar{\calH}}\bigg(\frac{1}{16\sqrt{N}},\bbB(r,\bar{\calH})\bigg)\cdot \calN_{\infty}\bigg(\frac{1}{16rL_{K}B_{k}\sqrt{N}},\tilde{\calW}\bigg)}{\delta}\bigg),\nonumber
\end{align}
with probability at least $1-\delta$. Thus, we have that
\begin{align}
    \text{(IX)}&=O\bigg(\frac{(B_{S}+rB_{K})^{2}}{\sqrt{N}}\log\frac{\calN_{\bar{\calH}}\bigg(\frac{1}{16\sqrt{N}},\bbB(r,\bar{\calH})\bigg)\cdot \calN_{\infty}\bigg(\frac{1}{16rL_{K}B_{k}\sqrt{N}},\tilde{\calW}\bigg)}{\delta}\nonumber\\
    &\qquad+\frac{(B_{S}+rB_{\tilK})^{2}}{\sqrt{N}}\log\frac{\calN_{\tilde{\calH}}\bigg(\frac{1}{16\sqrt{N}},\bbB(\tilr,\tilde{\calH})\bigg)\cdot \calN_{\infty}\bigg(\frac{1}{16rL_{K}B_{k}\sqrt{N}},\tilde{\calW}\bigg)}{\delta}\bigg)\label{ieq:36}
\end{align}

\textbf{Step 2: Bound the difference between the risk given positions.}

We adopt the similar procedures as the proof of Theorem~\ref{thm:fixest}. From inequalities~\eqref{ieq:8}, \eqref{ieq:9},and \eqref{ieq:16}, we define that
\begin{align}
    \text{(XI)}&=\frac{1}{NL}\sum_{\tau=1}^{L}\sum_{i=1}^{N}\bbE_{\rho_{\tau,h}^{i}}\bigg[\Big(s_{\tau,h+1}^{i}-\hatf_{h}\big(\omega_{\tau,h}^{i}(\hatW_{h})\big)\Big)^{2}-\Big(s_{\tau,h+1}^{i}-f_{h}^{*}\big(\omega_{\tau,h}^{i}(W_{h}^{*})\big)\Big)^{2}\bigg]\nonumber\\
    &\quad\qquad -2\frac{1}{NL}\sum_{\tau=1}^{L}\sum_{i=1}^{N}\Big(s_{\tau,h+1}^{i}-\hatf_{h}\big(\omega_{\tau,h}^{i}(\hatW_{h})\big)\Big)^{2}-\Big(s_{\tau,h+1}^{i}-f_{h}^{*}\big(\omega_{\tau,h}^{i}(W_{h}^{*})\big)\Big)^{2}\nonumber\\
    \text{(XII)}&=4(B_{S}+rB_{K})rL_{K}\sup_{W\in\tilde{\calW}}\frac{1}{NL}\sum_{\tau=1}^{L}\sum_{i=1}^{N}\big\|\homega_{\tau,h}^{i}(W)-\bar{\omega}_{\tau,h}^{i}(W)\big\|_{\calH}\nonumber\\
    \text{(XIII)}&=4(B_{S}+rB_{K})rL_{K}\sup_{W\in\tilde{\calW}}\frac{1}{NL}\sum_{\tau=1}^{L}\sum_{i=1}^{N}\big\|\bar{\omega}_{\tau,h}^{i}(W)-\omega_{\tau,h}^{i}(W)\big\|_{\calH}.\nonumber
\end{align}

For term (XIII), we adopt a different method with the proof of Theorem~\ref{thm:fixest}. Let $\varepsilon>0$, $\tilde{\calW}_{\varepsilon}$ be a $L_{\infty}$ $\varepsilon-$cover of $\tilde{\calW}$. Then for any $W\in\tilde{\calW}$, there exists $\barW\in\tilde{\calW}_{\varepsilon}$ such that $\|\barW-W\|_{\infty}\leq\varepsilon$. Then we have
\begin{align*}
     \|\omega_{\tau,h}^{i}(W)-\omega_{\tau,h}^{i}(\barW)\|_{\calH}
    & =\bigg\|\int_{0}^{1}\int_{\calS}\big(W(\xi_{i},\beta)-\barW(\xi_{i},\beta)\big)k\big(\cdot,(s_{\tau,h}^{i},a_{\tau,h}^{i},s)\big)\mu_{\tau,h}^{\beta}(s)\rmd s\rmd\beta\bigg\|_{\calH}\\
    & \leq \varepsilon B_{k},
\end{align*}
where the inequality results from the triangle inequality. Similarly, we have that $\|\bar{\omega}_{\tau,h}^{i}(W)-\bar{\omega}_{\tau,h}^{i}(\barW)\|_{\calH}\leq\varepsilon B_{k}$. For any $t>0$, we will set $\varepsilon=t/(4B_{k})$. Then the tail probability for (XIII) can be bounded as
\begin{align*}
    &\bbP\bigg(\sup_{W\in\tilde{\calW}}\frac{1}{NL}\sum_{\tau=1}^{L}\sum_{i=1}^{N}\big\|\bar{\omega}_{\tau,h}^{i}(W)-\omega_{\tau,h}^{i}(W)\big\|_{\calH}\geq t\bigg)\\
    &\quad\leq \bbP\bigg(\exists W\in\tilde{\calW}, \tau\in[L],i\in[N],\big\|\bar{\omega}_{\tau,h}^{i}(W)-\omega_{\tau,h}^{i}(W)\big\|_{\calH}\geq t\bigg)\\
    &\quad\leq NL\calN_{\infty}\bigg(\frac{t}{4B_{k}},\tilde{\calW}\bigg)\max_{W\in\tilde{\calW}_{\frac{t}{4B_{k}}}, \tau\in[L],i\in[N]}\bbP\bigg(\big\|\bar{\omega}_{\tau,h}^{i}(W)-\omega_{\tau,h}^{i}(W)\big\|_{\calH}\geq \frac{t}{2}\bigg)\\
    &\quad\leq 2NL\calN_{\infty}\bigg(\frac{t}{4B_{k}},\tilde{\calW}\bigg)\exp\bigg(-\frac{(N-1)t^{2}}{8B_{k}^{2}}\bigg),
\end{align*}
where the second inequality results from the union bound, and the last inequality resutls from Lemma~\ref{lem:hilconcen}. For any $0<\delta<1$, we set
\begin{align*}
    t=\frac{2\sqrt{2}B_{k}}{\sqrt{N-1}}\log\frac{2NL\calN_{\infty}\bigg(\frac{1}{\sqrt{N}},\tilde{\calW}\bigg)}{\delta}.
\end{align*}
Then we have that
\begin{align}
    \text{(XIII)}\leq O\bigg(\frac{(B_{S}+rB_{K})rL_{K}B_{k}}{\sqrt{N}}\log\frac{NL\calN_{\infty}\bigg(\frac{1}{\sqrt{N}},\tilde{\calW}\bigg)}{\delta}\bigg)\label{ieq:38}
\end{align}
with probability at least $1-\delta$.

For term (XII), we follow the proof of Theorem~\ref{thm:fixest} and condition on the values of $\bar{\xi}$ to bound the tail probability. We have that
\begin{align*}
    &\bbP\bigg(\sup_{W\in\tilde{\calW}}\frac{1}{NL}\sum_{\tau=1}^{L}\sum_{i=1}^{N}\big\|\homega_{\tau,h}^{i}(W)-\bar{\omega}_{\tau,h}^{i}(W)\big\|_{\calH}\geq t\bigg)\\
    &\quad=\bbE_{\bar{\xi}}\bigg[\bbP\bigg(\sup_{W\in\tilde{\calW}}\frac{1}{NL}\sum_{\tau=1}^{L}\sum_{i=1}^{N}\big\|\homega_{\tau,h}^{i}(W)-\bar{\omega}_{\tau,h}^{i}(W)\big\|_{\calH}\geq t\,\bigg|\,\bar{\xi}\bigg)\bigg]\\
    &\leq 2NL\calN_{\infty}(1/\sqrt{N},\tilde{\calW})\exp\bigg(-\frac{(N-1)t^{2}}{32B_{k}^{2}}\bigg),
\end{align*}
where we condition on the values of $\bar{\xi}$ in the first equality, and the inequality results from inequality~\eqref{ieq:15}. Thus, we have that
\begin{align}
    \text{(XII)}\leq O\bigg(\frac{(B_{S}+rB_{K})rL_{K}B_{k}}{\sqrt{N}}\log\frac{NL\calN_{\infty}(1/\sqrt{N},\tilde{\calW})}{\delta}\bigg)\label{ieq:39}
\end{align}
with probability at least $1-\delta$.

For term (XI), we just adopt the same conditional probability trick as shown in the bound of (XII). From inequality~\eqref{ieq:32}, we have that
\begin{align}
    \text{(XI)}\leq O\bigg(\frac{(B_{S}+rB_{K})^{4}}{NL}\log\frac{N_{\bbB_{r}}N_{\tilde{\calW}}}{\delta}\bigg)\label{ieq:40}
\end{align}
with probability at least $1-\delta$.

Combining the inequalities~\eqref{ieq:34}, \eqref{ieq:36}, \eqref{ieq:38}, \eqref{ieq:39}, and \eqref{ieq:40}, we have that
\begin{align*}
    &\calR(\hatf_{h},\hatg_{h},\hatW_{h})-\calR(f_{h}^{*},g_{h}^{*},W_{h}^{*})\\
    &\quad\leq O\bigg(\frac{\big(B_{S}+r\max\{B_{K},B_{\tilK}\}\big)^{2}}{\sqrt{N}}\log\frac{\calN_{\bar{\calH}}\bigg(\frac{1}{16\sqrt{N}},\bbB(r,\bar{\calH})\bigg)\cdot\calN_{\tilde{\calH}}\bigg(\!\!\frac{1}{16\sqrt{N}},\bbB(\tilr,\tilde{\calH})\!\!\bigg)\cdot \calN_{\infty}\bigg(\!\!\frac{1}{16rL_{K}B_{k}\sqrt{N}},\tilde{\calW}\!\!\bigg)}{\delta}\bigg)\\
    &\quad\qquad+\frac{\big(B_{S}+r\max\{B_{K},B_{\tilK}\}\big)r\max\{L_{K},L_{\tilK}B_{K}\}}{\sqrt{N}}\log\frac{NL\calN_{\infty}\bigg(\frac{1}{\sqrt{N}},\tilde{\calW}\bigg)}{\delta}\\
    &\quad\qquad+\frac{(B_{S}+r\max\{B_{K},B_{\tilK})^{4}}{NL}\log\frac{N_{\bbB_{r}}N_{\tilde{\bbB}_{\tilr}}N_{\tilde{\calW}}}{\delta}\bigg).
\end{align*}
Thus, we conclude the proof of Theorem~\ref{thm:randest}.
\end{proof}

\section{Proof of Theorem~\ref{thm:ufixest}}
\begin{proof}[Proof of Theorem~\ref{thm:ufixest}]
We first decompose the difference between the permutation-invariant risk as the sum of the generalization error of risk, the Estimation Error of Mean-embedding, and the empirical risk difference. Given the fact that the empirical risk difference is equal and less to zero, Our proof involves two steps:
\begin{itemize}
    \item Bound the estimation error of mean-embedding.
    \item Bound the generalization error of the risk.
\end{itemize}

\begin{align*}
    &\bar{\calR}_{\bar{\xi}}(\hatf_{h},\hatg_{h},\hatW_{h})-\bar{\calR}_{\bar{\xi}}(f_{h}^{*},g_{h}^{*},W_{h}^{*})\nonumber\\
    &\quad =\inf_{\phi\in\calB_{[0,1]}}\frac{1}{NL}\sum_{\tau=1}^{L}\sum_{i=1}^{N}\bbE_{\rho_{\tau,h}^{i}}\bigg[\Big(s_{\tau,h+1}^{i}-\hatf_{h}\big(\omega_{\tau,h}^{i}(\hatW_{h}^{\phi})\big)\Big)^{2}+\Big(r_{\tau,h}^{i}-\hatg_{h}\big(\omega_{\tau,h}^{i}(\hatW_{h}^{\phi})\big)\Big)^{2}\\
    &\quad\qquad-\Big(s_{\tau,h+1}^{i}-f_{h}^{*}\big(\omega_{\tau,h}^{i}(W_{h}^{*})\big)\Big)^{2}-\Big(r_{\tau,h}^{i}-g_{h}^{*}\big(\omega_{\tau,h}^{i}(W_{h}^{*})\big)\Big)^{2}\bigg]\nonumber\\
    &\quad\leq\inf_{\phi\in\calC_{[0,1]}^{N}}\frac{1}{NL}\sum_{\tau=1}^{L}\sum_{i=1}^{N}\bbE_{\rho_{\tau,h}^{i}}\bigg[\Big(s_{\tau,h+1}^{i}-\hatf_{h}\big(\omega_{\tau,h}^{i}(\hatW_{h}^{\phi})\big)\Big)^{2}+\Big(r_{\tau,h}^{i}-\hatg_{h}\big(\omega_{\tau,h}^{i}(\hatW_{h}^{\phi})\big)\Big)^{2}\\
    &\quad\qquad-\Big(s_{\tau,h+1}^{i}-f_{h}^{*}\big(\omega_{\tau,h}^{i}(W_{h}^{*})\big)\Big)^{2}-\Big(r_{\tau,h}^{i}-g_{h}^{*}\big(\omega_{\tau,h}^{i}(W_{h}^{*})\big)\Big)^{2}\bigg]\nonumber\\
    &=\text{Generalization Error of Risk}+\text{Estimation Error of Mean-embedding}+\text{Empirical Risk Difference},
\end{align*}
where each term is defined as
\begin{align*}
    &\text{Generalization Error of Risk}\\
    &\quad=\inf_{\phi\in\calC_{[0,1]}^{N}}\frac{1}{NL}\sum_{\tau=1}^{L}\sum_{i=1}^{N}\bbE_{\rho_{\tau,h}^{i}}\bigg[\Big(s_{\tau,h+1}^{i}-\hatf_{h}\big(\omega_{\tau,h}^{i}(\hatW_{h}^{\phi})\big)\Big)^{2}+\Big(r_{\tau,h}^{i}-\hatg_{h}\big(\omega_{\tau,h}^{i}(\hatW_{h}^{\phi})\big)\Big)^{2}\\
    &\quad\qquad-\Big(s_{\tau,h+1}^{i}-f_{h}^{*}\big(\omega_{\tau,h}^{i}(W_{h}^{*})\big)\Big)^{2}-\Big(r_{\tau,h}^{i}-g_{h}^{*}\big(\omega_{\tau,h}^{i}(W_{h}^{*})\big)\Big)^{2}\bigg]\nonumber\\
    &\quad\qquad -2\inf_{\phi\in\calC_{[0,1]}^{N}}\frac{1}{NL}\sum_{\tau=1}^{L}\sum_{i=1}^{N}\bigg[\Big(s_{\tau,h+1}^{i}-\hatf_{h}\big(\omega_{\tau,h}^{i}(\hatW_{h}^{\phi})\big)\Big)^{2}+\Big(r_{\tau,h}^{i}-\hatg_{h}\big(\omega_{\tau,h}^{i}(\hatW_{h}^{\phi})\big)\Big)^{2}\\
    &\quad\qquad-\Big(s_{\tau,h+1}^{i}-f_{h}^{*}\big(\omega_{\tau,h}^{i}(W_{h}^{*})\big)\Big)^{2}-\Big(r_{\tau,h}^{i}-g_{h}^{*}\big(\omega_{\tau,h}^{i}(W_{h}^{*})\big)\Big)^{2}\bigg].
\end{align*}
This generalization error of risk represents the error due to the fact that we optimize over the empirical estimation of the risk not the population risk.
\begin{align*}
    &\text{Estimation Error of Mean-embedding}\\
    &\quad=2\inf_{\phi\in\calC_{[0,1]}^{N}}\frac{1}{NL}\sum_{\tau=1}^{L}\sum_{i=1}^{N}\Big(s_{\tau,h+1}^{i}-\hatf_{h}\big(\omega_{\tau,h}^{i}(\hatW_{h}^{\phi})\big)\Big)^{2}+\Big(r_{\tau,h}^{i}-\hatg_{h}\big(\omega_{\tau,h}^{i}(\hatW_{h}^{\phi})\big)\Big)^{2}\\
    &\quad\qquad-2\inf_{\phi\in\calC_{[0,1]}^{N}}\frac{1}{NL}\sum_{\tau=1}^{L}\sum_{i=1}^{N}\Big(s_{\tau,h+1}^{i}-\hatf_{h}\big(\hbomega_{\tau,h}^{i}(\hatW_{h}^{\phi})\big)\Big)^{2}+\Big(r_{\tau,h}^{i}-\hatg_{h}\big(\hbomega_{\tau,h}^{i}(\hatW_{h}^{\phi})\big)\Big)^{2}\nonumber\\
    &\quad\qquad+2\frac{1}{NL}\sum_{\tau=1}^{L}\sum_{i=1}^{N}\Big(s_{\tau,h+1}^{i}-f_{h}^{*}\big(\hbomega_{\tau,h}^{i}(W_{h}^{*,\phi^{*}})\big)\Big)^{2}+\Big(r_{\tau,h}^{i}-g_{h}^{*}\big(\hbomega_{\tau,h}^{i}(W_{h}^{*,\phi^{*}})\big)\Big)^{2}\\
    &\quad\qquad-2\frac{1}{NL}\sum_{\tau=1}^{L}\sum_{i=1}^{N}\Big(s_{\tau,h+1}^{i}-f_{h}^{*}\big(\omega_{\tau,h}^{i}(W_{h}^{*})\big)\Big)^{2}+\Big(r_{\tau,h}^{i}-g_{h}^{*}\big(\omega_{\tau,h}^{i}(W_{h}^{*})\big)\Big)^{2}.
\end{align*}
Estimation error of mean-embedding represents the error due to the fact that we cannot observe the value of $\homega_{\tau,h}^{i}(\hatW_{h})$. Instead, we can only estimate the value of it through the states of sampled agents.
\begin{align*}
    &\text{Empirical Risk Difference}\\
    &\quad =2\inf_{\phi\in\calC_{[0,1]}^{N}}\frac{1}{NL}\sum_{\tau=1}^{L}\sum_{i=1}^{N}\Big(s_{\tau,h+1}^{i}-\hatf_{h}\big(\hbomega_{\tau,h}^{i}(\hatW_{h}^{\phi})\big)\Big)^{2}+\Big(r_{\tau,h}^{i}-\hatg_{h}\big(\hbomega_{\tau,h}^{i}(\hatW_{h}^{\phi})\big)\Big)^{2}\\
    &\quad\qquad-2\frac{1}{NL}\sum_{\tau=1}^{L}\sum_{i=1}^{N}\Big(s_{\tau,h+1}^{i}-f_{h}^{*}\big(\hbomega_{\tau,h}^{i}(W_{h}^{*,\phi^{*}})\big)\Big)^{2}+\Big(r_{\tau,h}^{i}-g_{h}^{*}\big(\hbomega_{\tau,h}^{i}(W_{h}^{*,\phi^{*}})\big)\Big)^{2},\nonumber
\end{align*}
where $\phi^{*}\in\calC_{[0,1]}^{N}$ is a permutation of $((i-1)/N,i/N]$ for $i\in[N]$ such that $\phi^{*}(i/N)=\xi_{i}$. From the estimation procedure of Algorithm~\eqref{algo:inestalgo2}, we have that
\begin{align*}
    \text{Empirical Risk Difference}\leq 0.
\end{align*}
Thus, we have that
\begin{align*}
    &\bar{\calR}_{\bar{\xi}}(\hatf_{h},\hatg_{h},\hatW_{h})-\bar{\calR}_{\bar{\xi}}(f_{h}^{*},g_{h}^{*},W_{h}^{*})\nonumber\\
    &\quad\leq \text{Generalization Error of Risk}+\text{Estimation Error of Mean-embedding}.
\end{align*}

\textbf{Step 1: Bound the Estimation Error of Mean-embedding.}

From the definition of the generalization error, we bound two terms separately
\begin{align}
    &\text{Estimation Error of Mean-embedding}\nonumber\\
    &\quad\leq 2\sup_{f,g,W}\bigg|\inf_{\phi\in\calC_{[0,1]}^{N}}\frac{1}{NL}\sum_{\tau=1}^{L}\sum_{i=1}^{N}\Big(s_{\tau,h+1}^{i}-f\big(\hbomega_{\tau,h}^{i}(W^{\phi})\big)\Big)^{2}+\Big(r_{\tau,h}^{i}-g\big(\hbomega_{\tau,h}^{i}(W^{\phi})\big)\Big)^{2}\nonumber\\
    &\quad\qquad-\inf_{\phi\in\calC_{[0,1]}^{N}}\frac{1}{NL}\sum_{\tau=1}^{L}\sum_{i=1}^{N}\Big(s_{\tau,h+1}^{i}-f\big(\omega_{\tau,h}^{i}(W^{\phi})\big)\Big)^{2}+\Big(r_{\tau,h}^{i}-g\big(\omega_{\tau,h}^{i}(W^{\phi})\big)\Big)^{2}\bigg|\nonumber\\
    &\quad\qquad+2\frac{1}{NL}\sum_{\tau=1}^{L}\sum_{i=1}^{N}\Big(s_{\tau,h+1}^{i}-f_{h}^{*}\big(\hbomega_{\tau,h}^{i}(W_{h}^{*,\phi^{*}})\big)\Big)^{2}+\Big(r_{\tau,h}^{i}-g_{h}^{*}\big(\hbomega_{\tau,h}^{i}(W_{h}^{*,\phi^{*}})\big)\Big)^{2}\nonumber\\
    &\quad\qquad-2\frac{1}{NL}\sum_{\tau=1}^{L}\sum_{i=1}^{N}\Big(s_{\tau,h+1}^{i}-f_{h}^{*}\big(\omega_{\tau,h}^{i}(W_{h}^{*})\big)\Big)^{2}+\Big(r_{\tau,h}^{i}-g_{h}^{*}\big(\omega_{\tau,h}^{i}(W_{h}^{*})\big)\Big)^{2}\nonumber\\
    &\quad= \text{(XIV)}+\text{(XV)}\label{ieq:41}
\end{align}
We first denote the composition of two measure-preserving bijections $\phi$ and $\psi$ as $\phi\circ\psi$. When appiled to a graphon $W$, the composition of bijections maps the values of the graphon as
\begin{align*}
    W^{\phi\circ\psi}(x,y)=W\Big(\phi\big(\psi(x)\big),\phi\big(\psi(y)\big)\Big).
\end{align*}

Then we bound the term (XIV) as
\begin{align}
    &\sup_{f,g,W}\bigg|\inf_{\phi\in\calC_{[0,1]}^{N}}\frac{1}{NL}\sum_{\tau=1}^{L}\sum_{i=1}^{N}\Big(s_{\tau,h+1}^{i}-f\big(\hbomega_{\tau,h}^{i}(W^{\phi\circ\phi^{*}})\big)\Big)^{2}+\Big(r_{\tau,h}^{i}-g\big(\hbomega_{\tau,h}^{i}(W^{\phi\circ\phi^{*}})\big)\Big)^{2}\nonumber\\
    &\quad\qquad-\inf_{\phi\in\calC_{[0,1]}^{N}}\frac{1}{NL}\sum_{\tau=1}^{L}\sum_{i=1}^{N}\Big(s_{\tau,h+1}^{i}-f\big(\omega_{\tau,h}^{i}(W^{\phi})\big)\Big)^{2}+\Big(r_{\tau,h}^{i}-g\big(\omega_{\tau,h}^{i}(W^{\phi})\big)\Big)^{2}\bigg|\nonumber\\
    &\quad=\sup_{f,W,\phi}\bigg|\frac{1}{NL}\sum_{\tau=1}^{L}\sum_{i=1}^{N}\Big(s_{\tau,h+1}^{i}-f\big(\hbomega_{\tau,h}^{i}(W^{\phi\circ\phi^{*}})\big)\Big)^{2}+\Big(r_{\tau,h}^{i}-g\big(\hbomega_{\tau,h}^{i}(W^{\phi\circ\phi^{*}})\big)\Big)^{2}\nonumber\\
    &\quad\qquad-\frac{1}{NL}\sum_{\tau=1}^{L}\sum_{i=1}^{N}\Big(s_{\tau,h+1}^{i}-f\big(\omega_{\tau,h}^{i}(W^{\phi})\big)\Big)^{2}+\Big(r_{\tau,h}^{i}-g\big(\omega_{\tau,h}^{i}(W^{\phi})\big)\Big)^{2}\bigg|\nonumber\\
    &\quad\leq 4(B_{S}+\barr\barB_{K})\barr\barL_{K}\sup_{W\in\tilde{\calW},\phi\in\calC_{[0,1]}^{N}}\frac{1}{NL}\sum_{\tau=1}^{L}\sum_{i=1}^{N}\big\|\hbomega_{\tau,h}^{i}(W^{\phi\circ\phi^{*}})-\omega_{\tau,h}^{i}(W^{\phi})\big\|_{\calH},\label{ieq:45}
\end{align}
where the equality results from the fact that $\phi^{*}$ is a measure-preserving bijection, and the inequality results from the same arguments in inequality~\eqref{ieq:9}.

We decompose the error as
\begin{align*}
    \sup_{W,\phi}\big\|\hbomega_{\tau,h}^{i}(W^{\phi\circ\phi^{*}})-\omega_{\tau,h}^{i}(W^{\phi})\big\|_{\calH}&\leq \sup_{W,\phi}\big\|\bar{\omega}_{\tau,h}^{i}(W^{\phi})-\omega_{\tau,h}^{i}(W^{\phi})\big\|_{\calH}+\sup_{W,\phi}\big\|\hbomega_{\tau,h}^{i}(W^{\phi\circ\phi^{*}})-\bar{\omega}_{\tau,h}^{i}(W^{\phi})\big\|_{\calH}\nonumber
\end{align*}
where
\begin{align*}
    \omega_{\tau,h}^{i}(W^{\phi})&=\int_{0}^{1}\int_{\calS}W\big(\phi(\xi_{i}),\phi(\beta)\big)k\big(\cdot,(s_{\tau,h}^{i},a_{\tau,h}^{i},s)\big)\mu_{\tau,h}^{\beta}(s)\rmd s\rmd\beta,\\
    \bar{\omega}_{\tau,h}^{i}(W^{\phi})&=\frac{1}{N-1}\sum_{j\neq i}W\big(\phi(\xi_{i}),\phi(\xi_{j})\big)\int_{\calS}k\big(\cdot,(s_{\tau,h}^{i},a_{\tau,h}^{i},s)\big)\mu_{\tau,h}^{j}(s)\rmd s,\\
    \hbomega_{\tau,h}^{i}(W^{\phi\circ\phi^{*}})&=\frac{1}{(N-1)L}\sum_{j\neq i}\sum_{\tau^{\prime}=1}^{L}W\big(\phi(\xi_{i}),\phi(\xi_{j})\big)k\big(\cdot,(s_{\tau,h}^{i},a_{\tau,h}^{i},s_{\tau^{\prime},h}^{j})\big).
\end{align*}

For term $\sup_{W,\phi}\big\|\bar{\omega}_{\tau,h}^{i}(W^{\phi})-\omega_{\tau,h}^{i}(W^{\phi})\big\|_{\calH}$, we define the interval $\calI_{i}=((i-1)/N,i/N]$ for $i\in[N]$. Then we have that
\begin{align*}
    &\big\|\bar{\omega}_{\tau,h}^{i}(W^{\phi})-\omega_{\tau,h}^{i}(W^{\phi})\big\|_{\calH}\\
    &\quad\leq \frac{2}{N}B_{k}+\sum_{j\neq i}\bigg\|\int_{\xi_{j}-\frac{1}{N}}^{\xi_{j}}\int_{\calS}W\big(\phi(\xi_{i}),\phi(\beta)\big)k\big(\cdot,(s_{\tau,h}^{i},a_{\tau,h}^{i},s)\big)\mu_{\tau,h}^{\beta}(s)\rmd s\rmd\beta\\
    &\quad\qquad\qquad\qquad\qquad -\frac{1}{N}\sum_{j\neq i}W\big(\phi(\xi_{i}),\phi(\xi_{j})\big)\int_{\calS}k\big(\cdot,(s_{\tau,h}^{i},a_{\tau,h}^{i},s)\big)\mu_{\tau,h}^{j}(s)\rmd s\bigg\|_{\calH},
\end{align*}
where the inequality results from the triangle inequality. For each term in the sum, we bound it as
\begin{align*}
    &\bigg\|\int_{\xi_{j}-\frac{1}{N}}^{\xi_{j}}\int_{\calS}W\big(\phi(\xi_{i}),\phi(\beta)\big)k\big(\cdot,(s_{\tau,h}^{i},a_{\tau,h}^{i},s)\big)\mu_{\tau,h}^{\beta}(s)\rmd s\rmd\beta\\
    &\qquad -\frac{1}{N}\sum_{j\neq i}W\big(\phi(\xi_{i}),\phi(\xi_{j})\big)\int_{\calS}k\big(\cdot,(s_{\tau,h}^{i},a_{\tau,h}^{i},s)\big)\mu_{\tau,h}^{j}(s)\rmd s\bigg\|_{\calH}\\
    &\quad\leq \bigg\|\int_{\xi_{j}-\frac{1}{N}}^{\xi_{j}}\int_{\calS}W\big(\phi(\xi_{i}),\phi(\beta)\big)k\big(\cdot,(s_{\tau,h}^{i},a_{\tau,h}^{i},s)\big)\big(\mu_{\tau,h}^{\beta}(s)-\mu_{\tau,h}^{j}(s)\big)\rmd s\rmd\beta\bigg\|_{\calH}\\
    &\quad\qquad +\bigg\|\int_{\xi_{j}-\frac{1}{N}}^{\xi_{j}}\int_{\calS}\Big(W\big(\phi(\xi_{i}),\phi(\beta)\big)-W\big(\phi(\xi_{i}),\phi(\xi_{j})\Big)k\big(\cdot,(s_{\tau,h}^{i},a_{\tau,h}^{i},s)\big)\mu_{\tau,h}^{j}(s)\rmd s\rmd\beta\bigg\|_{\calH}\\
    &\quad=O\bigg( \frac{B_{k}}{N^{2}}\bigg),
\end{align*}
where the first inequality results from the triangle inequality, and the second inequality results from the same argument in inequality~\eqref{ieq:13} and the fact that $\beta$ and $\xi_{j}$ are always in the same interval for any $\phi\in\calC_{[0,1]}^{N}$. Thus, we have that
\begin{align}
    \sup_{W,\phi}\big\|\bar{\omega}_{\tau,h}^{i}(W^{\phi})-\omega_{\tau,h}^{i}(W^{\phi})\big\|_{\calH}=O\bigg(\frac{B_{k}}{N}\bigg).\label{ieq:42}
\end{align}

For $\sup_{W,\phi}\big\|\hbomega_{\tau,h}^{i}(W^{\phi\circ\phi^{*}})-\bar{\omega}_{\tau,h}^{i}(W^{\phi})\big\|_{\calH}$, we adopt the similar procedure in the proof of inequality~\eqref{ieq:15}.
\begin{align*}
    &\bbP\Big(\sup_{W,\phi}\big\|\hbomega_{\tau,h}^{i}(W^{\phi\circ\phi^{*}})-\bar{\omega}_{\tau,h}^{i}(W^{\phi})\big\|_{\calH}\geq t\Big)\\
    &\quad\leq N!NL\calN_{\infty}(t/(4B_{k}),\tilde{\calW})\max_{j\in[\calN_{\infty}],i\in[N],\tau\in[L]}\bbP\Big( \big\|\hbomega_{\tau,h}^{i}(W_{j}^{\phi\circ\phi^{*}})-\bar{\omega}_{\tau,h}^{i}(W_{j}^{\phi})\big\|_{\calH}\geq t/2\Big)\\
    &\quad\leq 2N!NL\calN_{\infty}(t/(4B_{k}),\tilde{\calW})\exp\bigg(-\frac{NLt^{2}}{16B_{k}^{2}}\bigg),
\end{align*}
where the first inequality results from the proof of inequality~\eqref{ieq:15}, and the last inequality results from Lemma~\ref{lem:hilconcen}. Thus, we have that with probability at least $1-\delta$
\begin{align}
    \sup_{W,\phi}\big\|\hbomega_{\tau,h}^{i}(W^{\phi\circ\phi^{*}})-\bar{\omega}_{\tau,h}^{i}(W^{\phi})\big\|_{\calH}=O\bigg(B_{k}\sqrt{\frac{N}{L}}\log\frac{NL \calN_{\infty}(\sqrt{N/L},\tilde{\calW})}{\delta}\bigg).\label{ieq:43}
\end{align}

Combining inequalities~\eqref{ieq:45}, \eqref{ieq:42} and \eqref{ieq:43}, we have that
\begin{align}
    \text{(XIV)}=O\bigg(\frac{B_{k}\barr\barL_{K}(B_{S}+\barr\barB_{K})}{N}+(B_{S}+\barr\barB_{K})\barr\barL_{K}B_{k}\sqrt{\frac{N}{L}}\log\frac{NL \calN_{\infty}(\sqrt{N/L},\tilde{\calW})}{\delta}\bigg)\label{ieq:46}.
\end{align}

Following the similar arguments, we can derive that
\begin{align*}
    \text{(XV)}=O\bigg(\frac{B_{k}\barr\barL_{K}(B_{S}+\barr\barB_{K})}{N}+(B_{S}+\barr\barB_{K})\barr\barL_{K}B_{k}\frac{1}{\sqrt{NL}}\log\frac{NL \calN_{\infty}(\sqrt{N/L},\tilde{\calW})}{\delta}\bigg).
\end{align*}

\textbf{Step 2: Bound the generalization error of risk}

We follow the similar procedures in Step 2 of the proof of Theorem~\ref{thm:fixest}. We denote the quadruple $(s_{\tau,h}^{i},a_{\tau,h}^{i},\mu_{\tau,h}^{\calI},s_{\tau,h+1}^{i})$ as $e_{\tau,h}^{i}$. We define the function $f_{W}$ as 
\begin{align*}
    f(e_{\tau,h}^{i},W,\phi)=\Big(s_{\tau,h+1}^{i}-f\big(\omega_{\tau,h}^{i}(W^{\phi})\big)\Big)^{2}-\Big(s_{\tau,h+1}^{i}-f_{h}^{*}\big(\omega_{\tau,h}^{i}(W_{h}^{*})\big)\Big)^{2}.
\end{align*}
Then we have that
\begin{align*}
    &\bbP\bigg(\inf_{\phi\in\calC_{[0,1]}^{N}}\frac{1}{NL}\sum_{\tau=1}^{L}\sum_{i=1}^{N}\bbE_{\rho_{\tau,h}^{i}}\big[\hatf_{h}(e_{\tau,h}^{i},\hatW_{h},\phi)\big] -2\inf_{\phi\in\calC_{[0,1]}^{N}}\frac{1}{NL}\sum_{\tau=1}^{L}\hatf_{h}(e_{\tau,h}^{i},\hatW_{h},\phi)\geq t\bigg)\\
    &\quad\leq \bbP\bigg(\exists f\in\bbB(r,\bar{\calH}),W\in\tilde{\calW}, \max_{\phi\in\calC_{[0,1]}^{N}}\Big[\frac{1}{NL}\sum_{\tau=1}^{L}\sum_{i=1}^{N}\bbE_{\rho_{\tau,h}^{i}}\big[f(e_{\tau,h}^{i},W,\phi)\big] -2\frac{1}{NL}\sum_{\tau=1}^{L}f(e_{\tau,h}^{i},W,\phi)\Big]\geq t\bigg)\\
    &\quad\leq N!\max_{\phi\in\calC_{[0,1]}^{N}}\bbP\bigg(\exists f\in\bbB(r,\bar{\calH}),W\in\tilde{\calW}, \frac{1}{NL}\sum_{\tau=1}^{L}\sum_{i=1}^{N}\bbE_{\rho_{\tau,h}^{i}}\big[f(e_{\tau,h}^{i},W,\phi)\big] -2\frac{1}{NL}\sum_{\tau=1}^{L}f(e_{\tau,h}^{i},W,\phi)\geq t\bigg)\\
    &\quad\leq 14N!\calN_{\bar{\calH}}\bigg(\frac{t}{160(B_{S}+rB_{K})^{3}B_{K}},\bbB(r,\bar{\calH})\bigg)\cdot \calN_{\infty}\bigg(\frac{t}{160(B_{S}+rB_{K})^{3}rL_{K}B_{k}},\tilde{\calW}\bigg),
\end{align*}
where the second inequality results from the union bound and the fact that $\min_{x}f(x)-\min_{x}g(x)\leq \max_{x}f(x)-g(x)$, and the final inequality results from Proposition~\ref{prop:frconcen}. Thus, we have that with probability at least $1-\delta$
\begin{align}
    &\inf_{\phi\in\calC_{[0,1]}^{N}}\frac{1}{NL}\sum_{\tau=1}^{L}\sum_{i=1}^{N}\bbE_{\rho_{\tau,h}^{i}}\big[\hatf_{h}(e_{\tau,h}^{i},\hatW_{h},\phi)\big] -2\inf_{\phi\in\calC_{[0,1]}^{N}}\frac{1}{NL}\sum_{\tau=1}^{L}\hatf_{h}(e_{\tau,h}^{i},\hatW_{h},\phi)\nonumber\\
    &\quad=O\bigg(\frac{(B_{S}+rB_{K})^{4}}{L}\log\frac{N\tilN_{\bbB_{r}}\tilN_{\infty}}{\delta}\bigg),\label{ieq:44}
\end{align}
where
\begin{align*}
    \tilN_{\bbB_{r}}=\calN_{\bar{\calH}}\bigg(\frac{3}{L},\bbB(r,\bar{\calH})\bigg), \tilN_{\tilde{\calW}}=\calN_{\infty}\bigg(\frac{3}{L_{K}L},\tilde{\calW}\bigg).
\end{align*}

Combining inequalities~\eqref{ieq:44} and \eqref{ieq:46}, we have that
\begin{align*}
    &\bar{\calR}_{\bar{\xi}}(\hatf_{h},\hatg_{h},\hatW_{h})-\bar{\calR}_{\bar{\xi}}(f_{h}^{*},g_{h}^{*},W_{h}^{*})\nonumber\\
    &\quad=O\bigg(\frac{B_{k}\barr\barL_{K}(B_{S}+\barr\barB_{K})}{N}+(B_{S}+\barr\barB_{K})\barr\barL_{K}\barB_{k}\sqrt{\frac{N}{L}}\log\frac{NL \calN_{\infty}(\sqrt{N/L},\tilde{\calW})}{\delta}\\
    &\quad\qquad+\frac{(B_{S}+\barr\barB_{K})^{4}}{L}\log\frac{N\tilN_{\bbB_{r}}\tilN_{\tilde{\bbB}_{\tilr}}\tilN_{\infty}}{\delta}\bigg),
\end{align*}
where
\begin{align*}
    \tilN_{\tilde{\bbB}_{\tilr}}=\calN_{\tilde{\calH}}\bigg(\frac{3}{L},\bbB(\tilr,\tilde{\calH})\bigg).
\end{align*}

Thus, we conclude the proof of Theorem~\ref{thm:ufixest}.

\end{proof}

\section{Proof of Corollary~\ref{coro:fixest}}
\begin{proof}[Proof of Corollary~\ref{coro:fixest}]
    The proof of Corollary~\ref{coro:fixest} follows the same procedures as the proof of Theorem~\ref{thm:fixest}. The only difference is that inequality~\eqref{ieq:15} in the proof of Theorem~\ref{thm:fixest} is replaced by
    \begin{align*}
        \text{(IV)}=\sup_{W\in\tilde{\calW}}\big\|\homega_{\tau,h}^{i}(W)-\bar{\omega}_{\tau,h}^{i}(W)\big\|_{\calH}\leq \frac{4\sqrt{2}B_{k}}{\sqrt{(N-1)L}}\log\frac{2NL\calN_{\infty}(1/\sqrt{NL},\tilde{\calW})}{\delta}.
    \end{align*}
\end{proof}

\section{Proof of Corollary~\ref{coro:optknownfix}}
\begin{proof}[Proof of Corollary~\ref{coro:optknownfix}]
    Our proof mainly involves four steps
    \begin{itemize}
        \item Derive the performance guarantee of Algorithm~\eqref{eq:fixdistest}.
        \item Generalize the performance guarantee from $\{\xi_{i}\}_{i}^{N}$ to $[0,1]$ by lipschitzness.
        \item Bound the estimation error of distribution flow and action-value function estimate.
        \item Conclude the final result.
    \end{itemize}
    
    \textbf{Step 1: Derive the performance guarantee of Algorithm~\eqref{eq:fixdistest}.}
    
    We first derive the performance guarantee of Algorithm~\eqref{eq:fixdistest} when we we sample agents with known grid positions. In such setting, we implement $\pi^{\calI}$ for $L$ times on the \ac{mdp} induced by $\mu^{\calI}$ to collect the dataset $\calD_{\tau}=\{(s_{\tau,h}^{[N]},a_{\tau,h}^{[N]},r_{\tau,h}^{[N]},s_{\tau,h+1}^{[N]})\}_{h=1}^{H}$ for $\tau\in[L]$. We define $\mu^{+,\calI}=\Gamma_{3}(\pi^{\calI},\mu^{\calI},W^{*})$ as the distribution flow of implementing $\pi^{\calI}$ on the \ac{mdp} induced by $\mu^{\calI}$. Then the joint distribution of $(s_{\tau,h}^{i},a_{\tau,h}^{i},r_{\tau,h}^{i},s_{\tau,h+1}^{i})_{i=1}^{N}$ is $\prod_{i=1}^{N}\rho_{\tau,h}^{+,i}$, where $\rho_{\tau,h}^{+,i}=\mu_{\tau,h}^{+,i}\times\pi_{\tau,h}^{i}\times \delta_{r_{h}^{*}}\times P_{h}^{*}$. With a little abuse of notation, we define the risk of $(f,g,W)$ given $\bar{\xi}$ as
    \begin{align*}
        \calR_{\bar{\xi}}(f,g,W)&=\frac{1}{NL}\sum_{\tau=1}^{L}\sum_{i=1}^{N}\bbE_{\rho_{\tau,h}^{+,i}}\bigg[\Big(s_{\tau,h+1}^{i}-f\big(\omega_{\tau,h}^{i}(W)\big)\Big)^{2}+\Big(r_{\tau,h}^{i}-g\big(\omega_{\tau,h}^{i}(W)\big)\Big)^{2}\bigg]\\
        &=\frac{1}{N}\sum_{i=1}^{N}\bbE_{\rho_{h}^{+,i}}\bigg[\Big(s_{h+1}^{i}-f\big(\omega_{h}^{i}(W)\big)\Big)^{2}+\Big(r_{h+1}^{i}-g\big(\omega_{h}^{i}(W)\big)\Big)^{2}\bigg],
    \end{align*}
    where the second equality results from that we implement the same policy for $L$ times. The difference between this definition and Eqn.~\eqref{eq:condirisk} is that we take expectation with respect to $\rho_{\tau,h}^{+,i}$ instead of $\rho_{\tau,h}^{i}$. The reason is that in the setting where we specify Eqn.~\eqref{eq:condirisk}, the \ac{mdp} is induced by the distribution flow of the policy itself, not by a pre-specified distribution flow. We state the performance guarantee as
    \begin{corollary}\label{coro:knownfixest}
        Under Assumptions~\ref{assump:graphon}, \ref{assump:kernel}, \ref{assump:real}, and \ref{assump:lipconti}, if $\xi_{i}=i/N$ for $i\in[N]$, then the risk of estimate derived in Algorithm~\eqref{eq:fixdistest} can be bounded as
        \begin{align*}
            &\calR_{\bar{\xi}}(\hatf_{h},\hatg_{h},\hatW_{h})-\calR_{\bar{\xi}}(f_{h}^{*},g_{h}^{*},W_{h}^{*})\\
            &\quad\leq O\bigg(\frac{(B_{S}+\barr\barB_{K})^{4}}{NL}\log\frac{N_{\bbB_{r}}N_{\tilde{\bbB}_{\tilr}}N_{\tilde{\calW}}}{\delta}\bigg)
        \end{align*}
        with probability at least $1-\delta$, where $N_{\bbB_{r}}$, $N_{\tilde{\bbB}_{\tilr}}$, and $N_{\tilde{\calW}}$ are defined in Theorem~\ref{thm:fixest}.
    \end{corollary}
    \begin{proof}[Proof of Corollary~\ref{coro:knownfixest}]
        See Appendix~\ref{app:knownfixest}.
    \end{proof}
    
    \textbf{Step 2: Generalize the performance guarantee from $\{\xi_{i}\}_{i=1}^{N}$ to $[0,1]$ by lipschitzness.}
    
    Intuitively, when the implemented policy is lipschitz, we can generalize the performance guarantee of $\calR_{\bar{\xi}}(f,g,W)$ to that of $\calR(f,g,W)$. Here we consider the case where the \ac{mdp} is induced by the distribution flow of the policy itself, i.e., the case specified in Section~\ref{sec:modellearn}. The results for the case where the \ac{mdp} is induced by a pre-specified distribution flow can be similarly derived. We note that
    \begin{align}
        &\calR(\hatf_{h},\hatg_{h},\hatW_{h})-\calR(f_{h}^{*},g_{h}^{*},W_{h}^{*})\nonumber\\
        &\quad=\calR(\hatf_{h},\hatg_{h},\hatW_{h})-\calR_{\bar{\xi}}(\hatf_{h},\hatg_{h},\hatW_{h})-\big(\calR(f_{h}^{*},g_{h}^{*},W_{h}^{*})-\calR_{\bar{\xi}}(f_{h}^{*},g_{h}^{*},W_{h}^{*})\big)\nonumber\\
        &\quad\qquad+\calR_{\bar{\xi}}(\hatf_{h},\hatg_{h},\hatW_{h})-\calR_{\bar{\xi}}(f_{h}^{*},g_{h}^{*},W_{h}^{*})\nonumber\\
        &\quad\leq 2\sup_{f\in\bbB(r,\bar{\calH}),g\in\bbB(\tilr,\tilde{\calH}),W\in\tilde{\calW}}\big|\calR(f,g,W)-\calR_{\bar{\xi}}(f,g,W)\big|+\calR_{\bar{\xi}}(\hatf_{h},\hatg_{h},\hatW_{h})-\calR_{\bar{\xi}}(f_{h}^{*},g_{h}^{*},W_{h}^{*}).\label{ieq:72}
    \end{align}
    Then we attempt to bound the first term of the right-hand side of inequality~\eqref{ieq:72}. For any two positions $\alpha,\beta\in\calI$ and $f\in\bbB(r,\bar{\calH})$, we have
    \begin{align}
        &\bigg|\bbE_{\rho_{\tau,h}^{\alpha}}\bigg[\Big(s_{\tau,h+1}-f\big(\omega_{\tau,h}^{\alpha}(W)\big)\Big)^{2}\bigg]-\bbE_{\rho_{\tau,h}^{\beta}}\bigg[\Big(s_{\tau,h+1}-f\big(\omega_{\tau,h}^{\beta}(W)\big)\Big)^{2}\bigg]\bigg|\nonumber\\
        &\quad\leq \bigg|\bbE_{\rho_{\tau,h}^{\alpha}}\bigg[\Big(s_{\tau,h+1}-f\big(\omega_{\tau,h}^{\alpha}(W)\big)\Big)^{2}\bigg]-\bbE_{\rho_{\tau,h}^{\beta}}\bigg[\Big(s_{\tau,h+1}-f\big(\omega_{\tau,h}^{\alpha}(W)\big)\Big)^{2}\bigg]\bigg|\nonumber\\
        &\quad\qquad +\bigg|\bbE_{\rho_{\tau,h}^{\beta}}\bigg[\Big(s_{\tau,h+1}-f\big(\omega_{\tau,h}^{\alpha}(W)\big)\Big)^{2}\bigg]-\bbE_{\rho_{\tau,h}^{\beta}}\bigg[\Big(s_{\tau,h+1}-f\big(\omega_{\tau,h}^{\beta}(W)\big)\Big)^{2}\bigg]\bigg|,\label{ieq:73}
    \end{align}
    where the inequality results from the triangle inequality. For the first term in the right-hand side of inequality~\eqref{ieq:73}, we have that
    \begin{align*}
        &\bigg|\bbE_{\rho_{\tau,h}^{\alpha}}\bigg[\Big(s_{\tau,h+1}-f\big(\omega_{\tau,h}^{\alpha}(W)\big)\Big)^{2}\bigg]-\bbE_{\rho_{\tau,h}^{\beta}}\bigg[\Big(s_{\tau,h+1}-f\big(\omega_{\tau,h}^{\alpha}(W)\big)\Big)^{2}\bigg]\bigg|\\
        &\quad \leq (B_{S}+rB_{K})^{2}\Big[\|\mu_{\tau,h}^{\alpha}-\mu_{\tau,h}^{\beta}\|_{1}+\bbE_{\mu_{\tau,h}^{\alpha}}\big[\|\pi_{\tau,h}^{\alpha}(\cdot\,|\,s)-\pi_{\tau,h}^{\beta}(\cdot\,|\,s)\|_{1}\big]\\
        &\quad\qquad+L_{P}\big\|z_{h}^{\alpha}(\mu_{\tau,h}^{\calI},W_{h}^{*})-z_{h}^{\beta}(\mu_{\tau,h}^{\calI},W_{h}^{*})\big\|_{1}\Big]\\
        &\quad \leq C(B_{S}+rB_{K})^{2}\cdot|\alpha-\beta|,
    \end{align*}
    where $C>0$ is a constant, the first inequality results from the definition of $\rho_{\tau,h}^{\calI}$, and the last inequality adopts Proposition~\ref{prop:agentlip} and Assumption~\ref{assump:graphon} to bound these three terms. The second term in the right-hand side of inequality~\eqref{ieq:73} can be bounded as
    \begin{align*}
        &\bigg|\bbE_{\rho_{\tau,h}^{\beta}}\bigg[\Big(s_{\tau,h+1}-f\big(\omega_{\tau,h}^{\alpha}(W)\big)\Big)^{2}\bigg]-\bbE_{\rho_{\tau,h}^{\beta}}\bigg[\Big(s_{\tau,h+1}-f\big(\omega_{\tau,h}^{\beta}(W)\big)\Big)^{2}\bigg]\bigg|\\
        &\quad\leq 2(B_{S}+rB_{K})rL_{K}L_{\bar{\calW}}B_{k}|\alpha-\beta|,
    \end{align*}
    where the inequality results from Lemma~\ref{lem:rkhs} and Assumption~\ref{assump:graphon}. Thus, we conclude that
    \begin{align*}
        &\bigg|\bbE_{\rho_{\tau,h}^{\alpha}}\bigg[\Big(s_{\tau,h+1}-f\big(\omega_{\tau,h}^{\alpha}(W)\big)\Big)^{2}\bigg]-\bbE_{\rho_{\tau,h}^{\beta}}\bigg[\Big(s_{\tau,h+1}-f\big(\omega_{\tau,h}^{\beta}(W)\big)\Big)^{2}\bigg]\bigg|\nonumber\\
        &\quad=O\big((B_{S}+rB_{K})(B_{S}+rB_{K}+rL_{k}B_{k})|\alpha-\beta|\big).
    \end{align*}
    By decomposing the interval $[0,1]$ into the disjoint union of intervals $((i-1)/N,i/N]$ for $i\in[N]$ and using this result, we can bound the first term of the right-hand side of inequality~\eqref{ieq:72} as
    \begin{align}
        \sup_{f\in\bbB(r,\bar{\calH}),g\in\bbB(\tilr,\tilde{\calH}),W\in\tilde{\calW}}\big|\calR(f,g,W)-\calR_{\bar{\xi}}(f,g,W)\big|=O\bigg(\frac{(B_{S}+\barr\barB_{K})(B_{S}+\barr\barB_{K}+\barr\barL_{K}B_{k})}{N}\bigg).\label{eq:20}
    \end{align}
    Eqn.~\eqref{eq:20} implies that we can transfer the results in Corollary~\ref{coro:fixest} and Corollary~\ref{coro:knownfixest} to $\calR(\hatf_{h},\hatg_{h},\hatW_{h})$ with an additional term shown in Eqn.~\eqref{eq:20}. Thus, for the case where the \ac{mdp} is induced by the distribution flow of the policy itself, we have that
    \begin{align}
        &\calR(\hatf_{h},\hatg_{h},\hatW_{h})-\calR(f_{h}^{*},g_{h}^{*},W_{h}^{*})\nonumber\\
        &\quad= O\bigg(\frac{(B_{S}+\barr\barB_{K})(B_{S}+\barr\barB_{K}+\barr\barL_{K}B_{k})}{N}+\frac{(B_{S}+\barr\barB_{K})^{4}}{NL}\log\frac{N_{\bbB_{r}}N_{\tilde{\bbB}_{\tilr}}N_{\tilde{\calW}}}{\delta}.\label{eq:21}
    \end{align}
    For the case where the \ac{mdp} is induced by a pre-specified distribution flow, we have that
    \begin{align}
        &\calR(\hatf_{h}^{\prime},\hatg_{h}^{\prime},\hatW_{h}^{\prime})-\calR(f_{h}^{*},g_{h}^{*},W_{h}^{*})\nonumber\\
        &\quad= O\bigg(\frac{(B_{S}+\barr\barB_{K})(B_{S}+\barr\barB_{K}+\barr\barL_{K}B_{k})}{N}+\frac{(B_{S}+\barr\barB_{K})^{4}}{NL}\log\frac{N_{\bbB_{r}}N_{\tilde{\bbB}_{\tilr}}N_{\tilde{\calW}}}{\delta}\bigg).\label{eq:23}
    \end{align}
    
    \textbf{Step 3: Bound the estimation error of distribution flow and action-value function estimate.}
    
    For the estimation error of the distribution flow $\hmu_{t}^{\calI}$, we have the following proposition
    \begin{proposition}\label{prop:muerr}
        Given two \ac{gmfg}s $(P^{*},r^{*},W^{*})$ and $(\hatP,\hatr,\hatW)$, for a policy $\pi^{\calI}\in \tilde{\Pi}$, we define the distribution flows induced by this policy as $\mu^{\calI}=\Gamma_{2}(\pi^{\calI},W^{*})$ and $\hmu^{\calI}=\hat{\Gamma}_{2}(\pi^{\calI},\hatW)$. Assume that the transition kernels $P^{*}$ and $\hatP$ are equivalently defined by $f^{*}$ and $\hatf\in\bbB(r,\bar{\calH})$ from Eqn.~\eqref{eq:sd}. Under Assumption~\ref{assump:noise}, we have that
        \begin{align*}
            \|\hmu_{h}^{\alpha}-\mu_{h}^{\alpha}\|_{1}\leq H(1+rL_{K}L_{\varepsilon}B_{k})^{H}\sum_{m=1}^{H}\int_{0}^{1}e_{m}^{\pi,\beta}\rmd\beta+\sum_{m=1}^{H}e_{m}^{\pi,\alpha},
        \end{align*}
        where $e_{h}^{\pi,\alpha}$ is defined as
        \begin{align*}
            e_{h}^{\pi,\alpha}&=L_{\varepsilon}\sqrt{\bbE_{\rho_{h}^{\alpha}}\Big[\big(\hatf_{h}\big(\omega_{h}^{\alpha}(\hatW_{h})-f_{h}^{*}\big(\omega_{h}^{\alpha}(W_{h}^{*})\big)\big)^{2}\Big]},\\
            \omega_{h}^{\alpha}(W)&=\int_{0}^{1}\int_{\calS}W(\alpha,\beta)k\big(\cdot,(s_{\tau,h}^{i},a_{\tau,h}^{i},s)\big)\mu_{h}^{\beta}(s)\rmd s\rmd\beta,\\
            \rho_{h}^{\alpha}&=\mu_{h}^{\alpha}\times\pi_{h}^{\alpha} \text{ for }\alpha\in\calI.
        \end{align*}
    \end{proposition}
    \begin{proof}[Proof of Proposition~\ref{prop:muerr}]
        See Appendix~\ref{app:muerr}.
    \end{proof}
    From the definition of risk in Eqn.~\eqref{eq:popurisk}, we have that
    \begin{align*}
        &\calR(\hatf_{h},\hatg_{h},\hatW_{h})-\calR(f_{h}^{*},g_{h}^{*},W_{h}^{*})\\
        &\quad=\frac{1}{L}\sum_{\tau=1}^{L}\int_{0}^{1}\bbE_{\rho_{\tau,h}^{\alpha}}\bigg[\Big(f_{h}^{*}\big(\omega_{\tau,h}^{\alpha}(W_{h}^{*})\big)-\hatf_{h}\big(\omega_{\tau,h}^{\alpha}(\hatW_{h})\big)\Big)^{2}+\Big(g_{h}^{*}\big(\omega_{\tau,h}^{\alpha}(W_{h}^{*})\big)-\hatg_{h}\big(\omega_{\tau,h}^{\alpha}(\hatW_{h})\big)\Big)^{2}\bigg]\rmd \alpha.
    \end{align*}
    Since we implement the same policy $\pi_{t}^{\calI}$ for $L$ times in Step 1 of Algorithm~\ref{algo:gammaqest}, $\rho_{\tau,h}^{\alpha}$ for $\tau\in[L]$ are the same. Thus, we have
    \begin{align*}
        d(\hmu_{t}^{\calI},\mu_{t}^{\calI})=\sum_{h=1}^{H}\int_{0}^{1}\|\hmu_{t,h}^{\alpha}-\mu_{t,h}^{\alpha}\|_{1}\rmd\alpha\leq C\sum_{h=1}^{H}\sqrt{\calR(\hatf_{h},\hatg_{h},\hatW_{h})-\calR(f_{h}^{*},g_{h}^{*},W_{h}^{*})},
    \end{align*}
    where $C>0$ is a constant, and the inequality results from Proposition~\ref{prop:muerr} and H\"{o}lder inequality. The right-hande side of this inequality will play the role of $\varepsilon_{\mu}$ in the proof of Theorem~\ref{thm:contractopt}, which is bounded in Eqn.~\eqref{eq:21}.
    
    Next, we bound the estimation error of the action-value function. 
    \begin{proposition}\label{prop:qerr}
        Assume that we have two \ac{gmfg}s $(P^{*},r^{*},W^{*})$ and $(\hatP,\hatr,\hatW)$. For a policy $\pi^{\calI}\in \tilde{\Pi}$, a behavior policy $\pi^{\rmb,\calI}\in \tilde{\Pi}$, and a distribution flow $\mu^{\calI}\in\tilde{\Delta}$, we define the distribution flows induced by the behavior policy on the \ac{gmfg}  $(P^{*},r^{*},W^{*})$ with underlying distribution flow $\mu^{\calI}$ as $\mu^{\rmb,\calI}=\Gamma_{3}(\pi^{\rmb,\calI},\mu^{\calI},W^{*})$. Assume that the transition kernels $P^{*}$ and $\hatP$ are equivalently defined by $f^{*}$ and $\hatf\in\bbB(r,\bar{\calH})$ from Eqn.~\eqref{eq:sd}, and reward functions $r^{*}$ and $\hatr$ are equivalently defined by $g^{*}$ and $\hatg\in\bbB(\tilr,\tilde{\calH})$ from Eqn.~\eqref{eq:sd}. Assume that
        $\sup_{s\in\calS,a\in\calA,\alpha\in\calI,h\in[H]}\pi_{h}^{\alpha}(a\,|\,s)/\pi_{h}^{\rmb,\alpha}(a\,|\,s)\leq C$. Under Assumption~\ref{assump:noise}, we have that
        \begin{align*}
            &\bbE_{\rho_{h}^{\rmb,\alpha}}\Big[\Big|\hatQ_{h}^{\lambda,\alpha}(s,a,\pi^{\alpha},\mu^{\calI},\hatW)-Q_{h}^{\lambda,\alpha}(s,a,\pi^{\alpha},\mu^{\calI},W^{*})\Big|\Big]\\
            &\quad\leq C^{H}\sum_{m=h}^{H}\sqrt{\bbE_{\rho_{m}^{\rmb,\alpha}}\Big[\Big(\hatg_{m}\big(\omega_{h}^{\alpha}(\hatW_{m})\big)-g_{m}^{*}\big(\omega_{h}^{\alpha}(W_{m}^{*})\big)\Big)^{2}\Big]}\\
            &\quad\qquad+L_{\varepsilon}H(1+\lambda\log|\calA|)C^{H}\sum_{m=h}^{H}\sqrt{\bbE_{\rho_{m}^{\rmb,\alpha}}\Big[\Big(\hatf_{m}\big(\omega_{h}^{\alpha}(\hatW_{m})\big)-f_{m}^{*}\big(\omega_{h}^{\alpha}(W_{m}^{*})\big)\Big)^{2}\Big]},
        \end{align*}
        where $\rho_{h}^{\rmb,\alpha}$ is defined as $\rho_{h}^{\rmb,\alpha}=\mu_{h}^{\rmb,\alpha}\cdot\pi_{h}^{\rmb,\alpha}$, $\hatQ_{h}^{\lambda,\alpha}(s,a,\pi^{\alpha},\mu^{\calI},\hatW)$ and $Q_{h}^{\lambda,\alpha}(s,a,\pi^{\alpha},\mu^{\calI},W^{*})$ are the action-value functions of policy $\pi^{\calI}$ on two \ac{gmfg}s, and 
        \begin{align*}
             \omega_{h}^{\alpha}(W)&=\int_{0}^{1}\int_{\calS}W(\alpha,\beta)k\big(\cdot,(s_{\tau,h}^{i},a_{\tau,h}^{i},s)\big)\mu_{h}^{\beta}(s)\rmd s\rmd\beta.
        \end{align*}
    \end{proposition}
    \begin{proof}[Proof of Proposition~\ref{prop:qerr}]
        See Appendix~\ref{app:qerr}.
    \end{proof}
    
    Next, we will make use of Proposition~\ref{prop:qerr} to bound the estimation error of action-value function. Here, we adopt a different method to bound term (I) defined in Step 1 of the proof of Theorem~\ref{thm:contractopt}. From inequality~\eqref{ieq:71}, we have
    \begin{align}
        \text{(I)}
        &\leq \eta_{t+1}\Big|\big\langle Q_{h}^{\lambda,\alpha}(s_{h},\cdot,\pi_{t}^{\alpha},\barmu_{t}^{\calI},W^{*})-\hatQ_{h}^{\lambda,\alpha}(s_{h},\cdot,\pi_{t}^{\alpha},\hbmu_{t}^{\calI},\hatW),p-\pi_{t+1,h}^{\alpha}(\cdot\,|\,s_{h}) \big\rangle\Big|\nonumber\\
        &\qquad +\eta_{t+1}\Big|\big\langle \hatQ_{h}^{\lambda,\alpha}(s_{h},\cdot,\pi_{t}^{\alpha},\hbmu_{t}^{\calI},\hatW),\hpi_{t+1,h}^{\alpha}(\cdot\,|\,s_{h})-\pi_{t+1,h}^{\alpha}(\cdot\,|\,s_{h}) \big\rangle\Big|\nonumber\\
        &\leq 2\eta_{t+1}\big\|Q_{h}^{\lambda,\alpha}(s_{h},\cdot,\pi_{t}^{\alpha},\barmu_{t}^{\calI},W^{*})-Q_{h}^{\lambda,\alpha}(s_{h},\cdot,\pi_{t}^{\alpha},\hbmu_{t}^{\calI},W^{*})\big\|_{\infty}+2\eta_{t+1}H(1+\lambda\log|\calA|)\beta_{t+1}\nonumber\\
        &\qquad +\eta_{t+1}\Big|\big\langle Q_{h}^{\lambda,\alpha}(s_{h},\cdot,\pi_{t}^{\alpha},\hbmu_{t}^{\calI},W^{*})-\hatQ_{h}^{\lambda,\alpha}(s_{h},\cdot,\pi_{t}^{\alpha},\hbmu_{t}^{\calI},\hatW),p-\pi_{t+1,h}^{\alpha}(\cdot\,|\,s_{h})\big\rangle\Big|.\label{ieq:75}
    \end{align}
    For the third term in the right-hand side of inequality~\eqref{ieq:75}, if $p=\bar{\pi}_{t,h}^{*,\calI}(\cdot\,|\,s_{h})$, we have that
    \begin{align}
        &\Big|\big\langle Q_{h}^{\lambda,\alpha}(s_{h},\cdot,\pi_{t}^{\alpha},\hbmu_{t}^{\calI},W^{*})-\hatQ_{h}^{\lambda,\alpha}(s_{h},\cdot,\pi_{t}^{\alpha},\hbmu_{t}^{\calI},\hatW),\bar{\pi}_{t,h}^{*,\calI}(\cdot\,|\,s_{h})-\pi_{t+1,h}^{\alpha}(\cdot\,|\,s_{h})\big\rangle\Big|\nonumber\\
        &\quad=\bigg|\sum_{a\in\calA}\big[Q_{h}^{\lambda,\alpha}(s_{h},a_{h},\pi_{t}^{\alpha},\hbmu_{t}^{\calI},W^{*})-\hatQ_{h}^{\lambda,\alpha}(s_{h},a_{h},\pi_{t}^{\alpha},\hbmu_{t}^{\calI},\hatW)\big]\pi_{t,h}^{\rmb,\alpha}(a_{h}\,|\,s_{h})\cdot\frac{\bar{\pi}_{t,h}^{*,\calI}(a_{h}\,|\,s_{h})-\pi_{t+1,h}^{\alpha}(a_{h}\,|\,s_{h})}{\pi_{t,h}^{\rmb,\alpha}(a_{h}\,|\,s_{h})}\bigg|\nonumber\\
        &\quad\leq (C_{\pi}+C_{\pi}^{\prime})\sum_{a\in\calA}\big|Q_{h}^{\lambda,\alpha}(s_{h},a_{h},\pi_{t}^{\alpha},\hbmu_{t}^{\calI},W^{*})-\hatQ_{h}^{\lambda,\alpha}(s_{h},a_{h},\pi_{t}^{\alpha},\hbmu_{t}^{\calI},\hatW)\big|\cdot\pi_{t,h}^{\rmb,\alpha}(a_{h}\,|\,s_{h}),\label{ieq:77}
    \end{align}
    where the inequality results from Assumption~\ref{assump:policycover1}. We note that we can let $p=\bar{\pi}_{t,h}^{*,\calI}(\cdot\,|\,s_{h})$ in our whole proof, because we will use such bound to upper bound the right-hand side of inequality~\eqref{ieq:76}, which we take $p=\bar{\pi}_{t,h}^{*,\calI}(\cdot\,|\,s_{h})$ to prove. Now we can define a new $\Lambda_{t+1,h}^{\alpha}$ with the terms in inequalities~\eqref{ieq:75} and \eqref{ieq:77} replacing the original upper bound of term (I). In such case, the term $\varepsilon_{Q}$ in inequality~\eqref{ieq:78} can be replaced by the upper bound of the expectation of the third term in right-hand side in inequality~\eqref{ieq:75}.
    \begin{align*}
        &\bbE_{\barpi_{t}^{*,\alpha},\barmu_{t}^{\calI}}\bigg[\sum_{h=1}^{H}(C_{\pi}+C_{\pi}^{\prime})\sum_{a\in\calA}\big|Q_{h}^{\lambda,\alpha}(s_{h},a_{h},\pi_{t}^{\alpha},\hbmu_{t}^{\calI},W^{*})-\hatQ_{h}^{\lambda,\alpha}(s_{h},a_{h},\pi_{t}^{\alpha},\hbmu_{t}^{\calI},\hatW)\big|\cdot\pi_{t,h}^{\rmb,\alpha}(a_{h}\,|\,s_{h})\bigg]\\
        &\quad\leq (C_{\pi}+C_{\pi}^{\prime})C_{\pi}^{\prime\prime}\big(1+L_{\varepsilon}H(1+\lambda\log|\calA|)\big)C_{\pi}^{H}H\sum_{h=1}^{H}\sqrt{\calR(\hatf_{h}^{\prime},\hatg_{h}^{\prime},\hatW_{h}^{\prime})-\calR(f_{h}^{*},g_{h}^{*},W_{h}^{*})},
    \end{align*}
    where the inequality results from Propositions~\ref{prop:difdsamepv} and \ref{prop:qerr} and Assumption~\ref{assump:policycover2}. The right-hand side of this inequality can be further bounded with Eqn.~\eqref{eq:23} 
    
    \textbf{Step 4: Conclude the final result.}
    
    Replacing $\varepsilon_{\mu}$ and $\varepsilon_{Q}$ with the derived new bounds and using the union bound, we have that
    \begin{align*}
        &D\bigg(\frac{1}{T}\sum_{t=1}^{T}\pi_{t}^{\calI},\pi^{*,\calI}\bigg)+d\bigg(\frac{1}{T}\sum_{t=1}^{T}\hbmu_{t}^{\calI},\mu^{*,\calI}\bigg)\\
        &\quad=O\bigg(\frac{\sqrt{\log T}}{T^{1/3}}\bigg)+O\bigg(\frac{(B_{S}+\barr\barB_{K})^{1/4}(\barr\barL_{K}B_{k})^{1/4}}{(NL)^{1/8}}\log^{1/4}\frac{TNL\calN_{\infty}(1/\sqrt{N},\tilde{\calW})}{\delta}\\
        &\quad+\frac{B_{S}+\barr\barB_{K}}{(NL)^{1/4}}\log^{1/4}\frac{TN_{\bbB_{r}}N_{\tilde{\bbB}_{\tilr}}N_{\tilde{\calW}}}{\delta}+\frac{(B_{S}+\barr\barB_{K})^{1/4}(B_{S}+\barr\barB_{K}+\barr\barL_{K}B_{k})^{1/4}}{N^{1/4}}\bigg).
    \end{align*}
    Thus, we conclude the proof of Corollary~\ref{coro:optknownfix}.
\end{proof}
\section{Proof of Corollary~\ref{coro:optknownrand}}
\begin{proof}[Proof of Corollary~\ref{coro:optknownrand}]
    The proof of Corollary~\ref{coro:optknownrand} is same as the proof of Corollary~\ref{coro:optknownfix}, except that we use the bound in Theorem~\ref{thm:randest} instead of Theorem~\ref{thm:fixest}.
\end{proof}

\section{Proof of Corollary~\ref{coro:learnequi}}
\begin{proof}[Proof of Corollary~\ref{coro:learnequi}]
    We first define the inverse function of $\psi^{*}$ as $\phi^{*}$, i.e., $\phi^{*}(\psi^{*}(\alpha))=\alpha$ for all $\alpha\in\calI$. Similar to the proof of Theorem~\ref{thm:ufixest}, we can decompose the risk difference as
    \begin{align*}
        &\calR_{\bar{\xi}}(\hatf_{h},\hatg_{h},\hatW_{h}^{\hphi_{h}\circ\psi^{*}})-\calR_{\bar{\xi}}(f_{h}^{*},g_{h}^{*},W_{h}^{*})\nonumber\\
        &\quad=\text{Generalization Error of Risk}+\text{Estimation Error of Mean-embedding}+\text{Empirical Risk Difference}.
    \end{align*}
    For ease of notation, we only write the definition of each term for the transition kernel. The term for the reward functions can be easily derived.
    \begin{align*}
        &\text{Generalization Error of Risk}\\
        &\quad=\frac{1}{NL}\sum_{\tau=1}^{L}\sum_{i=1}^{N}\bbE_{\rho_{\tau,h}^{i}}\bigg[\Big(s_{\tau,h+1}^{i}-\hatf_{h}\big(\omega_{\tau,h}^{i}(\hatW_{h}^{\hphi_{h}\circ\psi^{*}})\big)\Big)^{2}-\Big(s_{\tau,h+1}^{i}-f_{h}^{*}\big(\omega_{\tau,h}^{i}(W_{h}^{*})\big)\Big)^{2}\bigg]\nonumber\\
        &\quad\qquad -2\frac{1}{NL}\sum_{\tau=1}^{L}\sum_{i=1}^{N}\Big(s_{\tau,h+1}^{i}-\hatf_{h}\big(\omega_{\tau,h}^{i}(\hatW_{h}^{\hphi_{h}\circ\psi^{*}})\big)\Big)^{2}-\Big(s_{\tau,h+1}^{i}-f_{h}^{*}\big(\omega_{\tau,h}^{i}(W_{h}^{*})\big)\Big)^{2}\nonumber\\
        &\text{Estimation Error of Mean-embedding}\\
        &\quad=2\frac{1}{NL}\sum_{\tau=1}^{L}\sum_{i=1}^{N}\Big(s_{\tau,h+1}^{i}-\hatf_{h}\big(\omega_{\tau,h}^{i}(\hatW_{h}^{\hphi_{h}\circ\psi^{*}})\big)\Big)^{2}-\Big(s_{\tau,h+1}^{i}-\hatf_{h}\big(\hbomega_{\tau,h}^{i}(\hatW_{h}^{\hphi_{h}})\big)\Big)^{2}\nonumber\\
        &\quad\qquad+2\frac{1}{NL}\sum_{\tau=1}^{L}\sum_{i=1}^{N}\Big(s_{\tau,h+1}^{i}-f_{h}^{*}\big(\hbomega_{\tau,h}^{i}(W_{h}^{*,\phi^{*}})\big)\Big)^{2}-\Big(s_{\tau,h+1}^{i}-f_{h}^{*}\big(\omega_{\tau,h}^{i}(W_{h}^{*})\big)\Big)^{2}\nonumber\\
        &\text{Empirical Risk Difference}\\
        &\quad=2\frac{1}{NL}\sum_{\tau=1}^{L}\sum_{i=1}^{N}\Big(s_{\tau,h+1}^{i}-\hatf_{h}\big(\hbomega_{\tau,h}^{i}(\hatW_{h}^{\hphi_{h}})\big)\Big)^{2}-\Big(s_{\tau,h+1}^{i}-f_{h}^{*}\big(\hbomega_{\tau,h}^{i}(W_{h}^{*,\phi^{*}})\big)\Big)^{2}.
    \end{align*}
    From the estimation procedure of Algorithm~\ref{algo:inestalgo2}, we have that
    \begin{align*}
        \text{Empirical Risk Difference}\leq 0.
    \end{align*}
    For Estimation Error of Mean-embedding, we can use the bound in inequality~\eqref{ieq:45} in the proof of Theorem~\ref{thm:ufixest} to bound it. In fact, since $\psi^{*}$ is the inverse function of $\phi^{*}$, the expression of the Estimation Error of Mean-embedding here is same as the term in inequality~\eqref{ieq:45}. For generalization error of risk, we can use inequality~\eqref{ieq:44} in the proof of Theorem~\ref{thm:ufixest} to bound it. Thus, we conclude the proof of Corollary~\ref{coro:learnequi}.
\end{proof}

\section{Proof of Corollary~\ref{coro:optunknownfix}}
\begin{proof}[Proof of Corollary~\ref{coro:optunknownfix}]
    We note that Line 6 of Algorithm~\ref{algo:gammaqest} involves the estimation of \ac{mdp} when the underlying distribution flow is given. However, different from the setting in Section~\ref{sec:knownopt}, here we can only specify the distribution flow through $\{\mu^{(\xi_{i}-1/N,\xi_{i}]}\}_{i=1}^{N}$. We concatenate these distribution flows to form $\tilde{\mu}^{\calI}$ that is defined as $\tilde{\mu}^{\alpha}=\mu^{\xi_{i}+\alpha-i/N}$ if $\alpha\in((i-1)/N,i/N]$. That is, we assume that $\xi_{i}=i/N$. Then we define the mean-embedding induced by $\tilde{\mu}^{\calI}$ as
    \begin{align*}
        \tilomega_{\tau,h}^{i}(W)=\int_{0}^{1}\int_{\calS}W(i/N,\beta)k\big(\cdot,(s_{\tau,h}^{i},a_{\tau,h}^{i},s)\big)\tilmu_{\tau,h}^{\beta}(s)\rmd s\rmd\beta.
    \end{align*}
    Then we estimate the transition kernels, reward functions, and graphons as
    \begin{align}(\hatf_{h},\hatg_{h},\hatW_{h},\hphi_{h})=\argmin_{\substack{f\in\bbB(r,\bar{\calH}), 
 \\ g\in\bbB(\tilr,\tilde{\calH}),\\ W\in\tilde{\calW},\\ \phi\in\calC_{[0,1]}^{N}} }\frac{1}{NL}\sum_{\tau=1}^{L}\sum_{i=1}^{N}\Big(s_{\tau,h+1}^{i}-f\big(\tilomega_{\tau,h}^{i}(W^{\phi})\big)\Big)^{2}+\Big(r_{\tau,h}^{i}-g\big(\tilomega_{\tau,h}^{i}(W^{\phi})\big)\Big)^{2}.\label{eq:ufixdistest}
    \end{align}

    Here we adopt the similar steps as the proof of Corollary~\ref{coro:optknownfix}. We note that the only different procedure is the first step. Next, we will derive the performance guarantee of Algorithm~\eqref{eq:ufixdistest}.
    
    In such setting, we implement $\{\pi^{(\xi_{i}-1/N,\xi_{i}]}\}_{i=1}^{N}$ for $L$ times on the \ac{mdp} induced by $\{\mu^{(\xi_{i}-1/N,\xi_{i}]}\}_{i=1}^{N}$ to collect the dataset $\calD_{\tau}=\{(s_{\tau,h}^{[N]},a_{\tau,h}^{[N]},r_{\tau,h}^{[N]},s_{\tau,h+1}^{[N]})\}_{h=1}^{H}$ for $\tau\in[L]$. We define $\mu^{+,\calI}=\Gamma_{3}(\pi^{\calI},\mu^{\calI},W^{*})$ as the distribution flow of implementing $\pi^{\calI}$ on the \ac{mdp} induced by $\mu^{\calI}$. We highlight that we will not use quantity in the estimation procedure, but use it only in the analysis. The joint distribution of $(s_{\tau,h}^{i},a_{\tau,h}^{i},r_{\tau,h}^{i},s_{\tau,h+1}^{i})_{i=1}^{N}$ is $\prod_{i=1}^{N}\rho_{\tau,h}^{+,i}$, where $\rho_{\tau,h}^{+,i}=\mu_{\tau,h}^{+,i}\times\pi_{\tau,h}^{i}\times \delta_{r_{h}^{*}}\times P_{h}^{*}$.  Same as the proof of Corollary~\ref{coro:learnequi}, we define two bijections $\psi^{*},\phi^{*}\in\calC_{[0,1]}^{N}$ as $\psi^{*}(\xi_{i})=i/N$ for all $i\in[N]$, and $\phi^{*}\circ\psi^{*}(\alpha)=\phi^{*}(\psi^{*}(\alpha))=\alpha$ for all $\alpha\in\calI$.
    
    With a little abuse of notation, we define the risk of $(f,g,W)$ given $\bar{\xi}$ as
    \begin{align*}
        \calR_{\bar{\xi}}(f,g,W)&=\frac{1}{N}\sum_{i=1}^{N}\bbE_{\rho_{h}^{+,i}}\bigg[\Big(s_{h+1}^{i}-f\big(\omega_{h}^{i}(W)\big)\Big)^{2}+\Big(r_{h+1}^{i}-g\big(\omega_{h}^{i}(W)\big)\Big)^{2}\bigg].
    \end{align*}
    \begin{corollary}\label{coro:unknownfixest}
        Under Assumptions~\ref{assump:graphon}, \ref{assump:kernel}, \ref{assump:real}, and \ref{assump:lipconti}, if $\{\xi_{i}\}_{i=1}^{N}=\{i/N\}_{i=1}^{N}$, then the risk of estimate derived in Algorithm~\eqref{eq:ufixdistest} can be bounded as
        \begin{align*}
            &\calR_{\bar{\xi}}(\hatf_{h},\hatg_{h},\hatW_{h}^{\hphi_{h}\circ\psi^{*}})-\calR_{\bar{\xi}}(f_{h}^{*},g_{h}^{*},W_{h}^{*})\leq O\bigg(\frac{(B_{S}+rB_{K})^{4}}{L}\log\frac{N\tilN_{\bbB_{r}}\tilN_{\infty}}{\delta}\bigg).
        \end{align*}
        with probability at least $1-\delta$.
    \end{corollary}
    \begin{proof}[Proof of Corollary~\ref{coro:unknownfixest}]
        See Appendix~\ref{app:unknownfixest}.
    \end{proof}
    
    Then we only need to exactly follow the steps 2, 3, and 4 in the proof of Corollary~\ref{coro:optknownfix} to prove the desired results. Thus, we conclude the proof of Corollary~\ref{coro:optunknownfix}.
\end{proof}

\section{Uniqueness of \ac{ne} Under Assumption~\ref{assump:contract}}\label{app:uniquene}
In the following, we adopt proof by contradiction. In this case, we admit the existence of multiple different \ac{ne}s. Then the expectations in distances $D(\cdot,\cdot)$ and $d(\cdot,\cdot)$ defined in Section~\ref{sec:optanaly} can be taken with respect to any \ac{ne}.
\begin{proposition}\label{prop:unique}
    Under Assumption~\ref{assump:contract}, the $\lambda$-regularized \ac{gmfg} admits at most one $\ac{ne}$ up to a set of zero-measure agents with respect to the Lebesgue measure on $[0,1]$.
\end{proposition}
\begin{proof}[Proof of Proposition~\ref{prop:unique}]
    Assume that there are two \ac{ne}s $(\pi^{*,\calI},\mu^{*,\calI})$ and $(\tilde{\pi}^{*,\calI},\tilde{\mu}^{*,\calI})$. From the Definition~\ref{def:regne} of the \ac{ne}, we have that
    \begin{align*}
        \pi^{*,\calI}=\Gamma_{1}^{\lambda}(\mu^{*,\calI},W^{*}),\quad \mu^{*,\calI}=\Gamma_{2}(\pi^{*,\calI},W^{*}),\quad \tilde{\pi}^{*,\calI}=\Gamma_{1}^{\lambda}(\tilde{\mu}^{*,\calI},W^{*}),\quad \tilde{\mu}^{*,\calI}=\Gamma_{2}(\tilde{\pi}^{*,\calI},W^{*}).
    \end{align*}
    Then Assumption~\ref{assump:contract} implies that
    \begin{align*}
        d(\mu^{*,\calI},\tilde{\mu}^{*,\calI})\leq d_{1}d_{2}d(\mu^{*,\calI},\tilde{\mu}^{*,\calI}).
    \end{align*}
    Thus, we have $d(\mu^{*,\calI},\tilde{\mu}^{*,\calI})=0$, which implies that they are different only on a set of zero-measure agents with respect to the Lebesgue measure on $[0,1]$. Thus, we conclude the proof of Proposition~\ref{prop:unique}.
\end{proof}
\section{Lipschitzness of \ac{ne}}\label{app:nelip}
\begin{proposition}\label{prop:nelip}
    Under Assumptions~\ref{assump:lipconti} and \ref{assump:graphon}, for any \ac{ne} of the $\lambda$-regularized \ac{gmfg} $(\pi^{\lambda,\calI},\mu^{\lambda,\calI})$ with $\lambda>0$, we have that
    \begin{align*}
        \big\|\pi_{h}^{\lambda,\alpha}(\cdot\,|\,s)-\pi_{h}^{\lambda,\beta}(\cdot\,|\,s)\big\|_{1}\leq \frac{2HL_{\bar{\calW}}\Big[L_{r}+H\big(1+\lambda\log|\calA|\big)L_{P}\Big]}{\lambda}|\alpha-\beta|\text{ for all }h\in[H],s\in\calS.
    \end{align*}
\end{proposition}
\begin{proof}[Proof of Proposition~\ref{prop:nelip}]
    For any distribution flow $\mu^{\calI}$, we denote the optimal value function in the $\lambda$-regularized \ac{mdp} induced by $\mu^{\calI}$ as $V^{*,\calI}=(V_{h}^{*,\calI})_{h=1}^{H}$. Then we prove the proposition in two steps:
    \begin{itemize}
        \item Given any distribution flow $\mu^{\calI}$, the optimal value function $V^{*,\calI}$ is Lipschitz in the positions of agents, i.e., $|V_{h}^{*,\alpha}(s)-V_{h}^{*,\beta}(s)|\leq H[L_{r}+H(1+\lambda\log|\calA|)L_{P}]L_{\bar{\calW}}|\alpha-\beta|$ for all $s\in\calS$ and $h\in[H]$.
        \item Any policy $\pi^{\calI}$ that achieves the optimal value function $V^{*,\calI}$ is Lipschitz in the positions of agents.
    \end{itemize}
    These two steps concludes the proof of Proposition~\ref{prop:nelip} by noting that for any $\lambda$-\ac{ne} $(\pi^{\lambda,\calI},\mu^{\lambda,\calI})$, the policy $\pi^{\lambda,\calI}$ achieves the maximal accumulative rewards in the \ac{mdp} induced by $\mu^{\lambda,\calI}$ according to Definition~\ref{def:regne}.
    
    \textbf{Step 1: Show the optimal value function $V^{*,\calI}$ is Lipschitz in the positions of agents.}
    
    For any distribution flow $\mu^{\calI}\in \Delta(\calS)^{\calI\times H}$, we define an operator acting on $\calS\rightarrow\bbR$ as
    \begin{align*}
        T_{h}^{\mu^{\calI},\alpha}u(s)&=\sup_{p\in\Delta(\calA)} \sum_{a\in\calA}p(a)r_{h}(s,a,z_{h}^{\alpha})-\lambda R(p)+\sum_{a\in\calA}\int_{\calS}p(a)P_{h}(s^{\prime}\,|\,s,a,z_{h}^{\alpha})u(s^{\prime})\rmd s^{\prime} \text{ for }h\in[H-1],\\
        T_{H}^{\mu^{\calI},\alpha}u(s)&=\sup_{p\in\Delta(\calA)} \sum_{a\in\calA}p(a)r_{H}(s,a,z_{H}^{\alpha})-\lambda R(p),
    \end{align*}
    where $R(\cdot)$ is the negative entropy function. Since $V^{*,\calI}$ is the optimal value function of the \ac{mdp} induced by $\mu^{\calI}$, we have that
    \begin{align*}
        T_{h}^{\mu^{\calI},\alpha}V_{h+1}^{*,\alpha}(s)=V_{h}^{*,\alpha}(s) \text{ and }V_{H+1}^{*,\alpha}(s)=0 \text{ for all }s\in\calS,h\in[H],\alpha\in\calI.
    \end{align*}
    For any $h\in[H-1]$, we have that
    \begin{align*}
        &\big|V_{h}^{*,\alpha}(s)-V_{h}^{*,\beta}(s)\big|\\
        &\quad\leq \sup_{p\in\Delta(\calA)} \bigg|\sum_{a\in\calA}p(a)\big(r_{h}(s,a,z_{h}^{\alpha})-r_{h}(s,a,z_{h}^{\beta})\big)\\
        &\quad\qquad+\sum_{a\in\calA}\int_{\calS}p(a)\big(P_{h}(s^{\prime}\,|\,s,a,z_{h}^{\alpha})V_{h+1}^{*,\alpha}(s^{\prime})-P_{h}(s^{\prime}\,|\,s,a,z_{h}^{\beta})V_{h+1}^{*,\beta}(s^{\prime})\big)\rmd s^{\prime}\bigg|\\
        &\quad\leq L_{r}\|z_{h}^{\alpha}-z_{h}^{\beta}\|_{1}+H(1+\lambda\log|\calA|)L_{P}\|z_{h}^{\alpha}-z_{h}^{\beta}\|_{1}+\sup_{s\in\calS}\big|V_{h+1}^{*,\alpha}(s)-V_{h+1}^{*,\beta}(s)\big|,
    \end{align*}
    where the first inequality results from Assumption~\ref{assump:lipconti}. Note that $\|z_{h}^{\alpha}-z_{h}^{\beta}\|_{1}\leq\|\int_{0}^{1}\big(W_{h}(\alpha,\gamma)-W_{h}(\beta,\gamma)\big)\mu_{h}^{\gamma}\rmd\gamma\|_{1}\leq L_{\bar{\calW}}|\alpha-\beta|$, we have
    \begin{align*}
        \sup_{s\in\calS}\big|V_{h}^{*,\alpha}(s)-V_{h}^{*,\beta}(s)\big|\leq \Big[L_{r}+H\big(1+\lambda\log|\calA|\big)L_{P}\Big]L_{\bar{\calW}}|\alpha-\beta|+\sup_{s\in\calS}\big|V_{h+1}^{*,\alpha}(s)-V_{h+1}^{*,\beta}(s)\big|.
    \end{align*}
    Summing this inequality for $t=h,\cdots,H$ and noting that $V_{H+1}^{*,\alpha}(s)=0$, we have
    \begin{align*}
        \sup_{s\in\calS}\big|V_{h}^{*,\alpha}(s)-V_{h}^{*,\beta}(s)\big|\leq H\Big[L_{r}+H\big(1+\lambda\log|\calA|\big)L_{P}\Big]L_{\bar{\calW}}|\alpha-\beta|.
    \end{align*}
    
    \textbf{Step 2: Any policy that achieves the optimal value function $V^{*,\calI}$ is Lipschitz in the positions of agents}
    
    Assume that policy $\pi^{\calI}$ achieves the optimal value function $V^{*,\calI}$. For any $\alpha,\beta\in\calI$, $s\in\calS$, and $h\in[H]$, we have that
    \begin{align*}
        \pi_{h}^{\alpha}(\cdot\,|\,s)&=\argmax_{p\in\Delta(\calA)}\sum_{a\in\calA}p(a)r_{h}(s,a,z_{h}^{\alpha})-\lambda R(p)+\sum_{a\in\calA}\int_{\calS}p(a)P_{h}(s^{\prime}\,|\,s,a,z_{h}^{\alpha})V_{h+1}^{*,\alpha}(s^{\prime})\rmd s^{\prime}\\
        \pi_{h}^{\beta}(\cdot\,|\,s)&=\argmax_{p\in\Delta(\calA)}\sum_{a\in\calA}p(a)r_{h}(s,a,z_{h}^{\beta})-\lambda R(p)+\sum_{a\in\calA}\int_{\calS}p(a)P_{h}(s^{\prime}\,|\,s,a,z_{h}^{\beta})V_{h+1}^{*,\beta}(s^{\prime})\rmd s^{\prime}
    \end{align*}
    Define $y^{\alpha}(s,a)=r_{h}(s,a,z_{h}^{\alpha})+\int_{\calS}p(a)P_{h}(s^{\prime}\,|\,s,a,z_{h}^{\alpha})V_{h+1}^{*,\alpha}(s^{\prime})\rmd s^{\prime}$ for all $\alpha\in\calI$. Lemma~\ref{lem:optapprox} shows that
    \begin{align*}
        \big\|\pi_{h}^{\alpha}(\cdot\,|\,s)-\pi_{h}^{\beta}(\cdot\,|\,s)\big\|_{1}\leq \frac{1}{\lambda}\|y^{\alpha}(s,\cdot)-y^{\beta}(s,\cdot)\|_{\infty}.
    \end{align*}
    Term $\|y^{\alpha}(s,\cdot)-y^{\beta}(s,\cdot)\|_{\infty}$ can be bounded as
    \begin{align*}
        &\|y^{\alpha}(s,\cdot)-y^{\beta}(s,\cdot)\|_{\infty}\\
        &\quad \leq\Big[L_{r}+H\big(1+\lambda\log|\calA|\big)L_{P}\Big]L_{\bar{\calW}}|\alpha-\beta|+H\Big[L_{r}+H\big(1+\lambda\log|\calA|\big)L_{P}\Big]L_{\bar{\calW}}|\alpha-\beta|\\
        &\quad\leq 2H\Big[L_{r}+H\big(1+\lambda\log|\calA|\big)L_{P}\Big]L_{\bar{\calW}}|\alpha-\beta|,
    \end{align*}
    where the first inequality results from the triangle inequality, and the second inequality results from Step 1. Thus, we conclude that
    \begin{align*}
        \big\|\pi_{h}^{\alpha}(\cdot\,|\,s)-\pi_{h}^{\beta}(\cdot\,|\,s)\big\|_{1}\leq \frac{2H\Big[L_{r}+H\big(1+\lambda\log|\calA|\big)L_{P}\Big]L_{\bar{\calW}}}{\lambda}|\alpha-\beta|,
    \end{align*}
    which proves the claim of Proposition~\ref{prop:nelip}.
\end{proof}

\section{Supporting Propositions and Lemmas}
\subsection{Propositions and Lemmas for Estimation}
\subsubsection{Proof of Proposition~\ref{prop:agentlip}}\label{app:agentlip}
\begin{proof}[Proof of Proposition~\ref{prop:agentlip}]
    For any $h\in[H-1]$, we have that
    \begin{align}
        &\|\mu_{h+1}^{\alpha}-\mu_{h+1}^{\beta}\|_{1}\nonumber\\
        &\quad=\int_{\calS}\bigg|\sum_{a\in\calA}\int_{\calS}P_{h}(s^{\prime}\,|\,s,a,z_{h}^{\alpha})\mu_{h}^{\alpha}(s)\pi_{h}^{\alpha}(a\,|\,s)\rmd s-\sum_{a\in\calA}\int_{\calS}P_{h}(s^{\prime}\,|\,s,a,z_{h}^{\beta})\mu_{h}^{\beta}(s)\pi_{h}^{\beta}(a\,|\,s)\rmd s\bigg|\rmd s^{\prime}\nonumber\\
        &\quad\leq L_{P}\|z_{h}^{\alpha}-z_{h}^{\beta}\|_{1}+\|\mu_{h}^{\alpha}-\mu_{h}^{\beta}\|_{1}+\sup_{s\in\calS}\big\|\pi_{h}^{\alpha}(\cdot\,|\,s)-\pi_{h}^{\beta}(\cdot\,|\,s)\big\|_{1},\label{ieq:10}
    \end{align}
    where the first inequality results from the triangle inequality, and the second inequality results from Assumptions~\ref{assump:graphon} and \ref{assump:lipconti}. We further bound the first term in the right-hand side of inequality~\eqref{ieq:10} as
    \begin{align*}
        \|z_{h}^{\alpha}-z_{h}^{\beta}\|_{1}&=\bigg\|\int_{0}^{1}W_{h}(\alpha,\gamma)\mu_{h}^{\gamma}\rmd\gamma-\int_{0}^{1}W_{h}(\beta,\gamma)\mu_{h}^{\gamma}\rmd\gamma\bigg\|_{1}\leq L_{\bar{\calW}}|\alpha-\beta|,
    \end{align*}
    where the inequality results from Assumption~\ref{assump:graphon}. Substituting this inequality to the right-hand side of inequality~\eqref{ieq:10}, we derive that
    \begin{align*}
        &\|\mu_{h+1}^{\alpha}-\mu_{h+1}^{\beta}\|_{1}\leq \|\mu_{h}^{\alpha}-\mu_{h}^{\beta}\|_{1}+L_{P}L_{\bar{\calW}}|\alpha-\beta|+\sup_{s\in\calS}\big\|\pi_{h}^{\alpha}(\cdot\,|\,s)-\pi_{h}^{\beta}(\cdot\,|\,s)\big\|_{1}.
    \end{align*}
    Summing these inequalities for $h=1,\cdots,t$, we have that
    \begin{align*}
        \|\mu_{t}^{\alpha}-\mu_{t}^{\beta}\|_{1}\leq (t-1)L_{P}L_{\bar{\calW}}|\alpha-\beta|+\sum_{h=1}^{t-1}\sup_{s\in\calS}\big\|\pi_{h}^{\alpha}(\cdot\,|\,s)-\pi_{h}^{\beta}(\cdot\,|\,s)\big\|_{1},
    \end{align*}
    which results from that $\mu_{1}^{\alpha}=\mu_{1}^{\beta}$. Thus, we concludes the proof of Proposition~\ref{prop:agentlip}.
\end{proof}

\subsubsection{Proof of Proposition~\ref{prop:frconcen}}\label{app:frconcen}
\begin{proof}[Proof of Proposition~\ref{prop:frconcen}]
    Our proof of Proposition~\ref{prop:frconcen} follows the pipeline of the proof of \citet[Theorem 11.4]{gyorfi2002distribution}. However, the random variables in our problem are not identically distributed, which requires additional techniques to control the tail probabilities. Our proof involves three steps:
    \begin{itemize}
        \item Symmetrization by a ghost sample.
        \item Additional randomization by random signs.
        \item Bounding the covering number
    \end{itemize}
    
    \textbf{Step 1: Symmetrization by a ghost sample.}
    
    We construct the random variables $\tilD_{h}=\{\tile_{\tau,h}^{i}\}_{\tau,i=1}^{L,N}$ that are independent of and identically distributed as $D_{h}=\{e_{\tau,h}^{i}\}_{\tau,i=1}^{L,N}$. It means that $\tile_{\tau,h}^{i}\overset{D}{=}e_{\tau,h}^{i}$ for all $\tau\in[L],i\in[N]$, and they are independent. For ease of notation, we write $\sum_{\tau=1}^{L}\sum_{i=1}^{N}$ as $\sum_{\tau,i}$ Choose a function $f_{W}$ that depends on $D_{h}$ such that
    \begin{align*}
        \frac{1}{NL}\sum_{\tau,i}\bbE_{\rho_{\tau,h}^{i}}\big[f_{W}(e_{\tau,h}^{i})\big]-\frac{1}{NL}\sum_{\tau,i}f_{W}(e_{\tau,h}^{i})\geq \varepsilon\bigg(\alpha+\beta+\frac{1}{NL}\sum_{\tau,i}\bbE_{\rho_{\tau,h}^{i}}\big[f_{W}(e_{\tau,h}^{i})\big]\bigg)
    \end{align*}
    holds. If such function does not exist, then $f_{W}$ is an arbitrary function in $\calF_{\tilde{\calW}}$. Then we have that
    \begin{align}
        &\bbE_{\rho_{\tau,h}^{i}}\bigg[\Big(f_{W}(\tile_{\tau,h}^{i})-\bbE_{\rho_{\tau,h}^{i}}\big[f_{W}(\tile_{\tau,h}^{i})\,\big|\,D_{h}\big]\Big)^{2}\,\bigg|\,D_{h}\bigg]\nonumber\\
        &\quad\leq \bbE_{\rho_{\tau,h}^{i}}\Big[\big(f_{W}(\tile_{\tau,h}^{i})\big)^{2}\,\Big|\,D_{h}\Big]\nonumber\\
        &\quad\leq 4(B_{S}+rB_{K})^{2}\bbE_{\rho_{\tau,h}^{i}}\bigg[\Big(f\big(\tilde{\omega}_{\tau,h}^{i}(W)\big)-f_{h}^{*}\big(\tilde{\omega}_{\tau,h}^{i}(W_{h}^{*})\big)\Big)^{2}\,\bigg|\,D_{h}\bigg]\nonumber\\
        &\quad= 4(B_{S}+rB_{K})^{2}\bbE_{\rho_{\tau,h}^{i}}\big[f_{W}(\tile_{\tau,h}^{i})\,\big|\,D_{h}\big],\label{ieq:18}
    \end{align}
    where the second inequality results from Lemma~\ref{lem:rkhs}, and the last equality results from that $\bbE_{\rho_{\tau,h}^{i}}[\tils_{\tau,h+1}^{i}\,|\,D_{h},\tilde{\omega}_{\tau,h}^{i}(W_{h}^{*})]=f_{h}^{*}(\tilde{\omega}_{\tau,h}^{i}(W_{h}^{*}))$. Then the tail probability for the ghost sample $\tilD_{h}$ is bouneded as
    \begin{align}
        &\bbP\bigg(\frac{1}{NL}\sum_{\tau,i}\bbE_{\rho_{\tau,h}^{i}}\big[f_{W}(\tile_{\tau,h}^{i})\,\big|\,D_{h}\big]-\frac{1}{NL}\sum_{\tau,i}f_{W}(\tile_{\tau,h}^{i})\geq \frac{\varepsilon}{2}\bigg(\alpha+\beta+\frac{1}{NL}\sum_{\tau,i}\bbE_{\rho_{\tau,h}^{i}}\big[f_{W}(\tile_{\tau,h}^{i})\,\big|\,D_{h}\big]\bigg)\bigg)\nonumber\\
        &\quad\leq \frac{\bbE\bigg[\bigg(\frac{1}{NL}\sum_{\tau,i}\bbE_{\rho_{\tau,h}^{i}}\big[f_{W}(\tile_{\tau,h}^{i})\,\big|\,D_{h}\big]-\frac{1}{NL}\sum_{\tau,i}f_{W}(\tile_{\tau,h}^{i})\bigg)^{2}\bigg]}{\bigg(\frac{\varepsilon(\alpha+\beta)}{2}+\frac{\varepsilon}{2NL}\sum_{\tau,i}\bbE_{\rho_{\tau,h}^{i}}\big[f_{W}(\tile_{\tau,h}^{i})\,\big|\,D_{h}\big]\bigg)^{2}}\nonumber\\
        &\quad\leq \frac{\frac{4(B_{S}+rB_{K})^{2}}{NL}\frac{1}{NL}\sum_{\tau,i}\bbE_{\rho_{\tau,h}^{i}}\big[f_{W}(\tile_{\tau,h}^{i})\,\big|\,D_{h}\big]}{\bigg(\frac{\varepsilon(\alpha+\beta)}{2}+\frac{\varepsilon}{2NL}\sum_{\tau,i}\bbE_{\rho_{\tau,h}^{i}}\big[f_{W}(\tile_{\tau,h}^{i})\,\big|\,D_{h}\big]\bigg)^{2}}\nonumber\\
        &\quad\leq \frac{4(B_{S}+rB_{K})^{2}}{(\alpha+\beta)NL\varepsilon^{2}},\nonumber
    \end{align}
    where the first inequality results from Chebyshev inequality, the second inequality results from inequality~\eqref{ieq:18}, and the last inequality results from $x/(a+x)^{2}\leq 1/(4a)$ for any $x,a>0$. When $NL\geq 32(B_{S}+rB_{K})^{2}/((\alpha+\beta)\varepsilon^{2})$, we have that
    \begin{align}
        &\bbP\bigg(\frac{1}{NL}\sum_{\tau,i}\bbE_{\rho_{\tau,h}^{i}}\big[f_{W}(\tile_{\tau,h}^{i})\,\big|\,D_{h}\big]-\frac{1}{NL}\sum_{\tau,i}f_{W}(\tile_{\tau,h}^{i})\nonumber\\
        &\quad\qquad\geq \frac{\varepsilon}{2}\bigg(\alpha+\beta+\frac{1}{NL}\sum_{\tau,i}\bbE_{\rho_{\tau,h}^{i}}\big[f_{W}(\tile_{\tau,h}^{i})\,\big|\,D_{h}\big]\bigg)\bigg)\leq\frac{1}{8}.\label{ieq:19}
    \end{align}
    Thus, we have that
    \begin{align}
        &\bbP\bigg(\exists f_{W}\in\calF_{\tilde{\calW}}, \frac{1}{NL}\sum_{\tau,i}\bbE_{\rho_{\tau,h}^{i}}\big[f_{W}(e_{\tau,h}^{i})\big]-\frac{1}{NL}\sum_{\tau,i}f_{W}(e_{\tau,h}^{i})\geq \varepsilon\bigg(\alpha+\beta+\frac{1}{NL}\sum_{\tau,i}\bbE_{\rho_{\tau,h}^{i}}\big[f_{W}(e_{\tau,h}^{i})\big]\bigg)\bigg)\nonumber\\
        &\quad\leq \frac{8}{7}\bbP\bigg(\exists f_{W}\in\calF_{\tilde{\calW}},\frac{1}{NL}\sum_{\tau,i}\bbE_{\rho_{\tau,h}^{i}}\big[f_{W}(\tile_{\tau,h}^{i})\big]-f_{W}(\tile_{\tau,h}^{i})< \frac{\varepsilon}{2}\bigg(\alpha+\beta+\frac{1}{NL}\sum_{\tau,i}\bbE_{\rho_{\tau,h}^{i}}\big[f_{W}(\tile_{\tau,h}^{i})\big]\bigg),\nonumber\\
        &\quad\qquad\frac{1}{NL}\sum_{\tau,i}\bbE_{\rho_{\tau,h}^{i}}\big[f_{W}(e_{\tau,h}^{i})\big]-f_{W}(e_{\tau,h}^{i})\geq \varepsilon\bigg(\alpha+\beta+\frac{1}{NL}\sum_{\tau,i}\bbE_{\rho_{\tau,h}^{i}}\big[f_{W}(e_{\tau,h}^{i})\big]\bigg)\bigg)\nonumber\\
        &\quad\leq \frac{8}{7}\bbP\bigg(\exists f_{W}\in\calF_{\tilde{\calW}},\frac{1}{NL}\sum_{\tau,i}f_{W}(\tile_{\tau,h}^{i})-f_{W}(e_{\tau,h}^{i})\geq \frac{\varepsilon}{2}\bigg(\alpha+\beta+\frac{1}{NL}\sum_{\tau,i}\bbE_{\rho_{\tau,h}^{i}}\big[f_{W}(e_{\tau,h}^{i})\big]\bigg)\bigg),\label{ieq:20}
    \end{align}
    where the first inequality results from inequality~\eqref{ieq:19}, the detailed proof of this step is in \citet[Theorem 11.4]{gyorfi2002distribution}. To derive the fast rate result, we want to replace the expectation of the function in the right-hand side of inequality~\eqref{ieq:20} by the expectation of the square of the function. Thus, We handle the right-hand side of inequality~\eqref{ieq:20} as
    \begin{align}
        &\bbP\bigg(\exists f_{W}\in\calF_{\tilde{\calW}},\frac{1}{NL}\sum_{\tau,i}f_{W}(\tile_{\tau,h}^{i})-\frac{1}{NL}\sum_{\tau,i}f_{W}(e_{\tau,h}^{i})\geq \frac{\varepsilon}{2}\bigg(\alpha+\beta+\frac{1}{NL}\sum_{\tau,i}\bbE_{\rho_{\tau,h}^{i}}\big[f_{W}(e_{\tau,h}^{i})\big]\bigg)\bigg)\nonumber\\
        &\quad\leq \bbP\bigg(\exists f_{W}\in\calF_{\tilde{\calW}},\frac{1}{NL}\sum_{\tau,i}f_{W}(\tile_{\tau,h}^{i})-\frac{1}{NL}\sum_{\tau,i}f_{W}(e_{\tau,h}^{i})\geq \frac{\varepsilon}{2}\bigg(\alpha+\beta+\frac{1}{NL}\sum_{\tau,i}\bbE_{\rho_{\tau,h}^{i}}\big[f_{W}(e_{\tau,h}^{i})\big]\bigg),\nonumber\\
        &\quad\qquad \frac{1}{NL}\sum_{\tau,i}f_{W}^{2}(e_{\tau,h}^{i})-\frac{1}{NL}\sum_{\tau,i}\bbE_{\rho_{\tau,h}^{i}}\big[f_{W}^{2}(e_{\tau,h}^{i})\big]\leq \varepsilon\bigg(\alpha+\beta+\frac{1}{NL}\sum_{\tau,i}\bbE_{\rho_{\tau,h}^{i}}\big[f_{W}^{2}(e_{\tau,h}^{i})\big]\bigg),\nonumber\\
        &\quad\qquad \frac{1}{NL}\sum_{\tau,i}f_{W}^{2}(\tile_{\tau,h}^{i})-\frac{1}{NL}\sum_{\tau,i}\bbE_{\rho_{\tau,h}^{i}}\big[f_{W}^{2}(\tile_{\tau,h}^{i})\big]\leq \varepsilon\bigg(\alpha+\beta+\frac{1}{NL}\sum_{\tau,i}\bbE_{\rho_{\tau,h}^{i}}\big[f_{W}^{2}(\tile_{\tau,h}^{i})\big]\bigg)\bigg)\nonumber\\
        &\quad\qquad+2\bbP\bigg(\exists f_{W}\in\calF_{\tilde{\calW}},\frac{1}{NL}\sum_{\tau,i}f_{W}^{2}(\tile_{\tau,h}^{i})-\bbE_{\rho_{\tau,h}^{i}}\big[f_{W}^{2}(\tile_{\tau,h}^{i})\big]\leq \varepsilon\bigg(\alpha+\beta+\frac{1}{NL}\sum_{\tau,i}\bbE_{\rho_{\tau,h}^{i}}\big[f_{W}^{2}(\tile_{\tau,h}^{i})\big]\bigg)\bigg)\nonumber\\
        &\quad=\text{(VII)}+\text{(VIII)},\label{ieq:30}
    \end{align}
    where the inequality follows from the union bound. For the term (VIII), Proposition~\ref{prop:sqreplace} shows that
    \begin{align}
        &\bbP\bigg(\exists f_{W}\in\calF_{\tilde{\calW}},\frac{1}{NL}\sum_{\tau,i}f_{W}^{2}(\tile_{\tau,h}^{i})-\frac{1}{NL}\sum_{\tau,i}\bbE_{\rho_{\tau,h}^{i}}\big[f_{W}^{2}(\tile_{\tau,h}^{i})\big]\leq \varepsilon\bigg(\alpha+\beta\nonumber\\
        &\quad\qquad+\frac{1}{NL}\sum_{\tau,i}\bbE_{\rho_{\tau,h}^{i}}\big[f_{W}^{2}(\tile_{\tau,h}^{i})\big]\bigg)\bigg)\nonumber\\
        &\quad\leq 2\bbE\bigg[\calN_{1}\bigg(\frac{\varepsilon(\alpha+\beta)}{5},\{f_{W}^{2}\,|\,f_{W}\in\calF_{\tilde{\calW}}\},\{\tile_{\tau,h}^{i}\}_{\tau,i=1}^{L,N}\bigg)\bigg]\exp\bigg(-\frac{3\varepsilon^{2}(\alpha+\beta)NL}{40(B_{S}+rB_{K})^{4}}\bigg)\nonumber\\
        &\quad\leq 2\bbE\bigg[\calN_{1}\bigg(\frac{\varepsilon(\alpha+\beta)}{10(B_{S}+rB_{K})^{2}},\calF_{\tilde{\calW}},\{\tile_{\tau,h}^{i}\}_{\tau,i=1}^{L,N}\bigg)\bigg]\exp\bigg(-\frac{3\varepsilon^{2}(\alpha+\beta)NL}{40(B_{S}+rB_{K})^{4}}\bigg).\label{ieq:31}
    \end{align}
    For term (VII), the last two events in (VII) are equivalent to that
    \begin{align}
        &(1+\varepsilon)\frac{1}{NL}\sum_{\tau,i}\bbE_{\rho_{\tau,h}^{i}}\big[f_{W}^{2}(e_{\tau,h}^{i})\big]\geq (1-\varepsilon)\frac{1}{NL}\sum_{\tau,i}f_{W}^{2}(e_{\tau,h}^{i})-\varepsilon(\alpha+\beta)\nonumber\\
        &(1+\varepsilon)\frac{1}{NL}\sum_{\tau,i}\bbE_{\rho_{\tau,h}^{i}}\big[f_{W}^{2}(\tile_{\tau,h}^{i})\big]\geq (1-\varepsilon)\frac{1}{NL}\sum_{\tau,i}f_{W}^{2}(\tile_{\tau,h}^{i})-\varepsilon(\alpha+\beta).\label{ieq:25}
    \end{align}
    Then term (VII) can be bounded as
    \begin{align}
        &\text{(VII)}\nonumber\\
        &\quad\leq \bbP\bigg(\exists f_{W}\in\calF_{\tilde{\calW}},\frac{1}{NL}\sum_{\tau,i}f_{W}(\tile_{\tau,h}^{i})-\frac{1}{NL}\sum_{\tau,i}f_{W}(e_{\tau,h}^{i})\geq\frac{\varepsilon(\alpha+\beta)}{2}\nonumber\\
        &\quad\qquad-\frac{\varepsilon^{2}(\alpha+\beta)}{4(B_{S}+rB_{K})^{2}(1+\varepsilon)}+\frac{\varepsilon(1-\varepsilon)}{8(B_{S}+rB_{K})^{2}(1+\varepsilon)}\frac{1}{NL}\sum_{\tau,i}f_{W}^{2}(e_{\tau,h}^{i})+f_{W}^{2}(\tile_{\tau,h}^{i})\bigg),\label{ieq:26}
    \end{align}
    where the inequality results from inequality~\eqref{ieq:25} and that $\bbE_{\rho_{\tau,h}^{i}}[f_{W}^{2}(e_{\tau,h}^{i})]\leq 4(B_{S}+rB_{K})^{2}\bbE_{\rho_{\tau,h}^{i}}[f_{W}(e_{\tau,h}^{i})]$.
    
\begin{proposition}\label{prop:sqreplace}
    Let $B\geq 1$, $\calG$ be a set of functions $g:\calX\rightarrow[0,B]$. Let $Z_{1},\cdots,Z_{n}$ be independent $\calX-$valued random variables that are distributed as $\rho_{1},\cdots,\rho_{n}$, respectively. Assume $\alpha>0$, $0<\varepsilon\leq 1$, $n\geq 1$. Then we have that
    \begin{align*}
        \bbP\bigg(\sup_{g\in\calG}\frac{\frac{1}{n}\sum_{i=1}^{n}g(Z_{i})-\frac{1}{n}\sum_{i=1}^{n}\bbE_{\rho_{i}}\big[g(Z)\big]}{\alpha+\frac{1}{n}\sum_{i=1}^{n}g(Z_{i})+\frac{1}{n}\sum_{i=1}^{n}\bbE_{\rho_{i}}\big[g(Z)\big]}>\varepsilon\bigg)\leq 2\bbE\bigg[\calN_{1}\bigg(\frac{\alpha\varepsilon}{5},\calG,Z_{1}^{n}\bigg)\bigg]\exp\bigg(-\frac{3\alpha n\varepsilon^{2}}{40B}\bigg)
    \end{align*}
    for $n\geq 16B/(\varepsilon^{2}\alpha)$.
\end{proposition}
\begin{proof}[Proof of Proposition~\ref{prop:sqreplace}]
    See Appendix~\ref{app:sqreplace}.
\end{proof}

    \textbf{Step 2: Additional randomization by random signs.}
    
    Let $\{U_{\tau,h}^{i}\}_{\tau,i=1}^{L,N}$ be independent and uniformly distributed over $\{+1,-1\}$ that are also independent of $D_{h}$ and $\tilD_{h}$. Then we have that
    \begin{align}
        &\bbP\bigg(\exists f_{W}\in\calF_{\tilde{\calW}},\frac{1}{NL}\sum_{\tau,i}f_{W}(\tile_{\tau,h}^{i})-\frac{1}{NL}\sum_{\tau,i}f_{W}(e_{\tau,h}^{i})\geq\frac{\varepsilon(\alpha+\beta)}{2}\nonumber\\
        &\quad\qquad-\frac{\varepsilon^{2}(\alpha+\beta)}{4(B_{S}+rB_{K})^{2}(1+\varepsilon)}+\frac{\varepsilon(1-\varepsilon)}{8(B_{S}+rB_{K})^{2}(1+\varepsilon)}\frac{1}{NL}\sum_{\tau,i}f_{W}^{2}(e_{\tau,h}^{i})+f_{W}^{2}(\tile_{\tau,h}^{i})\bigg)\nonumber\\
        &\quad\leq 2\bbE\bigg[\bbP\bigg(\exists f_{W}\in\calF_{\tilde{\calW}},\bigg|\frac{1}{NL}\sum_{\tau,i}U_{\tau,h}^{i}f_{W}(e_{\tau,h}^{i})\bigg|\geq\frac{\varepsilon(\alpha+\beta)}{4}\nonumber\\
        &\quad\qquad-\frac{\varepsilon^{2}(\alpha+\beta)}{8(B_{S}+rB_{K})^{2}(1+\varepsilon)}+\frac{\varepsilon(1-\varepsilon)}{8(B_{S}+rB_{K})^{2}(1+\varepsilon)}\frac{1}{NL}\sum_{\tau,i}f_{W}^{2}(e_{\tau,h}^{i})\,\bigg|\,\{e_{\tau,h}^{i}\}_{\tau,i=1}^{L,N}\bigg)\bigg]\label{ieq:27}
    \end{align}
    where the inequality results from the union bound. Let $\delta>0$, $\calF_{\delta}$ be a $L_{1}$ $\delta-$cover of $\calF_{\tilde{\calW}}$ on $\{e_{\tau,h}^{i}\}_{\tau,i=1}^{L,N}$. Then for any $f_{W}\in\calF_{\tilde{\calW}}$, there exists $\barf_{W}\in\calF_{\delta}$ such that
    \begin{align*}
        \frac{1}{NL}\sum_{\tau,i}\big|f_{W}(e_{\tau,h}^{i})-\barf_{W}(e_{\tau,h}^{i})\big|\leq\delta.
    \end{align*}
    This inequality implies that
    \begin{align*}
        \bigg|\frac{1}{NL}\sum_{\tau,i}U_{\tau,h}^{i}f_{W}(e_{\tau,h}^{i})\bigg|-\bigg|\frac{1}{NL}\sum_{\tau,i}U_{\tau,h}^{i}\barf_{W}(e_{\tau,h}^{i})\bigg|&\leq\delta\nonumber\\
        \frac{1}{NL}\sum_{\tau,i}f_{W}^{2}(e_{\tau,h}^{i})-\frac{1}{NL}\sum_{\tau,i}\barf_{W}^{2}(e_{\tau,h}^{i})&\geq -2(B_{S}+rB_{K})^{2}\delta,
    \end{align*}
    where these inequalities results from the triangle inequality. In the following, we take $\delta=\varepsilon\beta/5$. Thus, we can bound the right-hand side of inequality~\eqref{ieq:27} as
    \begin{align}
        &\bbP\bigg(\exists f_{W}\in\calF_{\tilde{\calW}},\bigg|\frac{1}{NL}\sum_{\tau,i}U_{\tau,h}^{i}f_{W}(e_{\tau,h}^{i})\bigg|\geq\frac{\varepsilon(\alpha+\beta)}{4}\nonumber\nonumber\\
        &\quad\qquad-\frac{\varepsilon^{2}(\alpha+\beta)}{8(B_{S}+rB_{K})^{2}(1+\varepsilon)}+\frac{\varepsilon(1-\varepsilon)}{8(B_{S}+rB_{K})^{2}(1+\varepsilon)}\frac{1}{NL}\sum_{\tau,i}f_{W}^{2}(e_{\tau,h}^{i})\,\bigg|\,\{e_{\tau,h}^{i}\}_{\tau,i=1}^{L,N}\bigg)\nonumber\\
        &\quad\leq \calN_{1}\bigg(\frac{\varepsilon\beta}{5},\calF_{\tilde{\calW}},\{e_{\tau,h}^{i}\}_{\tau,i=1}^{L,N}\bigg)\max_{f_{W}\in\calF_{\frac{\varepsilon\beta}{5}}}\bbP\bigg(\bigg|\frac{1}{NL}\sum_{\tau,i}U_{\tau,h}^{i}f_{W}(e_{\tau,h}^{i})\bigg|\geq\frac{\varepsilon\alpha}{4}\nonumber\\
        &\quad\qquad-\frac{\varepsilon^{2}\alpha}{8(B_{S}+rB_{K})^{2}(1+\varepsilon)}+\frac{\varepsilon(1-\varepsilon)}{8(B_{S}+rB_{K})^{2}(1+\varepsilon)}\frac{1}{NL}\sum_{\tau,i}f_{W}^{2}(e_{\tau,h}^{i})\,\bigg|\,\{e_{\tau,h}^{i}\}_{\tau,i=1}^{L,N}\bigg)\nonumber\\
        &\quad\leq 2\calN_{1}\bigg(\frac{\varepsilon\beta}{5},\calF_{\tilde{\calW}},\{e_{\tau,h}^{i}\}_{\tau,i=1}^{L,N}\bigg)\exp\bigg(-\frac{\varepsilon^{2}(1-\varepsilon)\alpha NL}{20(B_{S}+rB_{K})^{2}(1+\varepsilon)}\bigg)\label{ieq:28}
    \end{align}
    where the first inequality results from the union bound.
    
    \textbf{Step 3: Bounding the covering number.}
    
    In this step, we upper bound the covering number of $\calF_{\tilde{\calW}}$ by the covering numbers of $\bbB(r,\bar{\calH})$ and $\tilde{\calW}$ and conclude the tail probability. We note that
    \begin{align*}
        &\frac{1}{NL}\sum_{\tau,i}\big|f_{W}(e_{\tau,h}^{i})-\barf_{W}(e_{\tau,h}^{i})\big|\\
        &\quad\leq 2(B_{S}+rB_{K})\big[B_{K}\|f-\barf\|_{\bar{\calH}}+rL_{K}B_{k}\|W-\barW\|_{\infty}\big],
    \end{align*}
    where the inequality results from Lemma~\ref{lem:rkhs} and the triangle inequality. Thus, we have that
    \begin{align}
        \calN_{1}\big(\delta,\calF_{\tilde{\calW}},\{e_{\tau,h}^{i}\}_{\tau,i=1}^{L,N}\big)\leq \calN_{\bar{\calH}}\bigg(\frac{\delta}{4(B_{S}+rB_{K})B_{K}},\bbB(r,\bar{\calH})\bigg)\cdot \calN_{\infty}\bigg(\frac{\delta}{4(B_{S}+rB_{K})rL_{K}B_{k}},\tilde{\calW}\bigg)\label{ieq:29}
    \end{align}
    for any $\{e_{\tau,h}^{i}\}_{\tau,i=1}^{L,N}$. Combining the inequalities~\eqref{ieq:20}, \eqref{ieq:30}, \eqref{ieq:31}, \eqref{ieq:26}, \eqref{ieq:27}, and \eqref{ieq:28}, we have that for $NL\geq 32(B_{S}+rB_{K})^{2}/((\alpha+\beta)\varepsilon^{2})$
    \begin{align*}
        &\bbP\bigg(\exists f_{W}\in\calF_{\tilde{\calW}}, \frac{1}{NL}\sum_{\tau,i}\bbE_{\rho_{\tau,h}^{i}}\big[f_{W}(e_{\tau,h}^{i})\big]-\frac{1}{NL}\sum_{\tau,i}f_{W}(e_{\tau,h}^{i})\\
        &\quad\qquad\geq \varepsilon\bigg(\alpha+\beta+\frac{1}{NL}\sum_{\tau,i}\bbE_{\rho_{\tau,h}^{i}}\big[f_{W}(e_{\tau,h}^{i})\big]\bigg)\bigg)\\
        &\quad\leq\frac{64}{7}\bbE\bigg[\calN_{1}\bigg(\frac{\varepsilon(\alpha+\beta)}{10(B_{S}+rB_{K})^{2}},\calF_{\tilde{\calW}},\{e_{\tau,h}^{i}\}_{\tau,i=1}^{L,N}\bigg)\bigg]\exp\bigg(-\frac{3\varepsilon^{2}(\alpha+\beta)NL}{40(B_{S}+rB_{K})^{4}}\bigg)\\
        &\quad\qquad+\frac{32}{7}\bbE\bigg[\calN_{1}\bigg(\frac{\varepsilon\beta}{5},\calF_{\tilde{\calW}},\{e_{\tau,h}^{i}\}_{\tau,i=1}^{L,N}\bigg)\bigg]\exp\bigg(-\frac{\varepsilon^{2}(1-\varepsilon)\alpha NL}{20(B_{S}+rB_{K})^{2}(1+\varepsilon)}\bigg)\\
        &\quad\leq 14\bbE\bigg[\calN_{1}\bigg(\frac{\varepsilon\beta}{10(B_{S}+rB_{K})^{2}},\calF_{\tilde{\calW}},\{e_{\tau,h}^{i}\}_{\tau,i=1}^{L,N}\bigg)\bigg]\exp\bigg(-\frac{\varepsilon^{2}(1-\varepsilon)\alpha NL}{20(B_{S}+rB_{K})^{4}(1+\varepsilon)}\bigg)\\
        &\quad\leq 14\calN_{\bar{\calH}}\bigg(\frac{\varepsilon\beta}{40(B_{S}+rB_{K})^{3}B_{K}},\bbB(r,\bar{\calH})\bigg)\cdot \calN_{\infty}\bigg(\frac{\varepsilon\beta}{40(B_{S}+rB_{K})^{3}rL_{K}B_{k}},\tilde{\calW}\bigg)\\
        &\quad\qquad\cdot\exp\bigg(-\frac{\varepsilon^{2}(1-\varepsilon)\alpha NL}{20(B_{S}+rB_{K})^{4}(1+\varepsilon)}\bigg),
    \end{align*}
    where the last inequality results from inequality~\eqref{ieq:29}. For $NL\leq 32(B_{S}+rB_{K})^{2}/((\alpha+\beta)\varepsilon^{2})$, we have that
    \begin{align*}
        \exp\bigg(-\frac{\varepsilon^{2}(1-\varepsilon)\alpha NL}{20(B_{S}+rB_{K})^{4}(1+\varepsilon)}\bigg)\geq \exp\bigg(-\frac{32(1-\varepsilon)\alpha}{20(B_{S}+rB_{K})^{2}(1+\varepsilon)(\alpha+\beta)}\bigg)\geq \exp\bigg(-\frac{32}{80}\bigg)\geq\frac{1}{14}.
    \end{align*}
    Thus, we conclude the proof of Proposition~\ref{prop:frconcen}.
\end{proof}

\subsubsection{Proof of Proposition~\ref{prop:sqreplace}}\label{app:sqreplace}
\begin{proof}[Proof of Proposition~\ref{prop:sqreplace}]
    The proof of Proposition~\ref{prop:sqreplace} mainly follows the pipeline of the proof of \citet[Theorem 11.6]{gyorfi2002distribution}. However, the random variables in our problem are not identically distributed, which requires additional techniques to control the tail probabilities. Our proof involves two steps:
    \begin{itemize}
        \item Symmetrization by a ghost sample.
        \item Additional randomization by random signs
    \end{itemize}
    
    \textbf{Step 1: Symmetrization by a ghost sample.}
    
    We draw ghost samples $\tilZ_{1}^{n}=(\tilZ_{1},\cdots,\tilZ_{n})$ that are independent of and identically distributed as $Z_{1}^{n}=(Z_{1},\cdots,Z_{n})$. Then we have that
    \begin{align}
        &\bbP\bigg(\sup_{g\in\calG}\frac{\frac{1}{n}\sum_{i=1}^{n}g(\tilZ_{i})-\frac{1}{n}\sum_{i=1}^{n}\bbE_{\rho_{i}}\big[g(Z)\big]}{\alpha+\frac{1}{n}\sum_{i=1}^{n}g(\tilZ_{i})+\frac{1}{n}\sum_{i=1}^{n}\bbE_{\rho_{i}}\big[g(Z)\big]}>\beta\bigg)\nonumber\\
        &\quad\leq \frac{\bbE\Big[\Big(\sum_{i=1}^{n}g(\tilZ_{i})-\bbE_{\rho_{i}}\big[g(Z)\big]\Big)^{2}\Big]}{n^{2}\beta^{2}\Big(\alpha+\frac{1}{n}\sum_{i=1}^{n}\bbE_{\rho_{i}}\big[g(Z)\big]\Big)^{2}}\nonumber\\
        &\quad \leq \frac{\sum_{i=1}^{n}\Big(B-\bbE_{\rho_{i}}\big[g(Z)\big]\Big)\bbE_{\rho_{i}}\big[g(Z)\big]}{n^{2}\beta^{2}\Big(\alpha+\frac{1}{n}\sum_{i=1}^{n}\bbE_{\rho_{i}}\big[g(Z)\big]\Big)^{2}},\label{ieq:21}
    \end{align}
    where the first inequality results from Chebyshev inequality, and the last inequality results from that $g:\calX\rightarrow[0,B]$. For two constants $a,b>0$ and variables $0\leq x_{i}\leq b$ for $i\in[n]$, some basic calculus calculations show that
    \begin{align*}
        f(x_{1},\cdots,x_{n})=\frac{\sum_{i=1}^{n}(b-x_{i})x_{i}}{\big(a+\frac{1}{n}\sum_{i=1}^{n}x_{i}\big)}\leq \frac{nb}{2a}.
    \end{align*}
    Thus, inequality~\eqref{ieq:21} shows that
    \begin{align*}
        \bbP\bigg(\sup_{g\in\calG}\frac{\frac{1}{n}\sum_{i=1}^{n}g(\tilZ_{i})-\frac{1}{n}\sum_{i=1}^{n}\bbE_{\rho_{i}}\big[g(Z)\big]}{\alpha+\frac{1}{n}\sum_{i=1}^{n}g(\tilZ_{i})+\frac{1}{n}\sum_{i=1}^{n}\bbE_{\rho_{i}}\big[g(Z)\big]}>\beta\bigg)\leq \frac{B}{2\beta^{2}\alpha n}.
    \end{align*}
    We take $\beta=\varepsilon/4$. If $n\geq 16B/(\varepsilon^{2}\alpha)$, such probability is upper bounded by $1/2$. Then we have that 
    \begin{align}
        &\bbP\bigg(\sup_{g\in\calG}\frac{\frac{1}{n}\sum_{i=1}^{n}g(Z_{i})-\frac{1}{n}\sum_{i=1}^{n}\bbE_{\rho_{i}}\big[g(Z)\big]}{\alpha+\frac{1}{n}\sum_{i=1}^{n}g(Z_{i})+\frac{1}{n}\sum_{i=1}^{n}\bbE_{\rho_{i}}\big[g(Z)\big]}>\varepsilon\bigg)\nonumber\\
        &\quad\leq 2\bbP\bigg(\exists g\in\calG, \frac{1}{n}\sum_{i=1}^{n}\big(g(Z_{i})-g(\tilZ_{i})\big)\geq \frac{3\varepsilon}{8}\bigg(2\alpha+\frac{1}{n}\sum_{i=1}^{n}\big(g(Z_{i})+g(\tilZ_{i})\big)\bigg)\bigg),\label{ieq:22}
    \end{align}
    where the inequality results from the conditional probability trick. The detailed procedure can be found in~\citet[Theorem 11.6]{gyorfi2002distribution}.
    
    \textbf{Step 2: Additional randomization by random signs.}
    
    Let $\{U_{i}\}_{i=1}^{n}$ be independent and uniformly distributed random variables on $\{+1,1\}$ that are independent of $Z_{1}^{n}$ and $\tilZ_{1}^{n}$. Then we have that
    \begin{align}
        &\bbP\bigg(\exists g\in\calG, \frac{1}{n}\sum_{i=1}^{n}\big(g(Z_{i})-g(\tilZ_{i})\big)\geq \frac{3\varepsilon}{8}\bigg(2\alpha+\frac{1}{n}\sum_{i=1}^{n}\big(g(Z_{i})+g(\tilZ_{i})\big)\bigg)\bigg)\nonumber\\
        &\quad\leq 2\bbE\bigg[\bbP\bigg(\exists g\in\calG, \frac{1}{n}\sum_{i=1}^{n}U_{i}g(Z_{i})\geq \frac{3\varepsilon}{8}\bigg(\alpha+\frac{1}{n}\sum_{i=1}^{n}g(Z_{i})\bigg)\,\bigg|\,Z_{1}^{n}=z_{1}^{n}\bigg)\bigg],\label{ieq:23}
    \end{align}
    where the inequality results from the union bound. Let $\delta>0$, $\calG_{\delta}$ be a $L_{1}$ $\delta-$cover of $\calG$ on $z_{1}^{n}$. Then for any $g\in\calG$, there exists $\barg\in\calG_{\delta}$ such that $\sum_{i=1}^{n}|g(z_{i})-\barg(z_{i})|/n\leq \delta$. Thus, we have that
    \begin{align*}
        &\bbP\bigg(\exists g\in\calG, \frac{1}{n}\sum_{i=1}^{n}U_{i}g(Z_{i})\geq \frac{3\varepsilon}{8}\bigg(\alpha+\frac{1}{n}\sum_{i=1}^{n}g(Z_{i})\bigg)\,\bigg|\,Z_{1}^{n}=z_{1}^{n}\bigg)\\
        &\quad\leq \bbP\bigg(\exists g\in\calG_{\delta}, \delta+\frac{1}{n}\sum_{i=1}^{n}U_{i}g(Z_{i})\geq \frac{3\varepsilon}{8}\bigg(\alpha-\delta+\frac{1}{n}\sum_{i=1}^{n}g(Z_{i})\bigg)\,\bigg|\,Z_{1}^{n}=z_{1}^{n}\bigg)\\
        &\quad\leq |\calG_{\delta}|\max_{g\in\calG_{\delta}}\bbP\bigg(\frac{1}{n}\sum_{i=1}^{n}U_{i}g(Z_{i})\geq \frac{3\varepsilon\alpha}{8}-\frac{3\varepsilon\delta}{8}-\delta+\frac{3\varepsilon}{8}\frac{1}{n}\sum_{i=1}^{n}g(Z_{i})\,\bigg|\,Z_{1}^{n}=z_{1}^{n}\bigg),
    \end{align*}
    where the last inequality follows from the union bound. Take $\delta=\varepsilon\alpha/5$, then we have 
    \begin{align*}
        \frac{3\varepsilon\alpha}{8}-\frac{3\varepsilon\delta}{8}-\delta\geq \frac{\varepsilon\alpha}{10}.
    \end{align*}
    Thus, we can control the tail probability as
    \begin{align}
        &\bbP\bigg(\exists g\in\calG, \frac{1}{n}\sum_{i=1}^{n}U_{i}g(Z_{i})\geq \frac{3\varepsilon}{8}\bigg(\alpha+\frac{1}{n}\sum_{i=1}^{n}g(Z_{i})\bigg)\,\bigg|\,Z_{1}^{n}=z_{1}^{n}\bigg)\nonumber\\
        &\quad\leq \calN_{1}\bigg(\frac{\varepsilon\alpha}{5},\calG,z_{1}^{n}\bigg)\max_{g\in\calG_{\frac{\varepsilon\alpha}{5}}}\bbP\bigg(\frac{1}{n}\sum_{i=1}^{n}U_{i}g(Z_{i})\geq \frac{\varepsilon\alpha}{10}+\frac{3\varepsilon}{8}\frac{1}{n}\sum_{i=1}^{n}g(Z_{i})\,\bigg|\,Z_{1}^{n}=z_{1}^{n}\bigg)\nonumber\\
        &\quad\leq \calN_{1}\bigg(\frac{\varepsilon\alpha}{5},\calG,z_{1}^{n}\bigg)\exp\bigg(-\frac{9\varepsilon^{2}}{128B}\frac{\big(\frac{4}{15}n\alpha+\sum_{i=1}^{n}g(z_{i})\big)^{2}}{\sum_{i=1}^{n}g(z_{i})}\bigg)\nonumber\\
        &\quad\leq \calN_{1}\bigg(\frac{\varepsilon\alpha}{5},\calG,z_{1}^{n}\bigg)\exp\bigg(-\frac{3\alpha\varepsilon^{2}n}{40B}\bigg),\label{ieq:24}
    \end{align}
    where the second inequality results from the Hoeffding's inequality, and the last inequality results from that $(a+y)^{2}/y\geq 4a$ for any $a,y>0$. Combining the inequalities~\eqref{ieq:22}, \eqref{ieq:23} and \eqref{ieq:24}, we conclude the proof of Proposition~\ref{prop:sqreplace}.
\end{proof}

\subsection{Propositions and Lemmas for Optimization}
\subsubsection{Proof of Proposition~\ref{prop:firstordopt}}\label{app:firstordopt}
\begin{proof}[Proof of Proposition~\ref{prop:firstordopt}]
        From the definition of $R(\cdot)$ and $\kl(\cdot\|\cdot)$, we have that
        \begin{align*}
            \nabla_{p} R(p)=1+\log p\quad \nabla_{p}\kl(p\|q)=1+\log\frac{p}{q}.
        \end{align*}
        Then the first-order optimal condition of Eqn.~\eqref{eq:16} is that for any $p\in\Delta(\calA)$
        \begin{align*}
            \bigg\langle\eta_{t+1}\hatQ_{h}^{\lambda,\alpha}(s,\cdot,\pi_{t}^{\alpha},\hbmu_{t}^{\calI},\hatW)-\lambda\eta_{t+1}\log \hpi_{t+1,h}^{\alpha}(\cdot\,|\,s)-\log\frac{\hpi_{t+1,h}^{\alpha}(\cdot\,|\,s)}{\pi_{t,h}^{\alpha}(\cdot\,|\,s)},p- \hpi_{t+1,h}^{\alpha}(\cdot\,|\,s)\bigg\rangle\leq 0.
        \end{align*}
        Note that
        \begin{align*}
            \kl(p_{1}\|p_{2})&=\kl(p_{3}\|p_{2})+\big\langle\nabla_{p_{3}}\kl(p_{3}\|p_{2}),p_{1}-p_{3}\big\rangle+\kl(p_{1}\|p_{3})\\
            \kl(p_{1}\|p_{2})&=R(p_{1})-R(p_{2})+\big\langle\nabla R(p_{2}),p_{2}-p_{1}\big\rangle.
        \end{align*}
        Then we have
        \begin{align*}
            &\eta_{t+1}\big\langle \hatQ_{h}^{\lambda,\alpha}(s,\cdot,\pi_{t}^{\alpha},\hbmu_{t}^{\calI},\hatW),p-\hpi_{t+1,h}^{\alpha}(\cdot\,|\,s) \big\rangle+\lambda \eta_{t+1}\Big[ R\big(\hpi_{t+1,h}^{\alpha}(\cdot\,|\,s)\big)-R(p)\Big]+\kl\big(\hpi_{t+1,h}^{\alpha}(\cdot\,|\,s)\|\pi_{t,h}^{\alpha}(\cdot\,|\,s)\big)\\
            &\quad\leq \kl\big(p\|\pi_{t,h}^{\alpha}(\cdot\,|\,s)\big)-(1+\lambda\eta_{t+1})\kl\big(p\|\hpi_{t+1,h}^{\alpha}(\cdot\,|\,s)\big).
        \end{align*}
        Thus, we conclude the proof of Proposition~\ref{prop:firstordopt}.
    \end{proof}
\subsubsection{Proof of Proposition~\ref{prop:esterrsum}}\label{app:esterrsum}
\begin{proof}[Proof of Proposition~\ref{prop:esterrsum}]
    In the following, we upper bound these four terms separately. For term (I), we have that
    \begin{align}
        \text{(I)}&\leq \eta_{t+1}\Big|\big\langle Q_{h}^{\lambda,\alpha}(s_{h},\cdot,\pi_{t}^{\alpha},\barmu_{t}^{\calI},W^{*}),p-\pi_{t+1,h}^{\alpha}(\cdot\,|\,s_{h}) \big\rangle-\big\langle \hatQ_{h}^{\lambda,\alpha}(s_{h},\cdot,\pi_{t}^{\alpha},\hbmu_{t}^{\calI},\hatW),p-\hpi_{t+1,h}^{\alpha}(\cdot\,|\,s_{h}) \big\rangle\Big|\nonumber\\
        &\leq \eta_{t+1}\Big|\big\langle Q_{h}^{\lambda,\alpha}(s_{h},\cdot,\pi_{t}^{\alpha},\barmu_{t}^{\calI},W^{*})-\hatQ_{h}^{\lambda,\alpha}(s_{h},\cdot,\pi_{t}^{\alpha},\hbmu_{t}^{\calI},\hatW),p-\pi_{t+1,h}^{\alpha}(\cdot\,|\,s_{h}) \big\rangle\Big|\nonumber\\
        &\qquad +\eta_{t+1}\Big|\big\langle \hatQ_{h}^{\lambda,\alpha}(s_{h},\cdot,\pi_{t}^{\alpha},\hbmu_{t}^{\calI},\hatW),\hpi_{t+1,h}^{\alpha}(\cdot\,|\,s_{h})-\pi_{t+1,h}^{\alpha}(\cdot\,|\,s_{h}) \big\rangle\Big|\nonumber\\
        &\leq 2\eta_{t+1}\big\|Q_{h}^{\lambda,\alpha}(s_{h},\cdot,\pi_{t}^{\alpha},\barmu_{t}^{\calI},W^{*})-Q_{h}^{\lambda,\alpha}(s_{h},\cdot,\pi_{t}^{\alpha},\hbmu_{t}^{\calI},W^{*})\big\|_{\infty}\nonumber\\
        &\qquad +2\eta_{t+1}\big\|Q_{h}^{\lambda,\alpha}(s_{h},\cdot,\pi_{t}^{\alpha},\hbmu_{t}^{\calI},W^{*})-\hatQ_{h}^{\lambda,\alpha}(s_{h},\cdot,\pi_{t}^{\alpha},\hbmu_{t}^{\calI},\hatW)\big\|_{\infty}+2\eta_{t+1}H(1+\lambda\log|\calA|)\beta_{t+1},\label{ieq:71}
    \end{align}
    where the second inequality results from the triangle inequality, and the last inequality results from the H\"{o}lder inequality and the triangle inequality. To bound the second term in the right-hand side of inequality~\eqref{ieq:71}, we state the proposition
    \begin{proposition}\label{prop:difdsamepv}
        Under Assumption~\ref{assump:lipconti}, for any policy $\pi^{\calI}$ and two distribution flows $\mu^{\calI}$ and $\tilde{\mu}^{\calI}$, we have that
        \begin{align*}
            &\big|Q_{h}^{\lambda,\alpha}(s,a,\pi^{\alpha},\mu^{\calI},W^{*})-Q_{h}^{\lambda,\alpha}(s,a,\pi^{\alpha},\tilde{\mu}^{\calI},W^{*})\big|\leq \big[L_{r}+H(1+\lambda\log|\calA|)L_{P}\big]\!\sum_{m=h}^{H}\!\int_{0}^{1}\!\!\!\|\mu_{m}^{\beta}-\tilde{\mu}_{m}^{\beta}\|_{1}\rmd\beta,\\
            &\big|V_{h}^{\lambda,\alpha}(s,\pi^{\alpha},\mu^{\calI},W^{*})-V_{h}^{\lambda,\alpha}(s,\pi^{\alpha},\tilde{\mu}^{\calI},W^{*})\big|\leq \big[L_{r}+H(1+\lambda\log|\calA|)L_{P}\big]\!\sum_{m=h}^{H}\!\int_{0}^{1}\!\|\mu_{m}^{\beta}-\tilde{\mu}_{m}^{\beta}\|_{1}\rmd\beta
        \end{align*}
        for all $\alpha\in\calI$, $s\in\calS$, $a\in\calA$ and $h\in[H]$.
    \end{proposition}
    \begin{proof}[Proof of Proposition~\ref{prop:difdsamepv}]
        See Appendix~\ref{app:difdsamepv}.
    \end{proof}
    Thus, we have that
    \begin{align*}
        \text{(I)}&\leq 2\eta_{t+1}\big[L_{r}+H(1+\lambda\log|\calA|)L_{P}\big]\sum_{m=1}^{H}\int_{0}^{1}\|\barmu_{t,m}^{\beta}-\hbmu_{t,m}^{\beta}\|_{1}\rmd\beta\nonumber\\
        &\qquad +2\eta_{t+1}\big\|Q_{h}^{\lambda,\alpha}(s_{h},\cdot,\pi_{t}^{\alpha},\hbmu_{t}^{\calI},W^{*})-\hatQ_{h}^{\lambda,\alpha}(s_{h},\cdot,\pi_{t}^{\alpha},\hbmu_{t}^{\calI},\hatW)\big\|_{\infty}+2\eta_{t+1}H(1+\lambda\log|\calA|)\beta_{t+1}.
    \end{align*}
    Define $\alpha_{m,t}=\alpha_{m}^{\prime}\prod_{k=m+1}^{t-1}(1-\alpha_{k}^{\prime})$ for $m\in[t]$, where $\alpha_{m}^{\prime}=\alpha_{m}$ for $m\geq 2$ and $\alpha_{1}^{\prime}=1$ (since $\hbmu_{1}^{\calI}=\hat{\mu}_{1}^{\calI}$). Then it satisfies that $\sum_{m=1}^{t-1}\alpha_{m,t}=1$, and that 
    \begin{align}
        \barmu_{t}^{\calI}=\sum_{m=1}^{t-1}\alpha_{m,t}\cdot\mu_{m}^{\calI},\text{ and }\quad \hbmu_{t}^{\calI}=\sum_{m=1}^{t-1}\alpha_{m,t}\cdot\hat{\mu}_{m}^{\calI}.\label{eq:19}
    \end{align}
    Then we have that
    \begin{align*}
        \sum_{m=1}^{H}\int_{0}^{1}\|\barmu_{t,m}^{\beta}-\hbmu_{t,m}^{\beta}\|_{1}\rmd\beta=d(\barmu_{t}^{\calI},\hbmu_{t}^{\calI})\leq \sum_{m=1}^{t-1}\alpha_{m,t-1}d(\hat{\mu}_{m}^{\calI},\mu_{m}^{\calI})\leq \varepsilon_{\mu},
    \end{align*}
    where the inequality results from the triangle inequality. Thus, we have
    \begin{align*}
        \text{(I)}&\leq 2\eta_{t+1}\big\|Q_{h}^{\lambda,\alpha}(s_{h},\cdot,\pi_{t}^{\alpha},\hbmu_{t}^{\calI},W^{*})-\hatQ_{h}^{\lambda,\alpha}(s_{h},\cdot,\pi_{t}^{\alpha},\hbmu_{t}^{\calI},\hatW)\big\|_{\infty}\nonumber\\
        &\qquad +2\eta_{t+1}\big[L_{r}+H(1+\lambda\log|\calA|)L_{P}\big]\varepsilon_{\mu}+2\eta_{t+1}H(1+\lambda\log|\calA|)\beta_{t+1}.
    \end{align*}

    For term (II), Lemma~\ref{lem:monoentro} shows that $\text{(II)}\leq 0$.
    
    For term (III), we have that
    \begin{align*}
        \text{(III)}&=R\big(\pi_{t+1,h}^{\alpha}(\cdot\,|\,s_{h})\big)-R\big(\hpi_{t+1,h}^{\alpha}(\cdot\,|\,s_{h})\big)+\sum_{a\in\calA}\big(\pi_{t+1,h}^{\alpha}(a\,|\,s_{h})-\hpi_{t+1,h}^{\alpha}(a\,|\,s_{h})\big)\log \frac{1}{\pi_{t,h}^{\alpha}(a\,|\,s_{h})}\\
        &\leq \sum_{a\in\calA}\big|\pi_{t+1,h}^{\alpha}(a\,|\,s_{h})-\hpi_{t+1,h}^{\alpha}(a\,|\,s_{h})\big|\log \frac{|\calA|}{\beta_{t}}\\
        &\leq 2\beta_{t+1}\log \frac{|\calA|}{\beta_{t}},
    \end{align*}
    where the last inequality results from the definition of $\pi_{t+1,h}^{\alpha}$ and $\hpi_{t+1,h}^{\alpha}$.
    
    For term (IV), Lemma~\ref{lem:policyave} shows that for $\beta_{t+1}\leq 1/2$, we have that $\text{(IV)}\leq 2(1+\lambda\eta_{t+1})\beta_{t+1}.$ 
    
    Summing these four terms, we conclude the proof of the proposition.
\end{proof}
\subsubsection{Proof of Proposition~\ref{prop:valuefraction}}\label{app:valuefraction}
\begin{proof}[Proof of Proposition~\ref{prop:valuefraction}]
    Our proof involves two steps:
    \begin{itemize}
        \item Proof $\bbE_{\pi^{*}}[V_{h}^{\lambda}(s_{h},\pi^{*})-V_{h}^{\lambda}(s_{h},\pi)]\geq \gamma^{*}\bbE_{\pi^{*}}[V_{h+1}^{\lambda}(s_{h+1},\pi^{*})-V_{h+1}^{\lambda}(s_{h+1},\pi)]$ for all $h\in[H]$, where $\gamma^{*}>0$ is a constant.
        \item Proof the desired result from Step 1.
    \end{itemize}
    
    \textbf{Step 1: Proof $\bbE_{\pi^{*}}[V_{h}^{\lambda}(s_{h},\pi^{*})-V_{h}^{\lambda}(s_{h},\pi)]\geq \gamma^{*}\bbE_{\pi^{*}}[V_{h+1}^{\lambda}(s_{h+1},\pi^{*})-V_{h+1}^{\lambda}(s_{h+1},\pi)]$ for all $h\in[H]$.}
    
    If $\pi_{t}=\pi_{t}^{*}$ for all $t\geq h+1$, then the result trivially holds. In the following, we assume that $\pi_{t}\neq\pi_{t}^{*}$ for some $t\geq h+1$. This implies that
    \begin{align*}
        \bbE_{\pi^{*}}[V_{h+1}^{\lambda}(s_{h+1},\pi^{*})-V_{h+1}^{\lambda}(s_{h+1},\pi)]> 0.
    \end{align*}
    For ease of notation, we define that
    \begin{align*}
        y(s,a)&=r_{h}(s,a)+\int_{\calS}P_{h}(s^{\prime}\,|\,s,a)V_{h+1}(s^{\prime},\pi)\rmd s^{\prime}\\
        y^{*}(s,a)&=r_{h}(s,a)+\int_{\calS}P_{h}(s^{\prime}\,|\,s,a)V_{h+1}(s^{\prime},\pi^{*})\rmd s^{\prime}.
    \end{align*}
    Then we expand these two differences between value functions as
    \begin{align*}
        &\bbE_{\pi^{*}}\big[V_{h}^{\lambda}(s_{h},\pi^{*})-V_{h}^{\lambda}(s_{h},\pi)\big]\\
        &\quad=\bbE_{\pi^{*}}\big[\langle y(s_{h},\cdot),\pi_{h}^{*}(\cdot\,|\,s_{h})-\pi_{h}(\cdot\,|\,s_{h})\rangle+\lambda \Big[R\big(\pi_{h}(\cdot\,|\,s_{h})\big)-R\big(\pi_{h}^{*}(\cdot\,|\,s_{h})\big)\Big]+\langle y^{*}(s_{h},\cdot)-y(s_{h},\cdot),\pi_{h}^{*}(\cdot\,|\,s_{h}) \rangle\big]\\
        &\bbE_{\pi^{*}}\big[V_{h+1}^{\lambda}(s_{h+1},\pi^{*})-V_{h+1}^{\lambda}(s_{h+1},\pi)\big]\\
        &\quad=\bbE_{\pi^{*}}\big[\langle y^{*}(s_{h},\cdot)-y(s_{h},\cdot),\pi_{h}^{*}(\cdot\,|\,s_{h}) \rangle\big],
    \end{align*}
    where $R(p)=\langle p,\log p\rangle$. In the following, we will prove that for any $s\in\calS$
    \begin{align}
        &\langle y(s,\cdot),\pi_{h}^{*}(\cdot\,|\,s)-\pi_{h}(\cdot\,|\,s)\rangle+\lambda \Big[R\big(\pi_{h}(\cdot\,|\,s)\big)-R\big(\pi_{h}^{*}(\cdot\,|\,s)\big)\Big]+\langle y^{*}(s,\cdot)-y(s,\cdot),\pi_{h}^{*}(\cdot\,|\,s) \rangle\nonumber\\
        &\quad\geq \gamma^{*}\langle y^{*}(s,\cdot)-y(s,\cdot),\pi_{h}^{*}(\cdot\,|\,s) \rangle, \label{ieq:47}
    \end{align}
    and our desired result immediately follows from taking expectation on the both sides of inequality~\eqref{ieq:47}. For ease of notation, we define $p^{*}=\pi_{h}^{*}(\cdot\,|\,s)$. From the definition of the optimal policy, we have that
    \begin{align*}
        p^{*}=\argmax_{q\in\Delta(\calA)}\langle q,y^{*}(s,\cdot)\rangle-\lambda R(q),\quad p=\argmax_{q\in\Delta(\calA)}\langle q,y(s,\cdot)\rangle-\lambda R(q).
    \end{align*}
    They has the closed-form expression $p^{*}(a)=\exp(y^{*}(s,a)/\lambda)/Z^{*}(s)$ and $p(a)=\exp(y(s,a)/\lambda)/Z(s)$, where $Z^{*}(s)=\sum_{a}\exp(y^{*}(s,a)/\lambda)$ and $Z(s)=\sum_{a}\exp(y(s,a)/\lambda)$.
    To prove inequality~\eqref{ieq:47}, it suffices to prove that
    \begin{align}
        \langle y(s,\cdot),p^{*}-p\rangle+\lambda \big[R(p)-R(p^{*})\big]\geq (\gamma^{*}-1)\langle y^{*}(s,\cdot)-y(s,\cdot),p^{*} \rangle.\label{ieq:48}
    \end{align}
    The left-hand side the inequality~\eqref{ieq:48} is
    \begin{align}
        \langle y(s,\cdot),p^{*}-p\rangle+\lambda \big[R(p)-R(p^{*})\big]=\langle \lambda\log p,p^{*}-p\rangle+\lambda \big[R(p)-R(p^{*})\big]=-\lambda\bigg\langle p^{*},\log\frac{p^{*}}{p}\bigg\rangle,\label{ieq:49}
    \end{align}
    where the first equality results from the closed-form expression of $p$, and the second inequality results from the definition of $R(\cdot)$. We further expand this term as
    \begin{align}
        -\lambda\bigg\langle p^{*},\log\frac{p^{*}}{p}\bigg\rangle &=-\lambda \log\frac{Z(s)}{Z^{*}(s)}-\bigg\langle \frac{\exp\big(y^{*}(s,\cdot)/\lambda\big)}{Z^{*}(s)},y^{*}(s,\cdot)-y(s,\cdot)\bigg\rangle,\label{ieq:50}
    \end{align}
    where the equalitys result from the closed-form expressions of $p$ and $p^{*}$. The right-hand side of inequality~\eqref{ieq:48} is
    \begin{align}
        (\gamma^{*}-1)\langle y^{*}(s,\cdot)-y(s,\cdot),p^{*} \rangle=(\gamma^{*}-1)\lambda\bigg(\bigg\langle \log\frac{p^{*}}{p},p^{*} \bigg\rangle+\log\frac{Z^{*}(s)}{Z(s)}\bigg),\label{ieq:51}
    \end{align}
    where the equalitys result from the closed-form expressions of $p$ and $p^{*}$. Combining Eqn.~\eqref{ieq:49}, \eqref{ieq:50}, and \eqref{ieq:51}, we have
    \begin{align}
        &\langle y(s,\cdot),p^{*}-p\rangle+\lambda \big[R(p)-R(p^{*})\big]\geq (\gamma^{*}-1)\langle y^{*}(s,\cdot)-y(s,\cdot),p^{*} \rangle\nonumber\\
        &\quad \Leftrightarrow \frac{\gamma^{*}}{\lambda}\bigg\langle \exp\bigg(\frac{y^{*}(s,\cdot)}{\lambda}\bigg),y^{*}(s,\cdot)-y(s,\cdot)\bigg\rangle\leq Z^{*}(s)\log\frac{Z^{*}(s)}{Z(s)}.\label{ieq:52}
    \end{align}
    In the following, we prove inequality~\eqref{ieq:52}. The right-hand side of \eqref{ieq:52} can be lower-bounded as
    \begin{align}
        Z^{*}(s)\log\frac{Z^{*}(s)}{Z(s)}
        &\geq \frac{\log B}{B-1}\sum_{a\in\calA}\exp\big(y^{*}(s,a)/\lambda\big)\cdot\bigg[\frac{\sum_{a\in\calA}\exp\big(y^{*}(s,a)/\lambda\big)}{\sum_{a\in\calA}\exp\big(y(s,a)/\lambda\big)}-1\bigg]\nonumber\\
        & \geq \frac{\log B}{(B-1)\lambda}\cdot \sum_{a\in\calA}\exp\big(y(s,a)/\lambda\big)\big(y^{*}(s,a)-y(s,a)\big)\label{ieq:53}, 
    \end{align}
    where $B=\exp(H(1+\lambda\log|\calA|)/\lambda)$, the first inequality results from that $\log B/(B-1) \cdot (x-1)\leq \log x$ for $x\in [1,B]$ and the facts that $y^{*}(s,a)\geq y(s,a)$ and $|y^{*}(s,a)|\leq H(1+\lambda \log |\calA|)$ for all $s\in\calS$ and $a\in\calA$, and the second inequality results from that $\exp(x)-1\geq x$ and that $y^{*}(s,a)\geq y(s,a)$. The left-hand side of inequality~\eqref{ieq:52} can be upper bounded as
    \begin{align}
        &\frac{\gamma^{*}}{\lambda}\bigg\langle \exp\bigg(\frac{y^{*}(s,\cdot)}{\lambda}\bigg),y^{*}(s,\cdot)-y(s,\cdot)\bigg\rangle
        \leq \frac{\gamma^{*}}{\lambda} \cdot B\cdot \sum_{a\in\calA}\exp\big(y(s,a)/\lambda\big)\big(y^{*}(s,a)-y(s,a)\big),\label{ieq:54}
    \end{align}
    where the inequality results from that $y^{*}(s,a)/y(s,a)\leq B$ for all $s\in\calS$ and $a\in\calA$. Combining inequalities~\eqref{ieq:53} and \eqref{ieq:54}, we prove inequality~\eqref{ieq:52} given
    \begin{align*}
        0<\gamma^{*}\leq\frac{\log B}{B(B-1)}.
    \end{align*}
    
    \textbf{Step 2: Proof the desired result from Step 1.}
    
    We define that
    \begin{align*}
        D_{h}=\bbE_{\pi^{*}}[V_{h}^{\lambda}(s_{h},\pi^{*})-V_{h}^{\lambda}(s_{h},\pi)]
    \end{align*}
    for $h\in[H]$. Then Step 1 shows that $D_{h}\geq \gamma^{*}D_{h+1}$ for all $h\in[H]$. Thus, we have
    \begin{align*}
        \frac{\bbE_{\pi^{*}}\big[V_{1}^{\lambda}(s_{1},\pi^{*})-V_{1}^{\lambda}(s_{1},\pi)\big]}{\bbE_{\pi^{*}}\bigg[\sum_{h=2}^{H}V_{h}^{\lambda}(s_{h},\pi^{*})-V_{h}^{\lambda}(s_{h},\pi)\bigg]}=\frac{D_{1}}{\sum_{h=2}^{H}D_{h}}=\frac{1}{\sum_{h=2}^{H}D_{h}/D_{1}}.
    \end{align*}
    For each term, we have that
    \begin{align*}
        \frac{D_{h}}{D_{1}}=\frac{D_{h}}{D_{h-1}}\cdots \frac{D_{2}}{D_{1}}\leq \gamma^{* (2-h)}.
    \end{align*}
    Thus, we have that
    \begin{align*}
         \frac{\bbE_{\pi^{*}}\big[V_{1}^{\lambda}(s_{1},\pi^{*})-V_{1}^{\lambda}(s_{1},\pi)\big]}{\bbE_{\pi^{*}}\bigg[\sum_{h=2}^{H}V_{h}^{\lambda}(s_{h},\pi^{*})-V_{h}^{\lambda}(s_{h},\pi)\bigg]}\geq \frac{(1-\gamma^{*})\gamma^{*(H-2)}}{1-\gamma^{*(H-1)}}=\beta^{*}.
    \end{align*}
    The proof of Proposition~\ref{prop:valuefraction} is complete.
\end{proof}

\subsubsection{Proof of Proposition~\ref{prop:dynamerr}}\label{app:dynamerr}
\begin{proof}[Proof of Proposition~\ref{prop:dynamerr}]

    We first write
    \begin{align}
        \Delta_{t+1}^{\alpha}&=X_{t+1}^{\alpha}-\bbE_{\barpi_{t}^{*,\alpha},\barmu_{t}^{\calI}}\bigg[\sum_{h=1}^{H}V_{h}^{\lambda,\alpha}(s_{h},\barpi_{t}^{*,\alpha},\barmu_{t}^{\calI},W^{*})-V_{h}^{\lambda,\alpha}(s_{h},\pi_{t+1}^{\alpha},\barmu_{t}^{\calI},W^{*})\bigg]\nonumber\\
        &\qquad-\frac{1}{\eta\theta^{*}}\bbE_{\barpi_{t}^{*,\alpha},\barmu_{t}^{\calI}}\bigg[\sum_{h=1}^{H}\kl\big(\barpi_{t,h}^{*,\alpha}(\cdot\,|\,s_{h})\|\pi_{t+1,h}^{\alpha}(\cdot\,|\,s_{h})\big)\bigg]\nonumber\\
        & = \text{(V)}+\text{(VI)}+\text{(VII)}+\text{(VIII)}+\text{(IX)}.\label{ieq:61}
    \end{align}
    Term (V) is the error that measures the difference between the action-value function induced by the optimal policies of $\barmu_{t+1}^{\calI}$ and $\barmu_{t}^{\calI}$, which is defined as
    \begin{align*}
        \text{(V)}=\bbE_{\barpi_{t+1}^{*,\alpha},\barmu_{t+1}^{\calI}}\bigg[\sum_{h=1}^{H}V_{h}^{\lambda,\alpha}(s_{h},\barpi_{t+1}^{*,\alpha},\barmu_{t+1}^{\calI},W^{*})-V_{h}^{\lambda,\alpha}(s_{h},\barpi_{t}^{*,\alpha},\barmu_{t+1}^{\calI},W^{*})\bigg].
    \end{align*}
    
    To upper bound the term (V), we note that the optimal policies $\barmu_{t+1}^{\calI}$ and $\barmu_{t}^{\calI}$ satisfy the following property.
    \begin{proposition}\label{prop:optpolicybound}
        For a $\lambda$-regularized finite-horizon \ac{mdp} $(\calS,\calA,H,\{r_{h}\}_{h=1}^{H},\{P_{h}\}_{h=1}^{H})$ with $r_{h}\in[0,1]$ for all $h\in[H]$, we denote the optimal policy as $\pi^{*}=\{\pi^{*}_{h}\}_{h=1}^{H}$. Then we have that for any $s\in\calS$, and $h\in[H]$
        \begin{align*}
            \min_{a\in\calA}\pi_{h}^{*}(a\,|\,s)\geq \frac{1}{1+|\calA|\exp\big((H-h+1)(1+\lambda\log|\calA|)/\lambda\big)}.
        \end{align*}
    \end{proposition}
    \begin{proof}[Proof of Proposition~\ref{prop:optpolicybound}]
        See Appendix~\ref{app:optpolicybound}
    \end{proof}
    Then we have that
    \begin{align*}
        &\big|V_{h}^{\lambda,\alpha}(s_{h},\barpi_{t+1}^{*,\alpha},\barmu_{t+1}^{\calI},W^{*})-V_{h}^{\lambda,\alpha}(s_{h},\barpi_{t}^{*,\alpha},\barmu_{t+1}^{\calI},W^{*})\big|\\
        &\quad\leq \big(H(1+\lambda\log|\calA|)+\lambda L_{R}\big)\bbE_{\barpi_{t+1}^{*,\alpha},\barmu_{t+1}^{\calI}}\bigg[\sum_{m=h}^{H}\big\|\barpi_{t+1,m}^{*,\alpha}(\cdot\,|\,s_{m})-\barpi_{t,m}^{*,\alpha}(\cdot\,|\,s_{m})\big\|_{1}\bigg],
    \end{align*}
    where $L_{R}=\log(1+|\calA|\exp\big(H(1+\lambda\log|\calA|)/\lambda)$, the inequality results from the performance difference lemma, Lemma~\ref{lem:pdl}, proposition~\ref{prop:optpolicybound} and Lemma~\ref{lem:rlip}. Thus, we have that
    \begin{align}
        \text{(V)}\leq H\big(H(1+\lambda\log|\calA|)+\lambda L_{R}\big)\bbE_{\barpi_{t+1}^{*,\alpha},\barmu_{t+1}^{\calI}}\bigg[\sum_{m=1}^{H}\big\|\barpi_{t+1,m}^{*,\alpha}(\cdot\,|\,s_{m})-\barpi_{t,m}^{*,\alpha}(\cdot\,|\,s_{m})\big\|_{1}\bigg].\label{ieq:62}
    \end{align}
    
    Term (VI) is the error that measures the difference between the distribution of states induced by optimal policies of $\barmu_{t+1}^{\calI}$ and $\barmu_{t}^{\calI}$, which is defined as
    \begin{align*}
        \text{(VI)}&=\bbE_{\barpi_{t+1}^{*,\alpha},\barmu_{t+1}^{\calI}}-\bbE_{\barpi_{t}^{*,\alpha},\barmu_{t+1}^{\calI}}\bigg[\sum_{h=1}^{H}V_{h}^{\lambda,\alpha}(s_{h},\barpi_{t}^{*,\alpha},\barmu_{t+1}^{\calI},W^{*})-V_{h}^{\lambda,\alpha}(s_{h},\pi_{t+1}^{\alpha},\barmu_{t+1}^{\calI},W^{*})\bigg]\\
        &\qquad+\frac{1}{\eta\theta^{*}}\Big(\bbE_{\barpi_{t+1}^{*,\alpha},\barmu_{t+1}^{\calI}}-\bbE_{\barpi_{t}^{*,\alpha},\barmu_{t+1}^{\calI}}\Big)\bigg[\sum_{h=1}^{H}\kl\big(\barpi_{t,h}^{*,\alpha}(\cdot\,|\,s_{h})\|\pi_{t+1,h}^{\alpha}(\cdot\,|\,s_{h})\big)\bigg].
    \end{align*}
    
    \begin{proposition}\label{prop:difpsamedd}
        Given any two policies $\pi^{\calI},\tilde{\pi}^{\calI}$ and distribution flow $\mu^{\calI}$, we define $\mu^{+,\calI}=\Gamma_{3}(\pi^{\calI},\mu^{\calI},W)$ and $\tilde{\mu}^{+,\calI}=\Gamma_{3}(\tilde{\pi}^{\calI},\mu^{\calI},W)$ for any graphons $W=\{W_{h}\}_{h=1}^{H}$. Then we have 
        \begin{align*}
            \big\|\mu_{h}^{+,\alpha}-\tilde{\mu}_{h}^{+,\alpha}\big\|_{1}\leq \sum_{m=1}^{h-1}\bbE_{\mu_{m}^{+,\alpha}}\Big[\big\|\pi_{m}^{\alpha}(\cdot\,|\,s)-\tilde{\pi}_{m}^{\alpha}(\cdot\,|\,s)\big\|_{1}\Big].
        \end{align*}
        for all $\alpha\in\calI$ and $h\in[H]$.
    \end{proposition}
    \begin{proof}[Proof of Proposition~\ref{prop:difpsamedd}]
        See Appendix~\ref{app:difpsamedd}.
    \end{proof}
    
    Proposition~\ref{prop:difpsamedd} shows that
    \begin{align*}
        \text{(VI)}&\leq H\Big(H(1+\lambda\log|\calA|)+\max_{s\in\calS,h\in[H]}\kl\big(\barpi_{t,h}^{*,\alpha}(\cdot\,|\,s)\|\pi_{t+1,h}^{\alpha}(\cdot\,|\,s)\big)\Big)\\
        &\qquad\cdot\bbE_{\barpi_{t+1}^{*,\alpha},\barmu_{t+1}^{\calI}}\bigg[\sum_{m=1}^{H}\big\|\barpi_{t+1,m}^{*,\alpha}(\cdot\,|\,s_{m})-\barpi_{t,m}^{*,\alpha}(\cdot\,|\,s_{m})\big\|_{1}\bigg].
    \end{align*}
    Note that
    \begin{align*}
        \kl\big(\barpi_{t,h}^{*,\alpha}(\cdot\,|\,s)\|\pi_{t+1,h}^{\alpha}(\cdot\,|\,s)\big)\leq \log|\calA|+\log\frac{|\calA|}{\beta_{t+1}}.
    \end{align*}
    Thus, we have
    \begin{align}
        \text{(VI)}&\leq H\bigg(H\big(1+\lambda\log|\calA|\big)+\frac{1}{\eta\theta^{*}}\log\frac{|\calA|^{2}}{\beta_{t+1}}\bigg)\cdot\bbE_{\barpi_{t+1}^{*,\alpha},\barmu_{t+1}^{\calI}}\bigg[\sum_{m=1}^{H}\big\|\barpi_{t+1,m}^{*,\alpha}(\cdot\,|\,s_{m})-\barpi_{t,m}^{*,\alpha}(\cdot\,|\,s_{m})\big\|_{1}\bigg].\label{ieq:60}
    \end{align}
    
    Term (VII) is the error that measures the difference between the distribution of states induced by $\barmu_{t}^{\calI}$ on $\barmu_{t+1}^{\calI}$ and $\barmu_{t}^{\calI}$, which is defined as
    \begin{align*}
        \text{(VII)}&=\bbE_{\barpi_{t}^{*,\alpha},\barmu_{t+1}^{\calI}}-\bbE_{\barpi_{t}^{*,\alpha},\barmu_{t}^{\calI}}\bigg[\sum_{h=1}^{H}V_{h}^{\lambda,\alpha}(s_{h},\barpi_{t}^{*,\alpha},\barmu_{t+1}^{\calI},W^{*})-V_{h}^{\lambda,\alpha}(s_{h},\pi_{t+1}^{\alpha},\barmu_{t+1}^{\calI},W^{*})\bigg]\\
        &\qquad+\frac{1}{\eta\theta^{*}}\Big(\bbE_{\barpi_{t}^{*,\alpha},\barmu_{t+1}^{\calI}}-\bbE_{\barpi_{t}^{*,\alpha},\barmu_{t}^{\calI}}\Big)\bigg[\sum_{h=1}^{H}\kl\big(\barpi_{t,h}^{*,\alpha}(\cdot\,|\,s_{h})\|\pi_{t+1,h}^{\alpha}(\cdot\,|\,s_{h})\big)\bigg].
    \end{align*}
    
    \begin{proposition}\label{prop:difdsamepd}
        Given any policy $\pi^{\calI}$ and two distribution flows $\mu^{\calI}$ and $\tilde{\mu}^{\calI}$, we define $\mu^{+,\calI}=\Gamma_{3}(\pi^{\calI},\mu^{\calI},W^{*})$ and $\tilde{\mu}^{+,\calI}=\Gamma_{3}(\pi^{\calI},\tilde{\mu}^{\calI},W^{*})$. Under Assumption~\ref{assump:lipconti}, we have that
        \begin{align*}
            \big\|\mu_{h}^{+,\alpha}-\tilde{\mu}_{h}^{+,\alpha}\big\|_{1}\leq L_{P}\sum_{m=1}^{h-1}\int_{0}^{1}\|\mu_{m}^{\beta}-\tilde{\mu}_{m}^{\beta}\|_{1}\rmd \beta
        \end{align*}
        for all $\alpha\in\calI$ and $h\in[H]$.
    \end{proposition}
    \begin{proof}[Proof of Proposition~\ref{prop:difdsamepd}]
        See Appendix~\ref{app:difdsamepd}.
    \end{proof}
    Following the similar arguments in inequality~\eqref{ieq:60}, we have that
    \begin{align}
        \text{(VII)}\leq H\bigg(H\big(1+\lambda\log|\calA|\big)+\frac{1}{\eta\theta^{*}}\log\frac{|\calA|^{2}}{\beta_{t+1}}\bigg)L_{P}\cdot\sum_{m=1}^{H}\int_{0}^{1}\|\barmu_{t+1,m}^{\beta}-\barmu_{t,m}^{\beta}\|_{1}\rmd \beta.\label{ieq:63}
    \end{align}
    
     Term (VIII) is the error that measures the difference between the action-value function induced by difference distribution flows $\barmu_{t+1}^{\calI}$ and $\barmu_{t}^{\calI}$, which is defined as
     \begin{align*}
         \text{(VIII)}&=\bbE_{\barpi_{t}^{*,\alpha},\barmu_{t}^{\calI}}\bigg[\sum_{h=1}^{H}V_{h}^{\lambda,\alpha}(s_{h},\barpi_{t}^{*,\alpha},\barmu_{t+1}^{\calI},W^{*})-V_{h}^{\lambda,\alpha}(s_{h},\pi_{t+1}^{\alpha},\barmu_{t+1}^{\calI},W^{*})\bigg]\\
         &\qquad -\bbE_{\barpi_{t}^{*,\alpha},\barmu_{t}^{\calI}}\bigg[\sum_{h=1}^{H}V_{h}^{\lambda,\alpha}(s_{h},\barpi_{t}^{*,\alpha},\barmu_{t}^{\calI},W^{*})-V_{h}^{\lambda,\alpha}(s_{h},\pi_{t+1}^{\alpha},\barmu_{t}^{\calI},W^{*})\bigg]
     \end{align*}
    From Proposition~\ref{prop:difdsamepv}, we have that
    \begin{align}
        \text{(VIII)}&\leq 2H\big[L_{r}+H(1+\lambda\log|\calA|)L_{P}\big]\sum_{m=1}^{H}\int_{0}^{1}\|\mu_{m}^{\beta}-\tilde{\mu}_{m}^{\beta}\|_{1}\rmd\beta.\label{ieq:64}
    \end{align}
    
    Term (IX) is the error that measures the difference between the KL divergence related to the optimal policies of $\barmu_{t+1}^{\calI}$ and $\barmu_{t}^{\calI}$, which is defined as
    \begin{align*}
        \text{(IX)}=\frac{1}{\eta\theta^{*}}\bbE_{\barpi_{t+1}^{*,\alpha},\barmu_{t+1}^{\calI}}\bigg[\sum_{h=1}^{H}\kl\big(\barpi_{t+1,h}^{*,\alpha}(\cdot\,|\,s_{h})\|\pi_{t+1,h}^{\alpha}(\cdot\,|\,s_{h})\big)-\kl\big(\barpi_{t,h}^{*,\alpha}(\cdot\,|\,s_{h})\|\pi_{t+1,h}^{\alpha}(\cdot\,|\,s_{h})\big)\bigg].
    \end{align*}
    Lemma~\ref{lem:kllip} and Proposition~\ref{prop:optpolicybound} show that
    \begin{align}
        \text{(IX)}\leq \frac{2}{\eta\theta^{*}}\max\bigg\{\log\frac{|\calA|}{\beta_{t+1}},L_{R}\bigg\}\bbE_{\barpi_{t+1}^{*,\alpha},\barmu_{t+1}^{\calI}}\bigg[\sum_{m=1}^{H}\big\|\barpi_{t+1,m}^{*,\alpha}(\cdot\,|\,s_{m})-\barpi_{t,m}^{*,\alpha}(\cdot\,|\,s_{m})\big\|_{1}\bigg].\label{ieq:65}
    \end{align}
    Combining Eqn.~\eqref{ieq:61} and inequalities~\eqref{ieq:62}, \eqref{ieq:60}, \eqref{ieq:63}, \eqref{ieq:64}, \eqref{ieq:65}, we conclude the proof of this proposition.
\end{proof}
\subsubsection{Proof of Proposition~\ref{prop:optpolicybound}}\label{app:optpolicybound}
\begin{proof}[Proof of Proposition~\ref{prop:optpolicybound}]
    We denote the value function of the optimal policy $\pi^{*}$ as $V_{h}^{\lambda}(s,\pi^{*})$ for $h\in[H]$. From the definition of the optimal policy, we have that for any $s\in\calS$
    \begin{align*}
        \pi_{h}^{*}(\cdot\,|\,s)=\argmax_{p\in\Delta(\calA)}\langle r_{h}(s,\cdot),p\rangle -\lambda R(p)+\sum_{a\in\calA}\int_{\calS}p(a)P_{h}(s^{\prime}\,|\,s,a)V_{h+1}^{\lambda}(s^{\prime},\pi^{*})\rmd s^{\prime}.
    \end{align*}
    Then we have that
    \begin{align*}
        \pi_{h}^{*}(a\,|\,s)\propto \exp\bigg(\frac{1}{\lambda}\Big(r_{h}(s,a)+\int_{\calS}P_{h}(s^{\prime}\,|\,s,a)V_{h+1}^{\lambda}(s^{\prime},\pi^{*})\rmd s^{\prime}\Big)\bigg).
    \end{align*}
    The desired result follows from that $V_{h}^{\lambda}(s^{\prime},\pi^{*})\leq (H-h+1)(1+\lambda\log|\calA|)$. Thus, we conclude the proof of Proposition~\ref{prop:optpolicybound}.
\end{proof}

\subsubsection{Proof of Proposition~\ref{prop:difpsamedd}}\label{app:difpsamedd}
\begin{proof}[Proof of Proposition~\ref{prop:difpsamedd}]
    For any $h\in[H-1]$ and $\alpha\in\calI$, we have
    \begin{align*}
        &\big\|\mu_{h+1}^{+,\alpha}-\tilde{\mu}_{h+1}^{+,\alpha}\big\|_{1}\\
        &\quad\leq \sum_{a\in\calA}\int_{\calS}\int_{\calS}\big|\mu_{h}^{+,\alpha}(s)\pi_{h}^{\alpha}(a\,|\,s)-\tilde{\mu}_{h}^{+,\alpha}(s)\tilde{\pi}_{h}^{\alpha}(a\,|\,s) \big|P_{h}\big(s^{\prime}\,|\,s,a,z_{h}^{\alpha}(\mu_{h}^{\calI},W_{h})\big)\rmd s\rmd s^{\prime}\\
        &\quad\leq \sum_{a}\int_{S}\mu_{h}^{+,\alpha}(s)\big|\pi_{h}^{\alpha}(a\,|\,s)-\tilde{\pi}_{h}^{\alpha}(a\,|\,s)\big|\rmd s+\sum_{a}\int_{S}\big|\mu_{h}^{+,\alpha}(s)-\tilde{\mu}_{h}^{+,\alpha}(s)\big|\tilde{\pi}_{h}^{\alpha}(a\,|\,s)\rmd s\\
        &\quad =\bbE_{\mu_{h}^{+,\alpha}}\Big[\big\|\pi_{h}^{\alpha}(\cdot\,|\,s)-\tilde{\pi}_{h}^{\alpha}(\cdot\,|\,s)\big\|_{1}\Big]+\big\|\mu_{h}^{+,\alpha}-\tilde{\mu}_{h}^{+,\alpha}\big\|_{1},
    \end{align*}
    where the first inequality results from the definition of $\Gamma_{3}$ and the triangle inequality, and the second inequality results from the triangle inequality. Note that $\mu_{1}^{+,\alpha}=\tilde{\mu}_{1}^{+,\alpha}=\mu_{1}^{\alpha}$. Summing over $h$, we prove the desired result. This completes the proof of Proposition~\ref{prop:difpsamedd}.
\end{proof}
    
\subsubsection{Proof of Proposition~\ref{prop:difdsamepd}}\label{app:difdsamepd}
\begin{proof}[Proof of Proposition~\ref{prop:difdsamepd}]
    For any $h\in[H-1]$, we have that
    \begin{align*}
        &\big\|\mu_{h+1}^{+,\alpha}-\tilde{\mu}_{h+1}^{+,\alpha}\big\|_{1}\\
        &\quad\leq \int_{\calS}\bigg|\sum_{a\in\calA}\int_{\calS}\big(\mu_{h}^{+,\alpha}(s)-\tilde{\mu}_{h}^{+,\alpha}(s)\big)\pi_{h}^{\alpha}(a\,|\,s)P_{h}^{*}\big(s^{\prime}\,|\,s,a,z_{h}^{\alpha}(\mu_{h}^{\calI},W_{h}^{*})\big)\bigg|\rmd s^{\prime}\\
        &\quad\qquad+ \int_{\calS}\bigg|\sum_{a\in\calA}\int_{\calS}\tilde{\mu}_{h}^{+,\alpha}(s)\pi_{h}^{\alpha}(a\,|\,s)\Big(P_{h}^{*}\big(s^{\prime}\,|\,s,a,z_{h}^{\alpha}(\mu_{h}^{\calI},W_{h}^{*})\big)-P_{h}^{*}\big(s^{\prime}\,|\,s,a,z_{h}^{\alpha}(\tilde{\mu}_{h}^{\calI},W_{h}^{*})\big)\Big)\bigg|\rmd s^{\prime}\\
        &\quad \leq \|\mu_{h}^{+,\alpha}-\tilde{\mu}_{h}^{+,\alpha}\|_{1}+L_{P}\big\|z_{h}^{\alpha}(\mu_{h}^{\calI},W_{h}^{*})-z_{h}^{\alpha}(\tilde{\mu}_{h}^{\calI},W_{h}^{*})\big\|_{1},
    \end{align*}
    where the first inequality results from the definition of $\Gamma_{3}$ and triangle inequality, and the second inequality results from Assumption~\ref{assump:lipconti}. For the right-hand side term, we have that
    \begin{align*}
        \big\|z_{h}^{\alpha}(\mu_{h}^{\calI},W_{h}^{*})-z_{h}^{\alpha}(\tilde{\mu}_{h}^{\calI},W_{h}^{*})\big\|_{1}=\int_{\calS}\bigg|\int_{0}^{1}W_{h}^{*}(\alpha,\beta)\big(\mu_{h}^{\beta}(s)-\tilde{\mu}_{h}^{\beta}(s)\big)\rmd\beta\bigg|\rmd s\leq \int_{0}^{1}\|\mu_{h}^{\beta}-\tilde{\mu}_{h}^{\beta}\|_{1}\rmd \beta,
    \end{align*}
    where the inequality results from the triangle inequality and that $|W_{h}^{*}|\leq 1$. Summing over $h$, we prove the desired result. Thus, we conclude the proof of Proposition~\ref{prop:difdsamepd}.
\end{proof}

\subsubsection{Proof of Proposition~\ref{prop:difdsamepv}}\label{app:difdsamepv}
\begin{proof}[Proof of Proposition~\ref{prop:difdsamepv}]
    From the definitions of the value function and the action-value function, we have that for any $h\in[H]$
    \begin{align*}
        Q_{h}^{\lambda,\alpha}(s,a,\pi^{\alpha},\mu^{\calI},W^{*})
        &=r_{h}^{*}\big(s,a,z_{h}^{\alpha}(\mu^{\calI},W^{*})\big)+\int_{\calS}\Big[\sum_{a^{\prime}\in\calA}\pi_{h+1}^{\alpha}(a^{\prime}\,|\,s^{\prime})Q_{h+1}^{\lambda,\alpha}(s^{\prime},a^{\prime},\pi^{\alpha},\mu^{\calI},W^{*})\\
        &\qquad-\lambda R\big(\pi_{h+1}^{\alpha}(\cdot\, | \,s^{\prime})\big)\Big]P_{h}^{*}\big(s^{\prime}\,|\,s,a,z_{h}^{\alpha}(\mu^{\calI},W^{*})\big)\rmd s^{\prime}.
    \end{align*}
    Thus, we have that
    \begin{align*}
        &\big|Q_{h}^{\lambda,\alpha}(s,a,\pi^{\alpha},\mu^{\calI},W^{*})-Q_{h}^{\lambda,\alpha}(s,a,\pi^{\alpha},\tilde{\mu}^{\calI},W^{*})\big|\\
        &\quad\leq \big[L_{r}+H(1+\lambda\log|\calA|)L_{P}\big]\int_{0}^{1}\|\mu_{h}^{\beta}-\tilde{\mu}_{h}^{\beta}\|_{1}\rmd\beta\\
        &\quad\qquad+\sum_{a^{\prime}\in\calA}\int_{\calS}\big|Q_{h+1}^{\lambda,\alpha}(s^{\prime},a^{\prime},\pi^{\calI},\mu^{\calI},W^{*})-Q_{h+1}^{\lambda,\alpha}(s^{\prime},a^{\prime},\pi^{\calI},\tilde{\mu}^{\calI},W^{*})\big|\pi_{h+1}^{\alpha}(a^{\prime}\, | \,s^{\prime})P_{h}^{*}\big(s^{\prime}\,|\,s,a,z_{h}^{\alpha}(\mu^{\calI},W^{*})\big)\rmd s^{\prime},
    \end{align*}
    where the inequality results from the triangle inequality and Assumption~\ref{assump:lipconti}. By induction, it is easy to prove that
    \begin{align*}
        \big|Q_{h}^{\lambda,\alpha}(s,a,\pi^{\alpha},\mu^{\calI},W^{*})-Q_{h}^{\lambda,\alpha}(s,a,\pi^{\alpha},\tilde{\mu}^{\calI},W^{*})\big|\leq \big[L_{r}+H(1+\lambda\log|\calA|)L_{P}\big]\sum_{m=h}^{H}\int_{0}^{1}\|\mu_{m}^{\beta}-\tilde{\mu}_{m}^{\beta}\|_{1}\rmd\beta.
    \end{align*}
    From the relationship between the value function and action-value function, we have that
    \begin{align*}
        \big|V_{h}^{\lambda,\alpha}(s,\pi^{\alpha},\mu^{\calI},W^{*})-V_{h}^{\lambda,\alpha}(s,\pi^{\alpha},\tilde{\mu}^{\calI},W^{*})\big|\leq \big[L_{r}+H(1+\lambda\log|\calA|)L_{P}\big]\sum_{m=h}^{H}\int_{0}^{1}\|\mu_{m}^{\beta}-\tilde{\mu}_{m}^{\beta}\|_{1}\rmd\beta.
    \end{align*}
    Thus, we conclude the proof of Proposition~\ref{prop:difdsamepv}.
\end{proof}

\subsubsection{Proof of Proposition~\ref{prop:optpolicylip}}\label{app:optpolicylip}
\begin{proof}[Proof of Proposition~\ref{prop:optpolicylip}]
        We first prove the claim related to the value function. From the definition of the optimal policy, we have that
        \begin{align*}
            V_{h}^{\lambda,\alpha}(s,\pi^{*,\calI},\mu^{\calI},W^{*})&=\max_{p\in\Delta(\calA)}\big\langle r_{h}\big(s,\cdot,z_{h}^{\alpha}(\mu_{h}^{\calI},W_{h}^{*})\big),p\big\rangle-\lambda R(p)\\
            &\qquad+\sum_{a\in\calA}\int_{\calS}p(a)P_{h}\big(s^{\prime}\,|\,s,a,z_{h}^{\alpha}(\mu_{h}^{\calI},W_{h}^{*})\big)V_{h+1}^{\lambda,\alpha}(s,\pi^{*,\calI},\mu^{\calI},W^{*})\rmd s^{\prime}.
        \end{align*}
        Thus, we have that
        \begin{align*}
            &\big|V_{h}^{\lambda,\alpha}(s,\pi^{*,\calI},\mu^{\calI},W^{*})-V_{h}^{\lambda,\alpha}(s,\tilde{\pi}^{*,\calI},\tilde{\mu}^{\calI},W^{*})\big|\\
            &\quad\leq \big(H(1+\lambda\log|\calA|)L_{P}+L_{r}\big)\int_{0}^{1}\|\mu_{h}^{\beta}-\tilde{\mu}_{h}^{\beta}\|_{1}\rmd\beta+\max_{s\in\calS}\big|V_{h+1}^{\lambda,\alpha}(s,\pi^{*,\calI},\mu^{\calI},W^{*})-V_{h+1}^{\lambda,\alpha}(s,\tilde{\pi}^{*,\calI},\tilde{\mu}^{\calI},W^{*})\big|,
        \end{align*}
        where the inequality results from the fact that $|\max_{x}f(x)-\max_{x}g(x)|\leq \max_{x}|f(x)-g(x)|$ and Assumption~\ref{assump:lipconti}. By induction, it is easy to prove that
        \begin{align*}
            \max_{s\in\calS}\big|V_{h}^{\lambda,\alpha}(s,\pi^{*,\calI},\mu^{\calI},W^{*})-V_{h}^{\lambda,\alpha}(s,\tilde{\pi}^{*,\calI},\tilde{\mu}^{\calI},W^{*})\big|&\leq \big(H(1+\lambda\log|\calA|)L_{P}+L_{r}\big)\sum_{m=h}^{H}\int_{0}^{1}\|\mu_{m}^{\beta}-\tilde{\mu}_{m}^{\beta}\|_{1}\rmd\beta.
        \end{align*}
        
        Next, we prove the claim related to the optimal policies. From the definition of the optimal policies, we have that
        \begin{align*}
            \pi_{h}^{*,\alpha}(\cdot\,|\,s)&=\argmax_{p\in\Delta(\calA)}\big\langle r_{h}\big(s,\cdot,z_{h}^{\alpha}(\mu_{h}^{\calI},W_{h}^{*})\big),p\big\rangle-\lambda R(p)\\
            &\qquad+\sum_{a\in\calA}\int_{\calS}p(a)P_{h}\big(s^{\prime}\,|\,s,a,z_{h}^{\alpha}(\mu_{h}^{\calI},W_{h}^{*})\big)V_{h+1}^{\lambda,\alpha}(s,\pi^{*,\calI},\mu^{\calI},W^{*})\rmd s^{\prime}.
        \end{align*}
        We define that
        \begin{align*}
            y_{h}^{\alpha}(s,a)&=r_{h}\big(s,a,z_{h}^{\alpha}(\mu_{h}^{\calI},W_{h}^{*})\big)+\int_{\calS}P_{h}\big(s^{\prime}\,|\,s,a,z_{h}^{\alpha}(\mu_{h}^{\calI},W_{h}^{*})\big)V_{h+1}^{\lambda,\alpha}(s,\pi^{*,\calI},\mu^{\calI},W^{*})\rmd s^{\prime},\\
            \tily_{h}^{\alpha}(s,a)&=r_{h}\big(s,a,z_{h}^{\alpha}(\tilde{\mu}_{h}^{\calI},W_{h}^{*})\big)+\int_{\calS}P_{h}\big(s^{\prime}\,|\,s,a,z_{h}^{\alpha}(\tilde{\mu}_{h}^{\calI},W_{h}^{*})\big)V_{h+1}^{\lambda,\alpha}(s,\pi^{*,\calI},\tilde{\mu}^{\calI},W^{*})\rmd s^{\prime}.
        \end{align*}
        Then Lemma~\ref{lem:optapprox} shows that
        \begin{align*}
            \big\|\pi_{h}^{*,\alpha}(\cdot\,|\,s)-\tilde{\pi}_{h}^{*,\alpha}(\cdot\,|\,s)\big\|_{1}\leq \big\| y_{h}^{\alpha}(s,\cdot)-\tily_{h}^{\alpha}(s,\cdot)\big\|_{\infty}.
        \end{align*}
        From the triangle inequality and Assumption~\ref{assump:lipconti}, we have that
        \begin{align*}
            &\big\| y_{h}^{\alpha}(s,\cdot)-\tily_{h}^{\alpha}(s,\cdot)\big\|_{\infty}\\
            &\quad\leq \big(H(1+\lambda\log|\calA|)L_{P}+L_{r}\big)\int_{0}^{1}\|\mu_{h}^{\beta}-\tilde{\mu}_{h}^{\beta}\|_{1}\rmd\beta+\big(H(1+\lambda\log|\calA|)L_{P}+L_{r}\big)\sum_{m=h}^{H}\int_{0}^{1}\|\mu_{m}^{\beta}-\tilde{\mu}_{m}^{\beta}\|_{1}\rmd\beta,
        \end{align*}
        which proves the claim related to the optimal policies. Thus, we conclude the proof of Proposition~\ref{prop:optpolicylip}
    \end{proof}

\subsection{Propositions and Lemmas for Combination}
\subsubsection{Proof of Corollary~\ref{coro:knownfixest}}\label{app:knownfixest}
\begin{proof}[Proof of Corollary~\ref{coro:knownfixest}]
        Following the proof of Theorem~\ref{thm:fixest}, we decompose the risk difference as
        \begin{align*}
            &\calR_{\bar{\xi}}(\hatf_{h},\hatg_{h},\hatW_{h})-\calR_{\bar{\xi}}(f_{h}^{*},g_{h}^{*},W_{h}^{*})\nonumber\\
            &\quad=\text{Generalization Error of Risk}+\text{Empirical Risk Difference},
        \end{align*}
        where the generalization error of risk and the empirical risk difference are defined similarly as those in Theorem~\ref{thm:fixest}. From the procedure of Algorithm~\eqref{algo:estalgo}, we have
        \begin{align*}
            \text{Empirical Risk Difference}\leq 0.
        \end{align*}
        The generalization error of risk can be bounded by inequality~\eqref{ieq:32} in the proof of Theorem~\ref{thm:fixest}. Thus, we conclude the proof of Corollary~\ref{coro:knownfixest}.
    \end{proof}

\subsubsection{Proof of Proposition~\ref{prop:muerr}}\label{app:muerr}.
\begin{proof}[Proof of Proposition~\ref{prop:muerr}]
        For any $h\in[H-1]$, the definition of $\Gamma_{2}$ shows that
        \begin{align*}
            \mu_{h+1}^{\alpha}(s^{\prime})=\sum_{a\in\calA}\int_{\calS}P_{h}^{*}(s^{\prime}\,|\,s,a,z_{h}^{\alpha}(\mu_{h}^{\calI},W_{h}^{*}))\mu_{h}^{\alpha}(s)\pi_{h}^{\alpha}(a\,|\,s)\rmd s.
        \end{align*}
        Assumption~\ref{assump:noise} implies that we can bound the total variation between $\mu_{h+1}^{\alpha}$ and $\hmu_{h+1}^{\alpha}$ as 
        \begin{align}
            &\|\mu_{h+1}^{\alpha}-\hmu_{h+1}^{\alpha}\|_{1}\nonumber\\
            &\quad\leq L_{\varepsilon}\bbE_{\rho_{h}^{\alpha}}\Big[\big|\hatf_{h}\big(\omega_{h}^{\alpha}(\hatW_{h})\big)-f_{h}^{*}\big(\omega_{h}^{\alpha}(W_{h}^{*})\big)\big|\Big]+L_{\varepsilon}\bbE_{\rho_{h}^{\alpha}}\Big[\big|\hatf_{h}\big(\omega_{h}^{\alpha}(\hatW_{h})\big)-\hatf_{h}\big(\tilde{\omega}_{h}^{\alpha}(\hatW_{h})\big)\big|\Big] \nonumber\\
            &\quad\qquad+\|\mu_{h}^{\alpha}-\hmu_{h}^{\alpha}\|_{1},\label{ieq:74}
        \end{align}
        where 
        \begin{align*}
            \tilde{\omega}_{h}^{\alpha}(W)&=\int_{0}^{1}\int_{\calS}W(\alpha,\beta)k\big(\cdot,(s_{\tau,h}^{i},a_{\tau,h}^{i},s)\big)\hmu_{h}^{\beta}(s)\,\rmd s\,\rmd\beta.
        \end{align*}
        The first term in the right-hand side of inequality~\eqref{ieq:74} can be upper bounded as
        \begin{align*}
            L_{\varepsilon}\bbE_{\rho_{h}^{\alpha}}\Big[\big|\hatf_{h}\big(\omega_{h}^{\alpha}(\hatW_{h})-f_{h}^{*}\big(\omega_{h}^{\alpha}(W_{h}^{*})\big)\big|\Big]\leq e_{h}^{\pi,\alpha},
        \end{align*}
        where the inequality results from the H\"{o}lder inequality. The second term in the right-hand side of inequality~\eqref{ieq:74} can be upper bounded as
        \begin{align*}
            &L_{\varepsilon}\bbE_{\rho_{h}^{\alpha}}\Big[\big|\hatf_{h}\big(\omega_{h}^{\alpha}(\hatW_{h})\big)-\hatf_{h}\big(\tilde{\omega}_{h}^{\alpha}(\hatW_{h})\big)\big|\Big]\leq rL_{K}L_{\varepsilon}\|\omega_{h}^{\alpha}(\hatW_{h})-\tilde{\omega}_{h}^{\alpha}(\hatW_{h})\|_{\calH}\leq rL_{K}L_{\varepsilon}B_{k}\int_{0}^{1}\|\hmu_{h}^{\beta}-\mu_{h}^{\beta}\|_{1}\rmd\beta,
        \end{align*}
        where the first inequality results from Lemma~\ref{lem:rkhs}, and the second inequality results from the triangle inequality. Thus, we have that
        \begin{align*}
            \|\mu_{h+1}^{\alpha}-\hmu_{h+1}^{\alpha}\|_{1}\leq \|\mu_{h}^{\alpha}-\hmu_{h}^{\alpha}\|_{1}+rL_{K}L_{\varepsilon}B_{k}\int_{0}^{1}\|\hmu_{h}^{\beta}-\mu_{h}^{\beta}\|_{1}\rmd\beta+e_{h}^{\pi,\alpha}.
        \end{align*}
        By induction, it is easy to prove that
        \begin{align*}
            \|\mu_{h}^{\alpha}-\hmu_{h}^{\alpha}\|_{1}\leq \sum_{m=1}^{h-2}\sum_{k=0}^{h-m-2}(1+rL_{K}L_{\varepsilon}B_{k})^{k}\int_{0}^{1}e_{m}^{\pi,\beta}\rmd\beta+\sum_{m=1}^{h-1}e_{m}^{\pi,\alpha},
        \end{align*}
        which proves our desired results. Thus, we conclude the proof of Proposition~\ref{prop:muerr}.
    \end{proof}

\subsubsection{Proof of Proposition~\ref{prop:qerr}}\label{app:qerr}    
\begin{proof}[Proof of Proposition~\ref{prop:qerr}]
        From the definition of the action-value function, we have that 
        \begin{align}
            Q_{h}^{\lambda,\alpha}(s,a,\pi^{\alpha},\mu^{\calI},W^{*})&=r_{h}(s,a,z_{h}^{\alpha}(\mu^{\calI},W_{h}^{*}))+\bbE\big[V_{h+1}^{\lambda,\alpha}(s^{\prime},\pi^{\alpha},\mu^{\calI},W^{*})\,\big|\,s_{h}^{\alpha}=s,a_{h}^{\alpha}=a\big]\nonumber\\
            V_{h}^{\lambda,\alpha}(s^{\prime},\pi^{\alpha},\mu^{\calI},W^{*})&=\big\langle Q_{h}^{\lambda,\alpha}(s,\cdot,\pi^{\alpha},\mu^{\calI},W^{*}),\pi_{h}^{\alpha}(\cdot\,|\,s)\big\rangle+\lambda R(\pi_{h}^{\alpha}\big(\cdot\,|\,s)\big).\label{eq:22}
        \end{align}
        Thus for any $h\in[H]$, we have that
        \begin{align*}
            &\bbE_{\rho_{h}^{+,\alpha}}\Big[\Big|\hatQ_{h}^{\lambda,\alpha}(s,a,\pi^{\alpha},\mu^{\calI},\hatW)-Q_{h}^{\lambda,\alpha}(s,a,\pi^{\alpha},\mu^{\calI},W^{*})\Big|\Big]\\
            &\quad\leq \bbE_{\rho_{h}^{+,\alpha}}\Big[\Big|\hatr_{h}(s,a,z_{h}^{\alpha}(\mu^{\calI},\hatW_{h}))-r_{h}^{*}(s,a,z_{h}^{\alpha}(\mu^{\calI},W_{h}^{*}))\Big|\Big]\\
            &\quad\qquad+(H-h)(1+\lambda\log|\calA|)\bbE_{\rho_{h}^{+,\alpha}}\Big[\Big\|\hatP_{h}(\cdot\,|\,s,a,z_{h}^{\alpha}(\mu^{\calI},\hatW_{h}))-P_{h}^{*}(\cdot\,|\,s,a,z_{h}^{\alpha}(\mu^{\calI},W_{h}^{*}))\Big\|_{1}\Big]\\
            &\quad\qquad+C\cdot\bbE_{\rho_{h+1}^{\rmb,\alpha}}\Big[\Big|\hatQ_{h+1}^{\lambda,\alpha}(s,a,\pi^{\alpha},\mu^{\calI},\hatW)-Q_{h+1}^{\lambda,\alpha}(s,a,\pi^{\alpha},\mu^{\calI},W^{*})\Big|\Big],
        \end{align*}
        where the inequality results from the triangle inequality and Eqn.~\eqref{eq:22}. Since $\hatQ_{H+1}^{\lambda,\alpha}=Q_{H+1}^{\lambda,\alpha}=0$, we have that
        \begin{align*}
            &\bbE_{\rho_{h}^{\rmb,\alpha}}\Big[\Big|\hatQ_{h}^{\lambda,\alpha}(s,a,\pi^{\alpha},\mu^{\calI},\hatW)-Q_{h}^{\lambda,\alpha}(s,a,\pi^{\alpha},\mu^{\calI},W^{*})\Big|\Big]\\
            &\quad\leq \sum_{m=h}^{H}\bbE_{\rho_{m}^{\rmb,\alpha}}\Big[\Big|\hatg_{m}\big(\omega_{h}^{\alpha}(\hatW_{m})\big)-g_{m}^{*}\big(\omega_{h}^{\alpha}(W_{m}^{*})\big)\Big|\Big]\\
            &\quad\qquad+L_{\varepsilon}H(1+\lambda\log|\calA|)\sum_{m=h}^{H}\sqrt{\bbE_{\rho_{m}^{\rmb,\alpha}}\Big[\Big|\hatf_{m}\big(\omega_{h}^{\alpha}(\hatW_{m})\big)-f_{m}^{*}\big(\omega_{h}^{\alpha}(W_{m}^{*})\big)\Big|\Big]},
        \end{align*}
        where the inequality results from Assumption~\ref{assump:noise}. Our desired result follows from the H\"{o}lder inequality. Thus, we conclude the proof of Proposition~\ref{prop:qerr}.
    \end{proof}
    
\subsubsection{Proof of Corollary~\ref{coro:unknownfixest}}\label{app:unknownfixest}
\begin{proof}[Proof of Corollary~\ref{coro:unknownfixest}]
        We proof mainly takes two steps:
        \begin{itemize}
            \item Reformulate the algorithm~\eqref{eq:ufixdistest}.
            \item Decompose the risk difference and control each terms.
        \end{itemize}
        
        \textbf{Step 1: Reformulate the algorithm~\eqref{eq:ufixdistest}.}
        
        From the definition $\phi^{*}$, we have that for $\alpha\in((i-1)/N,i/N]$
        \begin{align*}
            \tilmu_{\tau,h}^{\alpha}=\mu_{\tau,h}^{\phi^{*}(i/N)+\alpha-i/N}=\mu_{\tau,h}^{\phi^{*}(\alpha)},
        \end{align*}
        where the first equality results from the definition of $\tilmu_{\tau}^{\calI}$, and the second equality results from that $\alpha\in((i-1)/N,i/N]$ and $\phi^{*}$ is a permutation of $\{((i-1)/N,i/N]\}_{i=1}^{N}$. Thus, we have that
        \begin{align*}
            \tilomega_{\tau,h}^{i}(W^{\phi})&=\int_{0}^{1}\int_{\calS}W\big(\phi(i/N),\phi(\beta)\big)k\big(\cdot,(s_{\tau,h}^{i},a_{\tau,h}^{i},s)\big)\mu_{\tau,h}^{\phi^{*}(\beta)}(s)\,\rmd s\,\rmd\beta\\
            &=\int_{0}^{1}\int_{\calS}W\big(\phi\circ\psi^{*}\circ\phi^{*}(i/N),\phi\circ\psi^{*}\circ\phi^{*}(\beta)\big)k\big(\cdot,(s_{\tau,h}^{i},a_{\tau,h}^{i},s)\big)\mu_{\tau,h}^{\phi^{*}(\beta)}(s)\,\rmd s\,\rmd\beta\\
            &=\int_{0}^{1}\int_{\calS}W\big(\phi\circ\psi^{*}(\xi_{i}),\phi\circ\psi^{*}(\gamma)\big)k\big(\cdot,(s_{\tau,h}^{i},a_{\tau,h}^{i},s)\big)\mu_{\tau,h}^{\gamma}(s)\,\rmd s\,\rmd\gamma\\
            &=\omega_{\tau,h}^{i}(W^{\phi\circ\psi^{*}}),
        \end{align*}
        where the second equality results from that $\phi^{*}$ is the inverse functio of $\psi^{*}$, and the third equality results from taking $\gamma=\phi^{*}(\beta)$ and $\phi^{*}(i/N)=\xi_{i}$. Thus, Algorithm~\eqref{eq:ufixdistest} can be equivalent formulated as
        \begin{align}
            &(\hatf_{h},\hatg_{h},\hatW_{h},\hphi_{h})\nonumber\\
            &\quad=\argmin_{f\in\bbB(r,\bar{\calH}),g\in\bbB(\tilr,\tilde{\calH}),W\in\tilde{\calW},\phi\in\calC_{[0,1]}^{N}} \frac{1}{NL}\sum_{\tau=1}^{L}\sum_{i=1}^{N}\Big(s_{\tau,h+1}^{i}-f\big(\omega_{\tau,h}^{i}(W^{\phi\circ\psi^{*}})\big)\Big)^{2}+\Big(r_{\tau,h}^{i}-g\big(\omega_{\tau,h}^{i}(W^{\phi\circ\psi^{*}})\big)\Big)^{2}.\label{eq:ufixdistest2}
        \end{align}
        
        \textbf{Step 2: Decompose the risk difference and control each terms.}
        
        \begin{align*}
            \calR_{\bar{\xi}}(\hatf_{h},\hatg_{h},\hatW_{h}^{\hphi_{h}\circ\psi^{*}})-\calR_{\bar{\xi}}(f_{h}^{*},g_{h}^{*},W_{h}^{*})=\text{Generalization Error of Risk}+\text{Empirical Risk Difference},
        \end{align*}
        where the generalization error of risk and the empirical risk difference are defined as
        \begin{align*}
            &\text{Generalization Error of Risk}\\
            &\quad=\frac{1}{NL}\sum_{\tau=1}^{L}\sum_{i=1}^{N}\bbE_{\rho_{\tau,h}^{i}}\bigg[\Big(s_{\tau,h+1}^{i}-\hatf_{h}\big(\omega_{\tau,h}^{i}(\hatW_{h}^{\hphi_{h}\circ\psi^{*}})\big)\Big)^{2}-\Big(s_{\tau,h+1}^{i}-f_{h}^{*}\big(\omega_{\tau,h}^{i}(W_{h}^{*})\big)\Big)^{2}\bigg]\nonumber\\
            &\quad\qquad -2\frac{1}{NL}\sum_{\tau=1}^{L}\sum_{i=1}^{N}\Big(s_{\tau,h+1}^{i}-\hatf_{h}\big(\omega_{\tau,h}^{i}(\hatW_{h}^{\hphi_{h}\circ\psi^{*}})\big)\Big)^{2}-\Big(s_{\tau,h+1}^{i}-f_{h}^{*}\big(\omega_{\tau,h}^{i}(W_{h}^{*})\big)\Big)^{2}\nonumber\\
            &\quad\qquad+\frac{1}{NL}\sum_{\tau=1}^{L}\sum_{i=1}^{N}\bbE_{\rho_{\tau,h}^{i}}\bigg[\Big(r_{\tau,h}^{i}-\hatg_{h}\big(\omega_{\tau,h}^{i}(\hatW_{h}^{\hphi_{h}\circ\psi^{*}})\big)\Big)^{2}-\Big(r_{\tau,h}^{i}-g_{h}^{*}\big(\omega_{\tau,h}^{i}(W_{h}^{*})\big)\Big)^{2}\bigg]\nonumber\\
            &\quad\qquad -2\frac{1}{NL}\sum_{\tau=1}^{L}\sum_{i=1}^{N}\Big(r_{\tau,h}^{i}-\hatg_{h}\big(\omega_{\tau,h}^{i}(\hatW_{h}^{\hphi_{h}\circ\psi^{*}})\big)\Big)^{2}-\Big(r_{\tau,h}^{i}-g_{h}^{*}\big(\omega_{\tau,h}^{i}(W_{h}^{*})\big)\Big)^{2},\nonumber\\
            &\text{Empirical Risk Difference}\\
            &\quad=2\frac{1}{NL}\sum_{\tau=1}^{L}\sum_{i=1}^{N}\Big(s_{\tau,h+1}^{i}-\hatf_{h}\big(\omega_{\tau,h}^{i}(\hatW_{h}^{\hphi_{h}\circ\psi^{*}})\big)\Big)^{2}-\Big(s_{\tau,h+1}^{i}-f_{h}^{*}\big(\omega_{\tau,h}^{i}(W_{h}^{*})\big)\Big)^{2}\\
            &\quad\qquad+2\frac{1}{NL}\sum_{\tau=1}^{L}\sum_{i=1}^{N}\Big(r_{\tau,h}^{i}-\hatg_{h}\big(\omega_{\tau,h}^{i}(\hatW_{h}^{\hphi_{h}\circ\psi^{*}})\big)\Big)^{2}-\Big(r_{\tau,h}^{i}-g_{h}^{*}\big(\omega_{\tau,h}^{i}(W_{h}^{*})\big)\Big)^{2}.
        \end{align*}
        From the procedure of Algorithm~\eqref{eq:ufixdistest2}, we have
        \begin{align*}
            \text{Empirical Risk Difference}\leq 0.
        \end{align*}
        The generalization error of risk can be controlled exactly as the proof of Theorem~\ref{thm:ufixest}. Thus, we conclude the proof of Corollary~\ref{coro:unknownfixest}.
    \end{proof}

\subsection{Technical Lemmas}
\begin{lemma}\label{lem:monoentro}
    For a finite alphabet $\calX$ and any distribution $p$ supported on it, we define $p_{\beta}=(1-\beta)p+\beta\unif(\calX)$. Then the function $f(\beta)=R(p_{\beta})$ is a decreasing function on $\beta\in[0,1]$.
\end{lemma}
\begin{proof}[Proof of Lemma~\ref{lem:monoentro}]
    From the calculus we can show that
    \begin{align*}
        f^{\prime\prime}(\beta)\geq 0 \text{ for }\beta\in[0,1].
    \end{align*}
    Since $f^{\prime}(1)=0$, we have that $f^{\prime}(\beta)\leq 0$ for $\beta\in[0,1]$. Thus, we conclude the proof of Lemma~\ref{lem:monoentro}.
\end{proof}
\begin{lemma}[Theorem 3.5 in~\citet{pinelis1994optimum}]\label{lem:hilconcen}
    Let $X_{1},\cdots,X_{n}$ be independent random variables that take values in a Hilbert space. If $\|X_{i}\|\leq M$ and $\bbE[X_{i}]=0$ for all $i\in[n]$. Then $\bbP(\|X_{1}+\cdots+X_{n}\|\geq t)\leq 2\exp(-t^{2}/(2nM))$.
\end{lemma}
\begin{lemma}\label{lem:rkhs}
    In a \ac{rkhs} $\calH$ with kernel $k:\calX\times\calX\rightarrow\bbR$ that satisfies: (i) $k(x,x)\leq B_{k}^{2}$ for all $x\in\calX$. (ii) $\|k(\cdot,x)-k(\cdot,x^{\prime})\|_{\calH}\leq L_{k}\|x-x^{\prime}\|_{\calX}$ for all $x,x^{\prime}\in\calX$. We have that for any $f\in\bbB(r,\calH)$: (i) $|f(x)|\leq rB_{k}$ for all $x\in\calX$. (ii) $|f(x)-f(x^{\prime})|\leq rL_{k}\|x-x^{\prime}\|_{\calX}$ for all $x,x^{\prime}\in\calX$.
\end{lemma}
\begin{proof}[Proof of Lemma~\ref{lem:rkhs}]
    For the first claim, we have that
    \begin{align*}
        \big|f(x)\big|=\Big|\big\langle f,k(\cdot,x)\big\rangle\Big|\leq \|f\|_{\calH}\cdot\big\|k(\cdot,x)\big\|_{\calH}\leq rB_{k}.
    \end{align*}
    For the second claim, we have that
    \begin{align*}
        \big|f(x)-f(x^{\prime})\big|=\Big|\big\langle f,k(\cdot,x)-k(\cdot,x^{\prime})\big\rangle\Big|\leq \|f\|_{\calH}\cdot\big\|k(\cdot,x)-k(\cdot,x^{\prime})\big\|_{\calH}\leq rL_{k}\|x-x^{\prime}\|_{\calX}.
    \end{align*}
    Thus, we conclude the proof of Lemma~\ref{lem:rkhs}.
\end{proof}

\begin{lemma}\label{lem:optapprox}
    Let $\calX$ be a nonempty compact convex set and $f:\calX\rightarrow\bbR$ be a differentiable $k$-strongly convex function, i.e., $f(x)\geq f(y)+\langle \nabla f(y),x-y \rangle +\frac{k}{2} \|x-y\|^{2}$ for all $x,y\in\calX$. For any two elements $y_{1},y_{2}$, we define
    \begin{align*}
        x_{i}=\argmax_{x\in\calX}\langle x,y_{i}\rangle-f(x) \text{ for }i=1,2.
    \end{align*}
    Then $\|x_{1}-x_{2}\|\leq \|y_{1}-y_{2}\|_{*}/k$, where $\|\cdot\|_{*}$ is the dual norm of $\|\cdot\|$.
\end{lemma}
\begin{proof}[Proof of Lemma~\ref{lem:optapprox}]
    Define $f_{i}(x)=\langle x,y_{i}\rangle-f(x)$ for $i=1,2$. Then \citet[Lemma 2.8]{shalev2012online} shows that
    \begin{align*}
        \frac{k}{2}\|x_{1}-x_{2}\|^{2}\leq f_{1}(x_{1})-f_{1}(x_{2})\quad  \text{ and } \quad \frac{k}{2}\|x_{1}-x_{2}\|^{2}\leq f_{2}(x_{2})-f_{2}(x_{1}).
    \end{align*}
    Summing these two inequalities, we have
    \begin{align*}
        k\|x_{1}-x_{2}\|^{2}\leq \langle x_{1}-x_{2},y_{1}-y_{2}\rangle\leq \|x_{1}-x_{2}\|\cdot \|y_{1}-y_{2}\|_{*},
    \end{align*}
    where the second inequality results from the definition of the dual norm. Thus, we conclude the proof of Lemma~\ref{lem:optapprox}.
\end{proof}

\begin{lemma}[Lemma 3.3 in~\cite{cai2020provably}]\label{lem:mdupdate}
    For any distribution $p,p^{*}\in\Delta(\calA)$ and any function $g:\calA\rightarrow [0,H]$, it holds for $q\in\Delta(\calA)$ with $q(\cdot)\propto p(\cdot)\exp\big(\alpha g(\cdot)\big)$ that
    \begin{align*}
        \langle g(\cdot),p^{*}(\cdot)-p(\cdot)\rangle \leq \alpha H^{2}/2+\alpha^{-1}\big[\kl(p^{*}\|p)-\kl(p^{*}\|q)\big].
    \end{align*}
\end{lemma}

\begin{lemma}\label{lem:policyave}
    For any two distributions $p^{*},p\in\Delta(\calA)$ and $\hatp=(1-\beta)p+\beta\unif(\calA)$ with $\beta\in(0,1)$. Then
    \begin{align*}
        \kl(p^{*}\|\hatp)&\leq\log\frac{|\calA|}{\beta}\\
        \kl(p^{*}\|\hatp)-\kl(p^{*}\|p)&\leq \beta/(1-\beta).
    \end{align*}
\end{lemma}
\begin{proof}[Proof of Lemma~\ref{lem:policyave}]
    \begin{align*}
        \kl(p^{*}\|\hatp)\leq \Big\langle p^{*},\log \frac{p^{*}}{(1-\beta)p+\beta/|\calA|}\Big\rangle\leq \Big\langle p^{*},\log \frac{1}{\beta/|\calA|}\Big\rangle=\log\frac{|\calA|}{\beta}.
    \end{align*}
    Thus, we prove the first inequality. For the second inequality, we have
    \begin{align*}
        \kl(p^{*}\|\hatp)-\kl(p^{*}\|p)=\Big\langle p^{*},\log \frac{p}{(1-\beta)p+\beta/|\calA|}\Big\rangle\leq \Big\langle p^{*},\log \frac{p}{(1-\beta)p}\Big\rangle \leq \Big\langle p^{*},\frac{\beta}{1-\beta}\Big\rangle=\frac{\beta}{1-\beta},
    \end{align*}
    where the second inequality results from $\log(x)\leq x-1$ for $x>0$. Thus, we conclude the proof of Lemma~\ref{lem:policyave}.
    
\end{proof}

\begin{lemma}[Performance Difference Lemma]\label{lem:pdl}
    Given a policy $\pi^{\calI}$ and the corresponding mean-field flow $\mu^{\calI}$, for any agent $\alpha\in\calI$ and any policy $\tilde{\pi}^{\alpha}$, we have
    \begin{align*}
        &V_{1}^{\lambda,\alpha}(s,\tilde{\pi}^{\alpha},\mu^{\calI},W)-V_{1}^{\lambda,\alpha}(s,\pi^{\alpha},\mu^{\calI},W)+\lambda \bbE_{\tilde{\pi}^{\alpha},\mu^{\calI}}\bigg[\sum_{h=1}^{H}\kl\big(\tilde{\pi}_{h}^{\alpha}(\cdot\,|\,s_{h}^{\alpha})\|\pi_{h}^{\alpha}(\cdot\,|\,s_{h}^{\alpha})\big)\,|\,s_{1}^{\alpha}=s\bigg]\\
        &\quad= \bbE_{\tilde{\pi}^{\alpha},\mu^{\calI}}\bigg[\sum_{h=1}^{H}\big\langle Q_{h}^{\lambda,\alpha}(s_{h}^{\alpha},\cdot,\pi^{\alpha},\mu^{\calI},W)-\lambda\log\pi_{h}^{\alpha}(\cdot\,|\,s_{h}^{\alpha}),\tilde{\pi}_{h}^{\alpha}(\cdot\,|\,s_{h}^{\alpha})-\pi_{h}^{\alpha}(\cdot\,|\,s_{h}^{\alpha})\big\rangle\,|\,s_{1}^{\alpha}=s\bigg],
    \end{align*}
    where the expectation $\bbE_{\tilde{\pi}^{\alpha},\mu^{\calI}}$ is taken with respect to the randomness in implementing policy $\tilde{\pi}^{\alpha}$ for agent $\alpha$ under the \ac{mdp} induced by $\mu^{\calI}$.
\end{lemma}
\begin{proof}[Proof of Lemma~\ref{lem:pdl}]
    From the definition of $V_{1}^{\lambda,\alpha}(s,\tilde{\pi}^{\alpha},\mu^{\calI},W)$, we have
    \begin{align}
        &V_{1}^{\lambda,\alpha}(s,\tilde{\pi}^{\alpha},\mu^{\calI},W)\nonumber\\
        &\quad=\bbE_{\tilde{\pi}^{\alpha},\mu^{\calI}}\bigg[\sum_{h=1}^{H}r_{h}(s_{h}^{\alpha},a_{h}^{\alpha},z_{h}^{\alpha})-\lambda\log\tilde{\pi}_{h}^{\alpha}(a_{h}^{\alpha}\,|\,s_{h}^{\alpha})+V_{h}^{\lambda,\alpha}(s_{h}^{\alpha},\pi^{\alpha},\mu^{\calI},W)-V_{h}^{\lambda,\alpha}(s_{h}^{\alpha},\pi^{\alpha},\mu^{\calI},W)\,\bigg|\,s_{1}^{\alpha}=s\bigg]\nonumber\\
        &\quad=\bbE_{\tilde{\pi}^{\alpha},\mu^{\calI}}\bigg[\sum_{h=1}^{H}r_{h}(s_{h}^{\alpha},a_{h}^{\alpha},z_{h}^{\alpha})-\lambda\log\tilde{\pi}_{h}^{\alpha}(a_{h}^{\alpha}\,|\,s_{h}^{\alpha})+V_{h+1}^{\lambda,\alpha}(s_{h+1}^{\alpha},\pi^{\alpha},\mu^{\calI},W)\nonumber\\
        &\quad\qquad-V_{h}^{\lambda,\alpha}(s_{h}^{\alpha},\pi^{\alpha},\mu^{\calI},W)\,\bigg|\,s_{1}^{\alpha}=s\bigg]+V_{1}^{\lambda,\alpha}(s,\pi^{\alpha},\mu^{\calI},W),\label{eq:8}
    \end{align}
    where the second equality results from the rearrangement from the terms. We then focus on a part of the right-hand side of Eqn.~\eqref{eq:8}.
    \begin{align}
        &\bbE_{\tilde{\pi}^{\alpha},\mu^{\calI}}\big[r_{h}(s_{h}^{\alpha},a_{h}^{\alpha},z_{h}^{\alpha})-\lambda\log\tilde{\pi}_{h}^{\alpha}(a_{h}^{\alpha}\,|\,s_{h}^{\alpha})+V_{h+1}^{\lambda,\alpha}(s_{h+1}^{\alpha},\pi^{\alpha},\mu^{\calI},W)\,|\,s_{1}^{\alpha}=s\big]\nonumber\\
        &\quad=\bbE_{\tilde{\pi}^{\alpha},\mu^{\calI}}\big[r_{h}(s_{h}^{\alpha},a_{h}^{\alpha},z_{h}^{\alpha})+V_{h+1}^{\lambda,\alpha}(s_{h+1}^{\alpha},\pi^{\alpha},\mu^{\calI},W)\,|\,s_{1}^{\alpha}=s\big]-\lambda\bbE_{\tilde{\pi}^{\alpha},\mu^{\calI}}\Big[ R\big(\tilde{\pi}_{h}^{\alpha}(\cdot\,|\,s_{h}^{\alpha})\big)\,|\,s_{1}^{\alpha}=s\Big]\nonumber\\
        &\quad=\bbE_{\tilde{\pi}^{\alpha},\mu^{\calI}}\Big[\big\langle Q_{h}^{\lambda,\alpha}(s_{h}^{\alpha},\cdot,\pi^{\alpha},\mu^{\calI},W),\tilde{\pi}_{h}^{\alpha}(\cdot\,|\,s_{h}^{\alpha})\big\rangle\,|\,s_{1}^{\alpha}=s\Big]-\lambda\bbE_{\tilde{\pi}^{\alpha},\mu^{\calI}}\Big[ R\big(\tilde{\pi}_{h}^{\alpha}(\cdot\,|\,s_{h}^{\alpha})\big)\,|\,s_{1}^{\alpha}=s\Big],\label{eq:9}
    \end{align}
    where $R(\cdot)$ is the negative entropy function, the inner product $\langle\cdot,\cdot\rangle$ is taken with respect to the action space $\calA$, and the second equality results from the definition of $Q_{h}^{\lambda,\alpha}$ and $V_{h+1}^{\lambda,\alpha}$. Substituting Eqn.~\eqref{eq:9} into Eqn.~\eqref{eq:8} and noting the fact that $V_{h}^{\lambda,\alpha}(s_{h}^{\alpha},\pi^{\alpha},\mu^{\calI},W)=\langle Q_{h}^{\lambda,\alpha}(s_{h}^{\alpha},\cdot,\pi^{\alpha},\mu^{\calI},W),\pi_{h}^{\alpha}(\cdot\,|\,s_{h}^{\alpha})\rangle-R(\pi_{h}^{\alpha}(\cdot\,|\,s_{h}^{\alpha}))$, we derive that
    \begin{align*}
        &V_{1}^{\lambda,\alpha}(s,\tilde{\pi}^{\alpha},\mu^{\calI},W)-V_{1}^{\lambda,\alpha}(s,\pi^{\alpha},\mu^{\calI},W)\nonumber\\
        &\quad =\bbE_{\tilde{\pi}^{\alpha},\mu^{\calI}}\bigg[\sum_{h=1}^{H}\big\langle Q_{h}^{\lambda,\alpha}(s_{h}^{\alpha},\cdot,\pi^{\alpha},\mu^{\calI},W),\tilde{\pi}_{h}^{\alpha}(\cdot\,|\,s_{h}^{\alpha})-\pi_{h}^{\alpha}(\cdot\,|\,s_{h}^{\alpha})\big\rangle\,|\,s_{1}^{\alpha}=s\bigg]\nonumber\\
        &\quad\qquad-\lambda\bbE_{\tilde{\pi}^{\alpha},\mu^{\calI}}\bigg[\sum_{h=1}^{H}R\big(\tilde{\pi}_{h}^{\alpha}(\cdot\,|\,s_{h}^{\alpha})\big)-R\big(\pi_{h}^{\alpha}(\cdot\,|\,s_{h}^{\alpha})\big)\,|\,s_{1}^{\alpha}=s\bigg]\\
        &\quad =\bbE_{\tilde{\pi}^{\alpha},\mu^{\calI}}\bigg[\sum_{h=1}^{H}\big\langle Q_{h}^{\lambda,\alpha}(s_{h}^{\alpha},\cdot,\pi^{\alpha},\mu^{\calI},W),\tilde{\pi}_{h}^{\alpha}(\cdot\,|\,s_{h}^{\alpha})-\pi_{h}^{\alpha}(\cdot\,|\,s_{h}^{\alpha})\big\rangle\,|\,s_{1}^{\alpha}=s\bigg]\nonumber\\
        &\quad\qquad -\lambda\bbE_{\tilde{\pi}^{\alpha},\mu^{\calI}}\bigg[\sum_{h=1}^{H}\kl\big(\tilde{\pi}_{h}^{\alpha}(\cdot\,|\,s_{h}^{\alpha})\|\pi_{h}^{\alpha}(\cdot\,|\,s_{h}^{\alpha})\big)+\big\langle \log\pi_{h}^{\alpha}(\cdot\,|\,s_{h}^{\alpha}),\tilde{\pi}_{h}^{\alpha}(\cdot\,|\,s_{h}^{\alpha})-\pi_{h}^{\alpha}(\cdot\,|\,s_{h}^{\alpha}) \big\rangle\,|\,s_{1}^{\alpha}=s\bigg],
    \end{align*}
    where the last equality results from the definition of the negative entropy $R(\cdot)$. This concludes the proof of Lemma~\ref{lem:pdl}.
    
\end{proof}

\begin{lemma}\label{lem:rlip}
    For a finite alphabet $\calX$, define $R$ as the negative entropy function. For two distributions $p,q$ supported on $\calX$, we have that
    \begin{align*}
        |R(p)-R(q)|\leq \max\Big\{\big\|\log(p)\big\|_{\infty},\big\|\log(q)\big\|_{\infty}\Big\}\|p-q\|_{1}.
    \end{align*}
\end{lemma}
\begin{proof}[Proof of Lemma~\ref{lem:rlip}]
    Then we have that
    \begin{align*}
        |R(p)-R(q)|\leq\int_{0}^{1}\Big|\Big\langle\nabla R\big(q+t(p-q)\big),p-q\Big\rangle\Big|\rmd t\leq \|p-q\|_{1}\int_{0}^{1}\Big\|\log\big(q+t(p-q)\big)\Big\|_{\infty}\rmd t,
    \end{align*}
    where the first inequality results from the definition of integral and the triangle inequality, and the second inequality results from H\"{o}lder's inequality. The desired result follows from the fact that for $t\in[0,1]$
    \begin{align*}
        \Big\|\log\big(q+t(p-q)\big)\Big\|_{\infty}\leq \max\Big\{\big\|\log(p)\big\|_{\infty},\big\|\log(q)\big\|_{\infty}\Big\}.
    \end{align*}
    Thus, we conclude the proof of Lemma~\ref{lem:rlip}.
\end{proof}

\begin{lemma}[Lemma 3 in~\citet{xie2021learning}]\label{lem:kllip}
    Let $p,q,u\in\Delta(\calX)$ be distributions supported on a finite set $\calX$. If $p(x)\geq \alpha_{1}$, $q(x)\geq\alpha_{1}$, and $u(x)\geq \alpha_{2}$ for all $x\in\calX$. Then 
    \begin{align*}
        \kl(p\|u)-\kl(q\|u)\leq \bigg(1+\log\frac{1}{\min\{\alpha_{1},\alpha_{2}\}}\bigg)\|p-q\|_{1}
    \end{align*}
\end{lemma}

\begin{lemma}[Lemma 39 in~\citet{wei2021last}]\label{lem:seqconv}
    Let $\{g_{t}\}_{t\geq 0}$ and $\{h_{t}\}_{t\geq 0}$ be non-negative sequences that satisft $g_{t}\leq (1-c)g_{t-1}+h_{t}$ for some $0<c<1$ for all $t\geq 1$. Then
    \begin{align*}
        g_{t}\leq g_{0}(1-c)^{t}+\frac{\max_{\tau\in[1,t/2]}h_{\tau}}{c}(1-c)^{t/2}+\frac{\max_{\tau\in[t/2,t]}h_{\tau}}{c}.
    \end{align*}
\end{lemma}

\end{document}